\documentclass[12pt,a4paper]{article}
\usepackage[a4paper,top=3cm,bottom=3cm,left=2cm,right=3cm,bindingoffset=5mm]{geometry}

\pdfoutput=1 

\usepackage[bottom]{footmisc}

\usepackage{jheppub}

\usepackage{xcolor}
\usepackage[makeroom]{cancel}
\usepackage{tcolorbox}
\tcbuselibrary{theorems}

\usepackage[utf8]{inputenc}

\usepackage{amsmath}
\usepackage{amsthm}
\usepackage{amssymb}
\usepackage{physics}
\usepackage{comment}
\usepackage{graphicx}
\usepackage{float}
\usepackage{hyperref}
\usepackage{amsfonts}
\usepackage{mathrsfs}
\usepackage{xcolor}

\usepackage{tcolorbox}

\usepackage{simpler-wick}
\usepackage{mathtools}
\usepackage{makecell}

\hypersetup{
colorlinks=true,         
linkcolor=blue,          
citecolor=red,        
urlcolor=blue            
}

\newcommand{\Square}{\Box}

\newcommand{\bw}{{\bar{w}}}

\newcommand{\bell}{\overline{\ell}}

\newcommand{\bz}{{\bar{z}}}

\newcommand{\mT}{\mathcal{T}}
\newcommand{\bt}{\mathbf{t}}
\newcommand{\bT}{\mathbf{T}}
\newcommand{\bu}{{\bar{u}}}
\newcommand{\mR}{\mathcal{R}}

\newcommand{\Bell}{\bar{\ell}}
\newcommand{\vep}{\varepsilon}

\newcommand{\by}{\bar{y}}
\newcommand{\bj}{{\bar{\jmath}}}

\newcommand{\kappat}{\tilde{\kappa}}
\newcommand{\ph}{\phantom{-}}
\newcommand{\tphi}{\tilde{\phi}}
\newcommand{\secret}{{\color{white}{.}}}

\newcommand{\mycdots}{.\:\!.\:\!.}

\newcommand{\mD}{\mathcal{D}}

\DeclareFontFamily{U}{jkpmia}{}
\DeclareFontShape{U}{jkpmia}{m}{it}{<->s*jkpmia}{}
\DeclareFontShape{U}{jkpmia}{bx}{it}{<->s*jkpbmia}{}
\DeclareMathAlphabet{\mathfrakalt}{U}{jkpmia}{m}{it}
\SetMathAlphabet{\mathfrakalt}{bold}{U}{jkpmia}{bx}{it}

\newcommand{\myw}{\mathfrakalt{w}}

\DeclareFontFamily{U}{MnSymbolC}{}
\DeclareSymbolFont{MnSyC}{U}{MnSymbolC}{m}{n}
\DeclareFontShape{U}{MnSymbolC}{m}{n}{
    <-6>  MnSymbolC5
   <6-7>  MnSymbolC6
   <7-8>  MnSymbolC7
   <8-9>  MnSymbolC8
   <9-10> MnSymbolC9
  <10-12> MnSymbolC10
  <12->   MnSymbolC12}{}
\DeclareMathSymbol{\meddiamond}{\mathbin}{MnSyC}{110}

\newcommand{\Lw}{\text{L}\myw_{1+\infty}}

\newcommand{\bigPhi}{\raisebox{-0.1\baselineskip}{\Large\ensuremath{\Phi}}}
\newcommand{\bigphi}{\raisebox{-0\baselineskip}{\Large\ensuremath{\phi}}}

\newtheorem{thm}{Theorem}

\theoremstyle{definition}
\newtheorem*{ex}{Example}
\newtheorem*{exs}{Examples}
\newtheorem{definition}{Definition}
\newtheorem{claim}{Claim}

\numberwithin{thm}{section}
\numberwithin{claim}{section}
\numberwithin{corollary}{section}

\title{Proof of the graviton MHV formula using Plebański's second heavenly equation}

\author{Noah Miller}

\affiliation{Center for the Fundamental Laws of Nature,
Harvard University, Cambridge, MA, USA}
\emailAdd{noahmiller@g.harvard.edu}

\preprint{}

\abstract{

Self-dual spacetimes can be thought of as spacetimes containing only positive helicity gravitons. In this work we give a perturbiner expansion for self-dual spacetimes based on Plebański's second heavenly equation. The expansion is naturally organized as a sum over ``marked tree graphs'' where each node corresponds to a positive helicity graviton and can have an arbitrary number of edges. Negative helicity gravitons must be added in by hand.        

We then use this perturbiner expansion to give a first principles derivation of the NSVW tree formula for the MHV amplitude in Einstein gravity. A unique feature of this proof is that it does not use BCFW recursion or twistor theory. It works by plugging the spacetime with arbitrarily many $+$ gravitons and two $-$ gravitons into the on-shell gravitational action and evaluating it. The action we use is the self-dual Plebański action plus an additional boundary term, and the amplitude itself comes entirely from the boundary term. Along the way, we also find an interesting new generalization of the NSVW formula which has not previously appeared in the literature.

In the appendix we give another way to express the perturbiner expansion using binary tree graphs instead of marked tree graphs, and prove the equivalence of these two expansions diagrammatically. We also provide an introduction to self-dual gravity aimed at non-experts, as well as a proof of the Parke-Taylor formula in Yang Mills theory analogous to our proof of the NSVW formula in gravity.

}

\begin{document}

\maketitle 

\newpage

\section{Introduction}

In pure Einstein gravity, the amplitude for an arbitrary number of positive helicity gravitons to scatter is zero at tree-level. If a single negative helicity graviton is added, the amplitude is still zero. If two negative helicity gravitons are added, however, the amplitude is non-zero. This is the so-called gravitational maximal helicity violating (MHV) amplitude, and it is the simplest non-vanishing amplitude in Einstein gravity.

While gravitational Feynman diagrammatics are notoriously complex even at tree-level, there are a number of miraculous closed form expressions for the MHV amplitude with arbitrarily many positive helicity gravitons \cite{Nguyen:2009jk, Berends:1988zp, Elvang:2007sg, Hodges:2012ym}. 

The first was given by Berends, Giele, and Kuijf (BGK) in 1988 \cite{Berends:1988zp}, although their formula did not make manifest the permutation symmetry of the positive helicity gravitons. A nicer formula was later found by Nguyen, Spradlin, Volovich, and Wen (NSVW) \cite{Nguyen:2009jk} in 2009. Their formula expressed the MHV amplitude as a sum over ``marked tree graphs'' (which are not the usual tree-level Feynman diagrams) where each node of the graph corresponds to a positive helicity graviton. In some sense, their formula was anticipated 10 years prior by Bern, Dixon, Perelstein, and Rozowsky \cite{Bern:1998sv}, who discovered a set of recursive rules to compute the MHV amplitude. While the NSVW formula does not appear explicitly in that earlier work, it can be obtained from the rules laid out in that paper \cite{Krasnov:2013wsa}. In 2012 \cite{Hodges:2012ym},  Hodges provided a more compact expression for the MHV amplitude as the determinant of a certain matrix which depends on two arbitrary reference spinors. Once the reference spinors are set to the negative helicity gravitons, it can be shown that Hodges' formula is equivalent to NSVW's through the use of a combinatorial ``matrix tree theorem,'' in which Hodges' determinant can be rewritten as a sum over NSVW's trees \cite{Feng:2012sy}. See \cite{Trnka:2020dxl} for a review of the various formulae.

It is not understood how the NSVW formula or Hodges' formula arise via the explicit summation of tree-level Feynman diagrams. Instead, they were both originally proven with a different method, by showing that they satisfied the BCFW recursion relation \cite{Britto:2004ap, Britto:2005fq}.

However, there do exist ``direct'' (i.e., non BCFW) proofs of the NSVW/Hodges' formulae. These proofs exploit the existence of an integrable subsector of Einstein gravity known as ``self-dual gravity'' (SDG), which is the main object of study of twistor theory \cite{penrose1976nonlinear, penrose1976nonlinear2}. These twistorial proofs of the MHV formulae include Mason and Skinner's 2008 proof of the original BGK formula \cite{Mason:2009afn} and Adamo, Mason, and Sharma's separate 2021 proof of the NSVW/Hodges' formula \cite{Adamo:2021bej}. 

Roughly speaking, both of these proofs operate by computing a classical gravitational action on-shell, whose value equals the tree-level MHV scattering amplitude itself. There is no fundamental reason that this action could not have been computed without the use of twistor theory, although in those works it was an indispensable tool.

We now summarize the two main results of this present paper:

\begin{enumerate}
    \item We explain how to perturbatively generate self-dual solutions to Einstein's equation using a perturbative expansion written in terms of ``marked tree graphs'' instead of Feynman diagrams. These self-dual metrics correspond to spacetimes containing an arbitrary number of positive helicity gravitons.
    
    \item Using the tree formula for self-dual metrics (and adding in two ASD perturbations by hand), we will provide a new first-principles proof of the NSVW formula by plugging this spacetime into the on-shell gravitational action. The only non-zero part of the action will come from a boundary term. This proof of NSVW is the first which does not make use of BCFW recursion or twistor theory.
\end{enumerate}

Throughout this paper, we will use a formulation of SDG due to Jerzy Plebański \cite{plebanski1975some}, who showed that any SD metric can be written in the form
\begin{equation}\label{eq2}
    ds^2 = 4 \big( d u \, d \bu - d w \, d \bw + (\partial_{w}^2 \phi ) \, d \bu^2 + (2 \partial_u \partial_w \phi) \, d \bu d \bw + (\partial_u^2 \phi) \, d \bw^2  \big)
\end{equation}
where $\phi$ is a scalar field that satisfies an equation of motion known as ``Plebański's second heavenly equation''
\begin{equation}\label{pleb_1}
    \Box \phi - \{ \partial_u \phi, \partial_w \phi \} = 0.
\end{equation}
Here, $\Box$ is the flat space wave operator and $\{ \cdot, \cdot \}$ is a spacetime Poisson bracket defined with the two coordinates $u, w$, as
\begin{equation}\label{pois_def_1}
    \Box \equiv \partial_u \partial_\bu - \partial_w \partial_\bw, \hspace{1 cm} \{ f, g \} \equiv \pdv{f}{u} \, \pdv{g}{w} - \pdv{f}{w} \, \pdv{g}{u}.
\end{equation}
Classical solutions to \eqref{pleb_1} can be generated with Feynman diagrams using standard methods, see for instance \cite{Monteiro:2011pc}, and \cite{ Monteiro:2022lwm, Armstrong-Williams:2022apo}. Having said that, we emphasize that we will not prove the validity of the marked tree expansion through the use of Feynman diagrams, even though it certainly must somehow be possible.

This paper is organized as follows.

\begin{itemize}
    \item In section \ref{sec2}, we present and prove the tree formula to perturbatively calculate self-dual metrics using marked tree diagrams.
    \item In section \ref{sec3}, we study linear SD and ASD perturbations to self-dual metrics using our tree formula. We then show how the tree formula can be used to explicitly calculate the action of a certain ``recursion operator’’ $\mR$, familiar to the study of integrability in SDG, on these linear perturbations. This will come in handy at multiple points when it comes time to prove the NSVW formula.
    \item In section \ref{sec4}, we briefly review the relationship between tree-level scattering amplitudes and on-shell actions.
    \item In section \ref{sec5} we provide the proof of the NSVW formula using the marked tree formula for the Plebański scalar.
    \item In section \ref{sec6} we discuss our results.
    \item In appendix \ref{appA}, we give a second way to build up self-dual metrics using binary tree graphs instead of the marked tree graphs from section \ref{sec2}. This binary tree expansion provides a recursive way to compute the self-dual metric containing $N+1$ positive helicity gravitons if one already has an expression for the metric containing $N$ positive helicity gravitons. We give a diagrammatic proof that this binary tree expansion is equivalent to our earlier marked tree expansion.
    \item In appendix \ref{appB} we provide the first principles derivations of many key equations of SDG used in this paper that are spread throughout the literature. This will hopefully be useful to readers new to the subject, and will culminate in the proofs of Plebański's first and second heavenly equations. 
    \item In appendix \ref{appD} we provide an analogous proof of the Parke-Taylor formula for the Yang Mills MHV amplitude.
\end{itemize}

In the language of Berends-Giele currents, we note that our calculation of the perturbiner expansion for Plebański's second scalar is essentially equivalent to the computation of the tree-level all-plus graviton current with one off-shell leg \cite{Bern:1998sv,Krasnov:2013wsa,Krasnov:2016emc}.

There are two possible advantages the proof we give here has over the BCFW proof. The first is that our perturbiner expansion can potentially be generalized to non-flat backgrounds \cite{Lipstein:2023pih,Bittleston:2024rqe,Adamo:2024xpc}, although we leave this for future work. The second is that it produces an interesting new generalization of the NSVW formula, equation \eqref{MHV_to_eval}, which is distinct from Hodges formula and has not previously appeared in the literature. (It would also be interesting if some connection could be made between our perturbiner expansion and Hodges formula, although we are unsure how this could be done.) 

In a partner paper \cite{mypaper}, we also use the marked tree expansion of the Plebański scalar to write down an explicit formula for the action of the so-called $\Lw$ algebra on self-dual spacetimes.

\section{A perturbiner expansion for Plebański's second heavenly equation}\label{sec2}

We are interested in solutions to the equation \eqref{pleb_1}. Once we have $\phi$, we can plug it into \eqref{eq2} to get a self-dual metric, so knowing $\phi$ is tantamount to knowing the metric. When $\phi = 0$, the metric reduces to flat space in lightcone coordinates
\begin{equation}\label{lc_coords}
\begin{aligned}
    u &= (X^0 - X^3)/2 \, ,\\
    \bu &= (X^0 + X^3)/2 \, ,\\
    w &= (X^1 + i X^2)/2 \, ,\\
    \bw &= (X^1 - i X^2)/2\, .
\end{aligned}
\end{equation}
Note that $\bu \neq (u)^*$, as both $u$ and $\bu$ are real independent variables. We have adopted this convention in the pursuit of notational symmetry.

It is straightforward to check that $\phi = A e^{i p \cdot X}$ solves \eqref{pleb_1} for any $A$ as long as $p^2$ = 0. This solution corresponds to a finite size (complex) positive helicity gravitational wave $h_{\mu \nu} = A \, \vep^+_{\mu \nu} e^{i p \cdot X}$ where $\vep^+_{\mu \nu}$ is given below. Because \eqref{pleb_1} is non-linear, the sum of two plane waves will in general not also satisfy \eqref{pleb_1}. One could imagine, however, taking a sum of plane waves and then recursively solving for higher order terms in the coefficients ``$A$'' in order to arrive at a full solution.

This motivates the idea of a ``perturbiner expansion'' \cite{Rosly:1997ap} which we now define. First, note that the equation of motion \eqref{pleb_1} can be split up into a free part (which is just the wave equation) and a quadratic interacting part. Imagine starting with a solution to the free equation 
\begin{equation}\label{linearpiece}
    \phi = \sum_{i = 1}^N \epsilon_i \, e^{i p_i \cdot X}
\end{equation}
where $p_i^2 = 0$, and the $\epsilon_i$'s are infinitesimal parameters. We will define the square of any particular infinitesimal parameter to be 0, but the product of distinct infinitesimal parameters to \emph{not} be zero:
\begin{equation}\label{epsilon_manual}
    (\epsilon_i)^2 = 0, \hspace{1 cm} \epsilon_i \epsilon_j \neq 0 \hspace{0.2 cm} \text{ if } i \neq j.
\end{equation}
With this in mind, we define the plane wave ``seed'' functions to be
\begin{equation}
    \phi_i \equiv \epsilon_i \, e^{i p_i \cdot X}.
\end{equation}

A perturbiner expansion is then defined as a solution to the full equation of motion which contains \eqref{linearpiece} as its linear piece. The linear piece is the boundary condition which specifies the rest of the solution. For example, when there are two seed functions, the perturbiner expansion will take the form
\begin{equation}
    \phi = \phi_1 + \phi_2 + \epsilon_1 \epsilon_2 (\ldots)
\end{equation}
and when there are three seed functions the expansion will take the form
\begin{equation}
    \phi = \phi_1 + \phi_2 + \phi_3 + \epsilon_1 \epsilon_2 (\ldots) + \epsilon_1 \epsilon_3 (\ldots) + \epsilon_2 \epsilon_3(\ldots) + \epsilon_1 \epsilon_2 \epsilon_3 (\ldots).
\end{equation}
The exact expression for the $(\ldots)$'s will be given soon. Note however that the expansion for $N$ seed functions terminates with a highest order term containing all $N$ infinitesimal parameters $\epsilon_1 \ldots \epsilon_N$.

We now define some useful conventions. We parameterize the momenta with an energy scale $\omega$ and two stereographic coordinates $\bz, z \in \mathbb{C}$ as
\begin{equation}\label{p_z_zb}
    p_i^\mu \equiv \omega_i q^\mu(\bz_i, z_i) \equiv \frac{\omega_i}{2}(1 + z_i \bz_i, z_i + \bz_i, - i(z_i - \bz_i), 1 - z_i \bz_i).
\end{equation}
In lightcone coordinates,
\begin{equation}\label{pdotX1}
    p_i \cdot X = \omega_i (u + z_i \bz_i \bu - \bz_i w - z_i \bw ).
\end{equation}
We also define spinor variables
\begin{equation}
    p^{A \dot A} = p^a \sigma_a^{A \dot A} = \begin{pmatrix} p^0 + p^3 & p^1 - i p^2 \\ p^1 + i p^2 & p^0 - p^3 \end{pmatrix}
\end{equation}
where $a = 0, 1, 2, 3$ are flat space indices. We then decompose the momenta as
\begin{equation}
    p^{A \dot A}_i = \kappa_i^A \kappat_i^{\dot A}
\end{equation}
with the explicit choice of spinor parameterization
\begin{equation}\label{kappa_def}
    \kappa_i^A = \begin{pmatrix} 1 \\ z_i \end{pmatrix}, \hspace{1 cm} \kappat_i^{\dot A} = \omega_i \begin{pmatrix} 1 \\ \bz_i \end{pmatrix}, \hspace{1 cm} \kappa^i_A = \begin{pmatrix} -z_i \\ 1 \end{pmatrix}, \hspace{1 cm} \kappat^i_{\dot A} = \omega_i \begin{pmatrix} -\bz_i \\ 1 \end{pmatrix}.
\end{equation}
We have the angle and square brackets
\begin{equation}\label{angsqa}
\begin{aligned}
    \langle i j \rangle &= \vep_{AB} \kappa_i^A \kappa_j^B = z_{ij}, \\
    [ i j ] &= \vep^{\dot A \dot B} \kappat^i_{\dot A} \kappat^j_{\dot B} = -\omega_i \omega_j \bz_{ij}.
\end{aligned}
\end{equation}
A full list of spinor conventions can be found in appendix \ref{app_spin}. We also define the spin-2 polarization tensors for the positive and negative helicity gravitons as
\begin{equation}\label{polarization}
    \vep^{\mu \nu}_{\pm,i} = \vep^\mu_{\pm,i} \vep^\nu_{\pm,i}, \hspace{1 cm} \vep^\mu_{+,i} = - 2 i \, \partial_z p^\mu_i,   \hspace{1 cm} \vep^\mu_{-,i} = \frac{2 i}{\omega^2} \partial_\bz p^\mu_i ,
\end{equation}
which when necessary are to be raised and lowered with the flat metric. These satisfy $\vep_{+,i}^{A\dot A} = 2 i \, r^A \kappat_i^{\dot A}/\langle \kappa_i r \rangle$ and $\vep_{-,i}^{A\dot A} = 2 i \, \kappa_i^A \tilde{r}^{\dot A}/[ \kappat_i \tilde{r} ]$ for $r^A = (0,1)$, $\tilde{r}^{\dot A} = (0,1)$.

Let us now return to our study of the perturbiner expansion for Plebański's second heavenly equation.

\begin{definition}
We denote the perturbiner expansion with
\begin{equation}\label{phi_pert_def}
    \bigphi ( \phi_1, \ldots, \phi_N ) \equiv \;\; \begin{matrix}\text{full perturbiner expansion of Plebański} \\ \text{scalar $\phi$ with seed functions }\phi_1, \ldots, \phi_N. \end{matrix} 
\end{equation}
Note the use of the large font.
\end{definition}

We will soon show how to write this perturbiner expansion as a sum over marked tree graphs. We now define these graphs.

\begin{definition}
    Let $\mT_N$ denote the set of marked, connected, tree graphs with $N$ or fewer nodes where each note has a distinct label contained in $\{1, \ldots, N\}$.
\end{definition}

\begin{ex} \secret
\begin{figure}[H]
    \centering
    \tikzset{every picture/.style={line width=0.75pt}} 

\begin{tikzpicture}[x=0.75pt,y=0.75pt,yscale=-1,xscale=1]

\draw    (295.4,83.7) -- (344.6,83.7) ;
\draw    (344.6,83.7) -- (387.4,55.7) ;
\draw    (24.1,133.5) -- (62.1,120.5) ;
\draw    (62.1,120.5) -- (97.8,135.6) ;
\draw  [fill={rgb, 255:red, 255; green, 255; blue, 255 }  ,fill opacity=1 ] (28.2,37.7) .. controls (28.2,30.8) and (33.8,25.2) .. (40.7,25.2) .. controls (47.6,25.2) and (53.2,30.8) .. (53.2,37.7) .. controls (53.2,44.6) and (47.6,50.2) .. (40.7,50.2) .. controls (33.8,50.2) and (28.2,44.6) .. (28.2,37.7) -- cycle ;
\draw  [fill={rgb, 255:red, 255; green, 255; blue, 255 }  ,fill opacity=1 ] (11.6,133.5) .. controls (11.6,126.6) and (17.2,121) .. (24.1,121) .. controls (31,121) and (36.6,126.6) .. (36.6,133.5) .. controls (36.6,140.4) and (31,146) .. (24.1,146) .. controls (17.2,146) and (11.6,140.4) .. (11.6,133.5) -- cycle ;
\draw  [fill={rgb, 255:red, 255; green, 255; blue, 255 }  ,fill opacity=1 ] (49.6,120.5) .. controls (49.6,113.6) and (55.2,108) .. (62.1,108) .. controls (69,108) and (74.6,113.6) .. (74.6,120.5) .. controls (74.6,127.4) and (69,133) .. (62.1,133) .. controls (55.2,133) and (49.6,127.4) .. (49.6,120.5) -- cycle ;
\draw  [fill={rgb, 255:red, 255; green, 255; blue, 255 }  ,fill opacity=1 ] (85.3,135.6) .. controls (85.3,128.7) and (90.9,123.1) .. (97.8,123.1) .. controls (104.7,123.1) and (110.3,128.7) .. (110.3,135.6) .. controls (110.3,142.5) and (104.7,148.1) .. (97.8,148.1) .. controls (90.9,148.1) and (85.3,142.5) .. (85.3,135.6) -- cycle ;
\draw    (295.4,83.7) -- (250.5,83.7) ;
\draw    (250.5,83.7) -- (221.4,110.1) ;
\draw    (221.4,51.2) -- (250.5,83.7) ;
\draw  [fill={rgb, 255:red, 255; green, 255; blue, 255 }  ,fill opacity=1 ] (208.9,51.2) .. controls (208.9,44.3) and (214.5,38.7) .. (221.4,38.7) .. controls (228.3,38.7) and (233.9,44.3) .. (233.9,51.2) .. controls (233.9,58.1) and (228.3,63.7) .. (221.4,63.7) .. controls (214.5,63.7) and (208.9,58.1) .. (208.9,51.2) -- cycle ;
\draw  [fill={rgb, 255:red, 255; green, 255; blue, 255 }  ,fill opacity=1 ] (238,83.7) .. controls (238,76.8) and (243.6,71.2) .. (250.5,71.2) .. controls (257.4,71.2) and (263,76.8) .. (263,83.7) .. controls (263,90.6) and (257.4,96.2) .. (250.5,96.2) .. controls (243.6,96.2) and (238,90.6) .. (238,83.7) -- cycle ;
\draw  [fill={rgb, 255:red, 255; green, 255; blue, 255 }  ,fill opacity=1 ] (208.9,110.1) .. controls (208.9,103.2) and (214.5,97.6) .. (221.4,97.6) .. controls (228.3,97.6) and (233.9,103.2) .. (233.9,110.1) .. controls (233.9,117) and (228.3,122.6) .. (221.4,122.6) .. controls (214.5,122.6) and (208.9,117) .. (208.9,110.1) -- cycle ;
\draw    (295.4,83.7) -- (326.2,118) ;
\draw    (327.3,46.9) -- (295.4,83.7) ;
\draw  [fill={rgb, 255:red, 255; green, 255; blue, 255 }  ,fill opacity=1 ] (282.9,83.7) .. controls (282.9,76.8) and (288.5,71.2) .. (295.4,71.2) .. controls (302.3,71.2) and (307.9,76.8) .. (307.9,83.7) .. controls (307.9,90.6) and (302.3,96.2) .. (295.4,96.2) .. controls (288.5,96.2) and (282.9,90.6) .. (282.9,83.7) -- cycle ;
\draw  [fill={rgb, 255:red, 255; green, 255; blue, 255 }  ,fill opacity=1 ] (314.8,46.9) .. controls (314.8,40) and (320.4,34.4) .. (327.3,34.4) .. controls (334.2,34.4) and (339.8,40) .. (339.8,46.9) .. controls (339.8,53.8) and (334.2,59.4) .. (327.3,59.4) .. controls (320.4,59.4) and (314.8,53.8) .. (314.8,46.9) -- cycle ;
\draw  [fill={rgb, 255:red, 255; green, 255; blue, 255 }  ,fill opacity=1 ] (332.1,83.7) .. controls (332.1,76.8) and (337.7,71.2) .. (344.6,71.2) .. controls (351.5,71.2) and (357.1,76.8) .. (357.1,83.7) .. controls (357.1,90.6) and (351.5,96.2) .. (344.6,96.2) .. controls (337.7,96.2) and (332.1,90.6) .. (332.1,83.7) -- cycle ;
\draw  [fill={rgb, 255:red, 255; green, 255; blue, 255 }  ,fill opacity=1 ] (313.7,118) .. controls (313.7,111.1) and (319.3,105.5) .. (326.2,105.5) .. controls (333.1,105.5) and (338.7,111.1) .. (338.7,118) .. controls (338.7,124.9) and (333.1,130.5) .. (326.2,130.5) .. controls (319.3,130.5) and (313.7,124.9) .. (313.7,118) -- cycle ;
\draw  [fill={rgb, 255:red, 255; green, 255; blue, 255 }  ,fill opacity=1 ] (374.9,55.7) .. controls (374.9,48.8) and (380.5,43.2) .. (387.4,43.2) .. controls (394.3,43.2) and (399.9,48.8) .. (399.9,55.7) .. controls (399.9,62.6) and (394.3,68.2) .. (387.4,68.2) .. controls (380.5,68.2) and (374.9,62.6) .. (374.9,55.7) -- cycle ;

\draw (55.2,37.7) node [anchor=west] [inner sep=0.75pt]    {$\ \in \mathcal{T}_{9}$};
\draw (40.7,37.7) node    {$3$};
\draw (62.1,120.5) node    {$1$};
\draw (97.8,135.6) node    {$9$};
\draw (24.1,133.5) node    {$5$};
\draw (108.4,128.1) node [anchor=west] [inner sep=0.75pt]    {$\ \in \mathcal{T}_{9}$};
\draw (326.2,118) node    {$9$};
\draw (327.3,46.9) node    {$2$};
\draw (221.4,51.2) node    {$7$};
\draw (250.5,83.7) node    {$3$};
\draw (221.4,110.1) node    {$5$};
\draw (295.4,83.7) node    {$8$};
\draw (344.6,83.7) node    {$6$};
\draw (387.4,55.7) node    {$4$};
\draw (396,85.7) node [anchor=west] [inner sep=0.75pt]    {$\ \in \mathcal{T}_{9}$};

\end{tikzpicture}
    \caption{Examples of graphs in $\mT_9$.}
    \label{fig:63y452t3}
\end{figure}

\end{ex}

Next we define a differential operator $\partial_\mu^{(i)}$ which only acts on the seed $\phi_i$, but not on $\phi_j$ for $i \neq j$.
\begin{equation}
    \partial_\mu^{(i)} \phi_i \equiv \partial_\mu \phi_i, \hspace{1 cm} \partial_\mu^{(i)} \phi_j \equiv 0 \hspace{0.5 cm} \text{ for } i \neq j
\end{equation}
For example
\begin{equation}
    \partial_\mu^{(1)} ( \phi_1 \phi_2 \phi_3 ) = (\partial_\mu \phi_1) \phi_2 \phi_3.
\end{equation}
We then define the differential operator $D^{ij}$ as
\begin{equation}
    D^{ij} \equiv \partial_u^{(i)} \partial_w^{(j)} - \partial_u^{(j)} \partial_w^{(i)}.
\end{equation}
Note, for instance,
\begin{equation}
    D^{12}(\phi_1 \phi_2 \phi_3) = \{ \phi_1, \phi_2 \} \phi_3 
\end{equation}
and
\begin{equation}\label{D12D12}
    D^{12} D^{12}(\phi_1 \phi_2 \phi_3) = \{ \partial_u \phi_1, \partial_w \phi_2 \} \phi_3 -\{ \partial_w \phi_1, \partial_u \phi_2 \} \phi_3.
\end{equation}

Using the operator $D^{ij}$, we will now explain what an ``edge'' in one of these graphs corresponds to. Two nodes, labelled $i$ and $j$, correspond to the two seed functions $\phi_i$ and $\phi_j$, while an edge between them corresponds to $D^{ij}/z_{ij}$, where $z_i$ and $z_j$ are just the same $z$'s which parameterized the momenta in \eqref{p_z_zb}.
\begin{figure}[H]
    \centering
    \tikzset{every picture/.style={line width=0.75pt}} 

\begin{tikzpicture}[x=0.75pt,y=0.75pt,yscale=-1,xscale=1]

\draw    (47,43.2) -- (84,43.2) ;
\draw   (22,43.2) .. controls (22,36.3) and (27.6,30.7) .. (34.5,30.7) .. controls (41.4,30.7) and (47,36.3) .. (47,43.2) .. controls (47,50.1) and (41.4,55.7) .. (34.5,55.7) .. controls (27.6,55.7) and (22,50.1) .. (22,43.2) -- cycle ;
\draw   (84,43.2) .. controls (84,36.3) and (89.6,30.7) .. (96.5,30.7) .. controls (103.4,30.7) and (109,36.3) .. (109,43.2) .. controls (109,50.1) and (103.4,55.7) .. (96.5,55.7) .. controls (89.6,55.7) and (84,50.1) .. (84,43.2) -- cycle ;

\draw (111,43.2) node [anchor=west] [inner sep=0.75pt]    {$\ =\ \dfrac{D^{ij}}{z_{ij}} \phi _{i} \phi _{j}$};
\draw (34.5,43.2) node    {$i$};
\draw (96.5,43.2) node    {$j$};

\end{tikzpicture}
    \caption{A edge in a graph corresponds to $D^{ij}/z_{ij}$.}
    \label{fig:7452634}
\end{figure}
More generally, we can express a term corresponding to an arbitrary graph in the following way.

\begin{definition}
For each graph $\bt \in \mathcal{T}_N$ we define the associated term $\phi_{\bt}$ via

\begin{equation}
    \phi_\bt \equiv \left( \prod_{e_{ij} \, \in \, \text{edges of } \bt} \frac{D^{ij}}{z_{ij}}\right) \left( \prod_{k \, \in \, \text{nodes of } \bt } \phi_k \right) \hspace{0.5 cm} \text{ for } \bt \in \mT_N.
\end{equation}
\end{definition}

\begin{ex} \secret
\begin{figure}[H]
    \centering
    \tikzset{every picture/.style={line width=0.75pt}} 

\begin{tikzpicture}[x=0.75pt,y=0.75pt,yscale=-1,xscale=1]

\draw    (221.2,111.81) -- (256.7,89.71) ;
\draw    (185.7,89.71) -- (221.2,111.81) ;
\draw    (185.7,89.71) -- (185.7,48.61) ;
\draw    (150.2,111.81) -- (185.7,89.71) ;
\draw  [fill={rgb, 255:red, 255; green, 255; blue, 255 }  ,fill opacity=1 ] (208.7,111.81) .. controls (208.7,104.91) and (214.3,99.31) .. (221.2,99.31) .. controls (228.1,99.31) and (233.7,104.91) .. (233.7,111.81) .. controls (233.7,118.71) and (228.1,124.31) .. (221.2,124.31) .. controls (214.3,124.31) and (208.7,118.71) .. (208.7,111.81) -- cycle ;
\draw  [fill={rgb, 255:red, 255; green, 255; blue, 255 }  ,fill opacity=1 ] (173.2,89.71) .. controls (173.2,82.81) and (178.8,77.21) .. (185.7,77.21) .. controls (192.6,77.21) and (198.2,82.81) .. (198.2,89.71) .. controls (198.2,96.61) and (192.6,102.21) .. (185.7,102.21) .. controls (178.8,102.21) and (173.2,96.61) .. (173.2,89.71) -- cycle ;
\draw  [fill={rgb, 255:red, 255; green, 255; blue, 255 }  ,fill opacity=1 ] (137.7,111.81) .. controls (137.7,104.91) and (143.3,99.31) .. (150.2,99.31) .. controls (157.1,99.31) and (162.7,104.91) .. (162.7,111.81) .. controls (162.7,118.71) and (157.1,124.31) .. (150.2,124.31) .. controls (143.3,124.31) and (137.7,118.71) .. (137.7,111.81) -- cycle ;
\draw  [fill={rgb, 255:red, 255; green, 255; blue, 255 }  ,fill opacity=1 ] (173.2,48.61) .. controls (173.2,41.71) and (178.8,36.11) .. (185.7,36.11) .. controls (192.6,36.11) and (198.2,41.71) .. (198.2,48.61) .. controls (198.2,55.51) and (192.6,61.11) .. (185.7,61.11) .. controls (178.8,61.11) and (173.2,55.51) .. (173.2,48.61) -- cycle ;
\draw  [fill={rgb, 255:red, 255; green, 255; blue, 255 }  ,fill opacity=1 ] (244.2,89.71) .. controls (244.2,82.81) and (249.8,77.21) .. (256.7,77.21) .. controls (263.6,77.21) and (269.2,82.81) .. (269.2,89.71) .. controls (269.2,96.61) and (263.6,102.21) .. (256.7,102.21) .. controls (249.8,102.21) and (244.2,96.61) .. (244.2,89.71) -- cycle ;

\draw (185.7,89.71) node    {$1$};
\draw (107,84.5) node [anchor=west] [inner sep=0.75pt]    {$\mathbf{t} =$};
\draw (150.2,111.81) node    {$2$};
\draw (185.7,48.61) node    {$3$};
\draw (221.2,111.81) node    {$4$};
\draw (358.5,84.5) node [anchor=west] [inner sep=0.75pt]    {$\phi _{\mathbf{t}} =\dfrac{D^{12}}{z_{12}}\dfrac{D^{13}}{z_{13}}\dfrac{D^{14}}{z_{14}}\dfrac{D^{45}}{z_{45}} \phi _{1} \phi _{2} \phi _{3} \phi _{4} \phi _{5}$};
\draw (256.7,89.71) node    {$5$};

\end{tikzpicture}
    \caption{An example of a graph $\bt$ and its corresponding $\phi_\bt$.}
    \label{fig:24t3}
\end{figure}
\end{ex}

\begin{thm}\label{thm_21}
    The complete perturbiner expansion of the Plebański scalar is given by a sum over all marked tree graphs in $\mT_N$ via

    \begin{equation}\label{theorem_tree_sum}
        \tcboxmath{ \bigphi ( \phi_1, \ldots, \phi_N ) = \sum^{}_{\bt \in \mT_N} \phi_\bt. }
    \end{equation} 
\end{thm}

\begin{proof}
To prove this statement, we will plug the above expression into Plebański's second heavenly equation \eqref{pleb_1} and check that it is a solution. This boils down to checking
\begin{equation}\label{to_prove_grav_1}
    \Box \left( \sum_{\bt \in \mT_N} \phi_{\bt} \right) \overset{?}{=}  \frac{1}{2}  \sum_{\bt \in \mT_N} \sum_{\bt' \in \mT_N} \big( \left\{ \partial_u \phi_{\bt}, \partial_w \phi_{\bt'}\right\} - \left\{ \partial_w \phi_{\bt}, \partial_u \phi_{\bt'}\right\}  \big).
\end{equation}
Let us start by evaluating the LHS of \eqref{to_prove_grav_1}. Recall that
\begin{equation}
    \Box = \partial_u \partial_\bu - \partial_w \partial_\bw
\end{equation}
and note the relations
\begin{equation}\label{z_diff_relation}
    \partial_\bu \phi_i = - z_i \partial_w \phi_i, \hspace{1 cm} \partial_\bw \phi_i = - z_i \partial_u \phi_i,
\end{equation}
which hold for our $\phi_i$ defined as plane waves (see \eqref{pdotX1}).

Let us think about what happens when we act $\Box$ on a string of $\phi_i$'s. Of course, for each $\phi_i$ we have $\Box \, \phi_i = 0$. Furthermore, using \eqref{z_diff_relation}, a single line of algebra shows that
\begin{equation}\label{boxphiiphij}
    \Box (\phi_i \phi_j) = z_{ij} D^{ij}( \phi_i \phi_j).
\end{equation}
From the product rule of differentiation, if we act $\Box$ on a longer string of seed functions, say $\phi_i \phi_j \phi_k$, we must sum over all pairs of the seed functions in the following way:
\begin{equation}
    \Box( \phi_i \phi_j \phi_k) = (z_{ij} D^{ij} + z_{jk} D^{jk} + z_{ik} D^{ik})(\phi_i \phi_j \phi_k).
\end{equation}
Therefore, in order to express the action of $\Box$ using our diagrammatic language, we shall define a new kind of edge, now denoted with a double line, that corresponds to the operator $z_{ij} D^{ij}$ instead of $D^{ij}/z_{ij}$.
\begin{figure}[H]
    \centering
    \tikzset{every picture/.style={line width=0.75pt}} 

\begin{tikzpicture}[x=0.75pt,y=0.75pt,yscale=-1,xscale=1]

\draw   (23,41.2) .. controls (23,34.3) and (28.6,28.7) .. (35.5,28.7) .. controls (42.4,28.7) and (48,34.3) .. (48,41.2) .. controls (48,48.1) and (42.4,53.7) .. (35.5,53.7) .. controls (28.6,53.7) and (23,48.1) .. (23,41.2) -- cycle ;
\draw   (85,41.2) .. controls (85,34.3) and (90.6,28.7) .. (97.5,28.7) .. controls (104.4,28.7) and (110,34.3) .. (110,41.2) .. controls (110,48.1) and (104.4,53.7) .. (97.5,53.7) .. controls (90.6,53.7) and (85,48.1) .. (85,41.2) -- cycle ;
\draw    (48,39.7) -- (85,39.7)(48,42.7) -- (85,42.7) ;

\draw (112,41.2) node [anchor=west] [inner sep=0.75pt]    {$\ =\ z_{ij} D^{ij} \phi _{i} \phi _{j}$};
\draw (35.5,41.2) node    {$i$};
\draw (97.5,41.2) node    {$j$};

\end{tikzpicture}
    \caption{A edge corresponding to $\Box(\phi_i \phi_j)$.}
    \label{fig:5698}
\end{figure}
If we act $\Box$ on a sum of graphs (like what we have on the LHS of \eqref{to_prove_grav_1}) we must therefore sum over a different set of marked graphs which we shall now define.

\begin{definition}
    Let $\mT^\Box_N$ denote the set of all graphs which can be constructed by taking a marked tree graph in $\mT_N$ and placing a single extra double line between two of the nodes.
\end{definition}
That is, the ``data'' of a graph $\bt \in \mT^\Box_N$ is identical to the ``data'' of a graph $\bt \in \mT_N$ save for the addition of a single double-lined edge. The term $\phi_\bt$ corresponding to these graphs is given by
\begin{equation}
    \phi_\bt \equiv\left( z_{ab} D^{ab}\right)\Bigg\rvert_{ \substack{e_{ab} \text{ is} \\ \text{double} \\ \text{edge of } \bt  }} \left( \prod_{e_{ij} \, \in \, \substack{\text{single-lined} \\ \text{edges of } \bt}} \frac{D^{ij}}{z_{ij}}\right) \left( \prod_{k \, \in \, \text{nodes of } \bt } \phi_k \right) \hspace{0.5 cm} \text{ for } \bt \in \mT_N^{\Box}.
\end{equation}
We can now write $\Box$ acting on our sum of tree graphs as
\begin{equation}\label{to_prove_LHS}
    \Box\left( \sum_{\bt \in \mT_N} \phi_{\bt} \right) = \sum_{\bt \in \mT^\Box_N} \phi_\bt.
\end{equation}

We now turn to the RHS of \eqref{to_prove_grav_1}. Let us begin by exploring what happens when two graphs are placed in two different slots of the Poisson bracket. To represent two seed functions in the Poisson bracket $\{ \phi_i, \phi_j\}$, we would need to create an edge corresponding to $D^{ij}(\phi_i \phi_j)$. However, due to the antisymmetry of the Poisson bracket ($D^{ij} = - D^{ji})$, this edge must be a bit different from the other edges because it must be given an orientation. We display this orientation with an arrow, and draw this particular edge with a wiggly line.
\begin{figure}[H]
    \centering
    \tikzset{every picture/.style={line width=0.75pt}} 

\begin{tikzpicture}[x=0.75pt,y=0.75pt,yscale=-1,xscale=1]

\draw   (37,35.2) .. controls (37,28.3) and (42.6,22.7) .. (49.5,22.7) .. controls (56.4,22.7) and (62,28.3) .. (62,35.2) .. controls (62,42.1) and (56.4,47.7) .. (49.5,47.7) .. controls (42.6,47.7) and (37,42.1) .. (37,35.2) -- cycle ;
\draw   (99,35.2) .. controls (99,28.3) and (104.6,22.7) .. (111.5,22.7) .. controls (118.4,22.7) and (124,28.3) .. (124,35.2) .. controls (124,42.1) and (118.4,47.7) .. (111.5,47.7) .. controls (104.6,47.7) and (99,42.1) .. (99,35.2) -- cycle ;
\draw    (62,35.2) .. controls (63.67,33.53) and (65.33,33.53) .. (67,35.2) .. controls (68.67,36.87) and (70.33,36.87) .. (72,35.2) .. controls (73.67,33.53) and (75.33,33.53) .. (77,35.2) .. controls (78.67,36.87) and (80.33,36.87) .. (82,35.2) .. controls (83.67,33.53) and (85.33,33.53) .. (87,35.2) .. controls (88.67,36.87) and (90.33,36.87) .. (92,35.2) .. controls (93.67,33.53) and (95.33,33.53) .. (97,35.2) -- (99,35.2) -- (99,35.2) ;
\draw [shift={(85.5,35.2)}, rotate = 180] [fill={rgb, 255:red, 0; green, 0; blue, 0 }  ][line width=0.08]  [draw opacity=0] (10.72,-5.15) -- (0,0) -- (10.72,5.15) -- (7.12,0) -- cycle    ;

\draw (126,35.2) node [anchor=west] [inner sep=0.75pt]    {$\ =\ D^{ij} \phi _{i} \phi _{j}$};
\draw (49.5,35.2) node    {$i$};
\draw (111.5,35.2) node    {$j$};

\end{tikzpicture}
    \caption{A edge corresponding to $D^{ij} \phi_i \phi_j$.}
    \label{fig:45635}
\end{figure}

Now, say we have two graphs $\bt_\alpha$ and $\bt_\beta$ which could look like the following.
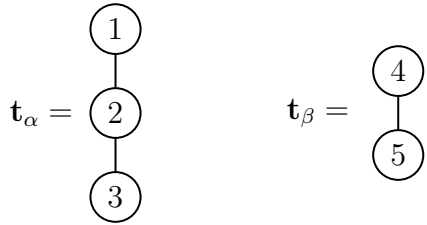
\begin{center}
    \tikzset{every picture/.style={line width=0.75pt}} 

\begin{tikzpicture}[x=0.75pt,y=0.75pt,yscale=-1,xscale=1]

\draw    (141.88,60.92) -- (141.88,103.13) ;
\draw    (141.88,18.71) -- (141.88,60.92) ;
\draw  [fill={rgb, 255:red, 255; green, 255; blue, 255 }  ,fill opacity=1 ] (129.38,18.71) .. controls (129.38,11.81) and (134.98,6.21) .. (141.88,6.21) .. controls (148.79,6.21) and (154.38,11.81) .. (154.38,18.71) .. controls (154.38,25.61) and (148.79,31.21) .. (141.88,31.21) .. controls (134.98,31.21) and (129.38,25.61) .. (129.38,18.71) -- cycle ;
\draw  [fill={rgb, 255:red, 255; green, 255; blue, 255 }  ,fill opacity=1 ] (129.38,60.92) .. controls (129.38,54.02) and (134.98,48.42) .. (141.88,48.42) .. controls (148.79,48.42) and (154.38,54.02) .. (154.38,60.92) .. controls (154.38,67.83) and (148.79,73.42) .. (141.88,73.42) .. controls (134.98,73.42) and (129.38,67.83) .. (129.38,60.92) -- cycle ;
\draw    (283,39.82) -- (283,82.03) ;
\draw  [fill={rgb, 255:red, 255; green, 255; blue, 255 }  ,fill opacity=1 ] (129.38,103.13) .. controls (129.38,96.23) and (134.98,90.63) .. (141.88,90.63) .. controls (148.79,90.63) and (154.38,96.23) .. (154.38,103.13) .. controls (154.38,110.04) and (148.79,115.63) .. (141.88,115.63) .. controls (134.98,115.63) and (129.38,110.04) .. (129.38,103.13) -- cycle ;
\draw  [fill={rgb, 255:red, 255; green, 255; blue, 255 }  ,fill opacity=1 ] (270.5,39.82) .. controls (270.5,32.91) and (276.1,27.32) .. (283,27.32) .. controls (289.9,27.32) and (295.5,32.91) .. (295.5,39.82) .. controls (295.5,46.72) and (289.9,52.32) .. (283,52.32) .. controls (276.1,52.32) and (270.5,46.72) .. (270.5,39.82) -- cycle ;
\draw  [fill={rgb, 255:red, 255; green, 255; blue, 255 }  ,fill opacity=1 ] (270.5,82.03) .. controls (270.5,75.12) and (276.1,69.53) .. (283,69.53) .. controls (289.9,69.53) and (295.5,75.12) .. (295.5,82.03) .. controls (295.5,88.93) and (289.9,94.53) .. (283,94.53) .. controls (276.1,94.53) and (270.5,88.93) .. (270.5,82.03) -- cycle ;

\draw (141.88,18.71) node    {$1$};
\draw (141.88,60.92) node    {$2$};
\draw (141.88,103.13) node    {$3$};
\draw (283,82.03) node    {$5$};
\draw (283,39.82) node    {$4$};
\draw (127.38,60.92) node [anchor=east] [inner sep=0.75pt]    {$\mathbf{t}_{\alpha } =\ $};
\draw (281,60.92) node [anchor=east] [inner sep=0.75pt]    {$\mathbf{t}_{\beta } =\ \ \ \ $};

\end{tikzpicture}
\end{center}
What sum of graphs will the expression $\{ \phi_{\mathbf{t}_\alpha} ,\phi_{\mathbf{t}_\beta} \}$ be? From the product rule of differentiation, the Poisson bracket will be equal to a sum of terms which will correspond to all possible ways a wiggly line could be drawn from a node on $\bt_\alpha$ to a node on $\bt_\beta$. Let us denote such a graph with a wiggly line starting at node $i$ and ending on node $j$ as $\bt_{\{\alpha, \beta\}_{ij}}$.
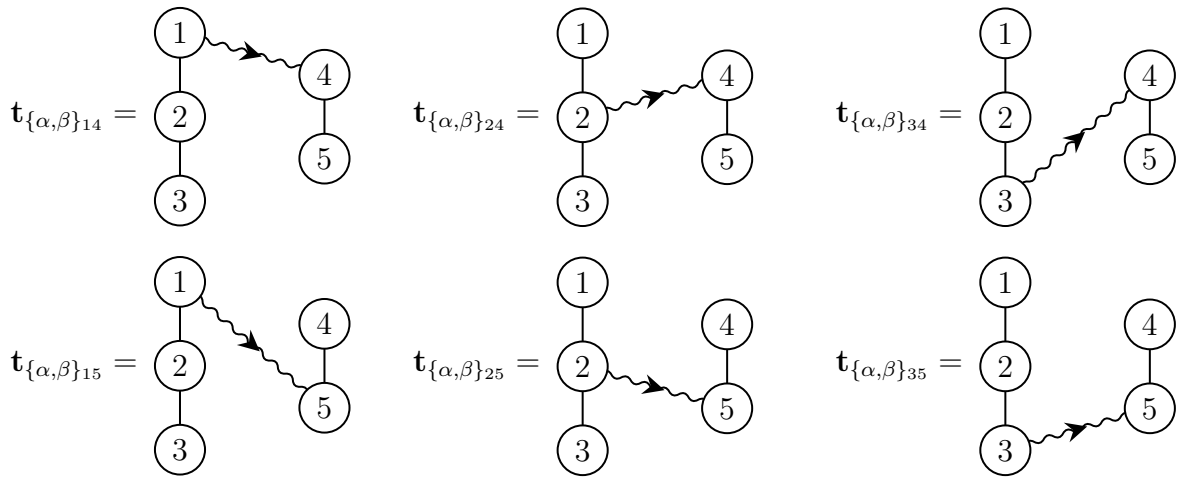
\begin{center}
    \input{figures/gravity_graph_7}
\end{center}
\begin{equation}
    \{ \phi_{\bt_\alpha}, \phi_{\bt_\beta} \} = \sum_{i \, \in \, \text{nodes of } \bt_\alpha } \sum_{j \, \in \, \text{nodes of } \bt_\beta } \phi_{ \bt_{\{\alpha, \beta \}_{ij} } }
\end{equation}

Of course, because $(\phi_i)^2 = 0$ due to the infinitesimal nature of $\epsilon_i$, if there are any shared nodes between the two graphs, the Poisson bracket will actually be zero.
\begin{equation}
    \{ \phi_{\bt_\alpha}, \phi_{\bt_\beta} \} = 0 \hspace{0.5 cm} \text{if } \bt_\alpha \text{ and } \bt_{\beta} \text{ both have a node with the same label}.
\end{equation}

Now we return to the RHS of \eqref{to_prove_grav_1} as written. We can see that this term contains not just a Poisson bracket, which would correspond to a single wiggly line being drawn between two graphs, but also contains an extra set of $\partial_u$, $\partial_w$ derivatives which have the effect of requiring that \emph{two} wiggly lines be drawn between the two graphs! (See for instance \eqref{D12D12}.) Let us now give a name to the set of all such graphs.

\begin{definition}
    Let $\mT^{P^2}_N$ denote the set of all graphs which can be constructed by taking two graphs in $\mT_N$ that share no common nodes and drawing \emph{two} wiggly lines between them, with each wiggly line connecting one node in one graph to another node in the other graph. The orientation of these two wiggly lines must be arranged such that they both point in the same direction (i.e., from the same one of the original graphs to the other).
\end{definition}

We write the term corresponding to a graph in $\mT^{P^2}_N$ as
\begin{equation}
    \phi_\bt \equiv \left( \prod_{\vec{e}_{ab} \, \in \, \substack{\text{wiggly} \\ \text{edges of }\bt } } D^{ab} \right) \left( \prod_{e_{ij} \, \in \, \substack{\text{single-lined} \\ \text{edges of } \bt} } \frac{D^{ij}}{z_{ij}}\right) \left( \prod_{k \, \in \, \text{nodes of } \bt } \phi_k \right) \hspace{0.5 cm} \text{ for } \bt \in \mT_N^{P^2}
\end{equation}
(where the vector notation $\vec{e}_{ab}$ is used to denote the fact that the wiggly edges have an orientation pointing from $a$ to $b$) such that
\begin{equation}\label{to_prove_RHS}
    \frac{1}{2}  \sum_{\bt \in \mT_N} \sum_{\bt' \in \mT_N} \big( \left\{ \partial_u \phi_{\bt}, \partial_w \phi_{\bt'}\right\} - \left\{ \partial_w \phi_{\bt}, \partial_u \phi_{\bt'}\right\}  \big) = \sum_{\bt \in \mT^{P^2}_N} \phi_{\bt}.
\end{equation}
The factor of $1/2$ accounts for the symmetry factor of 2 associated with the indistinguishability of the two wiggly lines.

Therefore, in order to prove \eqref{to_prove_grav_1}, using \eqref{to_prove_LHS} and \eqref{to_prove_RHS} we have now shown that it is equivalent to check
\begin{equation}\label{qmarkboxP2}
    \sum_{\bt \in \mT_N^{\Box}} \phi_{\bt} \overset{?}{=} \sum_{\bt \in \mT_N^{P^2}} \phi_{\bt}.
\end{equation}
Let's now look at some examples of graphs in $\mT_N^{\Box}$ and $\mT_N^{P^2}$ and show how we can re-organize the above sums in a way that enables us to prove the above equality.

\begin{exs} \secret
\begin{figure}[H]
    \centering
    \input{figures/gravity_graph_4}
    \caption{An example graph in $\mT_9^{\Box}$, as well as an example graph in $\mT_9^{P^2}$.}
    \label{fig:64523}
\end{figure}
There are also exceptional graphs in $\mT_N^\Box$ where the double-lined edge connects two nodes which were already connected by a single-lined edge, and exceptional graphs in $\mT^{P^2}_N$ where the two wiggly edges both attach to the same two nodes:
\begin{figure}[H]
    \centering
    \input{figures/gravity_graph_5}
    \caption{An exceptional case of graphs in $\mT^\Box_9$ and $\mT^{P^2}_9$.}
    \label{fig:546234}
\end{figure}
\end{exs}

The above examples in Figures \ref{fig:64523} and \ref{fig:546234} were chosen to emphasize a point. Notice that the graphs in both $\mT^\Box_N$ and $\mT^{P^2}_N$ all contain exactly one \emph{cycle}. In $\mT^\Box_N$, the double line is always contained in this cycle, and in $\mT^{P^2}_N$ the two wiggly lines are also contained in this cycle. 

Let us consider a group of graphs in $\mT_N^{\Box}$ and $\mT_N^{P^2}$ for which all the graphs in the group would look identical if all the double lines and wiggly lines were to be replaced with ordinary single lines. If the cycle in this group is comprised of $M$ nodes, then there will be $M$ such graphs from $\mT_N^{\Box}$ (as there are $M$ locations to place the double line in the cycle) and $M(M-1)/2$ such graphs from $\mT_N^{P^2}$ (as there are $M(M-1)/2$ ways to place the two wiggly lines in the cycle).

In Figures \ref{fig8} and \ref{fig9} we draw one such collection where $M = 5$. We have labelled the nodes in the cycle as $i,j,k,\ell, m$. These cycles should be understood as being contained in a larger graph, where the exact details of the larger graph will not impact the coming argument. 

In Figure \ref{fig8} we sum 5 double-lined graphs in $\mT^\Box_N$. \emph{Notice that a double line connecting nodes $i$ and $j$ is equal to a single line multiplied by $z_{ij}^2$.} (See Figures \ref{fig:7452634} and \ref{fig:5698}.) Therefore, upon summing up the graphs, we note that the sum is equal to a modified version of the graph containing only single lines in the cycle, given that we multiply the whole graph by the sum of squares of all $z_{ij}$'s in the cycle.

\begin{figure}
    \begin{center}
        \input{figures/gravity_graph_8}
    \end{center} \vspace{-0.6 cm}
        \caption{ \label{fig8}  Here we sum 5 graphs with a cycle of size 5 which could be in $\mT^\Box_N$. The graphs are identical except for the placement of the double line in the cycle. Notice the sum equals a graph with no double-line times a certain factor.}
        \vspace{0.4 cm}
    \begin{center}
        \input{figures/gravity_graph_9}
    \end{center} \vspace{-0.6 cm}
    \caption{ \label{fig9}  Here we sum 20 different graphs with a certain cycle of size 5, which could be contained in $\mT^{P^2}_N$. The graphs are identical except for the placement of two wiggly lines. Notice the sum equals a graph with no wiggly lines times a certain factor.} 
\end{figure}

In Figure \ref{fig9} we sum the corresponding graphs in $\mT^{P^2}_{N}$, which contain the exact same cycle as the previous group. There will be $5 (5-1)/2 = 10$ such graphs (or actually 20 if we account for both allowed orientations of the wiggly lines). \emph{Notice that a wiggly line connecting nodes $i$ and $j$ is can be replaced with a straight line if we multiply the graph by $z_{ij}$.} (See Figures \ref{fig:7452634} and \ref{fig:45635}.) Therefore, summing up the graphs, we note that the sum is equal to a version of the graph containing only straight lines in the cycle if we multiply the whole graph by $2$ times the sum of all products of distinct $z_{ij}$ edges in the cycle, with signs accounting for orientation.

If we now inspect Figures \ref{fig8} and \ref{fig9}, we see that the sum of graphs in both groups end up becoming the \emph{exact} same graph with seemingly different numerical prefactors. We will then complete the proof of our theorem once we show that these two numerical factors agree:
\begin{align}\label{cycle_to_check}
    & (z_{ij}^2 + z_{jk}^2 + z_{k \ell}^2 + z_{\ell i}^2 + z_{m i}^2) \overset{?}{=} \\
    & - 2 (z_{ij}(z_{jk} + z_{k\ell} + z_{\ell m} + z_{mi} ) + z_{jk}(z_{k \ell} + z_{\ell m} + z_{mi} ) + z_{k \ell}(z_{\ell m} + z_{mi}) + z_{\ell m} z_{mi}) . \nonumber
\end{align}
This equality can be proven as follows. Notice that if we sum all of the $z_{ij}$'s corresponding to the edges in the cycle, we get 0:
\begin{equation}
    z_{ij} + z_{jk} + z_{k \ell} + z_{\ell m} + z_{mi} = 0.
\end{equation}
Squaring the above equation gives exactly \eqref{cycle_to_check}. It is straightforward to see that the above argument works for not just cycles of size 5, but cycles of any size, completing the proof of \eqref{qmarkboxP2}, and therefore \eqref{theorem_tree_sum}.

\end{proof}

\begin{exs}
    Let us give some examples of perturbiner expansions of Plebański's second scalar for $N = 1, 2, 3$ using \eqref{theorem_tree_sum}. 

    \vspace{0.5 cm}

    \begin{minipage}{0.25 \textwidth}
    
    \begin{equation*}
        \bigphi(\phi_1) = \phi_1 \vphantom{\frac{1}{z_{12}}}
    \end{equation*}

    \begin{center}
        \tikzset{every picture/.style={line width=0.75pt}} 

\begin{tikzpicture}[x=0.75pt,y=0.75pt,yscale=-1,xscale=1]

\draw  [fill={rgb, 255:red, 255; green, 255; blue, 255 }  ,fill opacity=1 ] (141.44,32.21) .. controls (141.44,25.31) and (147.04,19.71) .. (153.94,19.71) .. controls (160.85,19.71) and (166.44,25.31) .. (166.44,32.21) .. controls (166.44,39.11) and (160.85,44.71) .. (153.94,44.71) .. controls (147.04,44.71) and (141.44,39.11) .. (141.44,32.21) -- cycle ;

\draw (153.94,32.21) node    {$1$};

\end{tikzpicture}
    \end{center}
    \end{minipage}
    \begin{minipage}{0.65 \textwidth}    
    \begin{equation*}
        \bigphi(\phi_1, \phi_2) = \phi_1 + \phi_2 + \frac{1}{z_{12}} \{ \phi_1, \phi_2\}
    \end{equation*}

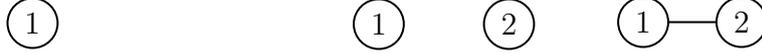
    \begin{center}
        \tikzset{every picture/.style={line width=0.75pt}} 

\begin{tikzpicture}[x=0.75pt,y=0.75pt,yscale=-1,xscale=1]

\draw    (283.94,42.21) -- (333,42.21) ;
\draw  [fill={rgb, 255:red, 255; green, 255; blue, 255 }  ,fill opacity=1 ] (139.44,43.21) .. controls (139.44,36.31) and (145.04,30.71) .. (151.94,30.71) .. controls (158.85,30.71) and (164.44,36.31) .. (164.44,43.21) .. controls (164.44,50.11) and (158.85,55.71) .. (151.94,55.71) .. controls (145.04,55.71) and (139.44,50.11) .. (139.44,43.21) -- cycle ;
\draw  [fill={rgb, 255:red, 255; green, 255; blue, 255 }  ,fill opacity=1 ] (204.44,43.21) .. controls (204.44,36.31) and (210.04,30.71) .. (216.94,30.71) .. controls (223.85,30.71) and (229.44,36.31) .. (229.44,43.21) .. controls (229.44,50.11) and (223.85,55.71) .. (216.94,55.71) .. controls (210.04,55.71) and (204.44,50.11) .. (204.44,43.21) -- cycle ;
\draw  [fill={rgb, 255:red, 255; green, 255; blue, 255 }  ,fill opacity=1 ] (320.5,42.21) .. controls (320.5,35.31) and (326.1,29.71) .. (333,29.71) .. controls (339.9,29.71) and (345.5,35.31) .. (345.5,42.21) .. controls (345.5,49.11) and (339.9,54.71) .. (333,54.71) .. controls (326.1,54.71) and (320.5,49.11) .. (320.5,42.21) -- cycle ;
\draw  [fill={rgb, 255:red, 255; green, 255; blue, 255 }  ,fill opacity=1 ] (271.44,42.21) .. controls (271.44,35.31) and (277.04,29.71) .. (283.94,29.71) .. controls (290.85,29.71) and (296.44,35.31) .. (296.44,42.21) .. controls (296.44,49.11) and (290.85,54.71) .. (283.94,54.71) .. controls (277.04,54.71) and (271.44,49.11) .. (271.44,42.21) -- cycle ;

\draw (151.94,43.21) node    {$1$};
\draw (216.94,43.21) node    {$2$};
\draw (333,42.21) node    {$2$};
\draw (283.94,42.21) node    {$1$};

\end{tikzpicture}
    \end{center}
    \end{minipage}

    \begin{align} \nonumber
        \bigphi(\phi_1, \phi_2, \phi_3)&= \phi_1 + \phi_2 + \phi_3 \\ 
        &+ \frac{1}{z_{12}} \{ \phi_1, \phi_2\} + \frac{1}{z_{13}} \{ \phi_1, \phi_3\} + \frac{1}{z_{23}} \{ \phi_2, \phi_3\} \\ \nonumber
        &+ \frac{D^{12}}{z_{12}}\frac{D^{23}}{z_{23}} \phi_1 \phi_2 \phi_3 + \frac{D^{13}}{z_{13}}\frac{D^{32}}{z_{32}} \phi_1 \phi_2 \phi_3 +  \frac{D^{21}}{z_{21}}\frac{D^{13}}{z_{13}} \phi_1 \phi_2 \phi_3
    \end{align}

    \begin{center}
        \tikzset{every picture/.style={line width=0.75pt}} 

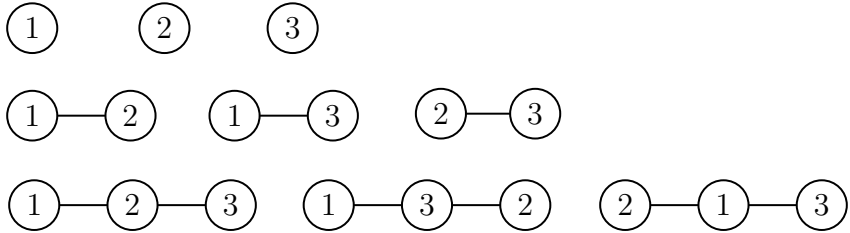
\begin{tikzpicture}[x=0.75pt,y=0.75pt,yscale=-1,xscale=1]

\draw    (215,126.71) -- (264.06,126.71) ;
\draw  [fill={rgb, 255:red, 255; green, 255; blue, 255 }  ,fill opacity=1 ] (152.44,37.71) .. controls (152.44,30.81) and (158.04,25.21) .. (164.94,25.21) .. controls (171.85,25.21) and (177.44,30.81) .. (177.44,37.71) .. controls (177.44,44.61) and (171.85,50.21) .. (164.94,50.21) .. controls (158.04,50.21) and (152.44,44.61) .. (152.44,37.71) -- cycle ;
\draw  [fill={rgb, 255:red, 255; green, 255; blue, 255 }  ,fill opacity=1 ] (218.44,37.71) .. controls (218.44,30.81) and (224.04,25.21) .. (230.94,25.21) .. controls (237.85,25.21) and (243.44,30.81) .. (243.44,37.71) .. controls (243.44,44.61) and (237.85,50.21) .. (230.94,50.21) .. controls (224.04,50.21) and (218.44,44.61) .. (218.44,37.71) -- cycle ;
\draw  [fill={rgb, 255:red, 255; green, 255; blue, 255 }  ,fill opacity=1 ] (282.44,37.71) .. controls (282.44,30.81) and (288.04,25.21) .. (294.94,25.21) .. controls (301.85,25.21) and (307.44,30.81) .. (307.44,37.71) .. controls (307.44,44.61) and (301.85,50.21) .. (294.94,50.21) .. controls (288.04,50.21) and (282.44,44.61) .. (282.44,37.71) -- cycle ;
\draw    (165.94,126.71) -- (215,126.71) ;
\draw  [fill={rgb, 255:red, 255; green, 255; blue, 255 }  ,fill opacity=1 ] (202.5,126.71) .. controls (202.5,119.81) and (208.1,114.21) .. (215,114.21) .. controls (221.9,114.21) and (227.5,119.81) .. (227.5,126.71) .. controls (227.5,133.61) and (221.9,139.21) .. (215,139.21) .. controls (208.1,139.21) and (202.5,133.61) .. (202.5,126.71) -- cycle ;
\draw  [fill={rgb, 255:red, 255; green, 255; blue, 255 }  ,fill opacity=1 ] (153.44,126.71) .. controls (153.44,119.81) and (159.04,114.21) .. (165.94,114.21) .. controls (172.85,114.21) and (178.44,119.81) .. (178.44,126.71) .. controls (178.44,133.61) and (172.85,139.21) .. (165.94,139.21) .. controls (159.04,139.21) and (153.44,133.61) .. (153.44,126.71) -- cycle ;
\draw  [fill={rgb, 255:red, 255; green, 255; blue, 255 }  ,fill opacity=1 ] (251.56,126.71) .. controls (251.56,119.81) and (257.15,114.21) .. (264.06,114.21) .. controls (270.96,114.21) and (276.56,119.81) .. (276.56,126.71) .. controls (276.56,133.61) and (270.96,139.21) .. (264.06,139.21) .. controls (257.15,139.21) and (251.56,133.61) .. (251.56,126.71) -- cycle ;
\draw    (363,126.71) -- (412.06,126.71) ;
\draw    (313.94,126.71) -- (363,126.71) ;
\draw  [fill={rgb, 255:red, 255; green, 255; blue, 255 }  ,fill opacity=1 ] (349.5,126.71) .. controls (349.5,119.81) and (355.1,114.21) .. (362,114.21) .. controls (368.9,114.21) and (374.5,119.81) .. (374.5,126.71) .. controls (374.5,133.61) and (368.9,139.21) .. (362,139.21) .. controls (355.1,139.21) and (349.5,133.61) .. (349.5,126.71) -- cycle ;
\draw  [fill={rgb, 255:red, 255; green, 255; blue, 255 }  ,fill opacity=1 ] (300.44,126.71) .. controls (300.44,119.81) and (306.04,114.21) .. (312.94,114.21) .. controls (319.85,114.21) and (325.44,119.81) .. (325.44,126.71) .. controls (325.44,133.61) and (319.85,139.21) .. (312.94,139.21) .. controls (306.04,139.21) and (300.44,133.61) .. (300.44,126.71) -- cycle ;
\draw  [fill={rgb, 255:red, 255; green, 255; blue, 255 }  ,fill opacity=1 ] (398.56,126.71) .. controls (398.56,119.81) and (404.15,114.21) .. (411.06,114.21) .. controls (417.96,114.21) and (423.56,119.81) .. (423.56,126.71) .. controls (423.56,133.61) and (417.96,139.21) .. (411.06,139.21) .. controls (404.15,139.21) and (398.56,133.61) .. (398.56,126.71) -- cycle ;
\draw    (510,126.71) -- (559.06,126.71) ;
\draw    (460.94,126.71) -- (510,126.71) ;
\draw  [fill={rgb, 255:red, 255; green, 255; blue, 255 }  ,fill opacity=1 ] (497.5,126.71) .. controls (497.5,119.81) and (503.1,114.21) .. (510,114.21) .. controls (516.9,114.21) and (522.5,119.81) .. (522.5,126.71) .. controls (522.5,133.61) and (516.9,139.21) .. (510,139.21) .. controls (503.1,139.21) and (497.5,133.61) .. (497.5,126.71) -- cycle ;
\draw  [fill={rgb, 255:red, 255; green, 255; blue, 255 }  ,fill opacity=1 ] (448.44,126.71) .. controls (448.44,119.81) and (454.04,114.21) .. (460.94,114.21) .. controls (467.85,114.21) and (473.44,119.81) .. (473.44,126.71) .. controls (473.44,133.61) and (467.85,139.21) .. (460.94,139.21) .. controls (454.04,139.21) and (448.44,133.61) .. (448.44,126.71) -- cycle ;
\draw  [fill={rgb, 255:red, 255; green, 255; blue, 255 }  ,fill opacity=1 ] (546.56,126.71) .. controls (546.56,119.81) and (552.15,114.21) .. (559.06,114.21) .. controls (565.96,114.21) and (571.56,119.81) .. (571.56,126.71) .. controls (571.56,133.61) and (565.96,139.21) .. (559.06,139.21) .. controls (552.15,139.21) and (546.56,133.61) .. (546.56,126.71) -- cycle ;
\draw    (164.94,81.71) -- (214,81.71) ;
\draw  [fill={rgb, 255:red, 255; green, 255; blue, 255 }  ,fill opacity=1 ] (201.5,81.71) .. controls (201.5,74.81) and (207.1,69.21) .. (214,69.21) .. controls (220.9,69.21) and (226.5,74.81) .. (226.5,81.71) .. controls (226.5,88.61) and (220.9,94.21) .. (214,94.21) .. controls (207.1,94.21) and (201.5,88.61) .. (201.5,81.71) -- cycle ;
\draw  [fill={rgb, 255:red, 255; green, 255; blue, 255 }  ,fill opacity=1 ] (152.44,81.71) .. controls (152.44,74.81) and (158.04,69.21) .. (164.94,69.21) .. controls (171.85,69.21) and (177.44,74.81) .. (177.44,81.71) .. controls (177.44,88.61) and (171.85,94.21) .. (164.94,94.21) .. controls (158.04,94.21) and (152.44,88.61) .. (152.44,81.71) -- cycle ;
\draw    (265.94,81.71) -- (315,81.71) ;
\draw  [fill={rgb, 255:red, 255; green, 255; blue, 255 }  ,fill opacity=1 ] (302.5,81.71) .. controls (302.5,74.81) and (308.1,69.21) .. (315,69.21) .. controls (321.9,69.21) and (327.5,74.81) .. (327.5,81.71) .. controls (327.5,88.61) and (321.9,94.21) .. (315,94.21) .. controls (308.1,94.21) and (302.5,88.61) .. (302.5,81.71) -- cycle ;
\draw  [fill={rgb, 255:red, 255; green, 255; blue, 255 }  ,fill opacity=1 ] (253.44,81.71) .. controls (253.44,74.81) and (259.04,69.21) .. (265.94,69.21) .. controls (272.85,69.21) and (278.44,74.81) .. (278.44,81.71) .. controls (278.44,88.61) and (272.85,94.21) .. (265.94,94.21) .. controls (259.04,94.21) and (253.44,88.61) .. (253.44,81.71) -- cycle ;
\draw    (368.94,80.71) -- (418,80.71) ;
\draw  [fill={rgb, 255:red, 255; green, 255; blue, 255 }  ,fill opacity=1 ] (403.5,80.71) .. controls (403.5,73.81) and (409.1,68.21) .. (416,68.21) .. controls (422.9,68.21) and (428.5,73.81) .. (428.5,80.71) .. controls (428.5,87.61) and (422.9,93.21) .. (416,93.21) .. controls (409.1,93.21) and (403.5,87.61) .. (403.5,80.71) -- cycle ;
\draw  [fill={rgb, 255:red, 255; green, 255; blue, 255 }  ,fill opacity=1 ] (356.44,80.71) .. controls (356.44,73.81) and (362.04,68.21) .. (368.94,68.21) .. controls (375.85,68.21) and (381.44,73.81) .. (381.44,80.71) .. controls (381.44,87.61) and (375.85,93.21) .. (368.94,93.21) .. controls (362.04,93.21) and (356.44,87.61) .. (356.44,80.71) -- cycle ;

\draw (164.94,37.71) node    {$1$};
\draw (230.94,37.71) node    {$2$};
\draw (294.94,37.71) node    {$3$};
\draw (215,126.71) node    {$2$};
\draw (165.94,126.71) node    {$1$};
\draw (264.06,126.71) node    {$3$};
\draw (362,126.71) node    {$3$};
\draw (312.94,126.71) node    {$1$};
\draw (411.06,126.71) node    {$2$};
\draw (510,126.71) node    {$1$};
\draw (460.94,126.71) node    {$2$};
\draw (559.06,126.71) node    {$3$};
\draw (214,81.71) node    {$2$};
\draw (164.94,81.71) node    {$1$};
\draw (315,81.71) node    {$3$};
\draw (265.94,81.71) node    {$1$};
\draw (416,80.71) node    {$3$};
\draw (368.94,80.71) node    {$2$};

\end{tikzpicture}
    \end{center}

\end{exs}

\section{Linear perturbations on a SD background}\label{sec3}

\subsection{Linear SD perturbations and the recursion operator}\label{sec3_3}

Here we study the set of \emph{linear perturbations} to Plebański's second scalar. A linear perturbation $\phi + \delta \phi$ satisfies the linearized e.o.m.

\begin{equation}\label{linearized_2nd_eq}
    \Box \delta \phi - \{ \partial_u \delta \phi, \partial_w \phi \} - \{ \partial_u \phi, \partial_w \delta  \phi \} = 0.
\end{equation}

\noindent If $\phi$ can be expressed as a perturbiner expansion of some set of seed functions, then we can always construct a $\delta \phi$ which is created by adding in \emph{one extra seed function} to the set. As a useful piece of notation, let us denote the change in the perturbiner associated to the addition of this single seed function with a vertical line by

\begin{equation}\label{lin_pert_def}
    \bigphi( \phi_1, \ldots, \phi_N \rvert \phi_{I} ) \equiv \bigphi(\phi_1, \ldots, \phi_N, \phi_{I} ) - \bigphi(\phi_1, \ldots, \phi_N)
\end{equation}
where $I$ is a new index not contained in $\{ 1, \ldots, N\}$. 

\begin{definition}
    Let $\mT_{N|I}$ denote the set of all marked tree graphs, necessarily containing the node $I$, with distinct node labels in the set $\{1, \ldots, N\} \cup \{I\}$.
\end{definition}

It is straightforward to see that
\begin{equation}\label{lin_graphs}
    \bigphi( \phi_1, \ldots, \phi_N \rvert \phi_{I} ) = \sum_{\bt \, \in \, \mT_{N | I}} \phi_{\bt} .
\end{equation}
Furthermore, as stated above, we know that
\begin{equation}\label{phi_deltaphi}
    \phi = \bigphi(\phi_1, \ldots, \phi_N) \, , \hspace{1 cm} \delta \phi = \bigphi(\phi_1, \ldots, \phi_N \rvert \phi_I) \, ,
\end{equation}
satisfies \eqref{linearized_2nd_eq}.

\begin{ex}
Let us explicitly write out $\delta \phi$ above when $N = 2$, $I = 3$, and draw all the corresponding graphs in $\mT_{2|3}$.
\begin{equation}
\begin{aligned}
    \bigphi(\phi_1, \phi_2 | \phi_3) &= \phi_3 + \frac{1}{z_{13}} \{ \phi_1, \phi_3 \} + \frac{1}{z_{23}} \{ \phi_2, \phi_3 \} \\
    &+ \frac{D^{12}}{z_{12}}\frac{D^{23}}{z_{23}} \phi_1 \phi_2 \phi_3 + \frac{D^{13}}{z_{13}}\frac{D^{32}}{z_{32}} \phi_1 \phi_2 \phi_3 +  \frac{D^{21}}{z_{21}}\frac{D^{13}}{z_{13}} \phi_1 \phi_2 \phi_3
\end{aligned}
\end{equation}

    \begin{center}
        \tikzset{every picture/.style={line width=0.75pt}} 

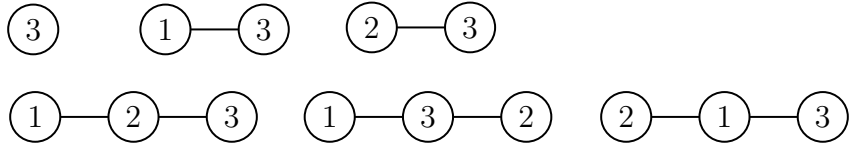
\begin{tikzpicture}[x=0.75pt,y=0.75pt,yscale=-1,xscale=1]

\draw    (230.94,37.71) -- (280,37.71) ;
\draw    (215,81.71) -- (264.06,81.71) ;
\draw  [fill={rgb, 255:red, 255; green, 255; blue, 255 }  ,fill opacity=1 ] (152.44,37.71) .. controls (152.44,30.81) and (158.04,25.21) .. (164.94,25.21) .. controls (171.85,25.21) and (177.44,30.81) .. (177.44,37.71) .. controls (177.44,44.61) and (171.85,50.21) .. (164.94,50.21) .. controls (158.04,50.21) and (152.44,44.61) .. (152.44,37.71) -- cycle ;
\draw  [fill={rgb, 255:red, 255; green, 255; blue, 255 }  ,fill opacity=1 ] (218.44,37.71) .. controls (218.44,30.81) and (224.04,25.21) .. (230.94,25.21) .. controls (237.85,25.21) and (243.44,30.81) .. (243.44,37.71) .. controls (243.44,44.61) and (237.85,50.21) .. (230.94,50.21) .. controls (224.04,50.21) and (218.44,44.61) .. (218.44,37.71) -- cycle ;
\draw    (165.94,81.71) -- (215,81.71) ;
\draw  [fill={rgb, 255:red, 255; green, 255; blue, 255 }  ,fill opacity=1 ] (202.5,81.71) .. controls (202.5,74.81) and (208.1,69.21) .. (215,69.21) .. controls (221.9,69.21) and (227.5,74.81) .. (227.5,81.71) .. controls (227.5,88.61) and (221.9,94.21) .. (215,94.21) .. controls (208.1,94.21) and (202.5,88.61) .. (202.5,81.71) -- cycle ;
\draw  [fill={rgb, 255:red, 255; green, 255; blue, 255 }  ,fill opacity=1 ] (153.44,81.71) .. controls (153.44,74.81) and (159.04,69.21) .. (165.94,69.21) .. controls (172.85,69.21) and (178.44,74.81) .. (178.44,81.71) .. controls (178.44,88.61) and (172.85,94.21) .. (165.94,94.21) .. controls (159.04,94.21) and (153.44,88.61) .. (153.44,81.71) -- cycle ;
\draw  [fill={rgb, 255:red, 255; green, 255; blue, 255 }  ,fill opacity=1 ] (251.56,81.71) .. controls (251.56,74.81) and (257.15,69.21) .. (264.06,69.21) .. controls (270.96,69.21) and (276.56,74.81) .. (276.56,81.71) .. controls (276.56,88.61) and (270.96,94.21) .. (264.06,94.21) .. controls (257.15,94.21) and (251.56,88.61) .. (251.56,81.71) -- cycle ;
\draw    (363,81.71) -- (412.06,81.71) ;
\draw    (313.94,81.71) -- (363,81.71) ;
\draw  [fill={rgb, 255:red, 255; green, 255; blue, 255 }  ,fill opacity=1 ] (349.5,81.71) .. controls (349.5,74.81) and (355.1,69.21) .. (362,69.21) .. controls (368.9,69.21) and (374.5,74.81) .. (374.5,81.71) .. controls (374.5,88.61) and (368.9,94.21) .. (362,94.21) .. controls (355.1,94.21) and (349.5,88.61) .. (349.5,81.71) -- cycle ;
\draw  [fill={rgb, 255:red, 255; green, 255; blue, 255 }  ,fill opacity=1 ] (300.44,81.71) .. controls (300.44,74.81) and (306.04,69.21) .. (312.94,69.21) .. controls (319.85,69.21) and (325.44,74.81) .. (325.44,81.71) .. controls (325.44,88.61) and (319.85,94.21) .. (312.94,94.21) .. controls (306.04,94.21) and (300.44,88.61) .. (300.44,81.71) -- cycle ;
\draw  [fill={rgb, 255:red, 255; green, 255; blue, 255 }  ,fill opacity=1 ] (398.56,81.71) .. controls (398.56,74.81) and (404.15,69.21) .. (411.06,69.21) .. controls (417.96,69.21) and (423.56,74.81) .. (423.56,81.71) .. controls (423.56,88.61) and (417.96,94.21) .. (411.06,94.21) .. controls (404.15,94.21) and (398.56,88.61) .. (398.56,81.71) -- cycle ;
\draw    (510,81.71) -- (559.06,81.71) ;
\draw    (460.94,81.71) -- (510,81.71) ;
\draw  [fill={rgb, 255:red, 255; green, 255; blue, 255 }  ,fill opacity=1 ] (497.5,81.71) .. controls (497.5,74.81) and (503.1,69.21) .. (510,69.21) .. controls (516.9,69.21) and (522.5,74.81) .. (522.5,81.71) .. controls (522.5,88.61) and (516.9,94.21) .. (510,94.21) .. controls (503.1,94.21) and (497.5,88.61) .. (497.5,81.71) -- cycle ;
\draw  [fill={rgb, 255:red, 255; green, 255; blue, 255 }  ,fill opacity=1 ] (448.44,81.71) .. controls (448.44,74.81) and (454.04,69.21) .. (460.94,69.21) .. controls (467.85,69.21) and (473.44,74.81) .. (473.44,81.71) .. controls (473.44,88.61) and (467.85,94.21) .. (460.94,94.21) .. controls (454.04,94.21) and (448.44,88.61) .. (448.44,81.71) -- cycle ;
\draw  [fill={rgb, 255:red, 255; green, 255; blue, 255 }  ,fill opacity=1 ] (546.56,81.71) .. controls (546.56,74.81) and (552.15,69.21) .. (559.06,69.21) .. controls (565.96,69.21) and (571.56,74.81) .. (571.56,81.71) .. controls (571.56,88.61) and (565.96,94.21) .. (559.06,94.21) .. controls (552.15,94.21) and (546.56,88.61) .. (546.56,81.71) -- cycle ;
\draw  [fill={rgb, 255:red, 255; green, 255; blue, 255 }  ,fill opacity=1 ] (267.5,37.71) .. controls (267.5,30.81) and (273.1,25.21) .. (280,25.21) .. controls (286.9,25.21) and (292.5,30.81) .. (292.5,37.71) .. controls (292.5,44.61) and (286.9,50.21) .. (280,50.21) .. controls (273.1,50.21) and (267.5,44.61) .. (267.5,37.71) -- cycle ;
\draw    (333.94,36.71) -- (383,36.71) ;
\draw  [fill={rgb, 255:red, 255; green, 255; blue, 255 }  ,fill opacity=1 ] (370.5,36.71) .. controls (370.5,29.81) and (376.1,24.21) .. (383,24.21) .. controls (389.9,24.21) and (395.5,29.81) .. (395.5,36.71) .. controls (395.5,43.61) and (389.9,49.21) .. (383,49.21) .. controls (376.1,49.21) and (370.5,43.61) .. (370.5,36.71) -- cycle ;
\draw  [fill={rgb, 255:red, 255; green, 255; blue, 255 }  ,fill opacity=1 ] (321.44,36.71) .. controls (321.44,29.81) and (327.04,24.21) .. (333.94,24.21) .. controls (340.85,24.21) and (346.44,29.81) .. (346.44,36.71) .. controls (346.44,43.61) and (340.85,49.21) .. (333.94,49.21) .. controls (327.04,49.21) and (321.44,43.61) .. (321.44,36.71) -- cycle ;

\draw (164.94,37.71) node    {$3$};
\draw (230.94,37.71) node    {$1$};
\draw (215,81.71) node    {$2$};
\draw (165.94,81.71) node    {$1$};
\draw (264.06,81.71) node    {$3$};
\draw (362,81.71) node    {$3$};
\draw (312.94,81.71) node    {$1$};
\draw (411.06,81.71) node    {$2$};
\draw (510,81.71) node    {$1$};
\draw (460.94,81.71) node    {$2$};
\draw (559.06,81.71) node    {$3$};
\draw (280,37.71) node    {$3$};
\draw (383,36.71) node    {$3$};
\draw (333.94,36.71) node    {$2$};

\end{tikzpicture}
    \end{center}

\end{ex}

In SDG there exists a certain ``recursion operator'' $\mR$ which is related to the integrability of the theory \cite{dunajski2009solitons}. This operator maps the space of linear perturbations of the Plebański scalar into itself, meaning if $\delta \phi$ satisfies \eqref{linearized_2nd_eq} then $\mR \delta \phi$ will also satisfy \eqref{linearized_2nd_eq}.

In particular, $\mR \delta \phi$ is defined as the solution to the pair of differential equations
\begin{equation}\label{R_grav}
\begin{aligned}
    \partial_u ( \mR \delta \phi) &= \partial_\bw \delta \phi + \{ \partial_u \phi, \delta \phi \}, \\
    \partial_w ( \mR \delta \phi) &= \partial_\bu \delta \phi + \{ \partial_w \phi, \delta \phi \}.
\end{aligned}
\end{equation}

The compatibility of the two equations \eqref{R_grav} follows from the fact that $\delta \phi$ is a linear perturbation of $\phi$:
\begin{equation}\label{Eq38}
    \partial_u ( \partial_w ( \mR \delta \phi) ) - \partial_w ( \partial_u ( \mR \delta \phi) ) = \Box \delta \phi - \{ \partial_u \delta \phi, \partial_w \phi \} - \{ \partial_u \phi, \partial_w \delta \phi \} = 0.
\end{equation}
Furthermore, $\mR \delta \phi$ is guaranteed to solve the linearized e.o.m.\ by 
\begin{equation}\label{Eq39}
    \Box \, \mR \delta \phi  - \{ \partial_u \mR \delta \phi, \partial_w \phi\} -\{ \partial_u \phi, \partial_w \mR \delta \phi \}= \{\Box \phi - \{\partial_u \phi, \partial_w \phi\}, \delta \phi\} = 0.
\end{equation}

A explicit set of solutions $\phi$ and $\delta \phi$ was provided by our perturbiner expansion in \eqref{phi_deltaphi}. One might wonder how $\mR$ acts on such a $\delta \phi$. The result is given by the following theorem.

\begin{thm}\label{thmRz}
    \begin{equation}
        \tcboxmath{\mR\,  \bigphi(\phi_1, \ldots, \phi_N \rvert \phi_I ) = - z_I \times  \bigphi(\phi_1, \ldots, \phi_N \rvert \phi_I ) } 
    \end{equation}
\end{thm}

\begin{proof} Proving this statement is equivalent to checking the following equality,
\begin{equation}\label{R_grav_toprove}
    \partial_u \left( - z_I \sum_{\bt \in \mT_{N | I}} \phi_{\bt} \right) \overset{?}{=} \partial_{\bw} \left( \sum_{\bt \in \mT_{N | I}} \phi_{\bt} \right) + \Big\{ \partial_u \left( \sum_{\bt \in \mT_{N}} \phi_{\bt}\right) , \sum_{\bt \in \mT_{N  | I}} \phi_{\bt}\Big\}
\end{equation}
which is the first equation of \eqref{R_grav}. The second equation will follow by an analogous argument.

Let us think about which graphs appear in each of the three terms in the above equation. In the first term, by the product rule of differentiation, we simply sum over all graphs in $\mT_{N|I}$ where $-z_I \partial_u$ is to ``decorate'' any of the nodes in the graph. (If $-z_I \partial_u$ decorates the node $i$, this amounts to acting with $-z_I \partial_u^{(i)}$ on the corresponding graph in the perturbiner expansion.) In the second term, we similarly sum over all graphs in $\mT_{N|I}$ but decorate one node with $\partial_{\bw}$ instead of $-z_I \partial_u$. For the third term, we sum over all graphs which consist of a wiggly line pointing from a graph in $\mT_N$ with a node decorated by $\partial_u$ to a graph containing the node $I$. 

Notice that \emph{all} of the graphs described above have one thing in common: they all contain the node $I$, and they all contain one ``decorated'' node, whether that be by $-z_I \partial_u$, $\partial_\bw$, or $\partial_u$. Of course, some of the graphs also contain a wiggly line. Let us therefore partition the graphs in groups where all the graphs in each group would look the same if we were to convert any wiggly line into a straight line, and if we were to regard all possible ``decorations'' as the same. If we can check that all of the graphs contained in such a group satisfy \eqref{R_grav_toprove}, then we are done. An example for a typical set of graphs in such a group is given below in the first line.
\begin{center}
    \input{figures/gravity_graph_16}
\end{center}
Because all of our graphs are tree graphs, there is always exactly one path between any two nodes in the graph. Therefore, there is always one path between the node $I$ and the decorated node. When it comes to the graphs containing a wiggly line, this wiggly line will always lie somewhere on this path. 

We then simplify the RHS of the above graphical equation using two facts. The first fact is that $\partial_{\bw} \phi_i = - z_i \partial_u \phi_i$, \eqref{z_diff_relation}. The second fact is that a wiggly line going from $i$ to $j$ can be replaced with a straight line if we multiply by $z_{ij}$. (See Figures \ref{fig:7452634} and \ref{fig:45635}.) The RHS then becomes a numerical factor times the very graph on the LHS. 
However, if we sum up the numerical factors along the chain from $i$ to $I$, we see they equal $-z_I$, matching precisely with the LHS and completing the proof.
\end{proof}

\subsection{Linear ASD perturbations}

An incredible fact is that linearized \emph{anti-self-dual} perturbations on a self-dual background satisfy the exact same equation of motion as the linearized \emph{self-dual} perturbations! This fact will become very important when it comes time to add two negative helicity gravitons to our spacetime and compute the MHV amplitude.

In appendix \ref{appB} we review many basic facts about SDG, and we refer the reader there to find the proofs of the statements below. If we define $\Psi_{ABCD}$ to be the ASD part of the Riemann curvature tensor, then this object satisfies the Bianchi identity
\begin{equation}
\begin{aligned}
    0 = \nabla^A_{\;\;\dot A} \Psi_{ABCD} = (e^A_{\;\; \dot A})^\mu \partial_\mu \Psi_{ABCD} &-  (\Gamma_{A}^{\;\;\; E})^A_{\;\;\dot A} \Psi_{EBCD} -  (\Gamma_{B}^{\;\;\; E})^A_{\;\; \dot A} \Psi_{AECD} \\
    &-(\Gamma_{C}^{\;\;\; E})^A_{\;\;\dot A} \Psi_{ABED} -  (\Gamma_{D}^{\;\;\; E})^A_{\;\; \dot A} \Psi_{ABCE}
    \end{aligned}
\end{equation}
where $\Gamma_{AB}$ is the ASD part of the spin connection 1-form and $e_{A \dot A}$ denotes a basis of four vierbein vector fields.  If $\theta^{A \dot A} = \theta^{A \dot A}_\mu d x^\mu$ are a basis of tetrad 1-forms satisfying $ds^2 = \tfrac{1}{2} \vep_{AB} \vep_{\dot A \dot B} \theta^{A \dot A} \theta^{B \dot B}$, then the vierbein is given by $e_{A \dot A}^\mu = g^{\mu \nu} \vep_{A B} \vep_{\dot A \dot B} \theta^{B \dot B}_\nu$. For a SD metric of the form \eqref{eq2},  a convenient choice of tetrads is given by
\begin{equation}\label{conv_tetrad}
\begin{aligned}
    \theta^{1 \dot 1} = 2 d \bu \, , & & & & & & & & &\theta^{2 \dot 1} = 2( d w - (\partial_u \partial_w \phi ) d \bu - (\partial_u^2 \phi) d \bw ) \, ,\\
    \theta^{1 \dot 2} =2 d \bw \, ,& & & & & & & & &\theta^{2 \dot 2} = 2 (d u + ( \partial_w^2 \phi) d \bu + (\partial_u \partial_w \phi) d \bw ) \, ,
\end{aligned}
\end{equation}
with the associated vierbein being
\begin{equation}\label{conv_vierbein}
\begin{aligned}
    e^\mu_{1 \dot 1}\partial_\mu &=  \partial_{\bu} +  (\partial_u \partial_w \phi) \partial_w - (\partial_w^2 \phi) \partial_u \, , & & & & & & & & & e^\mu_{2 \dot 1}\partial_\mu &= \partial_w \, ,\\
    e^\mu_{1 \dot 2}\partial_\mu &=  \partial_{\bw} + (\partial_u^2 \phi)\partial_w - (\partial_u \partial_w \phi)\partial_u \, ,& & & & & & & & & e^\mu_{2 \dot 2}\partial_\mu &= \partial_u \, . 
\end{aligned}
\end{equation}

In a SD background, $(\Psi_0)_{ABCD} = 0$ and $(\Gamma_0)_{AB} = 0$ for our above choice of vierbein. If we then denote $\psi_{ABCD}$ to be a linearized perturbation around $(\Psi_0)_{ABCD}$, it therefore satisfies
\begin{equation}
     0 = (e^A_{\;\;\dot A})^\mu \partial_\mu \psi_{ABCD}
\end{equation}
where $e_{A\dot A}$ above is the vierbein of the SD background.

Plugging in the formula for $e_{A \dot A}$, this reduces to the two equations, for $\dot A = \dot 1, \dot 2$,
\begin{equation}\label{dotA12}
\begin{aligned}
    \dot A = \dot 1:& \hspace{1 cm} 0 = -\partial_w \psi_{1BCD} + \partial_\bu \psi_{2BCD} + \{ \partial_w \phi, \psi_{2BCD} \} \, ,\\
    \dot A = \dot 2:& \hspace{1 cm} 0 = -\partial_u \psi_{1BCD} + \partial_\bw \psi_{2BCD} + \{ \partial_u \phi, \psi_{2BCD} \} \, .
\end{aligned}
\end{equation}
After a comparison with \eqref{R_grav}, we identify that \eqref{dotA12} is simply the equation
\begin{equation}\label{psiR}
    \psi_{1BCD} = \mR \, \psi_{2BCD}.
\end{equation}
The two equations \eqref{dotA12} can be rearranged to show that both $\psi_{1BCD}$ and $\psi_{2BCD}$ satisfy \eqref{linearized_2nd_eq} (see \eqref{Eq38} and \eqref{Eq39}):
\begin{equation}\label{eq316}
    \Box \psi_{ABCD} - \{ \partial_u \psi_{ABCD}, \partial_w \phi \} - \{ \partial_u \phi, \partial_w \psi_{ABCD} \} = 0.
\end{equation}

Furthermore, Theorem \ref{thmRz} tells us that \eqref{psiR} implies $\psi_{1 BCD} = - z_I \psi_{2BCD}$, and the index symmetry $\psi_{ABCD} = \psi_{(ABCD)}$ implies
\begin{equation}\label{psi_I}
    \psi^{I}_{ABCD} =  -\frac{1}{8} \, \kappa^{I}_A \kappa^{I}_B \kappa^{I}_C \kappa^{I}_D \, \bigphi(\phi_1, \ldots, \phi_N \rvert \phi_I )
\end{equation}
will always be a solution to \eqref{dotA12}, where $\kappa^I_A = (-z_I, 1)$. In fact, this is the solution that corresponds to adding one ASD graviton $h_{\mu \nu} = \epsilon_I \vep_{\mu \nu}^{-,I} e^{i p_I \cdot X}$ to our SD background! The prefactor of $-1/8$ can be fixed by matching with the trivial $N = 0$ case where one ASD graviton is added to flat space.

\subsection{Writing linear perturbations as an exponential}

In this section we will show that a linear perturbation to the Plebański scalar can be written as the exponential of a certain set of graphs.

We begin by defining a new node shaped like a diamond. The difference between the diamond node and the circular node is that, for a diamond node, one is not instructed to multiply the overall expression by the corresponding seed function. Here we give two examples of the use of the diamond node, one with a straight edge and one with a wiggly edge. Notice the absence of the seed function $\phi_I$ in the following expressions.

\begin{center}
    \tikzset{every picture/.style={line width=0.75pt}} 

\begin{tikzpicture}[x=0.75pt,y=0.75pt,yscale=-1,xscale=1]

\draw    (118,104.7) -- (191,104.7) ;
\draw  [fill={rgb, 255:red, 255; green, 255; blue, 255 }  ,fill opacity=1 ] (118,87.02) -- (135.68,104.7) -- (118,122.38) -- (100.32,104.7) -- cycle ;
\draw    (363.83,104.7) .. controls (365.5,103.03) and (367.16,103.03) .. (368.83,104.7) .. controls (370.5,106.37) and (372.16,106.37) .. (373.83,104.7) .. controls (375.5,103.03) and (377.16,103.03) .. (378.83,104.7) .. controls (380.5,106.37) and (382.16,106.37) .. (383.83,104.7) .. controls (385.5,103.03) and (387.16,103.03) .. (388.83,104.7) .. controls (390.5,106.37) and (392.16,106.37) .. (393.83,104.7) .. controls (395.5,103.03) and (397.16,103.03) .. (398.83,104.7) .. controls (400.5,106.37) and (402.16,106.37) .. (403.83,104.7) .. controls (405.5,103.03) and (407.16,103.03) .. (408.83,104.7) .. controls (410.5,106.37) and (412.16,106.37) .. (413.83,104.7) .. controls (415.5,103.03) and (417.16,103.03) .. (418.83,104.7) .. controls (420.5,106.37) and (422.16,106.37) .. (423.83,104.7) .. controls (425.5,103.03) and (427.16,103.03) .. (428.83,104.7) .. controls (430.5,106.37) and (432.16,106.37) .. (433.83,104.7) -- (436.83,104.7) -- (436.83,104.7) ;
\draw [shift={(405.33,104.7)}, rotate = 180] [fill={rgb, 255:red, 0; green, 0; blue, 0 }  ][line width=0.08]  [draw opacity=0] (10.72,-5.15) -- (0,0) -- (10.72,5.15) -- (7.12,0) -- cycle    ;
\draw  [fill={rgb, 255:red, 255; green, 255; blue, 255 }  ,fill opacity=1 ] (424.33,104.7) .. controls (424.33,97.8) and (429.93,92.2) .. (436.83,92.2) .. controls (443.74,92.2) and (449.33,97.8) .. (449.33,104.7) .. controls (449.33,111.6) and (443.74,117.2) .. (436.83,117.2) .. controls (429.93,117.2) and (424.33,111.6) .. (424.33,104.7) -- cycle ;
\draw  [fill={rgb, 255:red, 255; green, 255; blue, 255 }  ,fill opacity=1 ] (178.5,104.7) .. controls (178.5,97.8) and (184.1,92.2) .. (191,92.2) .. controls (197.9,92.2) and (203.5,97.8) .. (203.5,104.7) .. controls (203.5,111.6) and (197.9,117.2) .. (191,117.2) .. controls (184.1,117.2) and (178.5,111.6) .. (178.5,104.7) -- cycle ;
\draw  [fill={rgb, 255:red, 255; green, 255; blue, 255 }  ,fill opacity=1 ] (363.83,87.02) -- (381.51,104.7) -- (363.83,122.38) -- (346.16,104.7) -- cycle ;

\draw (363.83,104.7) node    {$I$};
\draw (436.83,104.7) node    {$i$};
\draw (451.33,104.7) node [anchor=west] [inner sep=0.75pt]    {$\ =\ [ Ii] \  \phi _{i}$};
\draw (118,104.7) node    {$I$};
\draw (191,104.7) node    {$i$};
\draw (205.5,104.7) node [anchor=west] [inner sep=0.75pt]    {$\ =\ \dfrac{[ Ii]}{\langle Ii\rangle } \phi _{i}$};

\end{tikzpicture}
\end{center}
Note that we are now writing things in terms of the angle and square brackets from \eqref{angsqa}.  To do this, we've used
\begin{equation}
D^{ij} \phi_i \phi_j = [ i j ] \phi_i \phi_j.
\end{equation}

For future convenience, we now define a new set of graphs using this diamond node.

\begin{definition}\label{def7}
    Let $\mT^{\meddiamond}_{N| I}$ denote the set of all graphs which can be formed by taking a graph in $\mT_{N}$ and appending one extra diamond node labelled $I$ to the graph. Note that the diamond node will always have just one edge.
\end{definition}

\begin{ex} All of the elements of $\mT^{ \meddiamond}_{2 | I}$, and their corresponding terms, are
\begin{center}
\tikzset{every picture/.style={line width=0.75pt}} 

\begin{tikzpicture}[x=0.75pt,y=0.75pt,yscale=-1,xscale=1]

\draw    (150.33,38.67) -- (150.33,93.98) ;
\draw  [fill={rgb, 255:red, 255; green, 255; blue, 255 }  ,fill opacity=1 ] (150.33,20.99) -- (168.01,38.67) -- (150.33,56.34) -- (132.66,38.67) -- cycle ;
\draw  [fill={rgb, 255:red, 255; green, 255; blue, 255 }  ,fill opacity=1 ] (137.83,93.98) .. controls (137.83,87.08) and (143.43,81.48) .. (150.33,81.48) .. controls (157.24,81.48) and (162.83,87.08) .. (162.83,93.98) .. controls (162.83,100.88) and (157.24,106.48) .. (150.33,106.48) .. controls (143.43,106.48) and (137.83,100.88) .. (137.83,93.98) -- cycle ;
\draw    (424.25,38.67) -- (424.25,93.98) ;
\draw  [fill={rgb, 255:red, 255; green, 255; blue, 255 }  ,fill opacity=1 ] (424.25,20.99) -- (441.93,38.67) -- (424.25,56.34) -- (406.57,38.67) -- cycle ;
\draw  [fill={rgb, 255:red, 255; green, 255; blue, 255 }  ,fill opacity=1 ] (411.75,93.98) .. controls (411.75,87.08) and (417.35,81.48) .. (424.25,81.48) .. controls (431.15,81.48) and (436.75,87.08) .. (436.75,93.98) .. controls (436.75,100.88) and (431.15,106.48) .. (424.25,106.48) .. controls (417.35,106.48) and (411.75,100.88) .. (411.75,93.98) -- cycle ;
\draw    (368.93,195.25) -- (423.53,195.25) ;
\draw    (95.67,139.93) -- (95.67,195.25) ;
\draw    (95.67,195.25) -- (150.27,195.25) ;
\draw  [fill={rgb, 255:red, 255; green, 255; blue, 255 }  ,fill opacity=1 ] (83.17,195.25) .. controls (83.17,188.34) and (88.76,182.75) .. (95.67,182.75) .. controls (102.57,182.75) and (108.17,188.34) .. (108.17,195.25) .. controls (108.17,202.15) and (102.57,207.75) .. (95.67,207.75) .. controls (88.76,207.75) and (83.17,202.15) .. (83.17,195.25) -- cycle ;
\draw  [fill={rgb, 255:red, 255; green, 255; blue, 255 }  ,fill opacity=1 ] (137.77,195.25) .. controls (137.77,188.34) and (143.36,182.75) .. (150.27,182.75) .. controls (157.17,182.75) and (162.77,188.34) .. (162.77,195.25) .. controls (162.77,202.15) and (157.17,207.75) .. (150.27,207.75) .. controls (143.36,207.75) and (137.77,202.15) .. (137.77,195.25) -- cycle ;
\draw  [fill={rgb, 255:red, 255; green, 255; blue, 255 }  ,fill opacity=1 ] (95.67,122.26) -- (113.34,139.93) -- (95.67,157.61) -- (77.99,139.93) -- cycle ;
\draw    (423.53,139.93) -- (423.53,195.25) ;
\draw  [fill={rgb, 255:red, 255; green, 255; blue, 255 }  ,fill opacity=1 ] (411.03,195.25) .. controls (411.03,188.34) and (416.63,182.75) .. (423.53,182.75) .. controls (430.44,182.75) and (436.03,188.34) .. (436.03,195.25) .. controls (436.03,202.15) and (430.44,207.75) .. (423.53,207.75) .. controls (416.63,207.75) and (411.03,202.15) .. (411.03,195.25) -- cycle ;
\draw  [fill={rgb, 255:red, 255; green, 255; blue, 255 }  ,fill opacity=1 ] (423.53,122.26) -- (441.21,139.93) -- (423.53,157.61) -- (405.86,139.93) -- cycle ;
\draw  [fill={rgb, 255:red, 255; green, 255; blue, 255 }  ,fill opacity=1 ] (356.43,195.25) .. controls (356.43,188.34) and (362.03,182.75) .. (368.93,182.75) .. controls (375.84,182.75) and (381.43,188.34) .. (381.43,195.25) .. controls (381.43,202.15) and (375.84,207.75) .. (368.93,207.75) .. controls (362.03,207.75) and (356.43,202.15) .. (356.43,195.25) -- cycle ;

\draw (150.33,38.67) node    {$I$};
\draw (150.33,93.98) node    {$1$};
\draw (164.83,65.98) node [anchor=west] [inner sep=0.75pt]    {$=\ \dfrac{[ Ii]}{\langle 1i\rangle } \phi _{1} ,$};
\draw (424.25,38.67) node    {$I$};
\draw (424.25,93.98) node    {$2$};
\draw (438.75,65.48) node [anchor=west] [inner sep=0.75pt]    {$=\ \dfrac{[ Ii]}{\langle 2i\rangle } \phi _{2} ,$};
\draw (164.77,170.25) node [anchor=west] [inner sep=0.75pt]    {$=\ \dfrac{[ I1]}{\langle I1\rangle }\dfrac{[ 12]}{\langle 12\rangle } \phi _{1} \phi _{2} ,$};
\draw (438.03,170.25) node [anchor=west] [inner sep=0.75pt]    {$=\ \ \dfrac{[ I2]}{\langle I2\rangle }\dfrac{[ 12]}{\langle 12\rangle } \phi _{1} \phi _{2} .$};
\draw (95.67,195.25) node    {$1$};
\draw (150.27,195.25) node    {$2$};
\draw (95.67,139.93) node    {$I$};
\draw (368.93,195.25) node    {$1$};
\draw (423.53,195.25) node    {$2$};
\draw (423.53,139.93) node    {$I$};

\end{tikzpicture}
\end{center}

\end{ex}

We now show that if we exponentiate the graphs in $\mT^{\meddiamond}_{N| I}$ and multiply the result by $\phi_I$, we will get a sum of graphs which equals the linear perturbation to the Plebański scalar from \eqref{lin_graphs}.

\begin{claim}\label{claimexp}
\begin{equation}\label{expI}
    \sum_{\bt \, \in \, \mT_{N | I}} \phi_{\bt} = \phi_I \; \exp( \sum_{\bt \, \in \, \mT^{\meddiamond}_{N | I}} \phi_{\bt} )
\end{equation}
\end{claim}

\begin{proof}
    The above claim is perhaps best justified by seeing how it holds in a particular example, say when $N=2$, as its generalization to all $N$ follows in a relatively straightforward manner.
    
    We'll expand the RHS of \eqref{expI} and show it equals LHS:
    \begin{center}
        \input{figures/gravity_graph_25.tex}
    \end{center}
    The final expression is indeed a sum over $\mT_{2 | I}$. Note the use of the identity below.
    \begin{center}
        \tikzset{every picture/.style={line width=0.75pt}} 

\begin{tikzpicture}[x=0.75pt,y=0.75pt,yscale=-1,xscale=1]

\draw    (114.67,31.69) -- (114.67,87) ;
\draw  [fill={rgb, 255:red, 255; green, 255; blue, 255 }  ,fill opacity=1 ] (114.67,14.01) -- (132.34,31.69) -- (114.67,49.36) -- (96.99,31.69) -- cycle ;
\draw  [fill={rgb, 255:red, 255; green, 255; blue, 255 }  ,fill opacity=1 ] (102.17,87) .. controls (102.17,80.1) and (107.76,74.5) .. (114.67,74.5) .. controls (121.57,74.5) and (127.17,80.1) .. (127.17,87) .. controls (127.17,93.9) and (121.57,99.5) .. (114.67,99.5) .. controls (107.76,99.5) and (102.17,93.9) .. (102.17,87) -- cycle ;
\draw    (184.2,31.69) -- (184.2,87) ;
\draw  [fill={rgb, 255:red, 255; green, 255; blue, 255 }  ,fill opacity=1 ] (184.2,14.01) -- (201.88,31.69) -- (184.2,49.36) -- (166.52,31.69) -- cycle ;
\draw  [fill={rgb, 255:red, 255; green, 255; blue, 255 }  ,fill opacity=1 ] (171.7,87) .. controls (171.7,80.1) and (177.3,74.5) .. (184.2,74.5) .. controls (191.1,74.5) and (196.7,80.1) .. (196.7,87) .. controls (196.7,93.9) and (191.1,99.5) .. (184.2,99.5) .. controls (177.3,99.5) and (171.7,93.9) .. (171.7,87) -- cycle ;
\draw    (324.58,30.69) -- (289.67,86) ;
\draw  [fill={rgb, 255:red, 255; green, 255; blue, 255 }  ,fill opacity=1 ] (277.17,86) .. controls (277.17,79.1) and (282.76,73.5) .. (289.67,73.5) .. controls (296.57,73.5) and (302.17,79.1) .. (302.17,86) .. controls (302.17,92.9) and (296.57,98.5) .. (289.67,98.5) .. controls (282.76,98.5) and (277.17,92.9) .. (277.17,86) -- cycle ;
\draw    (324.58,30.69) -- (359.5,86) ;
\draw  [fill={rgb, 255:red, 255; green, 255; blue, 255 }  ,fill opacity=1 ] (347,86) .. controls (347,79.1) and (352.6,73.5) .. (359.5,73.5) .. controls (366.4,73.5) and (372,79.1) .. (372,86) .. controls (372,92.9) and (366.4,98.5) .. (359.5,98.5) .. controls (352.6,98.5) and (347,92.9) .. (347,86) -- cycle ;
\draw  [fill={rgb, 255:red, 255; green, 255; blue, 255 }  ,fill opacity=1 ] (312.08,30.69) .. controls (312.08,23.78) and (317.68,18.19) .. (324.58,18.19) .. controls (331.49,18.19) and (337.08,23.78) .. (337.08,30.69) .. controls (337.08,37.59) and (331.49,43.19) .. (324.58,43.19) .. controls (317.68,43.19) and (312.08,37.59) .. (312.08,30.69) -- cycle ;
\draw  [draw opacity=0] (207.33,14.34) .. controls (217.08,25.95) and (222.95,40.91) .. (222.95,57.25) .. controls (222.95,73.59) and (217.08,88.55) .. (207.33,100.16) -- (156.2,57.25) -- cycle ; \draw   (207.33,14.34) .. controls (217.08,25.95) and (222.95,40.91) .. (222.95,57.25) .. controls (222.95,73.59) and (217.08,88.55) .. (207.33,100.16) ;  
\draw  [draw opacity=0] (90.62,100.16) .. controls (80.87,88.55) and (75,73.59) .. (75,57.25) .. controls (75,40.91) and (80.87,25.95) .. (90.62,14.34) -- (141.75,57.25) -- cycle ; \draw   (90.62,100.16) .. controls (80.87,88.55) and (75,73.59) .. (75,57.25) .. controls (75,40.91) and (80.87,25.95) .. (90.62,14.34) ;  
\draw  [fill={rgb, 255:red, 255; green, 255; blue, 255 }  ,fill opacity=1 ] (16.17,57.91) .. controls (16.17,51.01) and (21.76,45.41) .. (28.67,45.41) .. controls (35.57,45.41) and (41.17,51.01) .. (41.17,57.91) .. controls (41.17,64.82) and (35.57,70.41) .. (28.67,70.41) .. controls (21.76,70.41) and (16.17,64.82) .. (16.17,57.91) -- cycle ;

\draw (114.67,31.69) node    {$I$};
\draw (114.67,87) node    {$1$};
\draw (184.2,31.69) node    {$I$};
\draw (184.2,87) node    {$2$};
\draw (149.82,57.34) node  [font=\Large]  {$\times $};
\draw (254.82,58.34) node  [font=\Large]  {$=$};
\draw (324.58,30.69) node    {$I$};
\draw (289.67,86) node    {$1$};
\draw (359.5,86) node    {$2$};
\draw (28.67,57.91) node    {$I$};
\draw (46.17,57.91) node [anchor=west] [inner sep=0.75pt]  [font=\Large]  {$\times $};

\end{tikzpicture}
    \end{center}
        
    In order to see how this computation generalizes to all $N$, note that as $\exp$ is Taylor expanded to the $n^{\rm th}$ power, all graphs for which the diamond node $I$ is connected to $n$ other nodes are generated. The final multiplication by $\phi_I$ turns the diamond node into a circular node. 
    
\end{proof}

Using Claim \ref{claimexp}, we can see that we can rewrite a linearized perturbation as
\begin{equation}\label{phi_pert_FI}
    \bigphi( \phi_1, \ldots, \phi_N \rvert \phi_{I} ) = \epsilon_I \exp( i F_I)
\end{equation}
where we have defined $F_I$ as
\begin{equation}\label{FI_graph}
    F_I \equiv p_I \cdot X - i \sum_{\bt \in \mT^{\meddiamond}_{N | I}} \phi_{\bt}.
\end{equation}

It is possible to define an object $F_I^{\dot A}$ such that
\begin{equation}\label{FIindex}
    F_I = \kappat^I_{\dot A} F_I^{\dot A}.
\end{equation}
One can then contract $F_I^{\dot A}$ with an arbitrary reference spinor $\xi_{\dot A}$ instead of $\widetilde{\kappa}^I_{\dot A}$. The object $\xi_{\dot A} F_I^{\dot A}$ will become important later on, so we will discuss here how it can be expressed graphically. To do this, we'll need to define a new type of edge. This edge is labelled by the reference spinor $\xi$, and will only ever connect to the special node $I$ (or $J$, which will label a second ASD graviton). Here is how our new $\xi$-edge works: if it connects $I$ to $i$, then the numerator of the factor corresponding to this edge is $[\xi i]$ instead of $[I i]$. Four examples, with both the straight/wiggly edges and normal/diamond nodes are shown below.

\begin{center}
    \tikzset{every picture/.style={line width=0.75pt}} 

\begin{tikzpicture}[x=0.75pt,y=0.75pt,yscale=-1,xscale=1]

\draw    (72,157.25) -- (145,157.25) ;
\draw  [fill={rgb, 255:red, 255; green, 255; blue, 255 }  ,fill opacity=1 ] (72,139.57) -- (89.68,157.25) -- (72,174.93) -- (54.32,157.25) -- cycle ;
\draw    (327,77.7) .. controls (328.67,76.03) and (330.33,76.03) .. (332,77.7) .. controls (333.67,79.37) and (335.33,79.37) .. (337,77.7) .. controls (338.67,76.03) and (340.33,76.03) .. (342,77.7) .. controls (343.67,79.37) and (345.33,79.37) .. (347,77.7) .. controls (348.67,76.03) and (350.33,76.03) .. (352,77.7) .. controls (353.67,79.37) and (355.33,79.37) .. (357,77.7) .. controls (358.67,76.03) and (360.33,76.03) .. (362,77.7) .. controls (363.67,79.37) and (365.33,79.37) .. (367,77.7) .. controls (368.67,76.03) and (370.33,76.03) .. (372,77.7) .. controls (373.67,79.37) and (375.33,79.37) .. (377,77.7) .. controls (378.67,76.03) and (380.33,76.03) .. (382,77.7) .. controls (383.67,79.37) and (385.33,79.37) .. (387,77.7) .. controls (388.67,76.03) and (390.33,76.03) .. (392,77.7) .. controls (393.67,79.37) and (395.33,79.37) .. (397,77.7) -- (400,77.7) -- (400,77.7) ;
\draw [shift={(368.5,77.7)}, rotate = 180] [fill={rgb, 255:red, 0; green, 0; blue, 0 }  ][line width=0.08]  [draw opacity=0] (10.72,-5.15) -- (0,0) -- (10.72,5.15) -- (7.12,0) -- cycle    ;
\draw  [fill={rgb, 255:red, 255; green, 255; blue, 255 }  ,fill opacity=1 ] (387.5,77.7) .. controls (387.5,70.8) and (393.1,65.2) .. (400,65.2) .. controls (406.9,65.2) and (412.5,70.8) .. (412.5,77.7) .. controls (412.5,84.6) and (406.9,90.2) .. (400,90.2) .. controls (393.1,90.2) and (387.5,84.6) .. (387.5,77.7) -- cycle ;
\draw  [fill={rgb, 255:red, 255; green, 255; blue, 255 }  ,fill opacity=1 ] (314.5,77.7) .. controls (314.5,70.8) and (320.1,65.2) .. (327,65.2) .. controls (333.9,65.2) and (339.5,70.8) .. (339.5,77.7) .. controls (339.5,84.6) and (333.9,90.2) .. (327,90.2) .. controls (320.1,90.2) and (314.5,84.6) .. (314.5,77.7) -- cycle ;
\draw  [fill={rgb, 255:red, 255; green, 255; blue, 255 }  ,fill opacity=1 ] (132.5,157.25) .. controls (132.5,150.34) and (138.1,144.75) .. (145,144.75) .. controls (151.9,144.75) and (157.5,150.34) .. (157.5,157.25) .. controls (157.5,164.15) and (151.9,169.75) .. (145,169.75) .. controls (138.1,169.75) and (132.5,164.15) .. (132.5,157.25) -- cycle ;
\draw    (327,158.25) .. controls (328.67,156.58) and (330.33,156.58) .. (332,158.25) .. controls (333.67,159.92) and (335.33,159.92) .. (337,158.25) .. controls (338.67,156.58) and (340.33,156.58) .. (342,158.25) .. controls (343.67,159.92) and (345.33,159.92) .. (347,158.25) .. controls (348.67,156.58) and (350.33,156.58) .. (352,158.25) .. controls (353.67,159.92) and (355.33,159.92) .. (357,158.25) .. controls (358.67,156.58) and (360.33,156.58) .. (362,158.25) .. controls (363.67,159.92) and (365.33,159.92) .. (367,158.25) .. controls (368.67,156.58) and (370.33,156.58) .. (372,158.25) .. controls (373.67,159.92) and (375.33,159.92) .. (377,158.25) .. controls (378.67,156.58) and (380.33,156.58) .. (382,158.25) .. controls (383.67,159.92) and (385.33,159.92) .. (387,158.25) .. controls (388.67,156.58) and (390.33,156.58) .. (392,158.25) .. controls (393.67,159.92) and (395.33,159.92) .. (397,158.25) -- (400,158.25) -- (400,158.25) ;
\draw [shift={(368.5,158.25)}, rotate = 180] [fill={rgb, 255:red, 0; green, 0; blue, 0 }  ][line width=0.08]  [draw opacity=0] (10.72,-5.15) -- (0,0) -- (10.72,5.15) -- (7.12,0) -- cycle    ;
\draw  [fill={rgb, 255:red, 255; green, 255; blue, 255 }  ,fill opacity=1 ] (327,140.57) -- (344.68,158.25) -- (327,175.93) -- (309.32,158.25) -- cycle ;
\draw  [fill={rgb, 255:red, 255; green, 255; blue, 255 }  ,fill opacity=1 ] (387.5,158.25) .. controls (387.5,151.34) and (393.1,145.75) .. (400,145.75) .. controls (406.9,145.75) and (412.5,151.34) .. (412.5,158.25) .. controls (412.5,165.15) and (406.9,170.75) .. (400,170.75) .. controls (393.1,170.75) and (387.5,165.15) .. (387.5,158.25) -- cycle ;
\draw    (72,76.7) -- (145,76.7) ;
\draw  [fill={rgb, 255:red, 255; green, 255; blue, 255 }  ,fill opacity=1 ] (132.5,76.7) .. controls (132.5,69.8) and (138.1,64.2) .. (145,64.2) .. controls (151.9,64.2) and (157.5,69.8) .. (157.5,76.7) .. controls (157.5,83.6) and (151.9,89.2) .. (145,89.2) .. controls (138.1,89.2) and (132.5,83.6) .. (132.5,76.7) -- cycle ;
\draw  [fill={rgb, 255:red, 255; green, 255; blue, 255 }  ,fill opacity=1 ] (59.5,76.7) .. controls (59.5,69.8) and (65.1,64.2) .. (72,64.2) .. controls (78.9,64.2) and (84.5,69.8) .. (84.5,76.7) .. controls (84.5,83.6) and (78.9,89.2) .. (72,89.2) .. controls (65.1,89.2) and (59.5,83.6) .. (59.5,76.7) -- cycle ;

\draw (327,77.7) node    {$I$};
\draw (400,77.7) node    {$i$};
\draw (363.5,71.3) node [anchor=south] [inner sep=0.75pt]    {$\xi $};
\draw (414.5,77.7) node [anchor=west] [inner sep=0.75pt]    {$\ =\ [ \xi i] \phi _{I} \phi _{i}$};
\draw (72,157.25) node    {$I$};
\draw (145,157.25) node    {$i$};
\draw (108.5,151.85) node [anchor=south] [inner sep=0.75pt]    {$\xi $};
\draw (159.5,157.25) node [anchor=west] [inner sep=0.75pt]    {$\ =\ \dfrac{[ \xi i]}{\langle Ii\rangle } \phi _{i}$};
\draw (327,158.25) node    {$I$};
\draw (400,158.25) node    {$i$};
\draw (363.5,151.85) node [anchor=south] [inner sep=0.75pt]    {$\xi $};
\draw (414.5,158.25) node [anchor=west] [inner sep=0.75pt]    {$\ =\ [ \xi i] \phi _{i}$};
\draw (72,76.7) node    {$I$};
\draw (145,76.7) node    {$i$};
\draw (108.5,71.3) node [anchor=south] [inner sep=0.75pt]    {$\xi $};
\draw (159.5,76.7) node [anchor=west] [inner sep=0.75pt]    {$\ =\ \dfrac{[ \xi i]}{\langle Ii\rangle } \phi _{I} \phi _{i}$};

\end{tikzpicture}
\end{center}

\noindent Note that in the special case $\xi_{\dot A} = \kappat^I_{\dot A}$, this $\xi$-edge reduces to a normal edge.

\begin{center}
    \tikzset{every picture/.style={line width=0.75pt}} 

\begin{tikzpicture}[x=0.75pt,y=0.75pt,yscale=-1,xscale=1]

\draw    (214,75.7) -- (287,75.7) ;
\draw  [fill={rgb, 255:red, 255; green, 255; blue, 255 }  ,fill opacity=1 ] (274.5,75.7) .. controls (274.5,68.8) and (280.1,63.2) .. (287,63.2) .. controls (293.9,63.2) and (299.5,68.8) .. (299.5,75.7) .. controls (299.5,82.6) and (293.9,88.2) .. (287,88.2) .. controls (280.1,88.2) and (274.5,82.6) .. (274.5,75.7) -- cycle ;
\draw  [fill={rgb, 255:red, 255; green, 255; blue, 255 }  ,fill opacity=1 ] (201.5,75.7) .. controls (201.5,68.8) and (207.1,63.2) .. (214,63.2) .. controls (220.9,63.2) and (226.5,68.8) .. (226.5,75.7) .. controls (226.5,82.6) and (220.9,88.2) .. (214,88.2) .. controls (207.1,88.2) and (201.5,82.6) .. (201.5,75.7) -- cycle ;
\draw    (350,76.7) -- (423,76.7) ;
\draw  [fill={rgb, 255:red, 255; green, 255; blue, 255 }  ,fill opacity=1 ] (410.5,76.7) .. controls (410.5,69.8) and (416.1,64.2) .. (423,64.2) .. controls (429.9,64.2) and (435.5,69.8) .. (435.5,76.7) .. controls (435.5,83.6) and (429.9,89.2) .. (423,89.2) .. controls (416.1,89.2) and (410.5,83.6) .. (410.5,76.7) -- cycle ;
\draw  [fill={rgb, 255:red, 255; green, 255; blue, 255 }  ,fill opacity=1 ] (337.5,76.7) .. controls (337.5,69.8) and (343.1,64.2) .. (350,64.2) .. controls (356.9,64.2) and (362.5,69.8) .. (362.5,76.7) .. controls (362.5,83.6) and (356.9,89.2) .. (350,89.2) .. controls (343.1,89.2) and (337.5,83.6) .. (337.5,76.7) -- cycle ;

\draw (214,75.7) node    {$I$};
\draw (287,75.7) node    {$i$};
\draw (252.22,70) node [anchor=south] [inner sep=0.75pt]    {$\tilde{\kappa }^{I}$};
\draw (317.77,77.6) node    {$=$};
\draw (350,76.7) node    {$I$};
\draw (423,76.7) node    {$i$};

\end{tikzpicture}
\end{center}

\noindent With the $\xi$-edge, we will now define a new set of graphs which will let us compute $\xi_{\dot A} F_I^{\dot A}$.
\begin{definition}
    Let $\mT^{\meddiamond, \xi}_{N|I}$ be the set of all graphs which can be formed by taking a graph in $\mT_N$ and appending one diamond $I$ node to the graph with a $\xi$-edge.
\end{definition}
\begin{exs}
\begin{center} Here we give two examples of graphs in $\mT^{\meddiamond, \xi}_{N|I}$ for $N = 9$. 
    \tikzset{every picture/.style={line width=0.75pt}} 

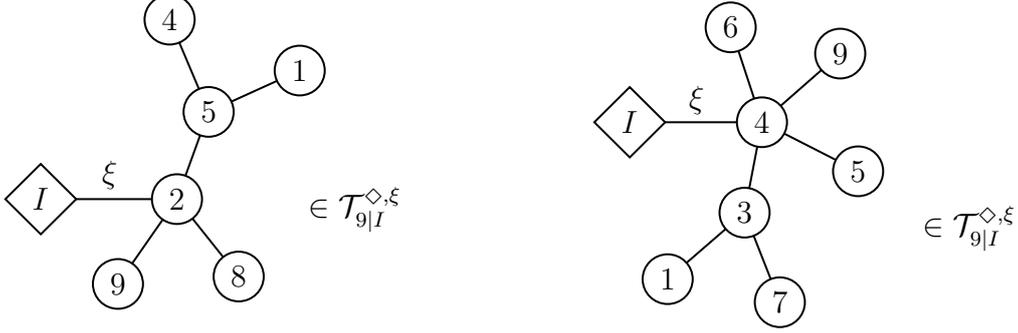
\begin{tikzpicture}[x=0.75pt,y=0.75pt,yscale=-1,xscale=1]

\draw    (428.33,164.05) -- (389.93,197.51) ;
\draw    (445.93,209.11) -- (428.33,164.05) ;
\draw    (437,118.58) -- (428.33,164.05) ;
\draw    (437,118.58) -- (421.4,70.85) ;
\draw    (485,143.25) -- (437,118.58) ;
\draw    (476.07,84.05) -- (437,118.58) ;
\draw    (175.8,196.18) -- (145,157.25) ;
\draw    (145,157.25) -- (115.53,199.91) ;
\draw    (160.8,113.38) -- (145,157.25) ;
\draw    (206.27,92.58) -- (160.8,113.38) ;
\draw    (141.33,66.98) -- (160.8,113.38) ;
\draw    (77,157.25) -- (150,157.25) ;
\draw  [fill={rgb, 255:red, 255; green, 255; blue, 255 }  ,fill opacity=1 ] (77,139.57) -- (94.68,157.25) -- (77,174.93) -- (59.32,157.25) -- cycle ;
\draw  [fill={rgb, 255:red, 255; green, 255; blue, 255 }  ,fill opacity=1 ] (132.5,157.25) .. controls (132.5,150.34) and (138.1,144.75) .. (145,144.75) .. controls (151.9,144.75) and (157.5,150.34) .. (157.5,157.25) .. controls (157.5,164.15) and (151.9,169.75) .. (145,169.75) .. controls (138.1,169.75) and (132.5,164.15) .. (132.5,157.25) -- cycle ;
\draw  [fill={rgb, 255:red, 255; green, 255; blue, 255 }  ,fill opacity=1 ] (148.3,113.38) .. controls (148.3,106.48) and (153.9,100.88) .. (160.8,100.88) .. controls (167.7,100.88) and (173.3,106.48) .. (173.3,113.38) .. controls (173.3,120.28) and (167.7,125.88) .. (160.8,125.88) .. controls (153.9,125.88) and (148.3,120.28) .. (148.3,113.38) -- cycle ;
\draw  [fill={rgb, 255:red, 255; green, 255; blue, 255 }  ,fill opacity=1 ] (128.83,66.98) .. controls (128.83,60.08) and (134.43,54.48) .. (141.33,54.48) .. controls (148.24,54.48) and (153.83,60.08) .. (153.83,66.98) .. controls (153.83,73.88) and (148.24,79.48) .. (141.33,79.48) .. controls (134.43,79.48) and (128.83,73.88) .. (128.83,66.98) -- cycle ;
\draw  [fill={rgb, 255:red, 255; green, 255; blue, 255 }  ,fill opacity=1 ] (193.77,92.58) .. controls (193.77,85.68) and (199.36,80.08) .. (206.27,80.08) .. controls (213.17,80.08) and (218.77,85.68) .. (218.77,92.58) .. controls (218.77,99.48) and (213.17,105.08) .. (206.27,105.08) .. controls (199.36,105.08) and (193.77,99.48) .. (193.77,92.58) -- cycle ;
\draw  [fill={rgb, 255:red, 255; green, 255; blue, 255 }  ,fill opacity=1 ] (103.03,199.91) .. controls (103.03,193.01) and (108.63,187.41) .. (115.53,187.41) .. controls (122.44,187.41) and (128.03,193.01) .. (128.03,199.91) .. controls (128.03,206.82) and (122.44,212.41) .. (115.53,212.41) .. controls (108.63,212.41) and (103.03,206.82) .. (103.03,199.91) -- cycle ;
\draw  [fill={rgb, 255:red, 255; green, 255; blue, 255 }  ,fill opacity=1 ] (163.3,196.18) .. controls (163.3,189.28) and (168.9,183.68) .. (175.8,183.68) .. controls (182.7,183.68) and (188.3,189.28) .. (188.3,196.18) .. controls (188.3,203.08) and (182.7,208.68) .. (175.8,208.68) .. controls (168.9,208.68) and (163.3,203.08) .. (163.3,196.18) -- cycle ;
\draw    (371,118.58) -- (437,118.58) ;
\draw  [fill={rgb, 255:red, 255; green, 255; blue, 255 }  ,fill opacity=1 ] (371,100.9) -- (388.68,118.58) -- (371,136.26) -- (353.32,118.58) -- cycle ;
\draw  [fill={rgb, 255:red, 255; green, 255; blue, 255 }  ,fill opacity=1 ] (424.5,118.58) .. controls (424.5,111.68) and (430.1,106.08) .. (437,106.08) .. controls (443.9,106.08) and (449.5,111.68) .. (449.5,118.58) .. controls (449.5,125.48) and (443.9,131.08) .. (437,131.08) .. controls (430.1,131.08) and (424.5,125.48) .. (424.5,118.58) -- cycle ;
\draw  [fill={rgb, 255:red, 255; green, 255; blue, 255 }  ,fill opacity=1 ] (415.83,164.05) .. controls (415.83,157.14) and (421.43,151.55) .. (428.33,151.55) .. controls (435.24,151.55) and (440.83,157.14) .. (440.83,164.05) .. controls (440.83,170.95) and (435.24,176.55) .. (428.33,176.55) .. controls (421.43,176.55) and (415.83,170.95) .. (415.83,164.05) -- cycle ;
\draw  [fill={rgb, 255:red, 255; green, 255; blue, 255 }  ,fill opacity=1 ] (408.9,70.85) .. controls (408.9,63.94) and (414.5,58.35) .. (421.4,58.35) .. controls (428.3,58.35) and (433.9,63.94) .. (433.9,70.85) .. controls (433.9,77.75) and (428.3,83.35) .. (421.4,83.35) .. controls (414.5,83.35) and (408.9,77.75) .. (408.9,70.85) -- cycle ;
\draw  [fill={rgb, 255:red, 255; green, 255; blue, 255 }  ,fill opacity=1 ] (463.57,84.05) .. controls (463.57,77.14) and (469.16,71.55) .. (476.07,71.55) .. controls (482.97,71.55) and (488.57,77.14) .. (488.57,84.05) .. controls (488.57,90.95) and (482.97,96.55) .. (476.07,96.55) .. controls (469.16,96.55) and (463.57,90.95) .. (463.57,84.05) -- cycle ;
\draw  [fill={rgb, 255:red, 255; green, 255; blue, 255 }  ,fill opacity=1 ] (472.5,143.25) .. controls (472.5,136.34) and (478.1,130.75) .. (485,130.75) .. controls (491.9,130.75) and (497.5,136.34) .. (497.5,143.25) .. controls (497.5,150.15) and (491.9,155.75) .. (485,155.75) .. controls (478.1,155.75) and (472.5,150.15) .. (472.5,143.25) -- cycle ;
\draw  [fill={rgb, 255:red, 255; green, 255; blue, 255 }  ,fill opacity=1 ] (377.43,197.51) .. controls (377.43,190.61) and (383.03,185.01) .. (389.93,185.01) .. controls (396.84,185.01) and (402.43,190.61) .. (402.43,197.51) .. controls (402.43,204.42) and (396.84,210.01) .. (389.93,210.01) .. controls (383.03,210.01) and (377.43,204.42) .. (377.43,197.51) -- cycle ;
\draw  [fill={rgb, 255:red, 255; green, 255; blue, 255 }  ,fill opacity=1 ] (433.43,209.11) .. controls (433.43,202.21) and (439.03,196.61) .. (445.93,196.61) .. controls (452.84,196.61) and (458.43,202.21) .. (458.43,209.11) .. controls (458.43,216.02) and (452.84,221.61) .. (445.93,221.61) .. controls (439.03,221.61) and (433.43,216.02) .. (433.43,209.11) -- cycle ;

\draw (77,157.25) node    {$I$};
\draw (145,157.25) node    {$2$};
\draw (111.5,153.85) node [anchor=south] [inner sep=0.75pt]    {$\xi $};
\draw (371,118.58) node    {$I$};
\draw (404.5,116.18) node [anchor=south] [inner sep=0.75pt]    {$\xi $};
\draw (160.8,113.38) node    {$5$};
\draw (206.27,92.58) node    {$1$};
\draw (115.53,199.91) node    {$9$};
\draw (175.8,196.18) node    {$8$};
\draw (141.33,66.98) node    {$4$};
\draw (428.33,164.05) node    {$3$};
\draw (437,118.58) node    {$4$};
\draw (476.07,84.05) node    {$9$};
\draw (389.93,197.51) node    {$1$};
\draw (485,143.25) node    {$5$};
\draw (421.4,70.85) node    {$6$};
\draw (445.93,209.11) node    {$7$};
\draw (208.93,149.67) node [anchor=north west][inner sep=0.75pt]    {$\in \mathcal{T}_{9|I}^{\meddiamond ,\xi }$};
\draw (515.6,158.93) node [anchor=north west][inner sep=0.75pt]    {$\in \mathcal{T}_{9|I}^{\meddiamond ,\xi }$};

\end{tikzpicture}
\end{center}
\end{exs}
\noindent We can now write $\xi_{\dot A} F_I^{\dot A}$ as a sum over $\mT^{\meddiamond, \xi}_{N|I}$:
\begin{equation}\label{FIA}
    \xi_{\dot A} F_I^{\dot A} = \xi_{\dot 1} ( w - z_I \bu) + \xi_{\dot 2}( u - z_I \bw ) - i \sum_{\bt \in \mT^{\meddiamond , \xi}_{N | I}} \phi_{\bt}.
\end{equation}
Because the graphs in $\mT^{\meddiamond, \xi}_{N|I}$ each contain only one $\xi$-edge, the RHS really can be written in the form of $\xi_{\dot A} F^{\dot A}_I$ for some $F^{\dot A}_I$.

As a consistency check, if we plug $\xi = \kappat^I$ into the above equation, we recover $F_I = \kappat_{\dot A}^I F_I^{\dot A}$ as given in \eqref{FI_graph}, as we should. This is because
\begin{equation}
    \mT^{\meddiamond, (\xi = \kappat^I)}_{N|I} = \mT^{\meddiamond}_{N|I}
\end{equation}
and, using $\kappat^I_{\dot A} = \omega_I(-z_I, 1)$ and \eqref{pdotX1},
\begin{equation}
    p_I \cdot X = - \omega_I \bz_I(w - z_I \bu) + \omega_I (u - z_I \bw).
\end{equation}

\section{Scattering amplitudes and on-shell actions}\label{sec4}

In this section we briefly review the connection between tree-level scattering amplitudes and on-shell actions. The general idea is that ``in'' and ``out'' states can be thought of as boundary conditions for the path integral, and by the stationary phase approximation, if we plug the classical solution with said boundary conditions into the action we compute the transition amplitude to the leading order in $\hbar$.\footnote{There is another way that classical solutions can be used to compute scattering amplitudes in which one leg of a Feynman diagram is taken to be off shell \cite{Boulware:1968zz,Monteiro:2011pc}. } A formalism for computing amplitudes in this way was developed by Arefeva, Faddeev, and Slavnov \cite{Arefeva:1974jv} and was revisited recently in \cite{Kim:2023qbl, Kraus:2024gso}. Further pedagogical treatments can be found in \cite{Shrauner:1977sk,Balian:1976vq}. This formalism is often used to calculate amplitudes in non-trivial backgrounds \cite{Adamo:2017nia, Adamo:2021rfq, Gonzo:2022tjm, Adamo:2020yzi, Adamo:2023fbj, Adamo:2024xpc}. See also \cite{Fabbrichesi:1993kz}.

For simplicity, say we have a scalar field $\varphi$ which has some self-interaction.  A scattering amplitude is defined as
\begin{equation}\label{amplitude_usual}
    \mathcal{A}(p_1, \ldots, p_N) = \bra{0} \hat{a}_{p_1} \ldots \hat{a}_{p_M} \, \mathcal{S} \, \hat{a}^\dagger_{p_{M+1}} \ldots \hat{a}^\dagger_{p_{N}} \ket{0}.
\end{equation}

One can also define finite-amplitude ``coherent states''
\begin{equation}
\begin{aligned}
    \bra{A_1, \ldots, A_M} &= \bra{0} e^{A_1 \hat{a}_{p_1}} \ldots  e^{A_M \hat{a}_{p_M}} 
    \, , \\
    \ket{A_{M+1}, \ldots, A_N} &= e^{A_{M+1} \hat{a}_{p_{M+1}}^\dagger} \ldots  e^{A_N \hat{a}_{p_N}^\dagger} \ket{0} \, ,
\end{aligned}
\end{equation}
which correspond, classically, to fields that are sums of plane waves with finite amplitudes $A_i$.  In a free theory, these are the solutions
\begin{equation}\label{varphi_free}
    \varphi(X) = \sum_{i = 1}^N A_i e^{i p_i \cdot X}.
\end{equation}
Here we are using a convention in which the energy of $p_i^\mu$ is positive (negative) for ingoing (outgoing) particles. Therefore, in the above field, the positive frequency part of $\varphi$ which we label $\varphi^+$ is determined by the ingoing momenta while the negative frequency part $\varphi^-$ is determined by the outgoing momenta.

Using the coherent state path integral, one can in principle compute transition amplitudes between such states in the interacting theory. In the coherent state path integral, the field ``paths'' one is integrating over are naturally complexified, and it turns out that the initial state is given by specifying the positive-frequency part of the field at $t = -T/2$ and the final state is given by specifying the negative-frequency part of the field at $t = T/2$:
\begin{equation}\label{mixed_bdy}
\begin{aligned}
    \lim_{T \to \infty} \varphi^-(+T/2) &= \sum_{i=1}^{M} A_i e^{i p_i X} , \hspace{0.5 cm}
    \lim_{T \to \infty} \varphi^+(-T/2) &= \sum_{i=M+1}^N A_i e^{i p_i X}  .
\end{aligned}
\end{equation}
One must be conscientious about what boundary terms are included in the action $S$ because the fields in the path integral have non-trivial boundary conditions. In principle the boundary terms should be fixed by demanding that the solutions to the e.o.m.\ render the action stationary given these boundary conditions, but they can also be determined by other means, like by using time-slicing regularization.
\begin{equation}
    \begin{aligned}
        & \bra{A_1, \ldots, A_M} \mathcal{S}  \ket{A_{M+1}, \ldots, A_N} = \lim_{T \to \infty} \int_{     \varphi^+(-T/2) = \sum_{i = M+1}^N A_i e^{i p_i X}     }^{\varphi^-(T/2) = \sum_{i = 1}^M A_i e^{i p_i X}} \mathcal{D} \varphi \; e^{i S[\varphi]}
    \end{aligned}
\end{equation}

In the stationary phase approximation, one is instructed to find the $\varphi$ that solves the equation of motion with boundary conditions \eqref{mixed_bdy} and plug it into the action. Then, to compute a standard scattering amplitude of elementary quanta at tree-level, one can simply differentiate the on-shell integrand by all of the $A_i$'s as
\begin{equation}\label{A_exponential}
    \mathcal{A}(p_1, \ldots, p_N) = \frac{\partial^N}{\partial A_1 \ldots \partial A_N} e^{i S[ \varphi]} \Big\rvert_{A_i = 0 \; \forall A_i}.
\end{equation}
Note that because we differentiate by $A_i$ and then evaluate $A_i = 0$, for all intents and purposes we can replace $A_i \mapsto \epsilon_i$ such that $(\epsilon_i)^2 = 0$ when we solve the classical e.o.m.. This justifies our prescription \eqref{epsilon_manual}.

Furthermore, as $A_i \to 0$, the boundary conditions are dominated by the leading ``free'' solution \eqref{varphi_free}, as higher order terms in the solution necessary to solve the interacting e.o.m.\ will be higher order in the $\epsilon_i$'s. Therefore we are just instructed to find a perturbative solution of the form
\begin{equation}
\varphi = \sum_{i = 1}^N \epsilon_i e^{i p_i X} + \sum_{i < j} \epsilon_i \epsilon_j (\ldots) e^{i (p_i + p_j) \cdot X} + \ldots
\end{equation}
which is nothing but a perturbiner expansion, and plug it into the action.

Upon Taylor expanding the exponential in \eqref{A_exponential}, there will be a zeroth order term 1 which will not contribute to the amplitude, a first order term $iS$, and higher order terms $(iS)^n$. The higher order terms are clearly products of lower-point amplitudes, so if we're only interested in the connected part of the amplitude we only need to evaluate the action itself:
\begin{equation}\label{A_S_connected}
    \mathcal{A}_{\rm connected}(p_1, \ldots, p_N) = \frac{\partial^N}{\partial \epsilon_1 \ldots \partial \epsilon_N} i S[\varphi]  \Big\rvert_{\epsilon_i = 0 \; \forall \epsilon_i} .
\end{equation}
Henceforth, it will be implicitly understood we're only interested in connected amplitudes and will drop the explicit subscript.

\section{The Gravitational MHV Amplitude}\label{sec5}

\subsection{MHV scattering and the Plebański Action}\label{sec4_3}

From the last section, we know that we can compute the graviton MHV amplitude in two steps. First, we must solve for the classical spacetime containing $N$ positive helicity gravitons and 2 negative helicity gravitons. Second, we must plug this spacetime into the action. Actually, using a standard trick we can circumvent having to do some of step 1, which we'll explain shortly.

The commonly used Einstein-Hilbert (EH) action (with possible boundary terms included) is ill suited for our task because its dependence on the metric is very complicated. A better action for us is the so-called chiral Plebański action, which we'll now review.

If $\theta^{A \dot A} = \theta^{A \dot A}_\mu d x^\mu$ is a basis of tetrad 1-forms,  we can define a basis of three ASD 2-forms $\Sigma^{AB} = \Sigma^{(AB)}$ and three SD 2-forms $\widetilde{\Sigma}^{\dot A \dot B} = \widetilde{\Sigma}^{(\dot A \dot B)}$
\begin{equation}
    \Sigma^{AB} \equiv \theta^{A \dot A} \wedge \theta^B_{\;\; \dot A} \, , \hspace{1 cm} \widetilde{\Sigma}^{\dot A \dot B} \equiv \theta^{A \dot A} \wedge \theta_{A}^{\;\; \dot B}\, .
\end{equation}We'll also take the spin-connection 1-form $\Gamma_{ab} = \Gamma_{ab \mu} d x^\mu$, satisfying $\Gamma_{ab} = \Gamma_{[ab]}$, and break it up into ASD and SD parts as 
\begin{equation}\Gamma_{A\dot A B \dot B} = 2\, \vep_{\dot A \dot B} \Gamma_{A B} + 2\, \vep_{AB} \widetilde{\Gamma}_{\dot A \dot B}
\end{equation}
with $\Gamma_{AB} = \Gamma_{(AB)}$ and $\widetilde{\Gamma}_{\dot A \dot B} = \widetilde{\Gamma}_{(\dot A \dot B)}$.
From the spin connection one can define the Riemann curvature 2-form $R_{ab} = \dd \Gamma_{ab} + \Gamma_{a c} \wedge \Gamma^{c}_{\;\; b}$ and similarly break it into ASD and SD parts $R_{AB} = R_{(AB)}$ and $\widetilde{R}_{\dot A \dot B} = \widetilde{R}_{(\dot A \dot B)}$. These objects are related to the spin connection by 
\begin{equation}
\begin{aligned}
    R_{AB} &= \dd \Gamma_{AB} + \Gamma_{A C} \wedge \Gamma^{C}_{\;\;\;B} \, , \\
    \widetilde{R}_{\dot A\dot B} &= \dd \widetilde{\Gamma}_{\dot A \dot B} + \widetilde{\Gamma}_{\dot A \dot C} \wedge \widetilde{\Gamma}^{\dot C}_{\;\;\; \dot B}\,.
\end{aligned}
\end{equation}
If one uses the vacuum Einstein equation, $R_{AB}$ and $\widetilde{R}_{\dot A \dot B}$ can also be expressed as
\begin{equation}
\begin{aligned}
    R_{AB} &= \Psi_{ABCD} \Sigma^{AB} \, ,\\
    \widetilde{R}_{\dot A \dot B} &= \widetilde{\Psi}_{\dot A \dot B \dot C \dot D} \widetilde{\Sigma}^{\dot C \dot D}\, ,
\end{aligned}
\end{equation}
using the scalars $\Psi_{ABCD} = \Psi_{(ABCD)}$ and $\widetilde{\Psi}_{\dot A \dot B \dot C \dot D} = \widetilde{\Psi}_{(\dot A \dot B \dot C \dot D)}$. See appendix \ref{appB} for a more thorough review of these objects.

The EH Lagrangian can be expressed with
\begin{equation}
\begin{aligned}
    2 \, R \, \mathrm{d vol} &= \vep_{abcd} R^{ab} \wedge \theta^c \wedge \theta^d \\
    &= i \; R_{AB} \wedge \Sigma^{AB} - i\; \widetilde{R}_{\dot A \dot B} \wedge \widetilde{\Sigma}^{\dot A \dot B}
\end{aligned}
\end{equation}
where we have used the identity for the Levi-Civita pseudotensor
\begin{equation}
    \vep_{A\dot A B \dot B C \dot C D \dot D} = 4 i\, \vep_{AC} \vep_{BD} \vep_{\dot A \dot B} \vep_{\dot C \dot D} - 4 i \, \vep_{AB} \vep_{CD} \vep_{\dot A \dot C} \vep_{\dot B \dot D}.
\end{equation}

The other linear combination of the SD and ASD parts of the curvature 2-form comes from the so-called Holst term,
\begin{equation}
    R_{ab} \wedge \theta^a \wedge \theta^b = \frac{1}{2} R_{AB} \wedge \Sigma^{AB} + \frac{1}{2} \widetilde{R}_{\dot A \dot B} \wedge \widetilde{\Sigma}^{\dot A \dot B}.
\end{equation}
In fact, the Holst term is identically zero due to the Bianchi identity $R_{ab} \wedge \theta^a = 0$. This implies we can express the EH action as
\begin{equation}
    \frac{1}{2} \int R \, \mathrm{dvol} = \frac{i}{2} \int R_{AB} \wedge \Sigma^{AB} = \frac{i}{2} \int \big( \dd \Gamma_{AB}   +  \Gamma_{AC}\wedge \Gamma^C_{\;\; B} \big) \wedge \Sigma^{AB} 
\end{equation}
where we set $\kappa_G^2 = 8 \pi G = 1$.

If we want to use the above action on-shell to compute scattering amplitudes, we run into the issue that $R = 0$ when $R_{\mu \nu} = 0$. A boundary term must be appended to this action to make it non-zero.

While we would like to provide the reader with a principled derivation of the appropriate boundary term, it is not yet known how to extend the AFS formalism to Einstein gravity using tetrad variables and we will unfortunately not be solving that problem here. However, an eminently natural choice of boundary term is the one which arises by an ``integration by parts'' in the action, 
\begin{equation}\label{pleb_act}
\begin{aligned}
    S[\theta^{A \dot A}, \Gamma_{AB}] &= \frac{i}{2} \int \big( R_{AB} \wedge \Sigma^{AB} - \dd (\Gamma_{AB} \wedge \Sigma^{AB} ) \big) \\
    &= \frac{i}{2} \int \big( \Gamma_{AB} \wedge \dd \Sigma^{AB}  +  \Gamma_{AC}\wedge \Gamma^C_{\;\; B} \wedge \Sigma^{AB}  \big) \, .
\end{aligned}
\end{equation}
This boundary term has appeared previously in different contexts \cite{Corichi:2010ur, Corichi:2015cqa, Ashtekar:2008jw, Corichi:2013zza}, see also \cite{DePaoli:2018erh}.

The action \eqref{pleb_act} (without boundary term) is known as the Plebański action, and though it only depends on the ASD 2-forms and ASD spin connection it remarkably is still an action for full Einstein gravity \cite{Capovilla:1991qb, plebanski1977separation, Capovilla:1993cvm}. It is in fact the Lagrangian counterpart to Ashtekar's SD Hamiltonian formulation for GR \cite{Ashtekar:1987gu}. The equations of motion arising from the variations $\delta \theta^{A \dot A}$ and $\delta \Gamma_{AB}$ are
\begin{equation}\label{pleb_eom}
    R_{AB} \wedge \theta^{A \dot A} = 0, \hspace{1 cm} \dd \Sigma^{AB} = - 2 \Gamma^{(A}_{\;\;\;\; C} \wedge \Sigma^{B) C}.
\end{equation}
The first equation is equivalent to the vacuum Einstein equation and the second equation is the first Cartan structure equation. Plugging the second equation of motion \eqref{pleb_eom} into the action, we get the further on-shell expression
\begin{equation}\label{S_useful}
    S = - \frac{i}{2} \int \big(\Sigma^{AB} \wedge \Gamma_{AC}\wedge \Gamma^{C}_{\;\; B}  \big).
\end{equation}

On a self-dual background, which we'll denote with the subscript 0, we have the equations
\begin{equation}
    (R_0)_{AB} = 0, \hspace{1 cm} (\Gamma_0)_{AB} = 0, \hspace{1 cm} \dd (\Sigma_0)^{AB} = 0, \hspace{1 cm} (\Psi_0)_{ABCD} = 0.
\end{equation}

Consider a SD background made out of $N$ positive helicity gravitons, and let us define the integers
\begin{equation}
    \begin{aligned}
        I &\equiv N + 1 \\
        J &\equiv N + 2
    \end{aligned}
\end{equation}
to be the labels corresponding to the two negative helicity gravitons. A linear perturbation to a SD background
$(\gamma_{AB}, \psi_{ABCD}) \equiv (\delta \Gamma_{AB},\delta \Psi_{ABCD})$ satisfies
\begin{equation}\label{linear_asd_pert}
    \dd \gamma_{AB} = \psi_{ABCD}  \Sigma_0^{CD}, \hspace{1.5 cm} \dd \Sigma^{AB} =- 2 \gamma^{(A}_{\;\;\;\; C} \wedge \Sigma_0^{B) C}.
\end{equation}

Something nice happens if we have \emph{two} ASD perturbations and wish to evaluate the action. In general, we would expect that we would have to solve for second order, non-linear corrections in $\gamma$ proportional to $\epsilon_I \epsilon_J$ coming from the two ASD perturbations interacting with each other. However, because $S$ from \eqref{S_useful} is quadratic in $\Gamma_{AB}$, and $(\Gamma_0)_{AB} = 0$, these second order corrections will not affect the scattering amplitude.

If we define $\gamma^I$ and $\gamma^J$ to be the linearized ASD perturbations corresponding to the negative helicity gravitons labelled $I$ and $J$, then by \eqref{A_S_connected} we have
\begin{equation}\label{MHV_eq1}
\begin{aligned}
    \mathcal{A}_{\rm MHV}(1^+, \ldots, N^+, I^-, J^-) &= \eval{ \frac{\partial^{N+2}}{\partial \epsilon_1 \ldots \partial \epsilon_{N+2}}  \int (\Sigma_0)^{AB} \wedge (\gamma^{I})_{AC} \wedge (\gamma^{J})^C_{\;\;B} }_{\epsilon_i = 0 \,  \forall \epsilon_i}.
\end{aligned}
\end{equation}
This is the well-known MHV generating function \cite{Mason:2009afn, Adamo:2021bej}. We shall now evaluate it using the marked-tree perturbiner expansion for the Plebański scalar we developed in section \ref{sec2}.

\subsection{Massaging the MHV generating function}

If we wish to evalaute \eqref{MHV_eq1}, the first task is to compute $\gamma^I$ corresponding to the ASD graviton labelled $I$. From \eqref{linear_asd_pert}, we can do this if we know what $\psi^I_{ABCD}$ is.

Fortunately, we have already found the expression for $\psi^I_{ABCD}$ in \eqref{psi_I} when we noted that $\psi$ solved the same equation of motion as a SD perturbation to $\phi$ in a SD background. We showed that such perturbations can be written as the exponential of a function $F_I$ \eqref{phi_pert_FI}, where $F_I$ is a sum over a certain special set of graphs \eqref{FI_graph} which were given in definition \ref{def7}.
\begin{equation}\label{psiI_for_proof}
    \psi_{ABCD}^I = - \frac{1}{8} \, \epsilon_I \, \kappa^I_A \kappa^I_B \kappa^I_C \kappa^I_D \exp(i F_I).
\end{equation}
With $\psi^I_{ABCD}$ in hand, we now just need to find a $\gamma^I$ that solves \eqref{linear_asd_pert}. There is a known formula for this \cite{Adamo:2022mev}, which review here in the following claim.

\begin{claim}\label{claim_gamma}
    \begin{equation}\label{gammaIput}
    (\gamma^I)_{AB} = -i \epsilon_I \kappa^I_A \kappa^I_B \frac{ \kappa^I_C \xi_{\dot C} (H_I)_{\dot A}^{\;\;\; \dot C} \theta^{C \dot A}}{2 [\xi I]}  \exp(i F_I)
\end{equation}
where the tensor $(H^I)_{\dot A}^{\;\; \dot B}$ is defined by
\begin{equation}\label{H_def}
    e_{A \dot A}^\mu \partial_\mu F_I^{\dot B} = \kappa_A^I (H^I)_{\dot A}^{\;\; \dot B} \, ,
\end{equation}
$\xi_{\dot A}$ is an arbitrary spinor,\footnote{The choice of $\xi$ corresonds to the choice of Lorentz frame for the tetrad.  A change in frame sends $\gamma^I_{AB} \mapsto \gamma^I_{AB} + \dd f_{AB}$ for some zero-form $f_{AB}$, just as the change $\xi \mapsto a \xi + b \kappat^I$ does. } and $\theta^{A \dot A}$/$e_{A \dot A}$ is the tetrad/vierbein of the SD background.
\end{claim}

\begin{proof}
To prove this claim, we must show that \eqref{gammaIput} and \eqref{psiI_for_proof} solve \eqref{linear_asd_pert}. 

We begin by expressing the exterior derivative $\dd$ as
    \begin{equation}
    \dd =  d x^\mu \wedge \partial_\mu  =  \theta^{a} \wedge e_{a}^\mu \partial_\mu= \frac{1}{2} \theta^{A \dot A} \wedge e_{A \dot A}^\mu \partial_\mu.
\end{equation}
The above expression can be used to rewrite \eqref{H_def} as
\begin{equation}
    2 \; \dd F_I^{\dot B} = \kappa^I_{A} (H^I)_{\dot A}^{\;\; \dot B} \theta^{A \dot A}
\end{equation}
which implies
\begin{equation}
    \dd \left( \kappa^I_A (H^I)_{\dot A}^{\;\; \dot B} \theta^{A \dot A} \right) = 0.
\end{equation}

We now act with the expression for $\gamma^I$ in \eqref{gammaIput} by $\dd$:
\begin{equation}\label{dgamma_work}
\begin{aligned}
    \dd (\gamma^I)_{AB} &= -i \epsilon_I  \kappa^I_A \kappa^I_B \frac{\kappa^I_C \xi_{\dot C} (H_I)_{\dot A}^{\;\;\; \dot C} \theta^{C \dot A}}{2 [\xi I]} \wedge  \dd \exp(i F_I) \\
    &= - i \epsilon_I \kappa^I_A \kappa^I_B \frac{\kappa^I_C \xi_{\dot C} (H_I)_{\dot A}^{\;\;\; \dot C} \theta^{C \dot A}}{2 [\xi I]} \wedge \left( \frac{i}{2} \kappat^I_{\dot B} \kappa^I_D (H_I)_{\dot D}^{\;\; \dot B} \theta^{D \dot D} \right) \exp(i F_I) \\
    &=\epsilon_I \kappa^I_A \kappa^I_B \kappa^I_C \kappa^I_D  \frac{\kappat^I_{\dot B} \xi_{\dot C}}{{4 [\xi I]}} \left(  (H_I)_{\dot A}^{\;\;\; \dot C}  (H_I)_{\dot D}^{\;\; \dot B} \theta^{C \dot A} \wedge \theta^{D \dot D} \right) \exp(i F_I)\, .
\end{aligned}
\end{equation}
It turns out that the tensor $(H_I)_{\dot A}^{\;\; \dot B}$ satisfies the property  
\begin{equation}\label{HH_given}
    (H_I)_{\dot A}^{\;\; \dot C} (H_I)_{\dot D}^{\; \; \dot B} \vep^{\dot A \dot D} = -\vep^{\dot B \dot C}.
\end{equation}
We will show the above equation is true later (see Claim \ref{claim_H_sl2c}) but for now use it as a given. \eqref{HH_given} implies
\begin{equation}
\begin{aligned}
    (H_I)_{\dot A}^{\;\;\; \dot C}  (H^I)_{\dot D}^{\;\; \dot B} \theta^{C \dot A} \wedge \theta^{D \dot D} &= - \frac{1}{2} (H_I)_{\dot A}^{\;\;\; \dot C}  (H^I)_{\dot D}^{\;\; \dot B} (\vep^{\dot A \dot D} \Sigma_0^{CD} + \vep^{CD} \widetilde{\Sigma}_0^{ \dot A \dot D} ) \\
    &= \frac{1}{2} \vep^{\dot B \dot C} \Sigma_0^{CD} - \frac{1}{2} (H_I)_{\dot A}^{\;\;\; \dot C}  (H^I)_{\dot D}^{\;\; \dot B} \vep^{CD} \widetilde{\Sigma}_0^{ \dot A \dot D} 
\end{aligned}
\end{equation}
which we can plug into \eqref{dgamma_work} to get
\begin{equation}
\begin{aligned}
    \dd (\gamma^I)_{AB} &=-\frac{1}{8} \epsilon_I \kappa^I_A \kappa^I_B \kappa^I_C \kappa^I_D   \exp(i F_I) \Sigma_0^{CD} \\
    &= \psi^I_{ABCD} \Sigma_0^{CD}
\end{aligned}
\end{equation}
as desired.

\end{proof}

The careful reader may have noticed that $(H_I)_{\dot A}^{\;\; \dot B}$, defined in \eqref{H_def}, was not yet proven to be well-defined. We rectify the situation below.

\begin{claim}\label{claim_H_def}
There exists some matrix $(H_I)_{\dot A}^{\;\; \dot B}$ such that \eqref{H_def} is satisfied.
    
\end{claim}

\begin{proof}
We begin with \eqref{FIindex}, reproduced below
\begin{equation}
    F_I = \kappat^I_{\dot A} F_I^{\dot A}
\end{equation}
where $\kappat^I_{\dot A} = \omega_I(-\bz_I, 1)$. Because $F_I$ can be written as a sum of graphs where the (diamond) $I$ node only has one edge (see \eqref{FI_graph}), we know that $F_I$ is linear in $\kappat^I$ and $F_I^{\dot A}$ is therefore independent of $\omega_I$ and $\bz_I$. The two components of $F_I^{\dot A}$ are then given by
\begin{align}
    F_I^{\dot 1} = -\frac{1}{\omega_I} \partial_{\bz_I} \eval{F_I}_{\bz_I = 0}\hspace{1 cm}
    F_I^{\dot 2} = \frac{1}{\omega_I} \eval{F_I}_{\bz_I = 0} 
\end{align}
which we can write as
\begin{align}
      F_I^{\dot B} = \frac{1}{\omega_I} \eval{ \begin{pmatrix} -\partial_{\bz_I} \\ 1 \end{pmatrix}^{\dot B} F_I }_{\bz_I = 0}.
\end{align}
We now use the explicit formula for the vierbein of the SD background \eqref{conv_vierbein},
\begin{equation}\label{vector_tetrad}
\begin{aligned}
    e^\mu_{1 \dot 1}\partial_\mu &=  \partial_{\bu} +  \{ \partial_w \phi , \cdot \}  =  \partial_w \mR \, , \\
    e^\mu_{1 \dot 2}\partial_\mu &=  \partial_{\bw} + \{ \partial_u \phi , \cdot \}  = \partial_u \mR \, ,\\
    e^\mu_{2 \dot 1}\partial_\mu &= \partial_w \, , \\
    e^\mu_{2 \dot 2}\partial_\mu &= \partial_u \, .
\end{aligned}
\end{equation}
Note that the $e_{1 \dot A}$ contain an extra application of the recursion operator $\mR$ than the $e_{2 \dot A}$ (see \eqref{R_grav}).

With the vierbein \eqref{vector_tetrad}, we may write
\begin{align}\label{diff_zb}
    e_{A \dot A}^\mu \partial_\mu F_I^{\dot B} = \frac{1}{\omega_I} \eval{ \begin{pmatrix} - \partial_{\bz_I} \\ 1 \end{pmatrix}^{\dot B} \left( \frac{-i}{e^{i F_I}} ( e_{A \dot A}^\mu \partial_\mu ) e^{i F_I} \right) }_{\bz_I = 0}.
\end{align}
The key feature of this equation is that, for the $e_{1 \dot A}$'s, we are acting with a recursion operator on the linear perturbation $e^{i F_I}$ on the RHS, which we know differs from $e_{2 \dot A}$ by a relative factor of $-z_I$ via claim \ref{thmRz}.
\begin{equation}
    e_{1 \dot A}^\mu \partial_\mu  e^{i F_I} = e_{2 \dot A}^\mu \partial_\mu ( \mR e^{i F_I} )= - z_I \, e_{2 \dot A}^\mu \partial_\mu e^{i F_I}
\end{equation}
From \eqref{diff_zb}, this implies
\begin{equation}
    e_{1 \dot A}^\mu \partial_\mu F_I^{\dot B} = - z_I \, e_{2 \dot A}^\mu \partial_\mu F_I^{\dot B}
\end{equation}
meaning we may now write
\begin{align}
    e_{1 \dot A}^\mu \partial_\mu F_I^{\dot B} &= -z_I (H^I)_{\dot A}^{\;\; \dot B} \\
    e_{2 \dot A}^\mu \partial_\mu F_I^{\dot B} &=  (H^I)_{\dot A}^{\;\; \dot B} \label{secondFH}
\end{align}
for some proportionality matrix $(H^I)_{\dot A}^{\;\; \dot B}$. This completes the proof.
\end{proof}

Note that from $e_{2 \dot 1} = \partial_w$, $e_{2 \dot 2} = \partial_u$, the equation \eqref{secondFH} reduces to the simple expression for $(H_I)_{\dot A}^{\;\; \dot B}$
\begin{equation}\label{H_der}
    (H_I)_{\dot A}^{\;\;\;\dot B} = \begin{pmatrix} \partial_w \\ \partial_u \end{pmatrix}_{\dot A} F_I^{\dot B}.
\end{equation}

\begin{claim}\label{claim_H_sl2c}
    \begin{equation}
        (H_I)_{\dot A}^{\;\; \dot C} (H_I)_{\dot D}^{\; \; \dot B} \vep^{\dot A \dot D} = -\vep^{\dot B \dot C}.
    \end{equation}
\end{claim}

\begin{proof}
The above statement is equivalent to
\begin{equation}
        \xi_{\dot B} \chi_{\dot C}(H_I)_{\dot A}^{\;\; \dot C} (H^I)_{\dot D}^{\; \; \dot B} \vep^{\dot A \dot D} = - \xi_{\dot B} \chi^{\dot B}
    \end{equation}
where $\xi$ and $\chi$ are arbitrary spinors. By \eqref{H_der}, this can be rewritten as
\begin{equation}
    \{ \xi_{\dot B} F_I^{\dot B}, \chi_{\dot C} F_I^{\dot C} \} =  -[ \xi \chi ] .
\end{equation}
The above equation will be proven diagrammatically in Claim \ref{claimFIFI}, completing the proof.

\end{proof}

We are now in a position to rewrite the MHV generating function $\Sigma_0 \wedge \gamma^I \wedge \gamma^J$ as an expression which is easier to work with. We know what $\Sigma_0$ is from \eqref{conv_tetrad}, we know what $\gamma^I$ is from \eqref{gammaIput}, and we know what $H_I$ is from \eqref{H_der}. Plugging in all of these formulas, we get
\begin{equation}\label{to_eval}
    \Sigma_0^{AB} \wedge (\gamma^I)_{AC} \wedge (\gamma^J)^{C}_{\;\;B} =  i \frac{\langle I J \rangle^3 }{[\xi I][\chi J]} \{ \xi_{\dot A} F_I^{\dot A}, \chi_{\dot B} F_J^{\dot B} \} \epsilon_I e^{i F_I} \, \epsilon_J e^{i F_J} \mathrm{d vol}
\end{equation}
where
\begin{equation}
\mathrm{dvol} = 4 i\, du \wedge d \bu \wedge dw \wedge d \bw 
\end{equation}
is the volume form of any SD metric written in Plebański form \eqref{eq2}, as $\mathrm{det}(g) = 16$.

\subsection{Proof of the NSVW formula}

Now that we have an expression for the MHV generating function \eqref{to_eval}, we can use it to calculate the MHV amplitude \eqref{MHV_eq1}. In the end, the dependence on the arbitrary reference spinors $\xi_{\dot A}$ and $\chi_{\dot B}$ will drop out. 

We begin by finding a visual representation of the term $\{ \xi_{\dot A} F_I^{\dot A}, \chi_{\dot B} F_J^{\dot B} \}$. It will end up being expressed in terms of the following graphs.

\begin{definition}
    Let $\mT^{\meddiamond, \xi, \chi}_{N|I, J}$ be the set of graphs which can be formed by taking a graph in $\mT_N$ and appending to it two diamond nodes with labels $I$ and $J$, which are connected by a $\xi$-edge and a $\chi$-edge, respectively. We also include in $\mT^{\meddiamond, \xi, \chi}_{N|I, J}$ the exceptional graph which is just the two diamond nodes $I$ and $J$ connected by a paired $\xi \chi$-edge. 
\end{definition}

\begin{exs}
    Here are some examples of graphs in $\mT^{\meddiamond, \xi, \chi}_{N|I, J}$ for $N = 7$.

    \begin{center}
        \tikzset{every picture/.style={line width=0.75pt}} 

\begin{tikzpicture}[x=0.75pt,y=0.75pt,yscale=-1,xscale=1]

\draw    (110,144.7) -- (166,144.7) ;
\draw    (81,184.7) -- (110,144.7) ;
\draw    (106,86.7) -- (110,144.7) ;
\draw    (106,86.7) -- (153,59.7) ;
\draw    (93,37.7) -- (106,86.7) ;
\draw    (50,86.7) -- (106,86.7) ;
\draw    (324,102.7) -- (379,102.7) ;
\draw    (269,102.7) -- (324,102.7) ;
\draw    (324,102.7) -- (303,57.7) ;
\draw    (324,102.7) -- (345,57.7) ;
\draw    (345,147.7) -- (324,102.7) ;
\draw    (303,147.7) -- (324,102.7) ;
\draw    (472.36,102.7) -- (545.36,102.7) ;
\draw  [fill={rgb, 255:red, 255; green, 255; blue, 255 }  ,fill opacity=1 ] (93.5,86.7) .. controls (93.5,79.8) and (99.1,74.2) .. (106,74.2) .. controls (112.9,74.2) and (118.5,79.8) .. (118.5,86.7) .. controls (118.5,93.6) and (112.9,99.2) .. (106,99.2) .. controls (99.1,99.2) and (93.5,93.6) .. (93.5,86.7) -- cycle ;
\draw  [fill={rgb, 255:red, 255; green, 255; blue, 255 }  ,fill opacity=1 ] (50,69.02) -- (67.68,86.7) -- (50,104.38) -- (32.32,86.7) -- cycle ;
\draw  [fill={rgb, 255:red, 255; green, 255; blue, 255 }  ,fill opacity=1 ] (80.5,37.7) .. controls (80.5,30.8) and (86.1,25.2) .. (93,25.2) .. controls (99.9,25.2) and (105.5,30.8) .. (105.5,37.7) .. controls (105.5,44.6) and (99.9,50.2) .. (93,50.2) .. controls (86.1,50.2) and (80.5,44.6) .. (80.5,37.7) -- cycle ;
\draw  [fill={rgb, 255:red, 255; green, 255; blue, 255 }  ,fill opacity=1 ] (140.5,59.7) .. controls (140.5,52.8) and (146.1,47.2) .. (153,47.2) .. controls (159.9,47.2) and (165.5,52.8) .. (165.5,59.7) .. controls (165.5,66.6) and (159.9,72.2) .. (153,72.2) .. controls (146.1,72.2) and (140.5,66.6) .. (140.5,59.7) -- cycle ;
\draw  [fill={rgb, 255:red, 255; green, 255; blue, 255 }  ,fill opacity=1 ] (97.5,144.7) .. controls (97.5,137.8) and (103.1,132.2) .. (110,132.2) .. controls (116.9,132.2) and (122.5,137.8) .. (122.5,144.7) .. controls (122.5,151.6) and (116.9,157.2) .. (110,157.2) .. controls (103.1,157.2) and (97.5,151.6) .. (97.5,144.7) -- cycle ;
\draw  [fill={rgb, 255:red, 255; green, 255; blue, 255 }  ,fill opacity=1 ] (68.5,184.7) .. controls (68.5,177.8) and (74.1,172.2) .. (81,172.2) .. controls (87.9,172.2) and (93.5,177.8) .. (93.5,184.7) .. controls (93.5,191.6) and (87.9,197.2) .. (81,197.2) .. controls (74.1,197.2) and (68.5,191.6) .. (68.5,184.7) -- cycle ;
\draw  [fill={rgb, 255:red, 255; green, 255; blue, 255 }  ,fill opacity=1 ] (166,127.02) -- (183.68,144.7) -- (166,162.38) -- (148.32,144.7) -- cycle ;
\draw  [fill={rgb, 255:red, 255; green, 255; blue, 255 }  ,fill opacity=1 ] (311.5,102.7) .. controls (311.5,95.8) and (317.1,90.2) .. (324,90.2) .. controls (330.9,90.2) and (336.5,95.8) .. (336.5,102.7) .. controls (336.5,109.6) and (330.9,115.2) .. (324,115.2) .. controls (317.1,115.2) and (311.5,109.6) .. (311.5,102.7) -- cycle ;
\draw  [fill={rgb, 255:red, 255; green, 255; blue, 255 }  ,fill opacity=1 ] (290.5,147.7) .. controls (290.5,140.8) and (296.1,135.2) .. (303,135.2) .. controls (309.9,135.2) and (315.5,140.8) .. (315.5,147.7) .. controls (315.5,154.6) and (309.9,160.2) .. (303,160.2) .. controls (296.1,160.2) and (290.5,154.6) .. (290.5,147.7) -- cycle ;
\draw  [fill={rgb, 255:red, 255; green, 255; blue, 255 }  ,fill opacity=1 ] (332.5,147.7) .. controls (332.5,140.8) and (338.1,135.2) .. (345,135.2) .. controls (351.9,135.2) and (357.5,140.8) .. (357.5,147.7) .. controls (357.5,154.6) and (351.9,160.2) .. (345,160.2) .. controls (338.1,160.2) and (332.5,154.6) .. (332.5,147.7) -- cycle ;
\draw  [fill={rgb, 255:red, 255; green, 255; blue, 255 }  ,fill opacity=1 ] (290.5,57.7) .. controls (290.5,50.8) and (296.1,45.2) .. (303,45.2) .. controls (309.9,45.2) and (315.5,50.8) .. (315.5,57.7) .. controls (315.5,64.6) and (309.9,70.2) .. (303,70.2) .. controls (296.1,70.2) and (290.5,64.6) .. (290.5,57.7) -- cycle ;
\draw  [fill={rgb, 255:red, 255; green, 255; blue, 255 }  ,fill opacity=1 ] (332.5,57.7) .. controls (332.5,50.8) and (338.1,45.2) .. (345,45.2) .. controls (351.9,45.2) and (357.5,50.8) .. (357.5,57.7) .. controls (357.5,64.6) and (351.9,70.2) .. (345,70.2) .. controls (338.1,70.2) and (332.5,64.6) .. (332.5,57.7) -- cycle ;
\draw  [fill={rgb, 255:red, 255; green, 255; blue, 255 }  ,fill opacity=1 ] (379,85.02) -- (396.68,102.7) -- (379,120.38) -- (361.32,102.7) -- cycle ;
\draw  [fill={rgb, 255:red, 255; green, 255; blue, 255 }  ,fill opacity=1 ] (269,85.02) -- (286.68,102.7) -- (269,120.38) -- (251.32,102.7) -- cycle ;
\draw  [fill={rgb, 255:red, 255; green, 255; blue, 255 }  ,fill opacity=1 ] (545.36,85.02) -- (563.03,102.7) -- (545.36,120.38) -- (527.68,102.7) -- cycle ;
\draw  [fill={rgb, 255:red, 255; green, 255; blue, 255 }  ,fill opacity=1 ] (472.36,85.02) -- (490.03,102.7) -- (472.36,120.38) -- (454.68,102.7) -- cycle ;

\draw (508.86,99.3) node [anchor=south] [inner sep=0.75pt]    {$\xi \ \chi $};
\draw (296.5,99.3) node [anchor=south] [inner sep=0.75pt]    {$\xi $};
\draw (351.5,99.3) node [anchor=south] [inner sep=0.75pt]    {$\chi $};
\draw (78,83.3) node [anchor=south] [inner sep=0.75pt]    {$\xi $};
\draw (138,141.3) node [anchor=south] [inner sep=0.75pt]    {$\chi $};
\draw (472.36,102.7) node    {$I$};
\draw (269,102.7) node    {$I$};
\draw (50,86.7) node    {$I$};
\draw (166,144.7) node    {$J$};
\draw (379,102.7) node    {$J$};
\draw (545.36,102.7) node    {$J$};
\draw (110,144.7) node    {$2$};
\draw (153,59.7) node    {$5$};
\draw (106,86.7) node    {$3$};
\draw (93,37.7) node    {$7$};
\draw (81,184.7) node    {$6$};
\draw (324,102.7) node    {$2$};
\draw (345,147.7) node    {$6$};
\draw (345,57.7) node    {$5$};
\draw (303,57.7) node    {$3$};
\draw (303,147.7) node    {$4$};

\end{tikzpicture}
    \end{center}
    
\end{exs}

\begin{claim}\label{claim_FIFJ}
    \begin{equation}\label{eq_claim_FIFJ}
        -\{ \xi_{\dot A} F_I^{\dot A}, \chi_{\dot B} F_J^{\dot B} \} =  \langle I J \rangle \sum_{\bt \, \in \, \mT^{\meddiamond, \xi, \chi}_{N|I, J}} \phi_{\bt } 
    \end{equation}
\end{claim}

\begin{proof}
    From the graphical equation for $\xi_{\dot A} F_I^{\dot A}$ \eqref{FIA}, we may write the LHS of the above expression as
    \begin{align}\label{poissonFIFJ}
    -\{ \xi_{\dot A} F_I^{\dot A},  \chi_{\dot B} F_J^{\dot B}\} =& \,   \{\sum_{\bt \in \mT^{\meddiamond, \xi }_{N | I}}  \phi_{\bt}, \sum_{\bt' \in \mT^{\meddiamond, \chi }_{N | J}} \phi_{\bt'} \} \\
    & - i (\xi_{\dot 1} \partial_u - \xi_{\dot 2} \partial_w ) \sum_{\bt' \in \mT^{\meddiamond, \chi }_{N | J}}  \phi_{\bt'}  + i ( \chi_{\dot 1} \partial_u - \chi_{\dot 2} \partial_w ) \sum_{\bt \in \mT^{\meddiamond, \xi}_{N | I}}  \phi_{\bt} \nonumber \\
    & + [\xi \chi] \nonumber
\end{align}
where we used the fact that
\begin{equation}
    -\{ \xi_{\dot 1} ( w - z_I \bu) + \xi_{\dot 2}( u - z_I \bw ), \chi_{\dot 1} ( w - z_I \bu) + \chi_{\dot 2}( u - z_I \bw ) \} =  [ \xi \chi].
\end{equation}
There are now four types of graphs to be summed, arising from the four terms in the RHS of \eqref{poissonFIFJ}.

The first type of graph, corresponding to the first term, is constructed by taking two graphs in $\mT^{\meddiamond, \xi}_{N|I}$ and $\mT^{\meddiamond, \chi}_{N|J}$ and drawing one wiggly line from a circle node in the first graph to a circle node in the second graph. A typical example of such a graph we should sum over is shown below.
\begin{center}
    \tikzset{every picture/.style={line width=0.75pt}} 

\begin{tikzpicture}[x=0.75pt,y=0.75pt,yscale=-1,xscale=1]

\draw    (201,33.58) -- (180,66.41) ;
\draw    (180,66.41) -- (201,99.23) ;
\draw    (132.5,66.41) -- (111.5,99.23) ;
\draw    (180,66.41) -- (227.5,66.41) ;
\draw    (76.03,30.41) -- (85,66.41) ;
\draw    (37.5,66.41) -- (85,66.41) ;
\draw    (85,66.41) -- (132.5,66.41) ;
\draw  [fill={rgb, 255:red, 255; green, 255; blue, 255 }  ,fill opacity=1 ] (76.03,66.41) .. controls (76.03,61.45) and (80.05,57.43) .. (85,57.43) .. controls (89.95,57.43) and (93.97,61.45) .. (93.97,66.41) .. controls (93.97,71.36) and (89.95,75.38) .. (85,75.38) .. controls (80.05,75.38) and (76.03,71.36) .. (76.03,66.41) -- cycle ;
\draw    (132.5,66.41) .. controls (134.17,64.74) and (135.83,64.74) .. (137.5,66.41) .. controls (139.17,68.08) and (140.83,68.08) .. (142.5,66.41) .. controls (144.17,64.74) and (145.83,64.74) .. (147.5,66.41) .. controls (149.17,68.08) and (150.83,68.08) .. (152.5,66.41) .. controls (154.17,64.74) and (155.83,64.74) .. (157.5,66.41) .. controls (159.17,68.08) and (160.83,68.08) .. (162.5,66.41) .. controls (164.17,64.74) and (165.83,64.74) .. (167.5,66.41) .. controls (169.17,68.08) and (170.83,68.08) .. (172.5,66.41) .. controls (174.17,64.74) and (175.83,64.74) .. (177.5,66.41) -- (180,66.41) -- (180,66.41) ;
\draw [shift={(160.05,66.41)}, rotate = 180] [fill={rgb, 255:red, 0; green, 0; blue, 0 }  ][line width=0.08]  [draw opacity=0] (8.93,-4.29) -- (0,0) -- (8.93,4.29) -- (5.93,0) -- cycle    ;
\draw  [fill={rgb, 255:red, 255; green, 255; blue, 255 }  ,fill opacity=1 ] (123.53,66.41) .. controls (123.53,61.45) and (127.55,57.43) .. (132.5,57.43) .. controls (137.45,57.43) and (141.47,61.45) .. (141.47,66.41) .. controls (141.47,71.36) and (137.45,75.38) .. (132.5,75.38) .. controls (127.55,75.38) and (123.53,71.36) .. (123.53,66.41) -- cycle ;
\draw  [fill={rgb, 255:red, 255; green, 255; blue, 255 }  ,fill opacity=1 ] (67.06,30.41) .. controls (67.06,25.45) and (71.07,21.43) .. (76.03,21.43) .. controls (80.98,21.43) and (85,25.45) .. (85,30.41) .. controls (85,35.36) and (80.98,39.38) .. (76.03,39.38) .. controls (71.07,39.38) and (67.06,35.36) .. (67.06,30.41) -- cycle ;
\draw  [fill={rgb, 255:red, 255; green, 255; blue, 255 }  ,fill opacity=1 ] (102.53,99.23) .. controls (102.53,94.27) and (106.55,90.26) .. (111.5,90.26) .. controls (116.45,90.26) and (120.47,94.27) .. (120.47,99.23) .. controls (120.47,104.18) and (116.45,108.2) .. (111.5,108.2) .. controls (106.55,108.2) and (102.53,104.18) .. (102.53,99.23) -- cycle ;
\draw  [fill={rgb, 255:red, 255; green, 255; blue, 255 }  ,fill opacity=1 ] (192.03,33.58) .. controls (192.03,28.63) and (196.05,24.61) .. (201,24.61) .. controls (205.95,24.61) and (209.97,28.63) .. (209.97,33.58) .. controls (209.97,38.54) and (205.95,42.55) .. (201,42.55) .. controls (196.05,42.55) and (192.03,38.54) .. (192.03,33.58) -- cycle ;
\draw  [fill={rgb, 255:red, 255; green, 255; blue, 255 }  ,fill opacity=1 ] (192.03,99.23) .. controls (192.03,94.27) and (196.05,90.26) .. (201,90.26) .. controls (205.95,90.26) and (209.97,94.27) .. (209.97,99.23) .. controls (209.97,104.18) and (205.95,108.2) .. (201,108.2) .. controls (196.05,108.2) and (192.03,104.18) .. (192.03,99.23) -- cycle ;
\draw  [fill={rgb, 255:red, 255; green, 255; blue, 255 }  ,fill opacity=1 ] (227.5,53.72) -- (240.19,66.41) -- (227.5,79.09) -- (214.81,66.41) -- cycle ;
\draw  [fill={rgb, 255:red, 255; green, 255; blue, 255 }  ,fill opacity=1 ] (171.03,66.41) .. controls (171.03,61.45) and (175.05,57.43) .. (180,57.43) .. controls (184.95,57.43) and (188.97,61.45) .. (188.97,66.41) .. controls (188.97,71.36) and (184.95,75.38) .. (180,75.38) .. controls (175.05,75.38) and (171.03,71.36) .. (171.03,66.41) -- cycle ;
\draw  [fill={rgb, 255:red, 255; green, 255; blue, 255 }  ,fill opacity=1 ] (37.5,53.72) -- (50.19,66.41) -- (37.5,79.09) -- (24.81,66.41) -- cycle ;

\draw (227.5,66.41) node  [font=\footnotesize]  {$J$};
\draw (61.25,63.01) node [anchor=south] [inner sep=0.75pt]    {$\xi $};
\draw (203.75,62.01) node [anchor=south] [inner sep=0.75pt]    {$\chi $};
\draw (37.5,66.41) node  [font=\footnotesize]  {$I$};
\draw (261.87,78) node [anchor=west] [inner sep=0.75pt]    {$\displaystyle \subset \;\;\; \{\sum_{\bt \, \in \, \mT^{\meddiamond, \xi }_{N | I}}  \phi_{\bt}, \sum_{\bt' \, \in \,  \mT^{\meddiamond, \chi }_{N | J}} \phi_{\bt'} \}$};

\end{tikzpicture}
\end{center}
For the second type of graph corresponding to the second term of \eqref{poissonFIFJ}, if we use the fact that
\begin{equation}
    -i(\xi_{\dot 1} \partial_u - \xi_{\dot 2} \partial_w ) \phi_i = [\xi i] \phi_i
\end{equation}
then we see that this term corresponds to appending one diamond $I$ node with a wiggly $\xi$-edge to a circle node in a graph in $\mT^{\meddiamond, \chi}_{N|J}$. An analogous statement holds for the third term, and typical graphs which contribute to these terms look like
\begin{center}
    \input{figures/gravity_graph_44.tex}
\end{center}
One could consider the equivalence class of graphs which would look identical if this wiggly line were converted into a straight line. Let us consider what happens when we sum over all graphs in an equivalence class.

Because a wiggly line between nodes $i$ and $j$ can be converted into a straight line at the cost of multiplication by $z_{ij}$, the result of this sum would be multiply a converted graph with no wiggly edges by a sum over all $z_{ij}$'s where $i$ and $j$ are all the nodes on the path connecting $I$ and $J$. (We used similar trick in the proof of Theorem \ref{thmRz}.) This is depicted in the figure below.

\begin{center}
    \input{figures/gravity_graph_34}
\end{center}
Note that the final graph in the above equation is indeed contained in $\mT^{\meddiamond, \xi, \chi}_{N|I, J}$, as we wanted to show. 

The above sum naturally included all graphs from the first three terms on the RHS of \eqref{poissonFIFJ}. To conclude the proof, we need to account for the fourth and final term $[\xi \chi]$. This, however, corresponds to the exceptional graph
\begin{center}
    \tikzset{every picture/.style={line width=0.75pt}} 

\begin{tikzpicture}[x=0.75pt,y=0.75pt,yscale=-1,xscale=1]

\draw    (122,79.2) .. controls (123.67,77.53) and (125.33,77.53) .. (127,79.2) .. controls (128.67,80.87) and (130.33,80.87) .. (132,79.2) .. controls (133.67,77.53) and (135.33,77.53) .. (137,79.2) .. controls (138.67,80.87) and (140.33,80.87) .. (142,79.2) .. controls (143.67,77.53) and (145.33,77.53) .. (147,79.2) .. controls (148.67,80.87) and (150.33,80.87) .. (152,79.2) .. controls (153.67,77.53) and (155.33,77.53) .. (157,79.2) .. controls (158.67,80.87) and (160.33,80.87) .. (162,79.2) .. controls (163.67,77.53) and (165.33,77.53) .. (167,79.2) .. controls (168.67,80.87) and (170.33,80.87) .. (172,79.2) .. controls (173.67,77.53) and (175.33,77.53) .. (177,79.2) .. controls (178.67,80.87) and (180.33,80.87) .. (182,79.2) .. controls (183.67,77.53) and (185.33,77.53) .. (187,79.2) .. controls (188.67,80.87) and (190.33,80.87) .. (192,79.2) -- (195,79.2) -- (195,79.2) ;
\draw [shift={(163.5,79.2)}, rotate = 180] [fill={rgb, 255:red, 0; green, 0; blue, 0 }  ][line width=0.08]  [draw opacity=0] (10.72,-5.15) -- (0,0) -- (10.72,5.15) -- (7.12,0) -- cycle    ;
\draw  [fill={rgb, 255:red, 255; green, 255; blue, 255 }  ,fill opacity=1 ] (122,61.52) -- (139.68,79.2) -- (122,96.88) -- (104.32,79.2) -- cycle ;
\draw  [fill={rgb, 255:red, 255; green, 255; blue, 255 }  ,fill opacity=1 ] (195,61.52) -- (212.68,79.2) -- (195,96.88) -- (177.32,79.2) -- cycle ;
\draw    (318,78.2) -- (391,78.2) ;
\draw  [fill={rgb, 255:red, 255; green, 255; blue, 255 }  ,fill opacity=1 ] (318,60.52) -- (335.68,78.2) -- (318,95.88) -- (300.32,78.2) -- cycle ;
\draw  [fill={rgb, 255:red, 255; green, 255; blue, 255 }  ,fill opacity=1 ] (391,60.52) -- (408.68,78.2) -- (391,95.88) -- (373.32,78.2) -- cycle ;

\draw (96.32,79.2) node [anchor=east] [inner sep=0.75pt]    {$[ \xi \chi ] \ =$};
\draw (122,79.2) node    {$I$};
\draw (195,79.2) node    {$J$};
\draw (158.5,72.8) node [anchor=south] [inner sep=0.75pt]  [font=\normalsize]  {$\xi \ \chi $};
\draw (318,78.2) node    {$I$};
\draw (391,78.2) node    {$J$};
\draw (354.5,71.8) node [anchor=south] [inner sep=0.75pt]  [font=\normalsize]  {$\xi \ \chi $};
\draw (298.32,78.2) node [anchor=east] [inner sep=0.75pt]    {$\ =\ \langle IJ\rangle \ \times \ $};

\end{tikzpicture}
\end{center}
completing the proof.

\end{proof}

Claim \ref{claim_FIFJ} also implies the following corollary.

\begin{claim}\label{claimFIFI}
    \begin{equation}
        \{ \xi_{\dot A} F^{\dot A}_I, \chi_{\dot B} F^{\dot B}_I \} = - [ \xi \chi]
    \end{equation}
\end{claim}

\begin{proof}
    This follows from Claim \ref{claim_FIFJ} if we replace $J$ with $I$. Because $\langle II \rangle = 0$, this expression isolates all graphs containing $1/\langle I I \rangle$. There is only one such graph, the exceptional graph
    \begin{center}
        \tikzset{every picture/.style={line width=0.75pt}} 

\begin{tikzpicture}[x=0.75pt,y=0.75pt,yscale=-1,xscale=1]

\draw    (257.36,72.7) -- (330.36,72.7) ;
\draw  [fill={rgb, 255:red, 255; green, 255; blue, 255 }  ,fill opacity=1 ] (330.36,55.02) -- (348.03,72.7) -- (330.36,90.38) -- (312.68,72.7) -- cycle ;
\draw  [fill={rgb, 255:red, 255; green, 255; blue, 255 }  ,fill opacity=1 ] (257.36,55.02) -- (275.03,72.7) -- (257.36,90.38) -- (239.68,72.7) -- cycle ;

\draw (293.86,69.3) node [anchor=south] [inner sep=0.75pt]    {$\xi \ \chi $};
\draw (257.36,72.7) node    {$I$};
\draw (330.36,72.7) node    {$I$};
\draw (237.68,72.7) node [anchor=east] [inner sep=0.75pt]    {$-\langle II\rangle \; \times $};
\draw (350.03,72.7) node [anchor=west] [inner sep=0.75pt]    {$=-\langle II\rangle \dfrac{[ \xi \chi ]}{\langle II\rangle } =-[ \xi \chi ]$};

\end{tikzpicture}
    \end{center}
    completing the proof.
\end{proof}

Now that we have graphical expressions for $\{ \xi_{\dot A} F_I^{\dot A}, \chi_{\dot B} F_J^{\dot B} \}$ and $\epsilon_I e^{i F_I}$ in \eqref{eq_claim_FIFJ},  \eqref{phi_pert_FI}, and \eqref{lin_graphs}, we will now plug them into the MHV generating function \eqref{to_eval} and simplify the result:

\begin{equation}
    \Sigma_0^{AB} \wedge (\gamma^I)_{AC} \wedge (\gamma^J)^{C}_{\;\;B} = -i \frac{\langle IJ \rangle^4 \mathrm{dvol} }{[\xi I] [\chi J]} \left( \sum_{\bt \, \in \, \mT^{\meddiamond, \xi, \chi}_{N|I, J}} \phi_{\bt } 
 \right) 
 \left( \sum_{\bt \in \mT_{N | I}} \phi_{\bt}\right) 
 \left(\sum_{\bt \in \mT_{N | J}} \phi_{\bt} \right).
\end{equation}

Visually, the above equation can be understood as 
\begin{center}
    \input{figures/gravity_graph_30.tex}
\end{center}

\noindent This product of three graphs can be simplified to 

\begin{center}
    \tikzset{every picture/.style={line width=0.75pt}} 

\begin{tikzpicture}[x=0.75pt,y=0.75pt,yscale=-1,xscale=1]

\draw    (407.3,122.07) -- (364.5,135.01) ;
\draw    (271.47,99.01) -- (269.5,135.01) ;
\draw    (317,135.01) -- (325.97,171.01) ;
\draw    (317,135.01) -- (364.5,135.01) ;
\draw    (222,135.01) -- (269.5,135.01) ;
\draw    (269.5,135.01) -- (317,135.01) ;
\draw  [fill={rgb, 255:red, 255; green, 255; blue, 255 }  ,fill opacity=1 ] (260.53,135.01) .. controls (260.53,130.05) and (264.55,126.04) .. (269.5,126.04) .. controls (274.45,126.04) and (278.47,130.05) .. (278.47,135.01) .. controls (278.47,139.96) and (274.45,143.98) .. (269.5,143.98) .. controls (264.55,143.98) and (260.53,139.96) .. (260.53,135.01) -- cycle ;
\draw    (383.47,168.02) -- (364.5,135.01) ;
\draw    (364.5,135.01) -- (381.77,104.42) ;
\draw    (203,177.51) -- (222,135.01) ;
\draw    (222,135.01) -- (178.5,125.51) ;
\draw    (178.5,125.51) -- (147,144.01) ;
\draw    (159,90.51) -- (178.5,125.51) ;
\draw  [draw opacity=0] (138.49,193.09) .. controls (128.64,176.07) and (123,156.3) .. (123,135.22) .. controls (123,114.14) and (128.64,94.38) .. (138.49,77.36) -- (238.71,135.22) -- cycle ; \draw   (138.49,193.09) .. controls (128.64,176.07) and (123,156.3) .. (123,135.22) .. controls (123,114.14) and (128.64,94.38) .. (138.49,77.36) ;  
\draw  [fill={rgb, 255:red, 255; green, 255; blue, 255 }  ,fill opacity=1 ] (150.03,90.51) .. controls (150.03,85.55) and (154.05,81.54) .. (159,81.54) .. controls (163.95,81.54) and (167.97,85.55) .. (167.97,90.51) .. controls (167.97,95.46) and (163.95,99.48) .. (159,99.48) .. controls (154.05,99.48) and (150.03,95.46) .. (150.03,90.51) -- cycle ;
\draw  [fill={rgb, 255:red, 255; green, 255; blue, 255 }  ,fill opacity=1 ] (169.53,125.51) .. controls (169.53,120.55) and (173.55,116.54) .. (178.5,116.54) .. controls (183.45,116.54) and (187.47,120.55) .. (187.47,125.51) .. controls (187.47,130.46) and (183.45,134.48) .. (178.5,134.48) .. controls (173.55,134.48) and (169.53,130.46) .. (169.53,125.51) -- cycle ;
\draw  [fill={rgb, 255:red, 255; green, 255; blue, 255 }  ,fill opacity=1 ] (194.03,177.51) .. controls (194.03,172.55) and (198.05,168.54) .. (203,168.54) .. controls (207.95,168.54) and (211.97,172.55) .. (211.97,177.51) .. controls (211.97,182.46) and (207.95,186.48) .. (203,186.48) .. controls (198.05,186.48) and (194.03,182.46) .. (194.03,177.51) -- cycle ;
\draw  [fill={rgb, 255:red, 255; green, 255; blue, 255 }  ,fill opacity=1 ] (213.03,135.01) .. controls (213.03,130.05) and (217.05,126.04) .. (222,126.04) .. controls (226.95,126.04) and (230.97,130.05) .. (230.97,135.01) .. controls (230.97,139.96) and (226.95,143.98) .. (222,143.98) .. controls (217.05,143.98) and (213.03,139.96) .. (213.03,135.01) -- cycle ;
\draw  [fill={rgb, 255:red, 255; green, 255; blue, 255 }  ,fill opacity=1 ] (138.03,144.01) .. controls (138.03,139.05) and (142.05,135.04) .. (147,135.04) .. controls (151.95,135.04) and (155.97,139.05) .. (155.97,144.01) .. controls (155.97,148.96) and (151.95,152.98) .. (147,152.98) .. controls (142.05,152.98) and (138.03,148.96) .. (138.03,144.01) -- cycle ;
\draw  [fill={rgb, 255:red, 255; green, 255; blue, 255 }  ,fill opacity=1 ] (308.03,135.01) .. controls (308.03,130.05) and (312.05,126.04) .. (317,126.04) .. controls (321.95,126.04) and (325.97,130.05) .. (325.97,135.01) .. controls (325.97,139.96) and (321.95,143.98) .. (317,143.98) .. controls (312.05,143.98) and (308.03,139.96) .. (308.03,135.01) -- cycle ;
\draw  [fill={rgb, 255:red, 255; green, 255; blue, 255 }  ,fill opacity=1 ] (262.5,99.01) .. controls (262.5,94.05) and (266.52,90.04) .. (271.47,90.04) .. controls (276.43,90.04) and (280.44,94.05) .. (280.44,99.01) .. controls (280.44,103.96) and (276.43,107.98) .. (271.47,107.98) .. controls (266.52,107.98) and (262.5,103.96) .. (262.5,99.01) -- cycle ;
\draw  [fill={rgb, 255:red, 255; green, 255; blue, 255 }  ,fill opacity=1 ] (372.8,104.42) .. controls (372.8,99.46) and (376.82,95.45) .. (381.77,95.45) .. controls (386.73,95.45) and (390.74,99.46) .. (390.74,104.42) .. controls (390.74,109.37) and (386.73,113.39) .. (381.77,113.39) .. controls (376.82,113.39) and (372.8,109.37) .. (372.8,104.42) -- cycle ;
\draw  [fill={rgb, 255:red, 255; green, 255; blue, 255 }  ,fill opacity=1 ] (398.33,122.07) .. controls (398.33,117.11) and (402.35,113.09) .. (407.3,113.09) .. controls (412.25,113.09) and (416.27,117.11) .. (416.27,122.07) .. controls (416.27,127.02) and (412.25,131.04) .. (407.3,131.04) .. controls (402.35,131.04) and (398.33,127.02) .. (398.33,122.07) -- cycle ;
\draw  [draw opacity=0] (416.92,77.43) .. controls (426.77,94.45) and (432.4,114.21) .. (432.4,135.29) .. controls (432.4,156.37) and (426.77,176.13) .. (416.92,193.16) -- (316.7,135.29) -- cycle ; \draw   (416.92,77.43) .. controls (426.77,94.45) and (432.4,114.21) .. (432.4,135.29) .. controls (432.4,156.37) and (426.77,176.13) .. (416.92,193.16) ;  
\draw  [fill={rgb, 255:red, 255; green, 255; blue, 255 }  ,fill opacity=1 ] (355.53,135.01) .. controls (355.53,130.05) and (359.55,126.04) .. (364.5,126.04) .. controls (369.45,126.04) and (373.47,130.05) .. (373.47,135.01) .. controls (373.47,139.96) and (369.45,143.98) .. (364.5,143.98) .. controls (359.55,143.98) and (355.53,139.96) .. (355.53,135.01) -- cycle ;
\draw  [fill={rgb, 255:red, 255; green, 255; blue, 255 }  ,fill opacity=1 ] (374.5,168.02) .. controls (374.5,163.06) and (378.52,159.05) .. (383.47,159.05) .. controls (388.43,159.05) and (392.44,163.06) .. (392.44,168.02) .. controls (392.44,172.97) and (388.43,176.99) .. (383.47,176.99) .. controls (378.52,176.99) and (374.5,172.97) .. (374.5,168.02) -- cycle ;
\draw  [fill={rgb, 255:red, 255; green, 255; blue, 255 }  ,fill opacity=1 ] (317,171.01) .. controls (317,166.05) and (321.02,162.04) .. (325.97,162.04) .. controls (330.93,162.04) and (334.94,166.05) .. (334.94,171.01) .. controls (334.94,175.96) and (330.93,179.98) .. (325.97,179.98) .. controls (321.02,179.98) and (317,175.96) .. (317,171.01) -- cycle ;

\draw (115.91,135.41) node [anchor=east] [inner sep=0.75pt]  [font=\normalsize]  {$=-i\dfrac{\langle IJ\rangle ^{4}\mathrm{dvol}}{[ \xi I][ \chi J]}$};
\draw (222,135.01) node  [font=\footnotesize]  {$I$};
\draw (364.5,135.01) node  [font=\footnotesize]  {$J$};
\draw (245.75,131.61) node [anchor=south] [inner sep=0.75pt]    {$\xi $};
\draw (340.75,131.61) node [anchor=south] [inner sep=0.75pt]    {$\chi $};
\draw (252,217.6) node [anchor=north west][inner sep=0.75pt]   [align=left] {+ all other similarly structured graphs};

\end{tikzpicture}
\end{center}

\noindent where, by our construction, we note that the $\xi$ edge connected to $I$ and the $\chi$ edge connected to $J$ both lie on the unique path from $I$ to $J$.  Note we have made use of the using the pictoral identity 
\begin{center}
    \tikzset{every picture/.style={line width=0.75pt}} 

\begin{tikzpicture}[x=0.75pt,y=0.75pt,yscale=-1,xscale=1]

\draw    (225,92.7) -- (298,92.7) ;
\draw  [fill={rgb, 255:red, 255; green, 255; blue, 255 }  ,fill opacity=1 ] (285.5,92.7) .. controls (285.5,85.8) and (291.1,80.2) .. (298,80.2) .. controls (304.9,80.2) and (310.5,85.8) .. (310.5,92.7) .. controls (310.5,99.6) and (304.9,105.2) .. (298,105.2) .. controls (291.1,105.2) and (285.5,99.6) .. (285.5,92.7) -- cycle ;
\draw  [fill={rgb, 255:red, 255; green, 255; blue, 255 }  ,fill opacity=1 ] (225,75.02) -- (242.68,92.7) -- (225,110.38) -- (207.32,92.7) -- cycle ;
\draw  [fill={rgb, 255:red, 255; green, 255; blue, 255 }  ,fill opacity=1 ] (152.6,92.7) .. controls (152.6,85.8) and (158.2,80.2) .. (165.1,80.2) .. controls (172,80.2) and (177.6,85.8) .. (177.6,92.7) .. controls (177.6,99.6) and (172,105.2) .. (165.1,105.2) .. controls (158.2,105.2) and (152.6,99.6) .. (152.6,92.7) -- cycle ;
\draw    (366,92.7) -- (439,92.7) ;
\draw  [fill={rgb, 255:red, 255; green, 255; blue, 255 }  ,fill opacity=1 ] (426.5,92.7) .. controls (426.5,85.8) and (432.1,80.2) .. (439,80.2) .. controls (445.9,80.2) and (451.5,85.8) .. (451.5,92.7) .. controls (451.5,99.6) and (445.9,105.2) .. (439,105.2) .. controls (432.1,105.2) and (426.5,99.6) .. (426.5,92.7) -- cycle ;
\draw  [fill={rgb, 255:red, 255; green, 255; blue, 255 }  ,fill opacity=1 ] (353.5,92.7) .. controls (353.5,85.8) and (359.1,80.2) .. (366,80.2) .. controls (372.9,80.2) and (378.5,85.8) .. (378.5,92.7) .. controls (378.5,99.6) and (372.9,105.2) .. (366,105.2) .. controls (359.1,105.2) and (353.5,99.6) .. (353.5,92.7) -- cycle ;

\draw (225,92.7) node    {$I$};
\draw (298,92.7) node    {$i$};
\draw (165.1,92.7) node    {$I$};
\draw (331.77,94.2) node  [font=\Large]  {$=$};
\draw (366,92.7) node    {$I$};
\draw (439,92.7) node    {$i$};
\draw (404.22,87) node [anchor=south] [inner sep=0.75pt]    {$\xi $};
\draw (192.77,93.2) node  [font=\Large]  {$\times $};
\draw (261.5,89.3) node [anchor=south] [inner sep=0.75pt]    {$\xi $};

\end{tikzpicture}
\end{center}
to glue the three graphs together.

We will define the set of graph which arise in this sum as follows.

\begin{definition}
    Let $\mT^{\xi, \chi}_{N|I,J}$ denote the set of all marked tree graphs with distinct nodes from the set $\{1, \ldots, N\} \cup \{ I, J\}$ which necessarily have the nodes $I$ and $J$ in graph, and for which the two edges connected to $I$ and $J$ that lie on the path from $I$ to $J$ have the $\xi$ and $\chi$ labels, respectively.
\end{definition}

\begin{ex} \secret

    \begin{figure}[H]
    \begin{center}
        \tikzset{every picture/.style={line width=0.75pt}} 

\begin{tikzpicture}[x=0.75pt,y=0.75pt,yscale=-1,xscale=1]

\draw    (112.5,278.7) -- (134,320.16) ;
\draw    (136.5,320.2) -- (139,369.2) ;
\draw    (136.5,320.2) -- (86.5,329.66) ;
\draw    (182,239.16) -- (185,288.7) ;
\draw    (185,288.7) -- (136.5,320.2) ;
\draw    (218.5,333.2) -- (185,288.7) ;
\draw    (269,311.2) -- (218.5,333.2) ;
\draw    (269,311.2) -- (296,359.7) ;
\draw    (317.5,222.16) -- (310.5,271.2) ;
\draw    (357,292.16) -- (310.5,271.2) ;
\draw    (310.5,271.2) -- (269,311.2) ;
\draw  [fill={rgb, 255:red, 255; green, 255; blue, 255 }  ,fill opacity=1 ] (124,320.2) .. controls (124,313.3) and (129.6,307.7) .. (136.5,307.7) .. controls (143.4,307.7) and (149,313.3) .. (149,320.2) .. controls (149,327.1) and (143.4,332.7) .. (136.5,332.7) .. controls (129.6,332.7) and (124,327.1) .. (124,320.2) -- cycle ;
\draw  [fill={rgb, 255:red, 255; green, 255; blue, 255 }  ,fill opacity=1 ] (172.5,288.7) .. controls (172.5,281.8) and (178.1,276.2) .. (185,276.2) .. controls (191.9,276.2) and (197.5,281.8) .. (197.5,288.7) .. controls (197.5,295.6) and (191.9,301.2) .. (185,301.2) .. controls (178.1,301.2) and (172.5,295.6) .. (172.5,288.7) -- cycle ;
\draw  [fill={rgb, 255:red, 255; green, 255; blue, 255 }  ,fill opacity=1 ] (206,333.2) .. controls (206,326.3) and (211.6,320.7) .. (218.5,320.7) .. controls (225.4,320.7) and (231,326.3) .. (231,333.2) .. controls (231,340.1) and (225.4,345.7) .. (218.5,345.7) .. controls (211.6,345.7) and (206,340.1) .. (206,333.2) -- cycle ;
\draw  [fill={rgb, 255:red, 255; green, 255; blue, 255 }  ,fill opacity=1 ] (256.5,311.2) .. controls (256.5,304.3) and (262.1,298.7) .. (269,298.7) .. controls (275.9,298.7) and (281.5,304.3) .. (281.5,311.2) .. controls (281.5,318.1) and (275.9,323.7) .. (269,323.7) .. controls (262.1,323.7) and (256.5,318.1) .. (256.5,311.2) -- cycle ;
\draw  [fill={rgb, 255:red, 255; green, 255; blue, 255 }  ,fill opacity=1 ] (298,271.2) .. controls (298,264.3) and (303.6,258.7) .. (310.5,258.7) .. controls (317.4,258.7) and (323,264.3) .. (323,271.2) .. controls (323,278.1) and (317.4,283.7) .. (310.5,283.7) .. controls (303.6,283.7) and (298,278.1) .. (298,271.2) -- cycle ;
\draw  [fill={rgb, 255:red, 255; green, 255; blue, 255 }  ,fill opacity=1 ] (344.5,292.16) .. controls (344.5,285.26) and (350.1,279.66) .. (357,279.66) .. controls (363.9,279.66) and (369.5,285.26) .. (369.5,292.16) .. controls (369.5,299.07) and (363.9,304.66) .. (357,304.66) .. controls (350.1,304.66) and (344.5,299.07) .. (344.5,292.16) -- cycle ;
\draw  [fill={rgb, 255:red, 255; green, 255; blue, 255 }  ,fill opacity=1 ] (305,222.16) .. controls (305,215.26) and (310.6,209.66) .. (317.5,209.66) .. controls (324.4,209.66) and (330,215.26) .. (330,222.16) .. controls (330,229.07) and (324.4,234.66) .. (317.5,234.66) .. controls (310.6,234.66) and (305,229.07) .. (305,222.16) -- cycle ;
\draw  [fill={rgb, 255:red, 255; green, 255; blue, 255 }  ,fill opacity=1 ] (169.5,239.16) .. controls (169.5,232.26) and (175.1,226.66) .. (182,226.66) .. controls (188.9,226.66) and (194.5,232.26) .. (194.5,239.16) .. controls (194.5,246.07) and (188.9,251.66) .. (182,251.66) .. controls (175.1,251.66) and (169.5,246.07) .. (169.5,239.16) -- cycle ;
\draw  [fill={rgb, 255:red, 255; green, 255; blue, 255 }  ,fill opacity=1 ] (74,329.66) .. controls (74,322.76) and (79.6,317.16) .. (86.5,317.16) .. controls (93.4,317.16) and (99,322.76) .. (99,329.66) .. controls (99,336.57) and (93.4,342.16) .. (86.5,342.16) .. controls (79.6,342.16) and (74,336.57) .. (74,329.66) -- cycle ;
\draw  [fill={rgb, 255:red, 255; green, 255; blue, 255 }  ,fill opacity=1 ] (283.5,359.7) .. controls (283.5,352.8) and (289.1,347.2) .. (296,347.2) .. controls (302.9,347.2) and (308.5,352.8) .. (308.5,359.7) .. controls (308.5,366.6) and (302.9,372.2) .. (296,372.2) .. controls (289.1,372.2) and (283.5,366.6) .. (283.5,359.7) -- cycle ;
\draw  [fill={rgb, 255:red, 255; green, 255; blue, 255 }  ,fill opacity=1 ] (126.5,369.2) .. controls (126.5,362.3) and (132.1,356.7) .. (139,356.7) .. controls (145.9,356.7) and (151.5,362.3) .. (151.5,369.2) .. controls (151.5,376.1) and (145.9,381.7) .. (139,381.7) .. controls (132.1,381.7) and (126.5,376.1) .. (126.5,369.2) -- cycle ;
\draw  [fill={rgb, 255:red, 255; green, 255; blue, 255 }  ,fill opacity=1 ] (100,278.7) .. controls (100,271.8) and (105.6,266.2) .. (112.5,266.2) .. controls (119.4,266.2) and (125,271.8) .. (125,278.7) .. controls (125,285.6) and (119.4,291.2) .. (112.5,291.2) .. controls (105.6,291.2) and (100,285.6) .. (100,278.7) -- cycle ;

\draw (136.5,320.2) node    {$I$};
\draw (310.5,271.2) node    {$J$};
\draw (285,288.3) node [anchor=south] [inner sep=0.75pt]    {$\chi $};
\draw (154.72,302.5) node [anchor=south] [inner sep=0.75pt]    {$\xi $};
\draw (296,359.7) node    {$1$};
\draw (317.5,222.16) node    {$2$};
\draw (218.5,333.2) node    {$3$};
\draw (185,288.7) node    {$4$};
\draw (86.5,329.66) node    {$5$};
\draw (357,292.16) node    {$6$};
\draw (182,239.16) node    {$7$};
\draw (269,311.2) node    {$8$};
\draw (139,369.2) node    {$9$};
\draw (112.5,278.7) node    {$10$};

\end{tikzpicture}
    \end{center}
    \caption{\label{graph_10_2}  An example of a graph in $\mT^{\xi, \chi}_{10|I,J}$ .  Notice that $\xi$ and $\chi$ lie on the path from $I$ to $J$.} 
    
    \end{figure}
\end{ex}

We now write the generating function as sum over the graphs in $\mT^{\xi, \chi}_{N|I,J}$, as

\begin{equation}\label{Sigma_tree}
    \Sigma_0^{AB} \wedge (\gamma^I)_{AC} \wedge (\gamma^J)^C_{\;\;B} = -i  \,\mathrm{dvol}  \frac{\langle I J \rangle^4 }{[\xi I] [\chi J]} \sum_{\bt \,  \in \, \mT^{\xi, \chi}_{N|I,J}} \phi_{\bt}.
\end{equation}

\noindent Let us now write out $\phi_\bt$ on the RHS of the above equation so that its dependence on the placement of the $I$ and $J$ nodes in each graph is written explicitly. To that end, we define a graph-dependent label $P_I(\bt)$ as follows.

\begin{definition}
    For any graph $\bt \in \mT^{\xi, \chi}_{N|I,J}$, let $P_I(\bt)$ be defined as the label of the node connected to $I$ which lies on the path to node $J$. Likewise, define $P_J(\bt)$ be defined as the node connected to $J$ lying on the path to $I$.
\end{definition}

\begin{ex}
    In the graph $\bt$ from Figure \ref{graph_10_2}, we have
    \begin{equation}
        P_I(\bt) = 4, \hspace{1 cm} P_J(\bt) = 8.
    \end{equation}
\end{ex}

\begin{definition}
    For any graph $\bt \in \mT^{\xi, \chi}_{N|I,J}$, define the factor $K(\bt)$ as
    \begin{equation}
        K(\bt) \equiv \begin{cases}
            \eval{\dfrac{[\xi a]}{[I a]} \dfrac{[\chi b]}{[J b]}}_{\substack{a = P_I(\bt) \\ b = P_J(\bt)  }} & \text{if $I$ and $J$ are not neighbors} \\
            \dfrac{[\xi \chi]}{[I J]} & \text{if $I$ and $J$ are neighbors} \\
        \end{cases}.
    \end{equation}
\end{definition}

The factor $K(\bt)$ has the effect of changing $[I| \to [\xi|$ and $[J| \to [\chi|$ in numerators of the edges connecting $I$ and $J$ to $P_I(\bt)$ and $P_J(\bt)$, respectively. We can use it write
\begin{equation}
    \phi_\bt = K(\bt) \left( \prod_{e_{ij} \, \in \, \text{edges of } \bt}  \frac{[ij]}{\langle i j \rangle } \right) \left( \prod_{k \, \in \, \text{nodes of } \bt } \phi_k \right) \hspace{0.5 cm} \text{ for } \bt \in \mT^{\xi, \chi}_{N|I,J}.
\end{equation}

We are now ready to plug \eqref{Sigma_tree} into \eqref{MHV_eq1} to get the MHV amplitude. The effect of the derivatives with respect to the $\epsilon_i$'s is to extract only the trees which contain all possible $N+2$ nodes.

\begin{definition}
    Let $\mT^{\xi, \chi, \mathrm{max}}_{N|I,J} \subset \mT^{\xi, \chi}_{N|I,J}$ denote all graphs in $\mT^{\xi, \chi}_{N|I,J}$ which contain all $N+2$ nodes in $\{1, \ldots, N \} \cup \{I, J\}$.
\end{definition}

Plugging \eqref{Sigma_tree} into \eqref{MHV_eq1}, we come to the naive expression
\begin{equation}
    \mathcal{A}_{\rm MHV} = - i \, (2 \pi)^4 \delta^{(4)}\left( \sum_{i = 1}^{N+2} p_i \right) \frac{\langle I J \rangle^4}{[\xi I][\chi J] } \sum_{\bt \in \mT^{\xi, \chi, \mathrm{max}}_{N|I,J} } K(\bt)  \left( \prod_{e_{ij} \, \in \, \text{edges of } \bt}  \frac{[ij]}{\langle i j \rangle } \right).
\end{equation}
However, this expression has an issue: it is not invariant under rescalings of the spinor variables $|i\rangle \mapsto t | i \rangle$, $|i] \mapsto t^{-1} |i]$ because the graphs we are summing over have different little group weights for the various spinor variables. The above equation only holds for our particular spinor parameterization \eqref{kappa_def}, and would fail if we were to rescale the spinors.

To remedy this, we shall now add two auxiliary spinors $|\iota\rangle$ and $|o \rangle$ to this expression to make it reparameterization invariant. We start by setting
\begin{equation}
    \iota^A = \begin{pmatrix}
        0 \\ 1
    \end{pmatrix}, \hspace{1 cm} o^A = \begin{pmatrix}
        0 \\ 1
    \end{pmatrix},
\end{equation}
such that in our parameterization \eqref{kappa_def}
\begin{equation}
    \langle \iota i \rangle = \langle o i \rangle = 1.
\end{equation}
We now multiply our amplitude by $1$ in such a way that the final expression has the correct little group weights. The resulting amplitude (with $ \mathcal{A} = i (2 \pi)^4 \delta^{(4)}( \ldots) \mathcal{M}$) is
\begin{align}\label{MHV_to_eval}
    &\mathcal{M}_{\rm MHV} = \\
    &-\frac{\langle I J \rangle^4 \langle I \iota \rangle \langle I o \rangle \langle J \iota \rangle  \langle J o \rangle }{[\xi I][\chi J] } \left( \sum_{\bt \in \mT^{\xi, \chi, \mathrm{max}}_{N|I,J} } K(\bt)   \prod_{e_{ij} \, \in \, \text{edges of } \bt}  \frac{[ij]}{\langle i j \rangle }   \prod_{a = 1 }^{N+2} \Big( \langle a \iota \rangle \langle a o \rangle \Big)^{\mathrm{deg}(a) - 2}\right) \nonumber
\end{align}
where $\mathrm{deg}(a)$ is defined as the number of edges connected to node $a$. The above expression now has the correct little group weights of $+4$ for $I$ and $J$,  and $-4$ for all $i \neq I, J$, and 0 for $\iota$, $o$, $\xi$, and $\chi$. To check the value of all the weights, one must use
\begin{equation}
    \sum_{a \, \in \, \text{nodes of } \bt } (\mathrm{deg}(a) - 2) = -2
\end{equation}
which holds for any connected tree diagram $\bt$.

\eqref{MHV_to_eval} is in fact a new generalization of the NSVW formula. It is expressed as a sum over trees where all $N+2$ gravitons have a corresponding node, instead of just the $N$ positive helicity gravitons as in the NSVW formula. It should be emphasized, however, that the two negative helicity graviton nodes $I$ and $J$ obey different rules from the other nodes, as one of each of their edges depends on a reference spinor $\xi$ and $\chi$ respectively. We now show how it can be reduced to the NSVW formula.

\begin{claim}\label{claim_const}
    Equation \eqref{MHV_to_eval} does not depend on the reference spinors $| \iota \rangle$ and $|o\rangle$.
\end{claim}

\begin{proof}
    Without loss of generality we focus on $| \iota \rangle$. Because the expression is invariant under the rescaling of $| \iota \rangle$, we can think of the two components $(\iota_{1}, \iota_{2} )$ as homogenous coordinates and the expression \eqref{MHV_to_eval} as a function on $\mathbb{C}P^1$. We can therefore prove that \eqref{MHV_to_eval} is constant as a function of $| \iota \rangle$ if we show that the expression has no poles in $| \iota \rangle$.

    The only potential poles occur for $\langle i \iota \rangle \to 0$ when $i \neq I, J$ and $i$ is a node with $\mathrm{deg}(i) = 1$. Clearly, all graphs where $\mathrm{deg}(i) = 1$ can be formed by taking a graph without $i$ and appending $i$ to any of the pre-existing $N+1$ nodes the graph. Summing over the $N+1$ possible nodes we could append $i$ to introduces a multiplicative factor of
    \begin{equation}
        \frac{1}{\langle i \iota \rangle \langle i o \rangle } \sum_{\substack{j = 1 \\ j \neq i } }^{N+2} \frac{[i j]}{\langle i j \rangle} \langle j \iota \rangle \langle j o \rangle
    \end{equation}
    to an otherwise nonsingular term. But then notice that the residue of the amplitude is
    \begin{equation}
    \begin{aligned}
    \underset{\langle i \iota \rangle \to 0}{\Res} \mathcal{M}_{\rm MHV} &= 
        \underset{\langle i \iota \rangle \to 0}{\Res} \left( \frac{1}{\langle i \iota \rangle \langle i o \rangle } \sum_{\substack{j = 1 \\ j \neq i } }^{N+2} \frac{[i j]}{\langle i j \rangle} \langle j \iota \rangle \langle j o \rangle \right) (\text{non sing.}) \\&= - \frac{1}{\langle i o \rangle} [ i| \left( \sum_{j = 1 }^{N+2} | j]  \langle j | \right) |o \rangle (\text{non sing.})= 0
    \end{aligned}
    \end{equation}
    and vanishes due to momentum conservation. The overall function therefore has no poles in $| \iota \rangle$, completing the proof.
\end{proof}

Using Claim \ref{claim_const}, we can now redefine the auxiliary spinors to be
\begin{equation}\label{iota_sub}
    | \iota \rangle \to | I \rangle, \hspace{1 cm} | o \rangle \to | J \rangle,
\end{equation}
without changing the amplitude.

Under this change, the only graphs in equation \eqref{MHV_to_eval} which give non-zero contributions are graphs where $\mathrm{deg}(I) = \mathrm{deg}(J) = 1$. An example is shown below.

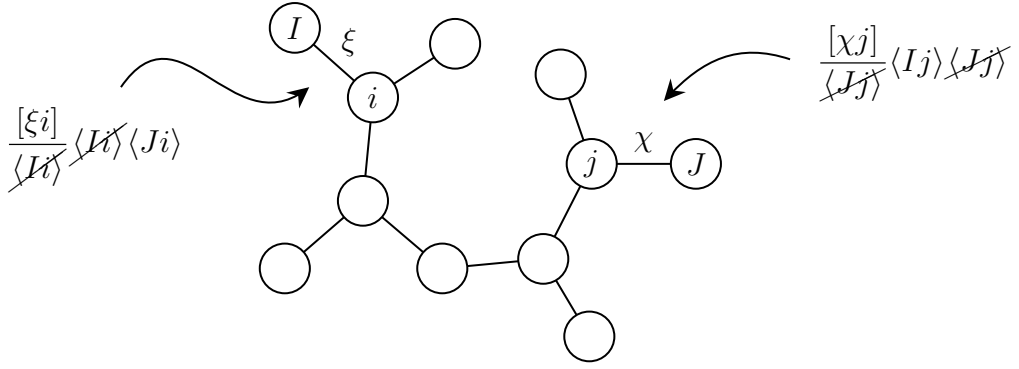
\begin{center}
    \tikzset{every picture/.style={line width=0.75pt}} 

\begin{tikzpicture}[x=0.75pt,y=0.75pt,yscale=-1,xscale=1]

\draw    (353.5,96.5) -- (369,141.5) ;
\draw    (369,141.5) -- (421,141.5) ;
\draw    (220.8,73.1) -- (259.8,108.1) ;
\draw    (259.8,108.1) -- (300.8,81.1) ;
\draw    (254.8,160.1) -- (259.8,108.1) ;
\draw    (254.8,160.1) -- (215.8,194.1) ;
\draw    (254.8,160.1) -- (294.4,194.1) ;
\draw    (344.8,189.3) -- (294.4,194.1) ;
\draw    (369,141.5) -- (344.8,189.3) ;
\draw    (344.8,189.3) -- (367.8,228.3) ;
\draw  [fill={rgb, 255:red, 255; green, 255; blue, 255 }  ,fill opacity=1 ] (208.3,73.1) .. controls (208.3,66.2) and (213.9,60.6) .. (220.8,60.6) .. controls (227.7,60.6) and (233.3,66.2) .. (233.3,73.1) .. controls (233.3,80) and (227.7,85.6) .. (220.8,85.6) .. controls (213.9,85.6) and (208.3,80) .. (208.3,73.1) -- cycle ;
\draw  [fill={rgb, 255:red, 255; green, 255; blue, 255 }  ,fill opacity=1 ] (247.3,108.1) .. controls (247.3,101.2) and (252.9,95.6) .. (259.8,95.6) .. controls (266.7,95.6) and (272.3,101.2) .. (272.3,108.1) .. controls (272.3,115) and (266.7,120.6) .. (259.8,120.6) .. controls (252.9,120.6) and (247.3,115) .. (247.3,108.1) -- cycle ;
\draw  [fill={rgb, 255:red, 255; green, 255; blue, 255 }  ,fill opacity=1 ] (288.3,81.1) .. controls (288.3,74.2) and (293.9,68.6) .. (300.8,68.6) .. controls (307.7,68.6) and (313.3,74.2) .. (313.3,81.1) .. controls (313.3,88) and (307.7,93.6) .. (300.8,93.6) .. controls (293.9,93.6) and (288.3,88) .. (288.3,81.1) -- cycle ;
\draw  [fill={rgb, 255:red, 255; green, 255; blue, 255 }  ,fill opacity=1 ] (242.3,160.1) .. controls (242.3,153.2) and (247.9,147.6) .. (254.8,147.6) .. controls (261.7,147.6) and (267.3,153.2) .. (267.3,160.1) .. controls (267.3,167) and (261.7,172.6) .. (254.8,172.6) .. controls (247.9,172.6) and (242.3,167) .. (242.3,160.1) -- cycle ;
\draw  [fill={rgb, 255:red, 255; green, 255; blue, 255 }  ,fill opacity=1 ] (203.3,194.1) .. controls (203.3,187.2) and (208.9,181.6) .. (215.8,181.6) .. controls (222.7,181.6) and (228.3,187.2) .. (228.3,194.1) .. controls (228.3,201) and (222.7,206.6) .. (215.8,206.6) .. controls (208.9,206.6) and (203.3,201) .. (203.3,194.1) -- cycle ;
\draw  [fill={rgb, 255:red, 255; green, 255; blue, 255 }  ,fill opacity=1 ] (281.9,194.1) .. controls (281.9,187.2) and (287.5,181.6) .. (294.4,181.6) .. controls (301.3,181.6) and (306.9,187.2) .. (306.9,194.1) .. controls (306.9,201) and (301.3,206.6) .. (294.4,206.6) .. controls (287.5,206.6) and (281.9,201) .. (281.9,194.1) -- cycle ;
\draw  [fill={rgb, 255:red, 255; green, 255; blue, 255 }  ,fill opacity=1 ] (332.3,189.3) .. controls (332.3,182.4) and (337.9,176.8) .. (344.8,176.8) .. controls (351.7,176.8) and (357.3,182.4) .. (357.3,189.3) .. controls (357.3,196.2) and (351.7,201.8) .. (344.8,201.8) .. controls (337.9,201.8) and (332.3,196.2) .. (332.3,189.3) -- cycle ;
\draw  [fill={rgb, 255:red, 255; green, 255; blue, 255 }  ,fill opacity=1 ] (356.5,141.5) .. controls (356.5,134.6) and (362.1,129) .. (369,129) .. controls (375.9,129) and (381.5,134.6) .. (381.5,141.5) .. controls (381.5,148.4) and (375.9,154) .. (369,154) .. controls (362.1,154) and (356.5,148.4) .. (356.5,141.5) -- cycle ;
\draw  [fill={rgb, 255:red, 255; green, 255; blue, 255 }  ,fill opacity=1 ] (355.3,228.3) .. controls (355.3,221.4) and (360.9,215.8) .. (367.8,215.8) .. controls (374.7,215.8) and (380.3,221.4) .. (380.3,228.3) .. controls (380.3,235.2) and (374.7,240.8) .. (367.8,240.8) .. controls (360.9,240.8) and (355.3,235.2) .. (355.3,228.3) -- cycle ;
\draw  [fill={rgb, 255:red, 255; green, 255; blue, 255 }  ,fill opacity=1 ] (408.5,141.5) .. controls (408.5,134.6) and (414.1,129) .. (421,129) .. controls (427.9,129) and (433.5,134.6) .. (433.5,141.5) .. controls (433.5,148.4) and (427.9,154) .. (421,154) .. controls (414.1,154) and (408.5,148.4) .. (408.5,141.5) -- cycle ;
\draw  [fill={rgb, 255:red, 255; green, 255; blue, 255 }  ,fill opacity=1 ] (341,96.5) .. controls (341,89.6) and (346.6,84) .. (353.5,84) .. controls (360.4,84) and (366,89.6) .. (366,96.5) .. controls (366,103.4) and (360.4,109) .. (353.5,109) .. controls (346.6,109) and (341,103.4) .. (341,96.5) -- cycle ;
\draw    (134,102.91) .. controls (163.4,59.79) and (188.96,130.96) .. (225.73,104.65) ;
\draw [shift={(228,102.91)}, rotate = 140.79] [fill={rgb, 255:red, 0; green, 0; blue, 0 }  ][line width=0.08]  [draw opacity=0] (10.72,-5.15) -- (0,0) -- (10.72,5.15) -- (7.12,0) -- cycle    ;
\draw    (468,88.91) .. controls (441.95,80.23) and (424.27,89.24) .. (406.89,107.85) ;
\draw [shift={(405,109.91)}, rotate = 311.99] [fill={rgb, 255:red, 0; green, 0; blue, 0 }  ][line width=0.08]  [draw opacity=0] (10.72,-5.15) -- (0,0) -- (10.72,5.15) -- (7.12,0) -- cycle    ;
\draw  [draw opacity=0] (46,71) .. controls (46,57.19) and (57.19,46) .. (71,46) .. controls (84.81,46) and (96,57.19) .. (96,71) .. controls (96,84.81) and (84.81,96) .. (71,96) .. controls (57.19,96) and (46,84.81) .. (46,71) -- cycle ;

\draw (242.58,88.7) node [anchor=south west] [inner sep=0.75pt]    {$\xi $};
\draw (395,137.1) node [anchor=south] [inner sep=0.75pt]    {$\chi $};
\draw (220.8,73.1) node    {$I$};
\draw (421,141.5) node    {$J$};
\draw (369,141.5) node    {$j$};
\draw (259.8,108.1) node    {$i$};
\draw (75,111.4) node [anchor=north west][inner sep=0.75pt]    {$\dfrac{[ \xi i]}{\cancel{\langle Ii\rangle }}\cancel{\langle Ii\rangle } \langle Ji\rangle $};
\draw (480,70.4) node [anchor=north west][inner sep=0.75pt]    {$\dfrac{[ \chi j]}{\cancel{\langle Jj\rangle }} \langle Ij\rangle \cancel{\langle Jj\rangle }$};

\end{tikzpicture}
\end{center}

\noindent All such graphs can be formed by taking a graph in a set we'll denote $\mT_N^{\rm max}$ and appending the $I$ and $J$ nodes with the $\xi$ and $\chi$ edges. 

\begin{definition}
    Let $\mT^{\rm max}_N \subset \mT_N$ denote the tree graphs which each contain the maximum number of all $N$ nodes in $\{1, \ldots, N \}$. 
\end{definition}

We will now write \eqref{MHV_to_eval} with the substitutions \eqref{iota_sub} as an explicit sum over graphs in $\mT_N^{\rm max}$ with an additional factor included for all the $N \times N$ ways to appropriately attach $I$ and $J$ nodes to those underlying graphs:

\begin{equation}\label{MHV_almost}
\begin{aligned}
    \mathcal{M}_{\rm MHV} = - \frac{\langle I J \rangle^4}{[\xi I][\chi J]} &\left( \sum_{i = 1}^N [\xi i] \langle i J \rangle \right) \left( \sum_{j = 1}^N [\chi j] \langle j I \rangle \right) \\
    &\times \left( \sum_{\bt \in \mT_N^{\rm max}} \prod_{e_{ij}\,  \in \, \text{edges of } \bt} \frac{[ij]}{\langle i j \rangle} \prod_{a = 1}^{N} \Big( \langle a I \rangle \langle a J \rangle \Big)^{\mathrm{deg}(a)-2} \right).
\end{aligned}
\end{equation}

\noindent For our final step, we note that from momentum conservation
\begin{equation}
\begin{aligned}
[ \xi| \left( | I ] \langle I | + | J ] \langle J| +  \sum_{i = 1}^N |i ] \langle i | \right)| J \rangle = 0 \\
[ \chi| \left( | I ] \langle I | + | J ] \langle J| +  \sum_{j = 1}^N |j ] \langle j | \right)| I \rangle = 0
\end{aligned}
\end{equation}
we have
\begin{equation}
    \frac{1}{[\xi I] [\chi J]} \left( \sum_{i = 1}^N [ \xi i ] \langle i J \rangle \right) \left( \sum_{j = 1}^N [ \chi j ] \langle j I \rangle \right) = -\langle I J \rangle^2.
\end{equation}
Plugging the above expression into \eqref{MHV_almost}, we arrive at the NSVW formula \cite{Nguyen:2009jk}, which is our second main result.
\begin{equation}\label{nsvw}
    \mathcal{M}_{\rm MHV} = \langle I J \rangle^6  \sum_{\bt \, \in \, \mT_N^{\rm max}} \left( \prod_{e_{ij} \in \, \text{edges of }\, \bt} \frac{[ij]}{\langle i j \rangle}  \right) \left( \prod_{a = 1}^N \Big( \langle a I \rangle \langle a J \rangle \Big)^{\mathrm{deg}(a)-2} \right)
\end{equation}

\section{Discussion}\label{sec6}

In this work, we developed a perturbiner expansion for Plebański's second heavenly equation which was written as a sum over marked tree graphs, given in Theorem \ref{thm_21}. This allowed us to explicitly write down expressions for spacetime metrics containing arbitrary numbers of positive helicity gravitons in the Plebański gauge \eqref{eq2}. This gauge is particularly natural because it can be thought of as an exact, non-linear version of transverse traceless gauge, satisfying
\begin{equation}
     \eta^{\mu \nu} h_{\mu \nu} = g^{\mu \nu} h_{\mu \nu} = 0 , \hspace{1 cm}
     \eta^{\mu \nu} \partial_\mu h_{\nu \rho} =  g^{\mu \nu} \partial_\mu h_{\nu \rho} =  0.
\end{equation}

Using this perturbiner expansion we were able to write down visual formulae for a variety of mathematical objects familiar to the study of self-dual gravity, and we were able to manipulate and combine these formulae in useful ways using a relatively small number of graphical tricks.  Ultimately, this allowed us to give a self contained proof of the NSVW formula for graviton MHV scattering.

While we have not computed any amplitudes aside from MHV, it seems likely that our marked tree expansion could be a useful tool in the computation of other amplitudes. We leave this for future work.

Our tree diagrams strongly resonate with other tree diagrams that have appeared previously \cite{Adamo:2021bej, Adamo:2012xe, Adamo:2013tja, Adamo:2013cra, Cao:2021dcd, Adamo:2024hme}. Notably, \cite{Adamo:2021bej} proved the NSVW/Hodges' formula utilizing a similar looking set of diagrams which have an unclear relationship to our diagrams. The main difference between their proof of NSVW and our proof is that their proof was based upon the use of the \emph{first} Plebański scalar (also known as the Kähler scalar), whereas our work is based upon the use of the \emph{second} Plebański scalar. Due to their use of the first scalar, they were able to choose a gauge in which the expressions for the ASD spin connections $\gamma^I$ and $\gamma^J$ became trivial, whereas in this work a good deal of effort was expended accounting for the spin connections.

The unique offering of this paper is a new theory of the second scalar.  What interpretation our formulae have in twistor-theoretic terms will be left for future work. Perhaps the use of the second scalar may be advantageous in certain situations.

\section*{Acknowledgements}

This author benefited tremendously from a large number in-depth conversations with Atul Sharma, Tim Adamo, and Eduardo Casali, to whom author is extremely grateful. The author further benefited from the gracious expertise of Lionel Mason, David Skinner, and Simone Speziale, for which he is also grateful. Further insightful conversations were had with Abhay Ashtekar, Adam Ball, Wei Bu, Pratik Chattopadhyay, Erin Crawley, Theo Coyne, Dan Freed, Alfredo Guevara, Elizabeth Himwich, Per Kraus, Jacob McNamara, Walker Melton, Hirosi Ooguri, Andrea Puhm, Marcus Spradlin, Andrew Strominger, Adam Tropper, Anastasia Volovich, and Tianli Wang.  This work was supported by  DOE Grant de-sc/0007870, the Simons Collaboration for Celestial Holography, and NSF GRFP
grant DGE1745303.

\appendix

\section{A binary tree expansion for Plebański's second heavenly equation}\label{appA}

In section \ref{sec2}, we saw the perturbiner expansion for $\phi$ given $N$ seed functions could be written as a sum over trees in $\mT_N$. We will now turn our attention to a different (but related) method of expressing the perturbiner expansion based a \emph{recursion relation}. Essentially, if one knows what the perturbiner expansion is for $N-1$ gravitons, there is a formula which allows one to find the expansion for $N$ gravitons.

In order to write down this recursive formula, we must first break up our perturbiner expansion into a sum of sub-terms labelled by the seed functions present in that term. We will denote these sub-terms as $\phi^{(k)}$, with $1 \leq k \leq N$, where $\phi^{(k)}$ depends linearly on $k$ seed functions appearing in its arguments.
\begin{equation}
    \bigphi\,(\phi_1, \ldots, \phi_N) = \sum_{k = 1}^N \sum_{ \substack{\{i_1, \ldots, i_k\}   \subset \{1, \ldots, N\} }  } \bigphi^{(k)}( \phi_{i_1}, \ldots, \phi_{i_k})
\end{equation}

\begin{ex}
    For $N=1$,
    \begin{equation}
        \bigphi(\phi_1) = \bigphi^{(1)}(\phi_1)
    \end{equation}
    where
    \begin{equation}
        \bigphi^{(1)}(\phi_i) = \phi_i.
    \end{equation}
    For $N=2$,
    \begin{equation}
        \bigphi(\phi_1, \phi_2) = \bigphi^{(1)}(\phi_1)+\bigphi^{(1)}(\phi_2)+\bigphi^{(2)}(\phi_1, \phi_2)
    \end{equation}
    where
    \begin{equation}
        \bigphi^{(2)}(\phi_i, \phi_j) = \frac{D^{ij}}{z_{ij}} \phi_i \phi_j.
    \end{equation}
    For $N=3$,
    \begin{equation}
    \begin{aligned}
        \bigphi(\phi_1, \phi_2, \phi_3) &= \bigphi^{(1)}(\phi_1)+\bigphi^{(1)}(\phi_2)+\bigphi^{(3)}(\phi_3)\\
        &+\bigphi^{(2)}(\phi_1, \phi_2)+ \bigphi^{(2)}(\phi_2, \phi_3)+\bigphi^{(2)}(\phi_1, \phi_3)\\
        &+\bigphi^{(3)}(\phi_1, \phi_2, \phi_3)
    \end{aligned}
    \end{equation}
    where
    \begin{equation}\label{phi_3}
        \bigphi^{(3)}(\phi_i, \phi_j, \phi_k) = \left(\frac{D^{ij}}{z_{ij}}\frac{D^{jk}}{z_{jk}} + \frac{D^{ik}}{z_{ik}}\frac{D^{kj}}{z_{kj}} +  \frac{D^{ji}}{z_{ji}}\frac{D^{ik}}{z_{ik}} \right) \phi_i \phi_j \phi_k.
    \end{equation}
    
\end{ex}

Notice that the $\phi^{(N)}$ expansion of $N$ seed functions is a sum of all marked tree graphs in $\mT^{\rm max}_N$, which are graphs that contain all $N$ labeled nodes.
\begin{equation}
    \bigphi^{(N)}(\phi_1, \ldots, \phi_N) = \sum_{\bt \, \in \, \mT^{\rm max}_N } \phi_{\bt}.
\end{equation}

We are now ready to give the recursive formula which allows one to solve for the perturbiner expansion of $N$ particles in terms of the perturbiner expansion of $N-1$ particles. We note that it is closely related to the ``inverse soft'' construction \cite{Krasnov:2013wsa}.

\begin{thm} 
\begin{equation}\label{gravity_soft_thm_to_prove}
    \tcboxmath{\bigphi^{(N)} ( \phi_1, \ldots, \phi_N ) = \sum_{i = 1}^{N-1} \bigphi^{(N-1)}  ( \phi_1, \ldots, \phi_{N-1} ) \bigg\rvert_{\displaystyle{  \phi_i \mapsto \frac{1}{z_{iN}} \{ \phi_i, \phi_N\} }  } }
\end{equation}

\end{thm}

\begin{proof}
    Let us first explain the notation used in the above formula. On the right we have a sum over all $N-1$ original seed functions $\phi_i$ where we replace them with $\tfrac{1}{z_{iN}} \{ \phi_i, \phi_N \}$. This means that if a term in the original expression includes a certain set of seed functions with a certain set of derivatives acting on $\phi_i$, then we must replace $\phi_i$ with $\tfrac{1}{z_{iN}} \{ \phi_i, \phi_N \}$ in the term such that those same derivatives are now to act on $\tfrac{1}{z_{iN}} \{ \phi_i, \phi_N \}$ instead. However, all of the pre-existing factors of $1/{z_{ij}}$, corresponding to nodes $j$ which previously attached to $i$, will remain the same.
    
    In order to understand this recursive operation diagrammatically, so that we may prove it, we will need to introduce two new objects to our graphical dictionary.

    The first new object we define is a straight lined edge with an arrow on it. This edge works quite differently than all previous edges.
    \begin{center}
        \tikzset{every picture/.style={line width=0.75pt}} 

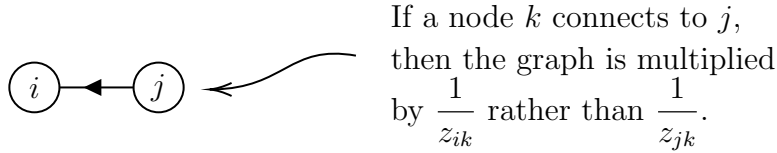
\begin{tikzpicture}[x=0.75pt,y=0.75pt,yscale=-1,xscale=1]

\draw    (109,68.2) -- (146,68.2) ;
\draw [shift={(121,68.2)}, rotate = 0] [fill={rgb, 255:red, 0; green, 0; blue, 0 }  ][line width=0.08]  [draw opacity=0] (8.93,-4.29) -- (0,0) -- (8.93,4.29) -- cycle    ;
\draw   (84,68.2) .. controls (84,61.3) and (89.6,55.7) .. (96.5,55.7) .. controls (103.4,55.7) and (109,61.3) .. (109,68.2) .. controls (109,75.1) and (103.4,80.7) .. (96.5,80.7) .. controls (89.6,80.7) and (84,75.1) .. (84,68.2) -- cycle ;
\draw   (146,68.2) .. controls (146,61.3) and (151.6,55.7) .. (158.5,55.7) .. controls (165.4,55.7) and (171,61.3) .. (171,68.2) .. controls (171,75.1) and (165.4,80.7) .. (158.5,80.7) .. controls (151.6,80.7) and (146,75.1) .. (146,68.2) -- cycle ;
\draw    (257,52.21) .. controls (223.51,46.3) and (222.04,68.53) .. (185.69,69.2) ;
\draw [shift={(184,69.21)}, rotate = 360] [color={rgb, 255:red, 0; green, 0; blue, 0 }  ][line width=0.75]    (10.93,-3.29) .. controls (6.95,-1.4) and (3.31,-0.3) .. (0,0) .. controls (3.31,0.3) and (6.95,1.4) .. (10.93,3.29)   ;

\draw (96.5,68.2) node    {$i$};
\draw (158.5,68.2) node    {$j$};
\draw (273,25) node [anchor=north west][inner sep=0.75pt]   [align=left] {If a node $\displaystyle k$ connects to $\displaystyle j$, \\then the graph is multiplied\\by $\displaystyle \dfrac{1}{z_{ik}}$ rather than $\displaystyle \dfrac{1}{z_{jk}}$.};

\end{tikzpicture}
    \end{center}
    In the edge above, if another node $k$ connects to the node $j$, then one should not multiply the graph by $1/z_{jk}$ as per usual, but one should rather multiply it by $1/z_{ik}$. The numerator, $D^{jk}$, should stay the same.
    
    \begin{ex}\secret
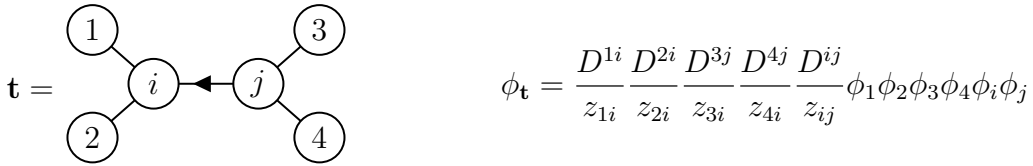
    \begin{center}
        \tikzset{every picture/.style={line width=0.75pt}} 

\begin{tikzpicture}[x=0.75pt,y=0.75pt,yscale=-1,xscale=1]

\draw    (104.28,57.95) -- (73.8,30.9) ;
\draw    (73.8,85) -- (104.28,57.95) ;
\draw    (187.08,85) -- (156.6,57.95) ;
\draw    (156.6,57.95) -- (187.08,30.9) ;
\draw    (104.28,57.95) -- (156.6,57.95) ;
\draw [shift={(123.94,57.95)}, rotate = 0] [fill={rgb, 255:red, 0; green, 0; blue, 0 }  ][line width=0.08]  [draw opacity=0] (8.93,-4.29) -- (0,0) -- (8.93,4.29) -- cycle    ;
\draw  [fill={rgb, 255:red, 255; green, 255; blue, 255 }  ,fill opacity=1 ] (91.78,57.95) .. controls (91.78,51.05) and (97.38,45.45) .. (104.28,45.45) .. controls (111.19,45.45) and (116.78,51.05) .. (116.78,57.95) .. controls (116.78,64.85) and (111.19,70.45) .. (104.28,70.45) .. controls (97.38,70.45) and (91.78,64.85) .. (91.78,57.95) -- cycle ;
\draw  [fill={rgb, 255:red, 255; green, 255; blue, 255 }  ,fill opacity=1 ] (144.1,57.95) .. controls (144.1,51.05) and (149.7,45.45) .. (156.6,45.45) .. controls (163.5,45.45) and (169.1,51.05) .. (169.1,57.95) .. controls (169.1,64.85) and (163.5,70.45) .. (156.6,70.45) .. controls (149.7,70.45) and (144.1,64.85) .. (144.1,57.95) -- cycle ;
\draw  [fill={rgb, 255:red, 255; green, 255; blue, 255 }  ,fill opacity=1 ] (61.3,30.9) .. controls (61.3,24) and (66.9,18.4) .. (73.8,18.4) .. controls (80.7,18.4) and (86.3,24) .. (86.3,30.9) .. controls (86.3,37.8) and (80.7,43.4) .. (73.8,43.4) .. controls (66.9,43.4) and (61.3,37.8) .. (61.3,30.9) -- cycle ;
\draw  [fill={rgb, 255:red, 255; green, 255; blue, 255 }  ,fill opacity=1 ] (61.3,85) .. controls (61.3,78.1) and (66.9,72.5) .. (73.8,72.5) .. controls (80.7,72.5) and (86.3,78.1) .. (86.3,85) .. controls (86.3,91.9) and (80.7,97.5) .. (73.8,97.5) .. controls (66.9,97.5) and (61.3,91.9) .. (61.3,85) -- cycle ;
\draw  [fill={rgb, 255:red, 255; green, 255; blue, 255 }  ,fill opacity=1 ] (174.58,30.9) .. controls (174.58,24) and (180.18,18.4) .. (187.08,18.4) .. controls (193.99,18.4) and (199.58,24) .. (199.58,30.9) .. controls (199.58,37.8) and (193.99,43.4) .. (187.08,43.4) .. controls (180.18,43.4) and (174.58,37.8) .. (174.58,30.9) -- cycle ;
\draw  [fill={rgb, 255:red, 255; green, 255; blue, 255 }  ,fill opacity=1 ] (174.58,85) .. controls (174.58,78.1) and (180.18,72.5) .. (187.08,72.5) .. controls (193.99,72.5) and (199.58,78.1) .. (199.58,85) .. controls (199.58,91.9) and (193.99,97.5) .. (187.08,97.5) .. controls (180.18,97.5) and (174.58,91.9) .. (174.58,85) -- cycle ;

\draw (104.28,57.95) node    {$i$};
\draw (156.6,57.95) node    {$j$};
\draw (73.8,30.9) node    {$1$};
\draw (73.8,85) node    {$2$};
\draw (187.08,30.9) node    {$3$};
\draw (187.08,85) node    {$4$};
\draw (276.1,58.95) node [anchor=west] [inner sep=0.75pt]    {$\phi _{\mathbf{t}} =\dfrac{D^{1i}}{z_{1i}}\dfrac{D^{2i}}{z_{2i}}\dfrac{D^{3j}}{z_{3i}}\dfrac{D^{4j}}{z_{4i}}\dfrac{D^{ij}}{z_{ij}} \phi _{1} \phi _{2} \phi _{3} \phi _{4} \phi _{i} \phi _{j}$};
\draw (56.08,60.2) node [anchor=east] [inner sep=0.75pt]    {$\mathbf{t} =$};

\end{tikzpicture}
    \end{center}
    \end{ex}

    The second new object we will introduce is the ``gray dotted circle.'' If we have some sort of internal sub-graph contained in the circle, as well as some external nodes connecting to the circle, then this graph corresponds to the \emph{sum of all possible connections} from the external nodes to the internal nodes. For instance, if there are $N_i$ nodes inside the circle and $N_e$ external nodes connected to the circle, then there will be $N_i^{N_e}$ graphs to sum.

    \begin{ex}\secret
    \begin{figure}[H]
\begin{center}
    \input{figures/gravity_graph_13}
\end{center}        \caption{An example of how the gray dotted circle is used. The circle is an instruction to sum over all ways for the external nodes to connect to the internal nodes.}
        \label{fig:564t24315}
    \end{figure}
        
    \end{ex}

We will use these two new graphical elements to diagrammatically express the RHS of \eqref{gravity_soft_thm_to_prove}. Taking $N = 5$, consider the action of replacing $\phi_4 \mapsto \frac{1}{z_{45}} \{ \phi_4, \phi_5\}$ in the following graph.
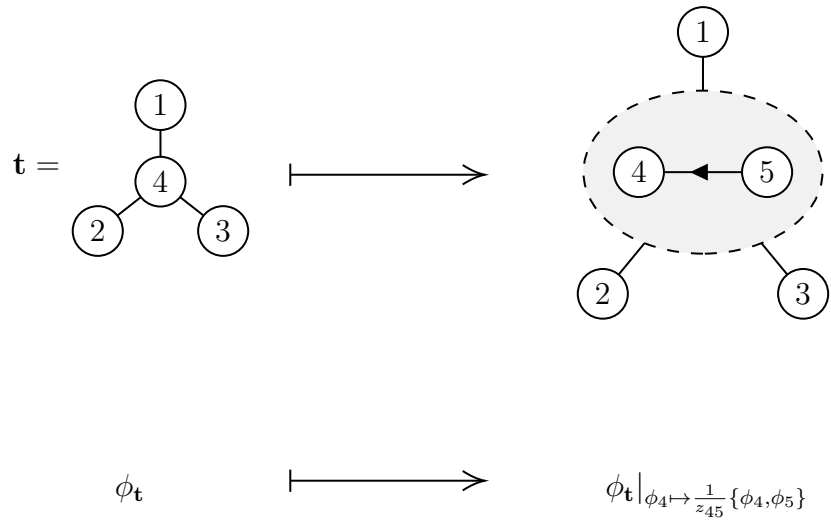
\begin{center}
    \tikzset{every picture/.style={line width=0.75pt}} 

\begin{tikzpicture}[x=0.75pt,y=0.75pt,yscale=-1,xscale=1]

\draw    (412,151.91) -- (462,90.7) ;
\draw    (512,151.91) -- (462,90.7) ;
\draw    (222.5,120.31) -- (191.2,95.49) ;
\draw    (159.9,120.31) -- (191.2,95.49) ;
\draw    (191.2,95.49) -- (191.2,57) ;
\draw  [fill={rgb, 255:red, 255; green, 255; blue, 255 }  ,fill opacity=1 ] (178.7,57) .. controls (178.7,50.1) and (184.3,44.5) .. (191.2,44.5) .. controls (198.1,44.5) and (203.7,50.1) .. (203.7,57) .. controls (203.7,63.9) and (198.1,69.5) .. (191.2,69.5) .. controls (184.3,69.5) and (178.7,63.9) .. (178.7,57) -- cycle ;
\draw  [fill={rgb, 255:red, 255; green, 255; blue, 255 }  ,fill opacity=1 ] (210,120.31) .. controls (210,113.41) and (215.6,107.81) .. (222.5,107.81) .. controls (229.4,107.81) and (235,113.41) .. (235,120.31) .. controls (235,127.21) and (229.4,132.81) .. (222.5,132.81) .. controls (215.6,132.81) and (210,127.21) .. (210,120.31) -- cycle ;
\draw  [fill={rgb, 255:red, 255; green, 255; blue, 255 }  ,fill opacity=1 ] (147.4,120.31) .. controls (147.4,113.41) and (153,107.81) .. (159.9,107.81) .. controls (166.8,107.81) and (172.4,113.41) .. (172.4,120.31) .. controls (172.4,127.21) and (166.8,132.81) .. (159.9,132.81) .. controls (153,132.81) and (147.4,127.21) .. (147.4,120.31) -- cycle ;
\draw  [fill={rgb, 255:red, 255; green, 255; blue, 255 }  ,fill opacity=1 ] (178.7,95.49) .. controls (178.7,88.59) and (184.3,82.99) .. (191.2,82.99) .. controls (198.1,82.99) and (203.7,88.59) .. (203.7,95.49) .. controls (203.7,102.39) and (198.1,107.99) .. (191.2,107.99) .. controls (184.3,107.99) and (178.7,102.39) .. (178.7,95.49) -- cycle ;
\draw  [fill={rgb, 255:red, 241; green, 241; blue, 241 }  ,fill opacity=1 ][dash pattern={on 4.5pt off 4.5pt}] (462,131.7) .. controls (429.06,131.7) and (402.36,113.34) .. (402.36,90.7) .. controls (402.36,68.06) and (429.06,49.7) .. (462,49.7) .. controls (494.94,49.7) and (521.64,68.06) .. (521.64,90.7) .. controls (521.64,113.34) and (494.94,131.7) .. (462,131.7) -- cycle ;
\draw    (440,90.7) -- (484,90.7) ;
\draw [shift={(455.5,90.7)}, rotate = 0] [fill={rgb, 255:red, 0; green, 0; blue, 0 }  ][line width=0.08]  [draw opacity=0] (8.93,-4.29) -- (0,0) -- (8.93,4.29) -- cycle    ;
\draw  [fill={rgb, 255:red, 255; green, 255; blue, 255 }  ,fill opacity=1 ] (430,103.2) .. controls (423.1,103.2) and (417.5,97.6) .. (417.5,90.7) .. controls (417.5,83.8) and (423.1,78.2) .. (430,78.2) .. controls (436.9,78.2) and (442.5,83.8) .. (442.5,90.7) .. controls (442.5,97.6) and (436.9,103.2) .. (430,103.2) -- cycle ;
\draw  [color={rgb, 255:red, 0; green, 0; blue, 0 }  ,draw opacity=1 ][fill={rgb, 255:red, 255; green, 255; blue, 255 }  ,fill opacity=1 ] (494,103.2) .. controls (487.1,103.2) and (481.5,97.6) .. (481.5,90.7) .. controls (481.5,83.8) and (487.1,78.2) .. (494,78.2) .. controls (500.9,78.2) and (506.5,83.8) .. (506.5,90.7) .. controls (506.5,97.6) and (500.9,103.2) .. (494,103.2) -- cycle ;
\draw    (462,20.2) -- (462,49.7) ;
\draw  [fill={rgb, 255:red, 255; green, 255; blue, 255 }  ,fill opacity=1 ] (462,32.7) .. controls (455.1,32.7) and (449.5,27.1) .. (449.5,20.2) .. controls (449.5,13.3) and (455.1,7.7) .. (462,7.7) .. controls (468.9,7.7) and (474.5,13.3) .. (474.5,20.2) .. controls (474.5,27.1) and (468.9,32.7) .. (462,32.7) -- cycle ;
\draw  [fill={rgb, 255:red, 255; green, 255; blue, 255 }  ,fill opacity=1 ] (512,164.41) .. controls (505.1,164.41) and (499.5,158.82) .. (499.5,151.91) .. controls (499.5,145.01) and (505.1,139.41) .. (512,139.41) .. controls (518.9,139.41) and (524.5,145.01) .. (524.5,151.91) .. controls (524.5,158.82) and (518.9,164.41) .. (512,164.41) -- cycle ;
\draw  [fill={rgb, 255:red, 255; green, 255; blue, 255 }  ,fill opacity=1 ] (399.5,151.91) .. controls (399.5,145.01) and (405.1,139.41) .. (412,139.41) .. controls (418.9,139.41) and (424.5,145.01) .. (424.5,151.91) .. controls (424.5,158.82) and (418.9,164.41) .. (412,164.41) .. controls (405.1,164.41) and (399.5,158.82) .. (399.5,151.91) -- cycle ;
\draw    (256,92) -- (351,92) ;
\draw [shift={(353,92)}, rotate = 180] [color={rgb, 255:red, 0; green, 0; blue, 0 }  ][line width=0.75]    (10.93,-4.9) .. controls (6.95,-2.3) and (3.31,-0.67) .. (0,0) .. controls (3.31,0.67) and (6.95,2.3) .. (10.93,4.9)   ;
\draw [shift={(256,92)}, rotate = 180] [color={rgb, 255:red, 0; green, 0; blue, 0 }  ][line width=0.75]    (0,5.59) -- (0,-5.59)   ;
\draw    (256,245.91) -- (351,245.91) ;
\draw [shift={(353,245.91)}, rotate = 180] [color={rgb, 255:red, 0; green, 0; blue, 0 }  ][line width=0.75]    (10.93,-4.9) .. controls (6.95,-2.3) and (3.31,-0.67) .. (0,0) .. controls (3.31,0.67) and (6.95,2.3) .. (10.93,4.9)   ;
\draw [shift={(256,245.91)}, rotate = 180] [color={rgb, 255:red, 0; green, 0; blue, 0 }  ][line width=0.75]    (0,5.59) -- (0,-5.59)   ;

\draw (116,80) node [anchor=north west][inner sep=0.75pt]    {$\mathbf{t} =$};
\draw (191.2,57) node    {$1$};
\draw (159.9,120.31) node    {$2$};
\draw (222.5,120.31) node    {$3$};
\draw (191.2,95.49) node    {$4$};
\draw (341.5,252.9) node    {$\phi _{\mathbf{t}} \ \ \ \ \ \ \ \ \ \ \ \ \ \ \ \ \ \ \ \ \ \ \ \ \ \ \ \ \ \ \ \ \ \ \ \ \ \ \ \ \ \ \ \ \phi _{\mathbf{t}} |_{\phi _{4} \mapsto \frac{1}{z_{45}}\{\phi _{4} ,\phi _{5}\}}$};
\draw (412,151.91) node    {$2$};
\draw (512,151.91) node    {$3$};
\draw (430,90.7) node    {$4$};
\draw (462,20.2) node    {$1$};
\draw (494,90.7) node    {$5$};

\end{tikzpicture}
\end{center}
We can see that the result of the operation $\phi_4 \mapsto \frac{1}{z_{45}} \{ \phi_4, \phi_5\}$ is to replace node 4 with a subgraph of two elements where 4 is connected to 5 by an arrowed edge, with the arrow pointing from 5 to 4.

\begin{definition}
    Let $\mT_N^{\rm max, \rightarrow}$ denote all marked tree graphs containing $N$ distinct labelled nodes in which an arrow points from the node ``$N$'' to exactly one adjacent node.
\end{definition}

\begin{ex}
    Here are some elements of $\mT_5^{\rm max, \rightarrow}.$
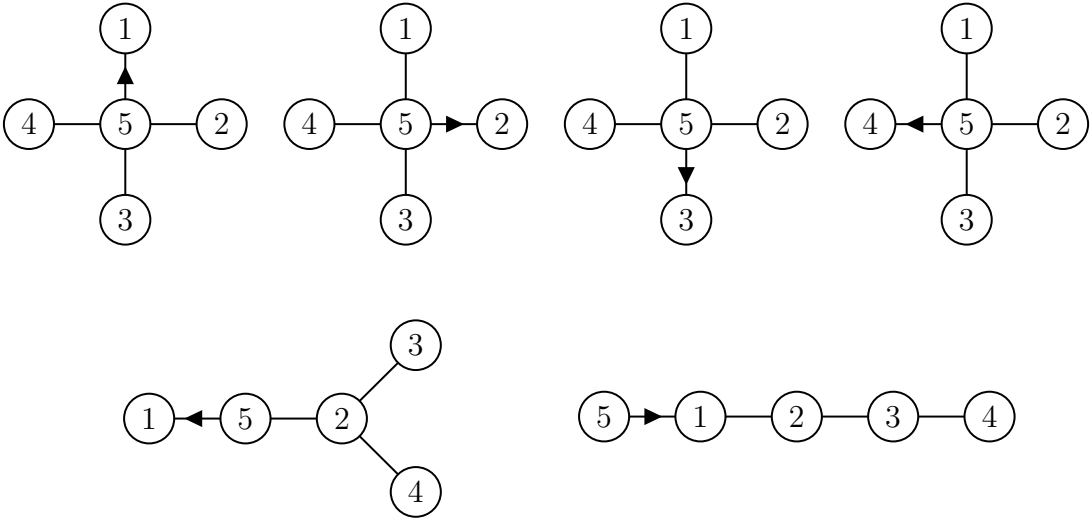
    \begin{center}
    \input{figures/gravity_graph_14}
\end{center}

\end{ex}

One can convince oneself that in the process of summing over $i = 1, \ldots, N-1$, and replacing $\phi_i \mapsto \tfrac{1}{z_{iN}} \{\phi_i, \phi_N\}$ in all of the graphs of $\mT_{N-1}^{\rm max}$, we will end up with a sum over all of the graphs in $\mT_N^{\rm max, \rightarrow}$. Therefore, the RHS of \eqref{gravity_soft_thm_to_prove} can be written as

\begin{equation}
    \sum_{i = 1}^{N-1} \bigphi^{(N-1)}  ( \phi_1, \ldots, \phi_{N-1} ) \bigg\rvert_{\displaystyle{  \phi_i \mapsto \frac{1}{z_{iN}} \{ \phi_i, \phi_N\} }  } = \sum_{\bt \in \mT_N^{\rm max, \rightarrow}} \phi_{\bt}.
\end{equation}
Therefore, our proof will be complete once we show that
\begin{equation}
    \sum_{\bt \in \mT_N^{\rm max, \rightarrow}} \phi_{\bt} \overset{?}{=}  \sum_{\bt \in \mT^{\rm max}_N} \phi_{\bt}.
\end{equation}

It is natural to put the graphs in $\mT_N^{\rm max, \rightarrow}$ in groups that look identical if the arrowed edge is converted into a normal edge. In fact, one can show that the sum of all such graphs in this group \emph{is} that very converted graph, as shown below!
\begin{figure}[H]
    \centering
\tikzset{every picture/.style={line width=0.75pt}} 

\begin{tikzpicture}[x=0.75pt,y=0.75pt,yscale=-1,xscale=1]

\draw    (245.5,103.31) -- (214.2,78.49) ;
\draw [shift={(234.94,94.94)}, rotate = 218.42] [fill={rgb, 255:red, 0; green, 0; blue, 0 }  ][line width=0.08]  [draw opacity=0] (8.93,-4.29) -- (0,0) -- (8.93,4.29) -- cycle    ;
\draw    (182.9,103.31) -- (214.2,78.49) ;
\draw    (214.2,78.49) -- (214.2,40) ;
\draw  [fill={rgb, 255:red, 255; green, 255; blue, 255 }  ,fill opacity=1 ] (201.7,40) .. controls (201.7,33.1) and (207.3,27.5) .. (214.2,27.5) .. controls (221.1,27.5) and (226.7,33.1) .. (226.7,40) .. controls (226.7,46.9) and (221.1,52.5) .. (214.2,52.5) .. controls (207.3,52.5) and (201.7,46.9) .. (201.7,40) -- cycle ;
\draw  [fill={rgb, 255:red, 255; green, 255; blue, 255 }  ,fill opacity=1 ] (233,103.31) .. controls (233,96.41) and (238.6,90.81) .. (245.5,90.81) .. controls (252.4,90.81) and (258,96.41) .. (258,103.31) .. controls (258,110.21) and (252.4,115.81) .. (245.5,115.81) .. controls (238.6,115.81) and (233,110.21) .. (233,103.31) -- cycle ;
\draw  [fill={rgb, 255:red, 255; green, 255; blue, 255 }  ,fill opacity=1 ] (170.4,103.31) .. controls (170.4,96.41) and (176,90.81) .. (182.9,90.81) .. controls (189.8,90.81) and (195.4,96.41) .. (195.4,103.31) .. controls (195.4,110.21) and (189.8,115.81) .. (182.9,115.81) .. controls (176,115.81) and (170.4,110.21) .. (170.4,103.31) -- cycle ;
\draw  [fill={rgb, 255:red, 255; green, 255; blue, 255 }  ,fill opacity=1 ] (201.7,78.49) .. controls (201.7,71.59) and (207.3,65.99) .. (214.2,65.99) .. controls (221.1,65.99) and (226.7,71.59) .. (226.7,78.49) .. controls (226.7,85.39) and (221.1,90.99) .. (214.2,90.99) .. controls (207.3,90.99) and (201.7,85.39) .. (201.7,78.49) -- cycle ;
\draw    (127.5,103.31) -- (96.2,78.49) ;
\draw    (64.9,103.31) -- (96.2,78.49) ;
\draw    (96.2,78.49) -- (96.2,40) ;
\draw [shift={(96.2,54.24)}, rotate = 90] [fill={rgb, 255:red, 0; green, 0; blue, 0 }  ][line width=0.08]  [draw opacity=0] (8.93,-4.29) -- (0,0) -- (8.93,4.29) -- cycle    ;
\draw  [fill={rgb, 255:red, 255; green, 255; blue, 255 }  ,fill opacity=1 ] (83.7,40) .. controls (83.7,33.1) and (89.3,27.5) .. (96.2,27.5) .. controls (103.1,27.5) and (108.7,33.1) .. (108.7,40) .. controls (108.7,46.9) and (103.1,52.5) .. (96.2,52.5) .. controls (89.3,52.5) and (83.7,46.9) .. (83.7,40) -- cycle ;
\draw  [fill={rgb, 255:red, 255; green, 255; blue, 255 }  ,fill opacity=1 ] (115,103.31) .. controls (115,96.41) and (120.6,90.81) .. (127.5,90.81) .. controls (134.4,90.81) and (140,96.41) .. (140,103.31) .. controls (140,110.21) and (134.4,115.81) .. (127.5,115.81) .. controls (120.6,115.81) and (115,110.21) .. (115,103.31) -- cycle ;
\draw  [fill={rgb, 255:red, 255; green, 255; blue, 255 }  ,fill opacity=1 ] (52.4,103.31) .. controls (52.4,96.41) and (58,90.81) .. (64.9,90.81) .. controls (71.8,90.81) and (77.4,96.41) .. (77.4,103.31) .. controls (77.4,110.21) and (71.8,115.81) .. (64.9,115.81) .. controls (58,115.81) and (52.4,110.21) .. (52.4,103.31) -- cycle ;
\draw  [fill={rgb, 255:red, 255; green, 255; blue, 255 }  ,fill opacity=1 ] (83.7,78.49) .. controls (83.7,71.59) and (89.3,65.99) .. (96.2,65.99) .. controls (103.1,65.99) and (108.7,71.59) .. (108.7,78.49) .. controls (108.7,85.39) and (103.1,90.99) .. (96.2,90.99) .. controls (89.3,90.99) and (83.7,85.39) .. (83.7,78.49) -- cycle ;
\draw    (366,103.31) -- (334.7,78.49) ;
\draw    (303.4,103.31) -- (334.7,78.49) ;
\draw [shift={(313.96,94.94)}, rotate = 321.58] [fill={rgb, 255:red, 0; green, 0; blue, 0 }  ][line width=0.08]  [draw opacity=0] (8.93,-4.29) -- (0,0) -- (8.93,4.29) -- cycle    ;
\draw    (334.7,78.49) -- (334.7,40) ;
\draw  [fill={rgb, 255:red, 255; green, 255; blue, 255 }  ,fill opacity=1 ] (322.2,40) .. controls (322.2,33.1) and (327.8,27.5) .. (334.7,27.5) .. controls (341.6,27.5) and (347.2,33.1) .. (347.2,40) .. controls (347.2,46.9) and (341.6,52.5) .. (334.7,52.5) .. controls (327.8,52.5) and (322.2,46.9) .. (322.2,40) -- cycle ;
\draw  [fill={rgb, 255:red, 255; green, 255; blue, 255 }  ,fill opacity=1 ] (353.5,103.31) .. controls (353.5,96.41) and (359.1,90.81) .. (366,90.81) .. controls (372.9,90.81) and (378.5,96.41) .. (378.5,103.31) .. controls (378.5,110.21) and (372.9,115.81) .. (366,115.81) .. controls (359.1,115.81) and (353.5,110.21) .. (353.5,103.31) -- cycle ;
\draw  [fill={rgb, 255:red, 255; green, 255; blue, 255 }  ,fill opacity=1 ] (290.9,103.31) .. controls (290.9,96.41) and (296.5,90.81) .. (303.4,90.81) .. controls (310.3,90.81) and (315.9,96.41) .. (315.9,103.31) .. controls (315.9,110.21) and (310.3,115.81) .. (303.4,115.81) .. controls (296.5,115.81) and (290.9,110.21) .. (290.9,103.31) -- cycle ;
\draw  [fill={rgb, 255:red, 255; green, 255; blue, 255 }  ,fill opacity=1 ] (322.2,78.49) .. controls (322.2,71.59) and (327.8,65.99) .. (334.7,65.99) .. controls (341.6,65.99) and (347.2,71.59) .. (347.2,78.49) .. controls (347.2,85.39) and (341.6,90.99) .. (334.7,90.99) .. controls (327.8,90.99) and (322.2,85.39) .. (322.2,78.49) -- cycle ;
\draw    (486.5,103.31) -- (455.2,78.49) ;
\draw    (423.9,103.31) -- (455.2,78.49) ;
\draw    (455.2,78.49) -- (455.2,40) ;
\draw  [fill={rgb, 255:red, 255; green, 255; blue, 255 }  ,fill opacity=1 ] (442.7,40) .. controls (442.7,33.1) and (448.3,27.5) .. (455.2,27.5) .. controls (462.1,27.5) and (467.7,33.1) .. (467.7,40) .. controls (467.7,46.9) and (462.1,52.5) .. (455.2,52.5) .. controls (448.3,52.5) and (442.7,46.9) .. (442.7,40) -- cycle ;
\draw  [fill={rgb, 255:red, 255; green, 255; blue, 255 }  ,fill opacity=1 ] (474,103.31) .. controls (474,96.41) and (479.6,90.81) .. (486.5,90.81) .. controls (493.4,90.81) and (499,96.41) .. (499,103.31) .. controls (499,110.21) and (493.4,115.81) .. (486.5,115.81) .. controls (479.6,115.81) and (474,110.21) .. (474,103.31) -- cycle ;
\draw  [fill={rgb, 255:red, 255; green, 255; blue, 255 }  ,fill opacity=1 ] (411.4,103.31) .. controls (411.4,96.41) and (417,90.81) .. (423.9,90.81) .. controls (430.8,90.81) and (436.4,96.41) .. (436.4,103.31) .. controls (436.4,110.21) and (430.8,115.81) .. (423.9,115.81) .. controls (417,115.81) and (411.4,110.21) .. (411.4,103.31) -- cycle ;
\draw  [fill={rgb, 255:red, 255; green, 255; blue, 255 }  ,fill opacity=1 ] (442.7,78.49) .. controls (442.7,71.59) and (448.3,65.99) .. (455.2,65.99) .. controls (462.1,65.99) and (467.7,71.59) .. (467.7,78.49) .. controls (467.7,85.39) and (462.1,90.99) .. (455.2,90.99) .. controls (448.3,90.99) and (442.7,85.39) .. (442.7,78.49) -- cycle ;

\draw (214.2,40) node    {$i_{1}$};
\draw (182.9,103.31) node    {$i_{3}$};
\draw (245.5,103.31) node    {$i_{2}$};
\draw (214.2,78.49) node    {$N$};
\draw (96.2,40) node    {$i_{1}$};
\draw (64.9,103.31) node    {$i_{3}$};
\draw (127.5,103.31) node    {$i_{2}$};
\draw (96.2,78.49) node    {$N$};
\draw (334.7,40) node    {$i_{1}$};
\draw (303.4,103.31) node    {$i_{3}$};
\draw (366,103.31) node    {$i_{2}$};
\draw (334.7,78.49) node    {$N$};
\draw (455.2,40) node    {$i_{1}$};
\draw (423.9,103.31) node    {$i_{3}$};
\draw (486.5,103.31) node    {$i_{2}$};
\draw (455.2,78.49) node    {$N$};
\draw (155.57,74.19) node  [font=\Large]  {$+$};
\draw (275.07,74.19) node  [font=\Large]  {$+$};
\draw (394.57,74.19) node  [font=\Large]  {$=$};

\end{tikzpicture}
    \caption{A node $N$ connects to $m = 3$ nodes $i_1, \ldots, i_m$. On the LHS we sum over all $m$ options for which one of the $m$ edges is chosen as the arrowed edge. On the RHS we have the graph with only normal edges. These graphs should be understood as a subpart of a larger graph, the details of which do not affect the present calculation.}
    \label{fig:63456}
\end{figure}
This equality is due to the elementary identity\footnote{This identity, which is a generalization of the partial fractions sum formula, is easiest to prove by noting that both sides have the same poles and residues in $z_N$, and have no zeros in $z_N$ aside from $z_N = \infty$.}
\begin{equation}
    \sum_{a = 1}^m \frac{1}{z_{Ni_a}} \prod^m_{\substack{b = 1 \\ b \neq a}} \frac{1}{z_{i_a i_b}} = \prod_{a=1}^m \frac{1}{z_{N i_a}}
\end{equation}
and this concludes the proof.

\end{proof}

\begin{ex}
Let us show how subsequent $\phi^{(N)}$'s can be computed using \eqref{gravity_soft_thm_to_prove}. We start with
\begin{equation}
    \bigphi^{(1)}(\phi_1) = \phi_1
\end{equation}
from which \eqref{gravity_soft_thm_to_prove} implies
\begin{equation}
    \bigphi^{(2)}(\phi_1, \phi_2) = \frac{1}{z_{12}}\{ \phi_1, \phi_2 \}.
\end{equation}
The first non-trivial example arises when computing $N = 3$ from $N=2$,
\begin{align}
    \bigphi^{(3)}(\phi_1, \phi_2, \phi_3) &= \frac{1}{z_{12}} \frac{1}{z_{13}} \{ \{ \phi_1, \phi_3\}, \phi_2 \} + \frac{1}{z_{12}} \frac{1}{z_{23}} \{\phi_1, \{ \phi_2, \phi_3\}\} \\
    &= \left( \frac{1}{z_{12}}\frac{1}{z_{13}} D^{13}(D^{12} + D^{32}) + \frac{1}{z_{12}}\frac{1}{z_{23}} D^{23}(D^{12} + D^{13}) \right) \phi_1 \phi_2 \phi_3 \nonumber
\end{align}
which after a rearrangement of terms can be shown to match \eqref{phi_3}. The expression for $\phi^{(4)}$ is
\begin{equation}
\begin{aligned}
    \bigphi^{(4)}(& \phi_1, \phi_2, \phi_3, \phi_4) \\
    &= \frac{1}{z_{14}} \left( \frac{1}{z_{12}} \frac{1}{z_{13}}  \{ \{ \{\phi_1, \phi_4 \} , \phi_3\}, \phi_2 \} + \frac{1}{z_{12}} \frac{1}{z_{23}} \{ \{ \phi_1, \phi_4 \}, \{ \phi_2, \phi_3\}\} \right) \\
    & + \frac{1}{z_{24}} \left( \frac{1}{z_{12}} \frac{1}{z_{13}} \{ \{ \phi_1, \phi_3\}, \{ \phi_2, \phi_4\} \} + \frac{1}{z_{12}} \frac{1}{z_{23}} \{\phi_1, \{ \{\phi_2, \phi_4\} , \phi_3\}\}  \right) \\
    & + \frac{1}{z_{34}} \left( \frac{1}{z_{12}} \frac{1}{z_{13}} \{ \{ \phi_1, \{\phi_3, \phi_4 \} \}, \phi_2 \} + \frac{1}{z_{12}} \frac{1}{z_{23}} \{\phi_1, \{ \phi_2, \{\phi_3, \phi_4 \} \}\} \right).
\end{aligned}
\end{equation}
In some sense, the recursive formula provides a ``more efficient'' way to calculate perturbiner expansions than the marked tree formula. For instance, the above expression for $\phi^{(4)}$ has 6 terms while the number of trees contained in $\mT_4^{\rm max}$ is 16. In general, number of terms which arise from the recursive formula is $(N-1)!$ while the number of corresponding trees is $N^{N-2}$, which grows much faster.

The way in which the nested Poisson brackets organize themselves after the repeated use of the recursive formula lends itself quite naturally to a representation in terms of \emph{binary} trees. We draw the first set of instances of these trees below.

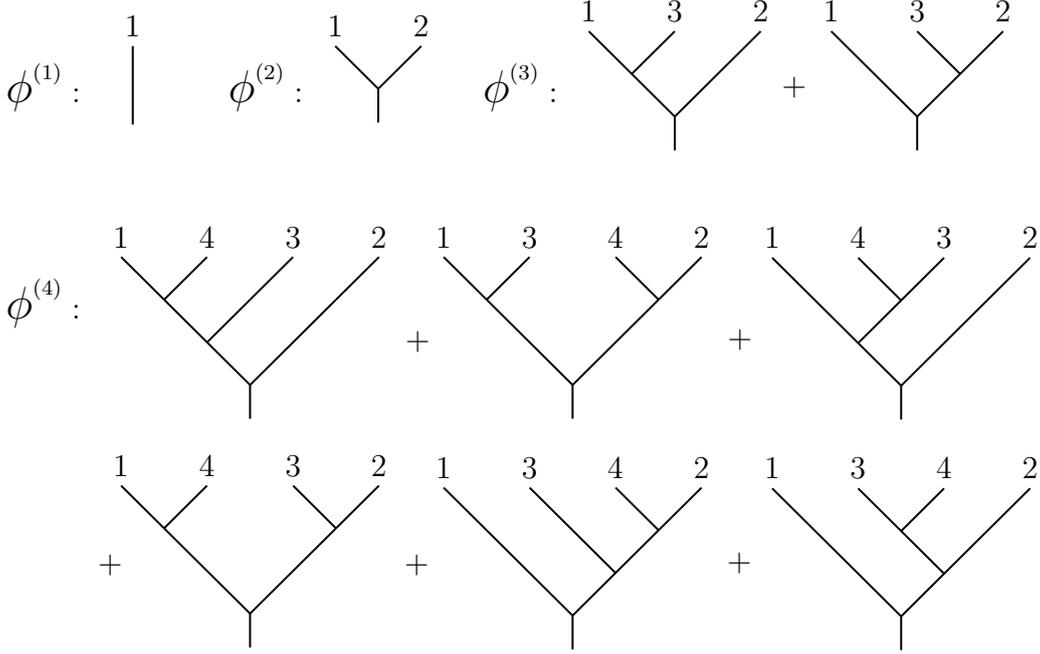
\begin{center}
\tikzset{every picture/.style={line width=0.75pt}} 

\begin{tikzpicture}[x=0.75pt,y=0.75pt,yscale=-1,xscale=1]

\draw    (254,115.91) -- (254,132.91) ;
\draw    (131.46,133.91) -- (131.46,94.46) ;
\draw    (402,129.91) -- (402,146.91) ;
\draw    (359.09,87) -- (402,129.91) ;
\draw    (444.91,87) -- (402,129.91) ;
\draw    (380.54,108.46) -- (402,87) ;
\draw    (523,129.91) -- (523,146.91) ;
\draw    (480.09,87) -- (523,129.91) ;
\draw    (565.91,87) -- (523,129.91) ;
\draw    (544.46,108.46) -- (523,87) ;
\draw    (254,115.91) -- (275.46,94.46) ;
\draw    (254,115.91) -- (232.54,94.46) ;
\draw    (190,264.91) -- (190,281.91) ;
\draw    (147.09,222) -- (190,264.91) ;
\draw    (232.91,222) -- (190,264.91) ;
\draw    (168.54,243.46) -- (190,222) ;
\draw    (190,379.91) -- (190,396.91) ;
\draw    (147.09,337) -- (190,379.91) ;
\draw    (232.91,337) -- (190,379.91) ;
\draw    (232.91,222) -- (254.37,200.54) ;
\draw    (232.91,337) -- (211.46,315.54) ;
\draw    (147.09,337) -- (168.54,315.54) ;
\draw    (147.09,222) -- (168.54,200.54) ;
\draw    (232.91,337) -- (254.37,315.54) ;
\draw    (190,222) -- (211.46,200.54) ;
\draw    (147.09,337) -- (125.63,315.54) ;
\draw    (147.09,222) -- (125.63,200.54) ;
\draw    (372.46,359.46) -- (351,338) ;
\draw    (351,338) -- (329.54,316.54) ;
\draw    (393.91,222) -- (372.46,200.54) ;
\draw    (308.09,222) -- (329.54,200.54) ;
\draw    (351,264.91) -- (351,281.91) ;
\draw    (308.09,222) -- (351,264.91) ;
\draw    (393.91,222) -- (351,264.91) ;
\draw    (351,380.91) -- (351,397.91) ;
\draw    (308.09,338) -- (351,380.91) ;
\draw    (393.91,338) -- (351,380.91) ;
\draw    (393.91,222) -- (415.37,200.54) ;
\draw    (393.91,338) -- (415.37,316.54) ;
\draw    (308.09,338) -- (286.63,316.54) ;
\draw    (308.09,222) -- (286.63,200.54) ;
\draw    (515,265.29) -- (515,282.29) ;
\draw    (472.09,222.37) -- (515,265.29) ;
\draw    (557.91,222.37) -- (515,265.29) ;
\draw    (557.91,222.37) -- (579.37,200.91) ;
\draw    (472.09,222.37) -- (450.63,200.91) ;
\draw    (515,381.29) -- (515,398.29) ;
\draw    (472.09,338.37) -- (515,381.29) ;
\draw    (557.91,338.37) -- (515,381.29) ;
\draw    (557.91,338.37) -- (579.37,316.91) ;
\draw    (472.09,338.37) -- (450.63,316.91) ;
\draw    (393.91,338) -- (372.46,316.54) ;
\draw    (536.46,359.83) -- (515,338.37) ;
\draw    (515,222.37) -- (493.54,200.91) ;
\draw    (515,338.37) -- (493.54,316.91) ;
\draw    (515,338.37) -- (536.46,316.91) ;
\draw    (515,222.37) -- (536.46,200.91) ;
\draw    (493.54,243.83) -- (515,222.37) ;

\draw (107,115) node [anchor=east] [inner sep=0.75pt]    {$\bigphi ^{( 1)} :$};
\draw (232.54,91.06) node [anchor=south] [inner sep=0.75pt]    {$1$};
\draw (275.46,91.06) node [anchor=south] [inner sep=0.75pt]    {$2$};
\draw (131.46,91.06) node [anchor=south] [inner sep=0.75pt]    {$1$};
\draw (359.09,83.6) node [anchor=south] [inner sep=0.75pt]    {$1$};
\draw (444.91,83.6) node [anchor=south] [inner sep=0.75pt]    {$2$};
\draw (480.09,83.6) node [anchor=south] [inner sep=0.75pt]    {$1$};
\draw (565.91,83.6) node [anchor=south] [inner sep=0.75pt]    {$2$};
\draw (402,83.6) node [anchor=south] [inner sep=0.75pt]    {$3$};
\draw (523,83.6) node [anchor=south] [inner sep=0.75pt]    {$3$};
\draw (125.63,197.14) node [anchor=south] [inner sep=0.75pt]    {$1$};
\draw (254.37,197.14) node [anchor=south] [inner sep=0.75pt]    {$2$};
\draw (211.46,197.14) node [anchor=south] [inner sep=0.75pt]    {$3$};
\draw (168.54,197.14) node [anchor=south] [inner sep=0.75pt]    {$4$};
\draw (125.63,312.14) node [anchor=south] [inner sep=0.75pt]    {$1$};
\draw (168.54,312.14) node [anchor=south] [inner sep=0.75pt]    {$4$};
\draw (211.46,312.14) node [anchor=south] [inner sep=0.75pt]    {$3$};
\draw (254.37,312.14) node [anchor=south] [inner sep=0.75pt]    {$2$};
\draw (286.63,197.14) node [anchor=south] [inner sep=0.75pt]    {$1$};
\draw (329.54,197.14) node [anchor=south] [inner sep=0.75pt]    {$3$};
\draw (372.46,197.14) node [anchor=south] [inner sep=0.75pt]    {$4$};
\draw (286.63,313.14) node [anchor=south] [inner sep=0.75pt]    {$1$};
\draw (372.46,313.14) node [anchor=south] [inner sep=0.75pt]    {$4$};
\draw (329.54,313.14) node [anchor=south] [inner sep=0.75pt]    {$3$};
\draw (415.37,313.14) node [anchor=south] [inner sep=0.75pt]    {$2$};
\draw (415.37,197.14) node [anchor=south] [inner sep=0.75pt]    {$2$};
\draw (450.63,197.51) node [anchor=south] [inner sep=0.75pt]    {$1$};
\draw (536.46,197.51) node [anchor=south] [inner sep=0.75pt]    {$3$};
\draw (493.54,197.51) node [anchor=south] [inner sep=0.75pt]    {$4$};
\draw (579.37,197.51) node [anchor=south] [inner sep=0.75pt]    {$2$};
\draw (450.63,313.51) node [anchor=south] [inner sep=0.75pt]    {$1$};
\draw (579.37,313.51) node [anchor=south] [inner sep=0.75pt]    {$2$};
\draw (493.54,313.51) node [anchor=south] [inner sep=0.75pt]    {$3$};
\draw (536.46,313.51) node [anchor=south] [inner sep=0.75pt]    {$4$};
\draw (218,115) node [anchor=east] [inner sep=0.75pt]    {$\bigphi ^{( 2)} :$};
\draw (345,115) node [anchor=east] [inner sep=0.75pt]    {$\bigphi ^{( 3)} :$};
\draw (107,223.74) node [anchor=east] [inner sep=0.75pt]    {$\bigphi ^{( 4)} :$};
\draw (461.68,115.2) node    {$+$};
\draw (273.68,243.2) node    {$+$};
\draw (273.27,355.41) node    {$+$};
\draw (434.68,242.2) node    {$+$};
\draw (434.27,354.41) node    {$+$};
\draw (120.27,355.41) node    {$+$};

\end{tikzpicture}
\end{center}

Each of these binary trees corresponds to a term comprised of nested Poisson brackets of seed functions, multiplied by a string of $1/z_{ij}$'s. Each binary tree can in fact be related to a sum of (a kind of) ``graph,'' but not quite like the ones we had before with $\mT_N^{\rm max}$. In order to elaborate on the connection between binary trees and this other kind of marked tree graph, let us define yet another graph element comprised of a line and an associated pair of numbers, say $a$ and $b$. These numbers dictate which $z_{ab}$ we should put in the denominator, which will be independent from the $D^{ij}$ in the numerator:
\begin{center}
\tikzset{every picture/.style={line width=0.75pt}} 

\begin{tikzpicture}[x=0.75pt,y=0.75pt,yscale=-1,xscale=1]

\draw    (67,63.2) -- (104,63.2) ;
\draw   (42,63.2) .. controls (42,56.3) and (47.6,50.7) .. (54.5,50.7) .. controls (61.4,50.7) and (67,56.3) .. (67,63.2) .. controls (67,70.1) and (61.4,75.7) .. (54.5,75.7) .. controls (47.6,75.7) and (42,70.1) .. (42,63.2) -- cycle ;
\draw   (104,63.2) .. controls (104,56.3) and (109.6,50.7) .. (116.5,50.7) .. controls (123.4,50.7) and (129,56.3) .. (129,63.2) .. controls (129,70.1) and (123.4,75.7) .. (116.5,75.7) .. controls (109.6,75.7) and (104,70.1) .. (104,63.2) -- cycle ;

\draw (131,63.2) node [anchor=west] [inner sep=0.75pt]    {$\ =\ \dfrac{D^{ij}}{z_{ab}} \phi _{i} \phi _{j}$};
\draw (54.5,63.2) node    {$i$};
\draw (116.5,63.2) node    {$j$};
\draw (85.5,59.8) node [anchor=south] [inner sep=0.75pt]  [font=\footnotesize]  {$a\ b$};

\end{tikzpicture}
\end{center}
If we then take some binary tree graph, we can convert it to a sum of a new kind of marked tree diagram as follows: in the sense that we can construct any binary tree by sequentially combining ``pairs'' of smaller binary sub-trees, we can construct our new kind of diagram by sequentially surrounding smaller sub-diagrams with grey dotted circles and connecting them with an ``$ab$''-labelled edge. The values of $a$ and $b$ on this edge must correspond to the smallest node label in each subgraph. Here is an example.
\begin{figure}[H]
    \begin{center}
    \tikzset{every picture/.style={line width=0.75pt}} 

\begin{tikzpicture}[x=0.75pt,y=0.75pt,yscale=-1,xscale=1]

\draw  [fill={rgb, 255:red, 216; green, 216; blue, 216 }  ,fill opacity=1 ][dash pattern={on 4.5pt off 4.5pt}] (276,90.96) .. controls (276,48.44) and (308.57,13.98) .. (348.75,13.98) .. controls (388.93,13.98) and (421.5,48.44) .. (421.5,90.96) .. controls (421.5,133.47) and (388.93,167.93) .. (348.75,167.93) .. controls (308.57,167.93) and (276,133.47) .. (276,90.96) -- cycle ;
\draw  [fill={rgb, 255:red, 241; green, 241; blue, 241 }  ,fill opacity=1 ][dash pattern={on 4.5pt off 4.5pt}] (337,90.96) .. controls (337,60.05) and (352.89,35) .. (372.5,35) .. controls (392.11,35) and (408,60.05) .. (408,90.96) .. controls (408,121.86) and (392.11,146.91) .. (372.5,146.91) .. controls (352.89,146.91) and (337,121.86) .. (337,90.96) -- cycle ;
\draw    (113,132) -- (113,149) ;
\draw    (70.09,89.09) -- (113,132) ;
\draw    (155.91,89.09) -- (113,132) ;
\draw    (177.37,67.63) -- (155.91,46.17) ;
\draw    (70.09,89.09) -- (91.54,67.63) ;
\draw    (155.91,89.09) -- (177.37,67.63) ;
\draw    (70.09,89.09) -- (48.63,67.63) ;
\draw    (177.37,67.63) -- (198.83,46.17) ;
\draw    (91.54,67.63) -- (113,46.17) ;
\draw    (91.54,67.63) -- (70.09,46.17) ;
\draw    (48.63,67.63) -- (27.17,46.17) ;
\draw  [fill={rgb, 255:red, 255; green, 255; blue, 255 }  ,fill opacity=1 ] (360,59.96) .. controls (360,53.05) and (365.6,47.46) .. (372.5,47.46) .. controls (379.4,47.46) and (385,53.05) .. (385,59.96) .. controls (385,66.86) and (379.4,72.46) .. (372.5,72.46) .. controls (365.6,72.46) and (360,66.86) .. (360,59.96) -- cycle ;
\draw  [fill={rgb, 255:red, 255; green, 255; blue, 255 }  ,fill opacity=1 ] (360,121.96) .. controls (360,115.05) and (365.6,109.46) .. (372.5,109.46) .. controls (379.4,109.46) and (385,115.05) .. (385,121.96) .. controls (385,128.86) and (379.4,134.46) .. (372.5,134.46) .. controls (365.6,134.46) and (360,128.86) .. (360,121.96) -- cycle ;
\draw    (421.5,90.96) -- (454.5,90.96) ;
\draw    (372.5,72.46) -- (372.5,109.46) ;
\draw  [fill={rgb, 255:red, 241; green, 241; blue, 241 }  ,fill opacity=1 ][dash pattern={on 4.5pt off 4.5pt}] (454.5,90.96) .. controls (454.5,60.05) and (470.39,35) .. (490,35) .. controls (509.61,35) and (525.5,60.05) .. (525.5,90.96) .. controls (525.5,121.86) and (509.61,146.91) .. (490,146.91) .. controls (470.39,146.91) and (454.5,121.86) .. (454.5,90.96) -- cycle ;
\draw  [fill={rgb, 255:red, 255; green, 255; blue, 255 }  ,fill opacity=1 ] (477.5,59.96) .. controls (477.5,53.05) and (483.1,47.46) .. (490,47.46) .. controls (496.9,47.46) and (502.5,53.05) .. (502.5,59.96) .. controls (502.5,66.86) and (496.9,72.46) .. (490,72.46) .. controls (483.1,72.46) and (477.5,66.86) .. (477.5,59.96) -- cycle ;
\draw  [fill={rgb, 255:red, 255; green, 255; blue, 255 }  ,fill opacity=1 ] (477.5,121.96) .. controls (477.5,115.05) and (483.1,109.46) .. (490,109.46) .. controls (496.9,109.46) and (502.5,115.05) .. (502.5,121.96) .. controls (502.5,128.86) and (496.9,134.46) .. (490,134.46) .. controls (483.1,134.46) and (477.5,128.86) .. (477.5,121.96) -- cycle ;
\draw    (490,72.46) -- (490,109.46) ;
\draw  [fill={rgb, 255:red, 255; green, 255; blue, 255 }  ,fill opacity=1 ] (290,91.96) .. controls (290,85.05) and (295.6,79.46) .. (302.5,79.46) .. controls (309.4,79.46) and (315,85.05) .. (315,91.96) .. controls (315,98.86) and (309.4,104.46) .. (302.5,104.46) .. controls (295.6,104.46) and (290,98.86) .. (290,91.96) -- cycle ;
\draw    (337,90.96) -- (315,90.96) ;

\draw (27.17,42.77) node [anchor=south] [inner sep=0.75pt]    {$1$};
\draw (113,42.77) node [anchor=south] [inner sep=0.75pt]    {$4$};
\draw (155.91,42.77) node [anchor=south] [inner sep=0.75pt]    {$3$};
\draw (198.83,42.77) node [anchor=south] [inner sep=0.75pt]    {$2$};
\draw (70.09,42.77) node [anchor=south] [inner sep=0.75pt]    {$5$};
\draw (372.5,59.96) node    {$5$};
\draw (372.5,121.96) node    {$4$};
\draw (370.5,90.96) node [anchor=east] [inner sep=0.75pt]  [font=\footnotesize]  {$ \begin{array}{l}
5\\
4
\end{array}$};
\draw (490,59.96) node    {$3$};
\draw (490,121.96) node    {$2$};
\draw (488,90.96) node [anchor=east] [inner sep=0.75pt]  [font=\footnotesize]  {$ \begin{array}{l}
3\\
2
\end{array}$};
\draw (302.5,91.96) node    {$1$};
\draw (324,87.56) node [anchor=south] [inner sep=0.75pt]  [font=\footnotesize]  {$1\ 4$};
\draw (438,87.56) node [anchor=south] [inner sep=0.75pt]  [font=\footnotesize]  {$1\ 2$};

\end{tikzpicture}
    \end{center}
    \caption{An example of how a binary tree graph can be converted into an object built from gray dotted circles, which should be thought of as a sum over a certain type of graph.}
    \label{fig:798934}
\end{figure}
To explain how the gray dotted circle notation works in this context, if an edge is shared between two gray dotted circles, then one is instructed to sum over all possible connections of that edge to all pairs of nodes, where one node is contained in one circle and the other node is contained in the other circle. One example of a graph to we should include in the sum is given below:
\begin{center}
\tikzset{every picture/.style={line width=0.75pt}} 

\begin{tikzpicture}[x=0.75pt,y=0.75pt,yscale=-1,xscale=1]

\draw    (411.5,57.96) -- (529,119.96) ;
\draw    (411.5,57.96) -- (341.5,89.96) ;
\draw  [fill={rgb, 255:red, 216; green, 216; blue, 216 }  ,fill opacity=1 ][dash pattern={on 4.5pt off 4.5pt}] (2,88.96) .. controls (2,46.44) and (34.57,11.98) .. (74.75,11.98) .. controls (114.93,11.98) and (147.5,46.44) .. (147.5,88.96) .. controls (147.5,131.47) and (114.93,165.93) .. (74.75,165.93) .. controls (34.57,165.93) and (2,131.47) .. (2,88.96) -- cycle ;
\draw  [fill={rgb, 255:red, 241; green, 241; blue, 241 }  ,fill opacity=1 ][dash pattern={on 4.5pt off 4.5pt}] (63,88.96) .. controls (63,58.05) and (78.89,33) .. (98.5,33) .. controls (118.11,33) and (134,58.05) .. (134,88.96) .. controls (134,119.86) and (118.11,144.91) .. (98.5,144.91) .. controls (78.89,144.91) and (63,119.86) .. (63,88.96) -- cycle ;
\draw  [fill={rgb, 255:red, 255; green, 255; blue, 255 }  ,fill opacity=1 ] (86,57.96) .. controls (86,51.05) and (91.6,45.46) .. (98.5,45.46) .. controls (105.4,45.46) and (111,51.05) .. (111,57.96) .. controls (111,64.86) and (105.4,70.46) .. (98.5,70.46) .. controls (91.6,70.46) and (86,64.86) .. (86,57.96) -- cycle ;
\draw  [fill={rgb, 255:red, 255; green, 255; blue, 255 }  ,fill opacity=1 ] (86,119.96) .. controls (86,113.05) and (91.6,107.46) .. (98.5,107.46) .. controls (105.4,107.46) and (111,113.05) .. (111,119.96) .. controls (111,126.86) and (105.4,132.46) .. (98.5,132.46) .. controls (91.6,132.46) and (86,126.86) .. (86,119.96) -- cycle ;
\draw    (147.5,88.96) -- (180.5,88.96) ;
\draw    (98.5,70.46) -- (98.5,107.46) ;
\draw  [fill={rgb, 255:red, 241; green, 241; blue, 241 }  ,fill opacity=1 ][dash pattern={on 4.5pt off 4.5pt}] (180.5,88.96) .. controls (180.5,58.05) and (196.39,33) .. (216,33) .. controls (235.61,33) and (251.5,58.05) .. (251.5,88.96) .. controls (251.5,119.86) and (235.61,144.91) .. (216,144.91) .. controls (196.39,144.91) and (180.5,119.86) .. (180.5,88.96) -- cycle ;
\draw  [fill={rgb, 255:red, 255; green, 255; blue, 255 }  ,fill opacity=1 ] (203.5,57.96) .. controls (203.5,51.05) and (209.1,45.46) .. (216,45.46) .. controls (222.9,45.46) and (228.5,51.05) .. (228.5,57.96) .. controls (228.5,64.86) and (222.9,70.46) .. (216,70.46) .. controls (209.1,70.46) and (203.5,64.86) .. (203.5,57.96) -- cycle ;
\draw  [fill={rgb, 255:red, 255; green, 255; blue, 255 }  ,fill opacity=1 ] (203.5,119.96) .. controls (203.5,113.05) and (209.1,107.46) .. (216,107.46) .. controls (222.9,107.46) and (228.5,113.05) .. (228.5,119.96) .. controls (228.5,126.86) and (222.9,132.46) .. (216,132.46) .. controls (209.1,132.46) and (203.5,126.86) .. (203.5,119.96) -- cycle ;
\draw    (216,70.46) -- (216,107.46) ;
\draw  [fill={rgb, 255:red, 255; green, 255; blue, 255 }  ,fill opacity=1 ] (16,89.96) .. controls (16,83.05) and (21.6,77.46) .. (28.5,77.46) .. controls (35.4,77.46) and (41,83.05) .. (41,89.96) .. controls (41,96.86) and (35.4,102.46) .. (28.5,102.46) .. controls (21.6,102.46) and (16,96.86) .. (16,89.96) -- cycle ;
\draw    (63,88.96) -- (41,88.96) ;
\draw  [fill={rgb, 255:red, 255; green, 255; blue, 255 }  ,fill opacity=1 ] (399,57.96) .. controls (399,51.05) and (404.6,45.46) .. (411.5,45.46) .. controls (418.4,45.46) and (424,51.05) .. (424,57.96) .. controls (424,64.86) and (418.4,70.46) .. (411.5,70.46) .. controls (404.6,70.46) and (399,64.86) .. (399,57.96) -- cycle ;
\draw  [fill={rgb, 255:red, 255; green, 255; blue, 255 }  ,fill opacity=1 ] (399,119.96) .. controls (399,113.05) and (404.6,107.46) .. (411.5,107.46) .. controls (418.4,107.46) and (424,113.05) .. (424,119.96) .. controls (424,126.86) and (418.4,132.46) .. (411.5,132.46) .. controls (404.6,132.46) and (399,126.86) .. (399,119.96) -- cycle ;
\draw    (411.5,70.46) -- (411.5,107.46) ;
\draw  [fill={rgb, 255:red, 255; green, 255; blue, 255 }  ,fill opacity=1 ] (516.5,57.96) .. controls (516.5,51.05) and (522.1,45.46) .. (529,45.46) .. controls (535.9,45.46) and (541.5,51.05) .. (541.5,57.96) .. controls (541.5,64.86) and (535.9,70.46) .. (529,70.46) .. controls (522.1,70.46) and (516.5,64.86) .. (516.5,57.96) -- cycle ;
\draw  [fill={rgb, 255:red, 255; green, 255; blue, 255 }  ,fill opacity=1 ] (516.5,119.96) .. controls (516.5,113.05) and (522.1,107.46) .. (529,107.46) .. controls (535.9,107.46) and (541.5,113.05) .. (541.5,119.96) .. controls (541.5,126.86) and (535.9,132.46) .. (529,132.46) .. controls (522.1,132.46) and (516.5,126.86) .. (516.5,119.96) -- cycle ;
\draw    (529,70.46) -- (529,107.46) ;
\draw  [fill={rgb, 255:red, 255; green, 255; blue, 255 }  ,fill opacity=1 ] (329,89.96) .. controls (329,83.05) and (334.6,77.46) .. (341.5,77.46) .. controls (348.4,77.46) and (354,83.05) .. (354,89.96) .. controls (354,96.86) and (348.4,102.46) .. (341.5,102.46) .. controls (334.6,102.46) and (329,96.86) .. (329,89.96) -- cycle ;

\draw (98.5,57.96) node    {$5$};
\draw (98.5,119.96) node    {$4$};
\draw (96.5,88.96) node [anchor=east] [inner sep=0.75pt]  [font=\footnotesize]  {$ \begin{array}{l}
5\\
4
\end{array}$};
\draw (216,57.96) node    {$3$};
\draw (216,119.96) node    {$2$};
\draw (214,88.96) node [anchor=east] [inner sep=0.75pt]  [font=\footnotesize]  {$ \begin{array}{l}
3\\
2
\end{array}$};
\draw (28.5,89.96) node    {$1$};
\draw (50,85.56) node [anchor=south] [inner sep=0.75pt]  [font=\footnotesize]  {$1\ 4$};
\draw (164,85.56) node [anchor=south] [inner sep=0.75pt]  [font=\footnotesize]  {$1\ 2$};
\draw (411.5,57.96) node    {$5$};
\draw (411.5,119.96) node    {$4$};
\draw (409.5,88.96) node [anchor=east] [inner sep=0.75pt]  [font=\footnotesize]  {$ \begin{array}{l}
5\\
4
\end{array}$};
\draw (529,57.96) node    {$3$};
\draw (529,119.96) node    {$2$};
\draw (527,88.96) node [anchor=east] [inner sep=0.75pt]  [font=\footnotesize]  {$ \begin{array}{l}
3\\
2
\end{array}$};
\draw (341.5,89.96) node    {$1$};
\draw (378.87,70.03) node [anchor=south] [inner sep=0.75pt]  [font=\footnotesize,rotate=-335]  {$1\ 4$};
\draw (467.62,84.9) node [anchor=south] [inner sep=0.75pt]  [font=\footnotesize,rotate=-28]  {$1\ 2$};
\draw (271.75,85.68) node [anchor=west] [inner sep=0.75pt]  [font=\normalsize]  {$=$};
\draw (491.45,157) node [anchor=north] [inner sep=0.75pt]   [align=left] {+ 11 other graphs};
\draw (272,228.88) node [anchor=west] [inner sep=0.75pt]    {$ \begin{array}{l}
=\ \dfrac{D^{15}}{z_{14}}\dfrac{D^{45}}{z_{45}}\dfrac{D^{52}}{z_{12}}\dfrac{D^{32}}{z_{32}} \phi _{1} \phi _{2} \phi _{3} \phi _{4} \phi _{5}\\
\ \ \ \ \ \ \ \ \\
\ \ \ \ \ \ \ \ \ \ \ \ \ \ \ \ \ \ \ \ \ \ \ \ \ \ \ \ \ +\ \text{11 other terms}
\end{array}$};

\end{tikzpicture}
\end{center}
We reiterate that the kinds of marked graphs which arise from this procedure are \emph{not} the graphs of $\mT_N^{\rm max}$, due to the special $ab$-labelled edges.

In any case, the full term which Figure \ref{fig:798934} corresponds to is
\begin{align}
    &\frac{1}{z_{14}} \frac{1}{z_{45}} \frac{1}{z_{12}} \frac{1}{z_{32}} \{ \{ \phi_1 , \{ \phi_4, \phi_5\} \}, \{ \phi_3, \phi_2\} \} = \\
    &\frac{1}{z_{14}} \frac{1}{z_{45}} \frac{1}{z_{12}} \frac{1}{z_{32}} D^{45} D^{32} (D^{15} + D^{14})(D^{12} + D^{13} + D^{42} + D^{43} + D^{52} + D^{53} ) \phi_1 \phi_2 \phi_3 \phi_4 \phi_5 . \nonumber 
\end{align}
\end{ex}
It is interesting to note that this recursive formula has a similar flavor to a Ward identity in a conformal field theory.

One might wonder if this binary tree expansion of $\phi^{(N)}$ can be used to prove any facts which are difficult to see using the marked tree graphs of $\mT_N^{\rm max}$. Indeed it can, and here is one such fact we'll give as an application.

\begin{claim}
    \begin{equation}
    \bigphi^{(N)}(\phi_1, \ldots, \phi_N) = 0 \hspace{0.5 cm} \text{if} \hspace{0.5 cm} \sum_{i = 1}^N p_i^\mu = 0
\end{equation}
\end{claim}
\begin{proof}
    Before beginning the proof, it is worth noting that the equation $\sum_{i = 1}^N p_i^\mu = 0$, which colloquially one might say corresponds to ``momentum conservation,'' is not a requirement for a perturbiner expansion to solve the classical equations of motion. One is always allowed to perturbatively solve the classical equations of motion using any collection of seed functions as long as $p_i^2 = 0$ for all $i$.
    
    In any case, the condition that the momenta sum to zero is equivalent to the four equations
    \begin{equation}
        \sum_{i = 1}^N \omega_i = 0, \hspace{0.75 cm}\sum_{i = 1}^N \omega_i \bz_i = 0,  \hspace{0.75 cm}\sum_{i = 1}^N \omega_i z_i = 0, \hspace{0.75 cm}\sum_{i = 1}^N \omega_i z_i \bz_i = 0.
    \end{equation}
    Using the first two of the above equations, one can show that, for any $j \in \{1, \ldots, N\}$, we have
    \begin{equation}\label{mom_cons_eq}
        \sum_{i=1}^N \omega_i \bz_{ij} = 0.
    \end{equation}

    \noindent Now, let us think about the final step in evaluating any particular binary tree in the diagrammatic expansion. At the highest level, every graph will always look like two large gray dotted circles with a ``12'' labelled edge connecting them. Let us call the collection of nodes in the two circles $A$ and $B$, such that $1 \in A$, $2 \in B$, and
    \begin{equation}
        A \cup B = \{1 , \ldots, N\}.
    \end{equation}
    A schematic representation of the graph, (where the exact sub-structure contained within each of the two largest dotted circles is supressed) is drawn below.
    \begin{center}
    \tikzset{every picture/.style={line width=0.75pt}} 

\begin{tikzpicture}[x=0.75pt,y=0.75pt,yscale=-1,xscale=1]

\draw  [fill={rgb, 255:red, 241; green, 241; blue, 241 }  ,fill opacity=1 ][dash pattern={on 4.5pt off 4.5pt}] (22,108.96) .. controls (22,66.44) and (54.57,31.98) .. (94.75,31.98) .. controls (134.93,31.98) and (167.5,66.44) .. (167.5,108.96) .. controls (167.5,151.47) and (134.93,185.93) .. (94.75,185.93) .. controls (54.57,185.93) and (22,151.47) .. (22,108.96) -- cycle ;
\draw  [fill={rgb, 255:red, 241; green, 241; blue, 241 }  ,fill opacity=1 ][dash pattern={on 4.5pt off 4.5pt}] (200.5,108.96) .. controls (200.5,66.44) and (233.07,31.98) .. (273.25,31.98) .. controls (313.43,31.98) and (346,66.44) .. (346,108.96) .. controls (346,151.47) and (313.43,185.93) .. (273.25,185.93) .. controls (233.07,185.93) and (200.5,151.47) .. (200.5,108.96) -- cycle ;
\draw  [fill={rgb, 255:red, 255; green, 255; blue, 255 }  ,fill opacity=1 ] (45.5,116.96) .. controls (45.5,110.05) and (51.1,104.46) .. (58,104.46) .. controls (64.9,104.46) and (70.5,110.05) .. (70.5,116.96) .. controls (70.5,123.86) and (64.9,129.46) .. (58,129.46) .. controls (51.1,129.46) and (45.5,123.86) .. (45.5,116.96) -- cycle ;
\draw  [fill={rgb, 255:red, 255; green, 255; blue, 255 }  ,fill opacity=1 ] (82.5,150.96) .. controls (82.5,144.05) and (88.1,138.46) .. (95,138.46) .. controls (101.9,138.46) and (107.5,144.05) .. (107.5,150.96) .. controls (107.5,157.86) and (101.9,163.46) .. (95,163.46) .. controls (88.1,163.46) and (82.5,157.86) .. (82.5,150.96) -- cycle ;
\draw  [fill={rgb, 255:red, 255; green, 255; blue, 255 }  ,fill opacity=1 ] (90.5,80.96) .. controls (90.5,74.05) and (96.1,68.46) .. (103,68.46) .. controls (109.9,68.46) and (115.5,74.05) .. (115.5,80.96) .. controls (115.5,87.86) and (109.9,93.46) .. (103,93.46) .. controls (96.1,93.46) and (90.5,87.86) .. (90.5,80.96) -- cycle ;
\draw  [fill={rgb, 255:red, 255; green, 255; blue, 255 }  ,fill opacity=1 ] (270.5,147.96) .. controls (270.5,141.05) and (276.1,135.46) .. (283,135.46) .. controls (289.9,135.46) and (295.5,141.05) .. (295.5,147.96) .. controls (295.5,154.86) and (289.9,160.46) .. (283,160.46) .. controls (276.1,160.46) and (270.5,154.86) .. (270.5,147.96) -- cycle ;
\draw  [fill={rgb, 255:red, 255; green, 255; blue, 255 }  ,fill opacity=1 ] (296.5,106.96) .. controls (296.5,100.05) and (302.1,94.46) .. (309,94.46) .. controls (315.9,94.46) and (321.5,100.05) .. (321.5,106.96) .. controls (321.5,113.86) and (315.9,119.46) .. (309,119.46) .. controls (302.1,119.46) and (296.5,113.86) .. (296.5,106.96) -- cycle ;
\draw  [fill={rgb, 255:red, 255; green, 255; blue, 255 }  ,fill opacity=1 ] (230.5,128.96) .. controls (230.5,122.05) and (236.1,116.46) .. (243,116.46) .. controls (249.9,116.46) and (255.5,122.05) .. (255.5,128.96) .. controls (255.5,135.86) and (249.9,141.46) .. (243,141.46) .. controls (236.1,141.46) and (230.5,135.86) .. (230.5,128.96) -- cycle ;
\draw  [fill={rgb, 255:red, 255; green, 255; blue, 255 }  ,fill opacity=1 ] (261.5,90.96) .. controls (261.5,84.05) and (267.1,78.46) .. (274,78.46) .. controls (280.9,78.46) and (286.5,84.05) .. (286.5,90.96) .. controls (286.5,97.86) and (280.9,103.46) .. (274,103.46) .. controls (267.1,103.46) and (261.5,97.86) .. (261.5,90.96) -- cycle ;
\draw  [fill={rgb, 255:red, 255; green, 255; blue, 255 }  ,fill opacity=1 ] (226.5,71.96) .. controls (226.5,65.05) and (232.1,59.46) .. (239,59.46) .. controls (245.9,59.46) and (251.5,65.05) .. (251.5,71.96) .. controls (251.5,78.86) and (245.9,84.46) .. (239,84.46) .. controls (232.1,84.46) and (226.5,78.86) .. (226.5,71.96) -- cycle ;
\draw  [fill={rgb, 255:red, 255; green, 255; blue, 255 }  ,fill opacity=1 ] (283.5,63.96) .. controls (283.5,57.05) and (289.1,51.46) .. (296,51.46) .. controls (302.9,51.46) and (308.5,57.05) .. (308.5,63.96) .. controls (308.5,70.86) and (302.9,76.46) .. (296,76.46) .. controls (289.1,76.46) and (283.5,70.86) .. (283.5,63.96) -- cycle ;
\draw  [fill={rgb, 255:red, 255; green, 255; blue, 255 }  ,fill opacity=1 ] (124.5,124.96) .. controls (124.5,118.05) and (130.1,112.46) .. (137,112.46) .. controls (143.9,112.46) and (149.5,118.05) .. (149.5,124.96) .. controls (149.5,131.86) and (143.9,137.46) .. (137,137.46) .. controls (130.1,137.46) and (124.5,131.86) .. (124.5,124.96) -- cycle ;
\draw  [fill={rgb, 255:red, 255; green, 255; blue, 255 }  ,fill opacity=1 ] (49.5,70.96) .. controls (49.5,64.05) and (55.1,58.46) .. (62,58.46) .. controls (68.9,58.46) and (74.5,64.05) .. (74.5,70.96) .. controls (74.5,77.86) and (68.9,83.46) .. (62,83.46) .. controls (55.1,83.46) and (49.5,77.86) .. (49.5,70.96) -- cycle ;
\draw    (167.5,108.96) -- (200.5,108.96) ;

\draw (94.75,28.58) node [anchor=south] [inner sep=0.75pt]  [font=\Large]  {$A$};
\draw (273.25,28.58) node [anchor=south] [inner sep=0.75pt]  [font=\Large]  {$B$};
\draw (184,105.56) node [anchor=south] [inner sep=0.75pt]  [font=\footnotesize]  {$1\ 2$};
\draw (143,199) node [anchor=north west][inner sep=0.75pt]   [align=left] {(schematic)};
\draw (95,150.96) node    {$1$};
\draw (274,90.96) node    {$2$};

\end{tikzpicture}
    \end{center}
    Noting that
    \begin{equation}
        D^{ij} \phi_1 \ldots \phi_N = -( \omega_i \omega_j \bz_{ij})  \phi_1  \ldots \phi_N
    \end{equation}
    we can see that, by implementing the final sum over these pairs of nodes, we have
    \begin{equation}\label{phiNprop}
        \bigphi^{(N)}(\phi_1, \ldots, \phi_N) \propto \frac{1}{z_{12}} \left(  -\sum_{i \in A} \sum_{j \in B} \omega_i \omega_j \bz_{ij} \right) \phi_1 \ldots \phi_N
    \end{equation}
    where the sub-structure of the graphs contained within each of the two largest gray dotted circles only affects $\phi^{(N)}$ by a multiplicative factor.

    Let us now rewrite \eqref{mom_cons_eq} as
    \begin{equation}
        \sum_{i \in A} \omega_i \omega_j \bz_{ij} + \sum_{i \in B} \omega_i \omega_j \bz_{ij} = 0 .
    \end{equation}
    Summing $j$ in the above equation over all values in $B$, we get
    \begin{equation}
        0 = \sum_{j \in B} \left( \sum_{i \in A} \omega_i \omega_j \bz_{ij} + \sum_{i \in B} \omega_i \omega_j \bz_{ij} \right) = \sum_{j \in B} \sum_{i \in A} \omega_i \omega_j \bz_{ij} 
    \end{equation}
    where the second term vanished due to antisymmetry. We finish the proof by plugging the above expression into \eqref{phiNprop}.
\end{proof}

\section{A crash course in self-dual gravity}\label{appB}

The purpose of this section is to provide the background necessary to understand and prove all of the key equations of self-dual gravity used in this paper.  We begin with a collection of spinor conventions in section \ref{app_spin}, review the tetrad formalism for Einstein gravity in section \ref{app_tet},  and conclude with the proof of Plebański's first and second heavenly equations in section \ref{app_pleba}.

We use Greek letters for curved space indices, $\mu, \nu = 0,1,2,3$, Latin letters $a, b = 0,1,2,3$ for vierbein indices, and upper case letters for spin indices, $A, B = 1, 2$, $\dot A, \dot B = \dot 1, \dot 2$.

\subsection{Spinor conventions}\label{app_spin}

We use the $\eta_{a b} = \text{diag}(+$$-$$-$$-)$ convention. We define the 2d antisymmetric tensors $\vep_{AB}$, $\vep^{AB}$, $\vep_{\dot A \dot B}$, $\vep^{\dot A \dot B}$ by
\begin{equation}
    \vep^{1 2} = - \vep_{12} = \vep^{\dot 1 \dot 2} = - \vep_{\dot 1 \dot 2} = 1
\end{equation}
such that
\begin{equation}
    \vep_{AB} \vep^{BC} = \delta_A^C, \hspace{1 cm} \vep_{\dot A \dot B} \vep^{\dot B \dot C} = \delta_{\dot A}^{\dot C}.
\end{equation}
We raise and lower spinors from the left as
\begin{equation}
    \lambda_A = \vep_{AB} \lambda^B = \vep_{AB} (\vep^{BC} \lambda_C).
\end{equation}
We define the four-vector of Pauli matrices as
\begin{align}
    \sigma_a^{A\dot A} &\equiv (1, \sigma_x, \sigma_y, \sigma_z) \, , \\
    \sigma^a_{A\dot A} &= (1, \sigma_x, -\sigma_y, \sigma_z) \, ,
\end{align}
where $\sigma^a_{A \dot A} = \vep_{A B} \vep_{\dot A \dot B} \eta^{ab} \sigma_{b}^{B \dot B}$. Note the identities
\begin{align}
    \sigma^a_{A \dot A} \sigma_a^{B \dot B} = 2 \delta_A^B \delta_{\dot A}^{\dot B}\, , \hspace{1 cm} \sigma^a_{A \dot A} \sigma^{A \dot A}_b = 2 \delta^a_b\, ,
\end{align}
which allow us to go back and forth between objects with flat space indices $V^a$ and their corresponding objects with spinor indices $V^{A \dot A}$ via
\begin{equation}
    V^{A \dot A} = \sigma_a^{A \dot A} V^a, \hspace{1 cm} V^a = \frac{1}{2} \sigma^a_{A \dot A} V^{A \dot A}, \hspace{1cm} V_{A \dot A} = \sigma^a_{A \dot A} V_a, \hspace{1 cm} V_a = \frac{1}{2} \sigma_a^{A \dot A} V_{A \dot A}.
\end{equation}
Note the relation $V^a V_a = \frac{1}{2} V^{A \dot A} V_{A \dot A}$. We further define the matrices
\begin{equation}
\begin{aligned}
    (\sigma_{a b})^{A}_{\;\; B} &\equiv \frac{1}{4} \left( \sigma_a^{A \dot C} \sigma_{b B \dot C} -  \sigma_b^{A \dot C} \sigma_{a B \dot C} \right) \, , \\
    (\widetilde{\sigma}_{a b})_{\dot A}^{\;\; \dot B} &\equiv \frac{1}{4} \left( {\sigma}_{a C \dot A} \sigma_b^{C \dot B} -  {\sigma}_{b C \dot A} \sigma_a^{C \dot B} \right)\, ,
\end{aligned}
\end{equation}
which satisfy the respective anti-self-duality and self-duality equations
\begin{equation}
\begin{aligned}
    (\sigma_{ab})_A^{\;\; B}  &= -\frac{i}{2} \vep_{abcd} (\sigma^{cd})_A^{\;\; B} \, , \\
    (\widetilde{\sigma}_{ab})^{\dot A}_{\;\; \dot B} &= \frac{i}{2} \vep_{abcd} (\widetilde{\sigma}^{cd})^{\dot A}_{\;\; \dot B}\, .
\end{aligned}
\end{equation}
The above identities can be used to decompose any antisymmetric rank-2 tensor $F_{ab} = F_{[ab]}$ into its ASD and SD parts using the following observation. Note that $F_{A \dot A B \dot B}$ can always be written as
\begin{align}
    F_{A \dot A B \dot B} = M_{(AB) (\dot A \dot B)} + N_{[A B] [\dot A \dot B]} + P_{(A B) [\dot A \dot B]} + Q_{[A B] (\dot A \dot B)}.
\end{align}
The antisymmetry condition $F_{A \dot A B \dot B} = - F_{B \dot B A \dot A}$ requires that the first two terms must be zero. Furthermore, because $\vep$ is the unique 2d antisymmetric tensor up to scaling, we know that the above equation may be written as
\begin{equation}\label{F_decomp}
    F_{A \dot A B \dot B} = 2 \, F_{AB} \vep_{\dot A \dot B} + 2 \, \widetilde{F}_{\dot A \dot B} \vep_{AB}
\end{equation}
where $F_{AB} = F_{(AB)}$ and $\widetilde{F}_{\dot A \dot B} = \widetilde{F}_{(\dot A \dot B)}$ are some symmetric tensors. In fact, these objects can be computed from $F_{A \dot A B \dot B}$ via
\begin{equation}
    F_{AB} = -\frac{1}{4} \vep^{\dot A \dot B} F_{A \dot A B \dot B}, \hspace{1 cm} \widetilde{F}_{\dot A \dot B} = -\frac{1}{4} \vep^{A B} F_{A \dot A B \dot B}.
\end{equation}

Writing
\begin{equation}
\begin{aligned}
    F_{ab} &= \frac{1}{2} \sigma_a^{\dot A A} \sigma_b^{\dot B B} (F_{AB} \vep_{\dot A \dot B} + \widetilde{F}_{\dot A \dot B} \vep_{A B} ) \\
    &= (\sigma_{ab})^{AB} F_{AB} -  ( \widetilde{\sigma}_{ab})^{\dot A \dot B} \widetilde{F}_{\dot A \dot B}
\end{aligned}
\end{equation}
we identify $F_{AB}$ as the ASD part of $F_{ab}$ and $\widetilde{F}_{\dot A \dot B}$ as the SD part of $F_{ab}$.

\subsection{The tetrad formalism with spinor indicies}\label{app_tet}

Here we review the vierbein/tetrad formulation of general relativity.

We take the vierbeins $e_a = e_a^\mu \partial_\mu$ to be a set of four vector fields satisfying
\begin{equation}
    g_{\mu \nu} e_a^\mu e_b^\nu = \eta_{a b}.
\end{equation}
If we also define the tetrad 1-forms $\theta^a = \theta^a_\mu \, d x^\mu$ via
\begin{equation}
    \theta^a_\mu \equiv g_{\mu \nu} \eta^{ab} e_b^\nu
\end{equation}
then we can write the metric components as
\begin{equation}
    g_{\mu \nu} = \eta_{ab} \theta^a_\mu \theta^b_\nu.
\end{equation}
Note the identities $\delta^a_b = \theta_\mu^a e_b^\mu $ and $\delta^\mu_\nu = e_a^\mu \theta^a_\nu$.

One can define the object $\Gamma^a_{cb}$ via the parallel transport of the vierbeins along each other as
\begin{equation}\label{B17}
    \nabla_{e_c} e_b = \Gamma^a_{cb} e_a
\end{equation}
where the explicit formula for $\Gamma^a_{cb}$ is given by
\begin{align}
    \Gamma^a_{cb} = e^a_{\; \mu} e_c^{\; \nu} (\partial_\nu e_b^{\; \mu} + \Gamma^\mu_{\nu \rho} e_b^{\; \rho}).
\end{align}
Using $\Gamma^a_{cb}$, we can define the spin-connection 1-form $\Gamma^a_{\;\; b}$ by
\begin{equation}
    \Gamma^a_{\; b} \equiv \Gamma^a_{cb} \theta^c.
\end{equation}
One can show (see for instance \cite{tong}) that
\begin{equation}
    \Gamma_{ab} = - \Gamma_{ba}
\end{equation}
meaning that the components of $\Gamma_{ab}$ represent infinitesimal Lorentz generators which lie in $\mathfrak{so}(1,3)$.

We now define the curvature 2-form
\begin{equation}
    R_{ab} \equiv \frac{1}{2} e_a^\mu e_b^\nu R_{\mu \nu \rho \sigma}  d x^\rho \wedge d x^\sigma.
\end{equation}
Cartan's structure equations relate $\theta^a$ and $R_{ab}$ to the spin connection via
\begin{align}\label{cartan2}
    \dd \theta^a &= - \Gamma^a_{\;\; b} \wedge \theta^b \, ,\\
    R_{ab} &= \dd \Gamma_{ab} + \Gamma_{a c} \wedge \Gamma^{c}_{\;\; b} \, .
\end{align}

In the tetrad formalism, we can write the first Bianchi identity, the vacuum Einstein equation, and the second Bianchi identity, as
\begin{align} \label{bianchi1}
     R_{ab} \wedge \theta^a = 0 \hspace{1 cm} &\Longleftrightarrow \hspace{1 cm} R_{\mu [\nu \rho \sigma]} = 0 \\ \label{einstein1}
    \vep_{abcd} R^{ab} \wedge \theta^c = 0 \hspace{1 cm} &\Longleftrightarrow \hspace{1 cm} R_{\mu \nu} - \frac{1}{2} g_{\mu \nu} R = 0 \\
    \dd R_{ab} + 2 \, \Gamma^{c}_{\;\; (a} \wedge R_{b) c} = 0 \hspace{1 cm} &\Longleftrightarrow \hspace{1 cm} \nabla_{[\mu } R_{\nu \rho ]\sigma \lambda} = 0
\end{align}
where the above conversions require the identity
\begin{equation}\label{epsilon_id}
    \vep^{\mu \nu \alpha \beta} \vep_{\gamma \delta \rho \sigma} = - 4! \, \delta^{[\mu}_{\gamma} \delta^{\nu}_{\delta} \delta^{\alpha}_{\rho} \delta^{\beta]}_{\sigma}.
\end{equation}
as well as
\begin{equation}
    \vep_{abcd} = e_a^\mu e_b^\nu e_c^\rho e_d^\sigma \vep_{\mu \nu \rho \sigma}, \hspace{1 cm} d x^\mu \wedge d x^\nu \wedge d x^\rho \wedge d x^\sigma = - \vep^{\mu \nu \rho \sigma} \mathrm{dvol}.
\end{equation}

Because $R_{ab} = R_{[ab]}$, we can decompose the ASD and SD parts of the $a$, $b$ indices
\begin{equation}
    R_{A \dot A B \dot B} = 2 \, R_{AB} \vep_{\dot A \dot B} + 2 \, \widetilde{R}_{\dot A \dot B} \vep_{AB}
\end{equation}
with $R_{AB} = R_{(AB)}$ and $\widetilde{R}_{\dot A \dot B} = \widetilde{R}_{(\dot A \dot B)}$. We do the same thing with the spin connection
\begin{equation}
    \Gamma_{A \dot A B \dot B} = 2  \, \Gamma_{AB} \vep_{\dot A \dot B} + 2 \,\widetilde{\Gamma}_{\dot A \dot B} \vep_{AB}
\end{equation}
where $\Gamma_{AB} = \Gamma_{(AB)}$ and $\widetilde{\Gamma}_{\dot A \dot B} = \widetilde{\Gamma}_{(\dot A \dot B)}$.

It is worth reviewing what this split of Lorentz generators into ASD and SD parts means explicitly. Recall that $\mathfrak{so}(1,3) = \mathfrak{sl}(2,\mathbb{R}) \times \mathfrak{sl}(2,\mathbb{R})$, which can be seen by redefining the generators in the following way. If $j_i, k_i \in \mathfrak{so}(1,3)$, for $i=1,2,3$, are the usual $4 \times 4$ rotation and boost matrices satisfying
\begin{equation}
    [j_i, j_j] = \vep_{ijk} j_k, \hspace{1 cm} [j_i, k_j] = \vep_{ijk} k_k, \hspace{1 cm} [k_i, k_j] = - \vep_{ijk} j_k 
\end{equation}
then
\begin{equation}\label{sl2c_generators}
  \begin{split}
    \ell_0  &= \frac{1}{2} (-k_3 - i j_3)\,, \\
    \ell_{1} &= \frac{1}{2} (-k_1 + j_2 - i (k_2 + j_1) )\,,  \\
    \ell_{-1} &= \frac{1}{2} (k_1 + j_2 - i (k_2 - j_1) )\,,
  \end{split}
\quad \quad \quad
  \begin{split}
    \Bell_0  &= \frac{1}{2} (-k_3 + i j_3)\,,  \\
    \Bell_{1} &= \frac{1}{2} (-k_1 + j_2 + i (k_2 + j_1 ) )\,,  \\
    \Bell_{-1}  &= \frac{1}{2} (k_1 + j_2 + i (k_2 - j_1) )\,,
  \end{split}
\end{equation}
satisfy
\begin{equation}
    [\ell_m, \ell_n] = (m-n) \ell_{m+n}\,, \hspace{1 cm} [\Bell_m, \Bell_n] = (m-n) \Bell_{m+n}\,,  \hspace{1 cm} 
    [\ell_m, \Bell_n] = 0\,,
\end{equation}
and of course $(\ell_n)_{ab} = - (\ell_n)_{ba}$, $(\Bell_n)_{ab} = - (\Bell_n)_{ba}$. The entries of these matrices are
\footnotesize
\begin{equation*}
\begin{aligned}
    (\ell_0)^a_{\;\; b} = \frac{1}{2} \begin{pmatrix} 
    \ph 0 & \ph0 & \ph0 & -1 \\ 
    \ph 0 & \ph0 & \ph i& \ph0 \\ 
    \ph 0 & -i   & \ph0 & \ph0 \\
    -1    & \ph0 & \ph0 & \ph0
    \end{pmatrix}, \hspace{0.25 cm} (\ell_1)^a_{\;\; b} =  \frac{1}{2} 
    \begin{pmatrix} 
    \ph0 & -1 & -i & \ph0 \\ 
    -1 & \ph0 & \ph0 & \ph1 \\ 
    -i & \ph0 & \ph0 & \ph i \\
    \ph0 & -1 & -i & \ph0
    \end{pmatrix}, \hspace{0.25 cm}
    (\ell_{-1})^a_{\;\; b} =  \frac{1}{2}
    \begin{pmatrix} 
    \ph0 & 1    & -i   & \ph0 \\ 
    \ph1 & \ph0 & \ph0 & \ph1 \\ 
    -i   & \ph0 & \ph0 & -i \\
    \ph0 & -1   & \ph i & \ph0
    \end{pmatrix}, \\
    (\bell_0)^a_{\;\; b} = \frac{1}{2} \begin{pmatrix} 
    \ph0 & \ph0 & \ph0 & -1 \\ 
    \ph0 & \ph0 & -i & \ph0 \\ 
    \ph0 & \ph i & \ph0 & \ph0 \\
    -1 & \ph0 & \ph0 & \ph0
    \end{pmatrix}, \hspace{0.25 cm} (\bell_1)^a_{\;\; b} = \frac{1}{2}
    \begin{pmatrix} 
    \ph0 & -1 & \ph i & \ph0 \\ 
    -1 & \ph0 & \ph0 & \ph1 \\ 
    \ph i & \ph0 & \ph0 & -i \\
    \ph0 & -1 & \ph i & \ph0
    \end{pmatrix}, \hspace{0.25 cm}
    (\bell_{-1})^a_{\;\; b} = \frac{1}{2} \begin{pmatrix} 
    \ph0 & \ph1 & \ph i & \ph0 \\ 
    \ph1 & \ph0 & \ph0 & \ph1 \\ 
    \ph i & \ph0 & \ph0 & \ph i \\
    \ph0 & -1 & -i & \ph0
    \end{pmatrix}.
\end{aligned}
\end{equation*}
\normalsize

The independent $\mathfrak{sl}(2, \mathbb{R})$ algebras generated by $\ell_n$ and $\Bell_n$ correspond to the ``ASD'' and ``SD'' halves of the Lorentz algebra because
\begin{equation}\label{sd_asd_lorentz}
\begin{aligned}
    (\ell_n)_{ab} &= -\frac{i}{2} \vep_{abcd} (\ell_n)^{cd} , \\
    (\Bell_n)_{ab} &= \frac{i}{2} \vep_{abcd} (\Bell_n)^{cd}.
\end{aligned}
\end{equation}

A straightforward exercise in spinor variables shows that the Cartan structure equations can be decomposed as
\begin{align}\label{cartantheta}
    \dd \theta^{A \dot A} = -\Gamma^A_{\;\; B} \wedge \theta^{B \dot A} - \widetilde{\Gamma}^{\dot A}_{\;\; \dot B} \wedge \theta^{A \dot B}
\end{align}
and
\begin{equation} \label{cartan_R_spin}
\begin{aligned}
    R_{AB} &= \dd \Gamma_{AB} + \Gamma_{A C} \wedge \Gamma^{C}_{\;\;\;B}\, , \\
    \widetilde{R}_{\dot A\dot B} &= \dd \widetilde{\Gamma}_{\dot A \dot B} + \widetilde{\Gamma}_{\dot A \dot C} \wedge \widetilde{\Gamma}^{\dot C}_{\;\;\; \dot B}\, .
\end{aligned}
\end{equation}

We now note an identity satisfied by the tetrads \cite{Capovilla:1991qb}
\begin{equation}\label{symtheta}
    \theta^{(A}_{\;\;\; \dot A} \wedge \theta^{B}_{\;\; \dot B} \wedge \theta^{C)}_{\;\;\; \dot C} = 0
\end{equation}
which holds because the above expression is totally antisymmetric in $[\dot A \dot B \dot C]$ and therefore must vanish because spin indices are 2 dimensional. The above expression further implies the if-and-only-if statement\footnote{It is if-and-only-if because the $\theta^{A \dot A}$ comprise a basis of four 1-forms, and only wedge products of three distinct basis 1-forms are non-vanishing. One can check that the object $\theta^{A}_{\;\; \dot A} \wedge \theta^{B}_{\;\; \dot B} \wedge \theta^{C}_{\;\; \dot C}$ is non-zero only for four options of $(A,B,C,\dot A, \dot B, \dot C)$, with, say, $(1,1,2,\dot 1, \dot 2, \dot 2)$ being one such option. Expanding out \eqref{symtheta} for any of these options shows that it simply reduces to the sum of two non-zero 3-forms cancelling out, and therefore corresponds to the \emph{unique} expression characterizing the cancellation of such 3-forms up to scaling.}
\begin{equation}\label{tetrad_sym_statement}
    M_{ABC} \, \theta^{A}_{\;\; \dot A} \wedge \theta^{B}_{\;\; \dot B} \wedge \theta^{C}_{\;\; \dot C} = 0 \hspace{0.75 cm} \Longleftrightarrow \hspace{0.75 cm} M_{ABC} = M_{(ABC)}.
\end{equation}

We now define the ASD 2-forms $\Sigma^{AB} = \Sigma^{(AB)}$ and the SD 2-forms $\widetilde{\Sigma}^{\dot A \dot B} = \widetilde{\Sigma}^{(\dot A \dot B)}$
\begin{equation}\label{sigma_def_both}
\begin{aligned}
    \Sigma^{AB} &\equiv \theta^{A \dot A} \wedge \theta^B_{\;\; \dot A} \\
    \widetilde{\Sigma}^{\dot A \dot B} &\equiv \theta^{A \dot A} \wedge \theta_A^{\;\;\; \dot B}
\end{aligned}
\end{equation}
which can be written in component form as
\begin{equation}
    \Sigma^{AB} = \frac{1}{2} \Sigma^{AB}_{\mu \nu} d x^\mu \wedge d x^\nu, \hspace{1 cm} \widetilde{\Sigma}^{\dot A \dot B} = \frac{1}{2} \widetilde{\Sigma}^{\dot A \dot B}_{\mu \nu} d x^\mu \wedge d x^\nu,
\end{equation}
and can straightforwardly be shown to satisfy 
\begin{equation}\label{sigma_dual}
    \Sigma^{AB}_{\mu \nu} = -\frac{i}{2} \vep_{\mu \nu \rho \sigma} (\Sigma^{AB})^{\rho \sigma}, \hspace{1 cm} \widetilde{\Sigma}^{\dot A \dot B}_{\mu \nu} = \frac{i}{2} \vep_{\mu \nu \rho \sigma} (\widetilde{\Sigma}^{\dot A \dot B})^{\rho \sigma}.
\end{equation}
Equation \eqref{cartantheta} also gives
\begin{align}\label{dsigma}
    \dd \Sigma^{AB} &= - 2 \Gamma^{(A}_{\;\;\;\; C} \wedge \Sigma^{B) C}, \\
    \dd \widetilde{\Sigma}^{\dot A \dot B} &= - 2 \widetilde{\Gamma}^{(\dot A}_{\;\;\;\; \dot C} \wedge \widetilde{\Sigma}^{\dot B) \dot C}. 
\end{align}

Combining the vacuum Einstein equation \eqref{einstein1} and the first Bianchi identity \eqref{bianchi1}, we find that the Einstein equation can be written as
\begin{align}\label{einstein_asd}
    R_{AB} \wedge \theta^{A \dot A} &= 0, \\
    \widetilde{R}_{\dot A \dot B} \wedge \theta^{A \dot A} &= 0.
\end{align}

We are always free to decompose $R_{AB}$, as a general spacetime 2-form, into its spacetime ASD and SD parts
\begin{equation}
    R_{AB} = \Psi_{ABCD} \Sigma^{CD} + \Phi_{AB \dot A \dot B} \widetilde{\Sigma}^{\dot A \dot B}.
\end{equation}
Combining \eqref{einstein_asd} with the statement \eqref{tetrad_sym_statement}, we see that the Einstein equation is only satisfied if $\Psi_{ABCD} = \Psi_{(ABCD)}$ and $\Phi_{AB \dot A \dot B}$ = 0. We therefore have
\begin{equation}\label{R_sigma}
    R_{AB} = \Psi_{ABCD} \Sigma^{CD}, \hspace{1 cm} \widetilde{R}_{\dot A \dot B} = \widetilde{\Psi}_{\dot A \dot B \dot C \dot D} \widetilde{\Sigma}^{\dot C \dot D},
\end{equation}
where $ \widetilde{\Psi}_{\dot A \dot B \dot C \dot D} =  \widetilde{\Psi}_{(\dot A \dot B \dot C \dot D)}$ as well.

The scalars $\Psi_{ABCD}$ and $\widetilde{\Psi}_{\dot A \dot B \dot C \dot D}$ are called the Weyl scalars, which in our conventions give the coefficients in the decomposition of the Riemann tensor $R_{A \dot A B \dot B C \dot C D \dot D} = e_{A \dot A}^\mu e_{B \dot B}^\nu e_{C \dot C}^\rho e_{D \dot D}^\sigma R_{\mu \nu \rho \sigma}$ as
\begin{equation}\label{riemann_weyl}
    R_{A \dot A B \dot B C \dot C D \dot D} = 16 \, \vep_{\dot A \dot B} \vep_{\dot C \dot D} \Psi_{ABCD} + 16 \, \vep_{A B} \vep_{C D} \widetilde{\Psi}_{\dot A \dot B \dot C \dot D}.
\end{equation}
We now define the spin covariant derivative $\nabla_a$, which we define to act on objects with vierbein indices as $\nabla_a V_b \equiv e_a^\mu \, e_b^\nu \, \nabla_\mu V_\nu$. It is straightforward to show it satisfies the natural equation
\begin{equation}
    \nabla_a V_b = e_a^\mu \partial_\mu V_b - \Gamma^c_{\;\; a b} V_c.
\end{equation}
The above equation also generalizes to tensors with a larger number of indices.

We can then take the differential Bianchi identity
\begin{equation}
    \nabla_{[a} R_{bc] de} = 0,
\end{equation}
plug in \eqref{riemann_weyl} (and contract the appropriate indices) to find
\begin{align}
    0 &= \nabla^{A}_{\;\; \dot A} \Psi_{ABCD}
\end{align}
where, using the notation $\Gamma_{BC} = \frac{1}{2} (\Gamma_{BC})_{A \dot A} \theta^{A \dot A}$, we explicitly have
\begin{equation}
\begin{aligned}
    \nabla_{A \dot A} \Psi_{BCDE} = (e_{A \dot A})^\mu \partial_\mu \Psi_{BCDE} &-  (\Gamma_{B}^{\;\;\; F})_{A \dot A} \Psi_{FCDE} -  (\Gamma_{C}^{\;\;\; F})_{A \dot A} \Psi_{BFDE} \\
    &-(\Gamma_{D}^{\;\;\; F})_{A \dot A} \Psi_{BCFE} -  (\Gamma_{C}^{\;\;\; E})_{A \dot A} \Psi_{BCDF}\, .
    \end{aligned}
\end{equation}

\subsection{Proof of Plebański's heavenly equations}\label{app_pleba}

The equation for SDG is
\begin{equation}\label{sdg_app}
    R_{\mu \nu \rho \sigma} = \frac{i}{2} \vep_{\mu \nu \alpha \beta} R^{\alpha \beta}_{\;\;\;\; \rho \sigma}.
\end{equation}
The symmetry of the Riemann tensor then implies
\begin{equation}
\begin{aligned}
    R^{\mu \nu}_{\;\;\;\; \rho \sigma} &= - \frac{1}{4} \vep^{\mu \nu \alpha \beta} \vep_{\rho \sigma \gamma \delta} R_{\alpha \beta}^{\;\;\;\;\; \gamma \delta}
\end{aligned}
\end{equation}
and by identity \eqref{epsilon_id} and the Bianchi identity, we see that \eqref{sdg_app} automatically implies the vacuum Einstein equation $R_{\nu \sigma} = 0$:
\begin{equation}
\begin{aligned}
    R^{\mu }_{\;\;\nu \mu \sigma} &= - \frac{3}{2} \delta^{[\gamma}_\sigma \delta^{\alpha}_{\nu} \delta^{\beta]}_{\delta} R_{\alpha \beta \gamma}^{\;\;\;\;\;\;\; \delta} \\
    &= 0.
\end{aligned}
\end{equation}

Using the decomposition of $R_{AB}$ \eqref{R_sigma}, as well as \eqref{sigma_dual}, the self-duality equation \eqref{sdg_app} then implies that the ASD part of the curvature 2-form vanishes:
\begin{equation}
    R_{AB} = 0.
\end{equation}
We will now argue that the above condition implies that it is always possible to choose a vierbein frame on the spacetime such that $\Gamma_{AB} = 0$.

We can always decompose $\Gamma_a^{\;\; b}$ into SD and ASD parts as
\begin{equation}
    \Gamma_a^{\;\; b} = (\Gamma_{\rm sd})_a^{\;\; b} + (\Gamma_{\rm asd})_a^{\;\; b}.
\end{equation}
We have previously shown (see \eqref{sd_asd_lorentz}) that one can decompose the Lorentz algebra $\mathfrak{so}(1,3)$ into two $\mathfrak{sl}(2, \mathbb{R})$ halves, one SD and one ASD. Crucially, these subalgebras commute with each other, as $[\ell_n, \Bell_m] = 0$. Therefore, any path ordered exponential of the spin connection will factor into two commuting matrices which are path ordered exponentials of $\Gamma_{\rm sd}$ and $\Gamma_{\rm asd}$ respectively.
\begin{equation}
    \mathcal{P} \exp(\int_{x_i}^{x_f} \Gamma )_{\!\! a}^{\;\, b} = \mathcal{P} \exp(\int_{x_i}^{x_f} \Gamma_{\rm sd} )_{\!\! a}^{\;\, c} \, \mathcal{P} \exp(\int_{x_i}^{x_f} \Gamma_{\rm asd} )_{\!\! c}^{\;\, b}
\end{equation}

We can then redefine our vierbein frames at each point $x$ with the ASD path ordered integral
\begin{equation}
    e_{A \dot A}(x) \mapsto e'_{A \dot A}(x) = \mathcal{P} \exp( \int_{x_0}^x \Gamma) _{\!\! A}^{\;\, B} e_{B \dot A}(x_0)
\end{equation}
where here we path-integrate from a chosen basepoint $x_0$ to $x$. The overall integral does not depend on the particular path chosen because $R_{AB} = 0$, and $R_{AB}$ necessarily quantifies the ASD $\mathfrak{sl}(2, \mathbb{R})$ holonomy around a tiny loop.

This vierbein redefinition sends
\begin{equation}
    \Gamma_A^{\;\; B} \mapsto (\Gamma_A^{\;\; B})' = 0.
\end{equation}
as desired. This also shows that the holonomy group of a SD metric is always the restricted holonomy group $SL_{\rm sd}(2, \mathbb{R}) \subset SO(1,3)$, in a complexified sense.

If we set $\Gamma_{AB} = 0$, equation \eqref{dsigma} implies that all three ASD 2-forms are closed when the metric is SD:
\begin{equation}\label{sigma_d_zero}
    \dd \Sigma^{AB} = 0.
\end{equation}
The SD 2-forms also satisfy
\begin{equation}\label{Sigma_sym}
    \Sigma^{(AB} \wedge \Sigma^{CD)} = 0
\end{equation}
which is just a consequence of their tetrad definition \eqref{sigma_def_both} and identity \eqref{symtheta}.

A surprising fact is that knowing of the triplet of ASD 2-forms alone is completely sufficient to reconstruct $g_{\mu \nu}$! This is true for all metrics, not just self-dual ones. This fact follows from Urbantke's remarkable identity \cite{urbantke1984integrability,capovilla1991self}, which using our conventions reads
\begin{equation}\label{Urbantke}
    \sqrt{-g} g_{\mu \nu} = -\frac{i}{3\times 2^6} \sqrt{-g} \vep^{\alpha \beta \gamma \delta} \Sigma^{AB}_{\mu \alpha} \Sigma_{\beta \gamma \; B}^{\;\;\;\;\;\;\;\; C} \Sigma_{\delta \nu C A}.
\end{equation}
Note that the RHS is determined entirely by the $\Sigma^{AB}
$'s (as $\sqrt{-g} \vep^{0123} = -1$ in our definition). The Urbantke formula can be proven by raising the indices of the middle 2-form using \eqref{sigma_dual}.

Let us now regard \eqref{sigma_d_zero} and \eqref{Sigma_sym} as our new starting points of SDG and show how the metric is determined from these equations.

Using the explicit formula
\begin{equation}\label{B64}
    \Sigma^{AB}_{\mu \nu} = 2 \vep_{\dot A \dot B} \theta^{A \dot A}_{[ \mu} \theta^{B \dot B}_{\nu ]}
\end{equation}
we can define the ``almost complex structure'' matrix
\begin{equation}
    J^\mu_{\;\; \nu} \equiv \frac{i}{2} \, g^{\mu \rho}(\Sigma^{12})_{\rho \nu}
\end{equation}
which can readily be shown to satisfy the properties
\begin{equation}\label{J2}
    J^\mu_{\; \; \rho} J^{\rho}_{\; \; \nu} = - \delta^\mu_\nu
\end{equation}
and
\begin{equation}\label{SigmaJ}
\begin{aligned}
    (\Sigma^{11})_{\rho \sigma} J^\rho_{\;\; \mu} J^\sigma_{\; \; \nu} &= - (\Sigma^{11})_{\mu \nu}, \\
    (\Sigma^{22})_{\rho \sigma} J^\rho_{\;\; \mu} J^\sigma_{\; \; \nu} &= - (\Sigma^{22})_{\mu \nu}.
\end{aligned}
\end{equation}
We also note that the 2-forms $\Sigma^{AB}_{\mu \nu}$ satisfy an important property. Adopting the notation $\Gamma_{AB \rho} dx^\rho = \Gamma_{AB}$,  $ \widetilde{\Gamma}_{\dot A \dot B \rho} dx^\rho = \widetilde{\Gamma}_{\dot A \dot B}$, equation \eqref{B17} becomes $\nabla_\rho \theta^{A \dot A}_\mu = - \Gamma^{\;\;\;A}_{C\;\;\; \rho} \theta^{C \dot A}_\mu - \widetilde{\Gamma}_{\dot C\;\;\; \rho }^{\;\;\; \dot A} \theta^{A \dot C}_\mu$, and from this and \eqref{B64} one can calculate
\begin{equation}\label{nablaSigma}
    \nabla_\rho \Sigma^{AB}_{\mu \nu} = - \Gamma^{\;\;\; A}_{C\;\;\; \rho} \Sigma^{CB}_{\mu \nu} - \Gamma^{\;\;\; B}_{C\;\;\; \rho} \Sigma^{AC}_{\mu \nu} = 0
\end{equation}
using $\Gamma_{AB} = 0$. We note the above equation is stronger than $\dd \Sigma^{AB} = 0$ which only gives $\partial_{[\rho} \Sigma^{AB}_{\mu \nu]} = \nabla_{[\rho} \Sigma^{AB}_{\mu \nu]} = 0$. Importantly for us, it implies $\nabla_\rho J^\mu_{\;\; \nu} = 0$.

The Newlander–Nirenberg theorem states that if the Nijenhuis tensor of $J$, defined as
\begin{equation}
    N^\alpha_{\;\; \mu \nu} \equiv  J^\rho_{\;\; \mu}\partial_{[\rho} J^{\alpha}_{\;\; \nu ]} - J^\rho_{\;\; \nu}\partial_{[\rho} J^{\alpha}_{\;\; \mu ]} = J^\rho_{\;\; \mu}\nabla_{[\rho} J^{\alpha}_{\;\; \nu ]} - J^\rho_{\;\; \nu}\nabla_{[\rho} J^{\alpha}_{\;\; \mu ]}
\end{equation}
is equal to 0, then $J$ is an actual complex structure. From $\nabla_\rho J^\mu_{\;\; \nu} = 0$ we can see that $N^\alpha_{\;\; \mu \nu} = 0$ for any self-dual spacetime.

Now that we know $J$ gives a complex structure, we can choose a coordinate basis $(y^1, y^2, \by^1, \by^2)$ that diagonalizes $J$ with eigenvalues $\pm i$ via
\begin{equation}\label{Jcoords}
    J \pdv{y^j} = i \pdv{y^j} \, , \hspace{1 cm}J \pdv{\by^\bj} = -i \pdv{\by^\bj} \, ,
\end{equation}
for $j = 1,2$ and $\bj = 1,2$. \eqref{SigmaJ} can then be understood as the equation
\begin{equation}
    \Sigma^{11}(J v_1, J v_2) = -\Sigma^{11}(v_1, v_2), \hspace{1 cm} \Sigma^{22}(J v_1, J v_2 ) = -\Sigma^{22}(v_1, v_2)
\end{equation}
where $v_1$ and $v_2$ are arbitrary vectors, which from \eqref{Jcoords} implies that the ``mixed'' holomorphic-antiholomorphic components of $(\Sigma^{11})_{\mu \nu}$ and $(\Sigma^{22})_{\mu\nu}$ are 0, with
\begin{equation}
    (\Sigma^{11})_{j \bj } = 0, \hspace{1 cm} (\Sigma^{22})_{j \bj} = 0.
\end{equation}
If we then use $\Sigma^{11} \wedge \Sigma^{11} = 0$, $\Sigma^{22} \wedge \Sigma^{22} = 0$, as well as $\dd \Sigma^{11} = 0$, $\dd \Sigma^{22} = 0$, we know that $\Sigma^{11}$ and $\Sigma^{22}$ must be able to be written in the form
\begin{equation}
    \Sigma^{11} = f(y^1, y^2) d y^1 \wedge d y^2, \hspace{1 cm} \Sigma^{22} = g(\by^1, \by^2) d \by^1 \wedge d \by^2,
\end{equation}
for some arbitrary functions $f(y^j)$ and $g(\by^{\bj})$. However, once again using $\dd \Sigma^{11} = 0$ and $\dd \Sigma^{22} = 0$, Darboux's theorem  tells us there exist coordinate redefinitions $y^j \mapsto \tilde{y}^j(y^1, y^2)$ and $\by^{\bj} \mapsto \tilde{\by}^{\bj}(\by^1, \by^2)$ such that, in these new coordinates,
\begin{equation}
    \Sigma^{11} = 8 \, d y^1 \wedge d y^2, \hspace{1 cm} \Sigma^{22} = 8\, d \by^1 \wedge d \by^2.
\end{equation}
If we now use $\Sigma^{11} \wedge \Sigma^{12} = \Sigma^{22} \wedge \Sigma^{12} = 0$, we know that $\Sigma^{12}$ must be of the form
\begin{align}
    \Sigma^{12} = 4\sum_{j = 1, 2} \sum_{\bj = 1, 2} \Omega_{j \bj} d y^j \wedge d\bar{y}^{\bj}
\end{align}
for some set of functions $\Omega_{j \bj}$. $\dd \Sigma^{12} = 0$ then implies that these functions must be the derivatives of a single scalar $\Omega$ called the Kähler potential,
\begin{equation}
    \Omega_{j\bj} = \pdv{y^j} \pdv{\bar{y}^\bj} \Omega \, .
\end{equation}
The equation $2 \, \Sigma^{12} \wedge \Sigma^{12} + \Sigma^{11} \wedge \Sigma^{22} = 0$ imposes the following constraint on the Kähler potential:
\begin{equation}\label{plebanskis_first}
    \det( \Omega_{j \bj} ) = 1.
\end{equation}
This is known as Plebański's first heavenly equation, and $\Omega$ is known as Plebański's first heavenly scalar.

Now that we have all of the $\Sigma^{AB}$'s, we could in principle use Urbantke's formula \eqref{Urbantke} to find the metric. We can also directly write a corresponding tetrad
\begin{equation}
\begin{aligned}
    \theta^{1 \dot 1} = 2 d \by^1 \, , & & & & & & & & &\theta^{2 \dot 1} = 2\Omega_{1\bar{2}} d y^1 + 2 \Omega_{2\bar{2}} d y^2 \, ,\\
    \theta^{1 \dot 2} = 2 d \by^2 \, , & & & & & & & & &\theta^{2 \dot 2} = -2\Omega_{1\bar{1}} d y^1 - 2\Omega_{2\bar{1}} d y^2 \, ,
\end{aligned}
\end{equation}
which yields the metric
\begin{equation}\label{Kähler}
\begin{aligned}
    ds^2 &= \frac{1}{2} \vep_{A B} \vep_{\dot A \dot B} \theta^{A \dot A} \theta^{B \dot B} = - 4 \, \Omega_{j \bj} \, dy^j d\by^{\bj}.
\end{aligned}
\end{equation}
One can show that $\det(g_{\mu \nu}) = 16$, and that our 2-forms $\Sigma^{AB}$ really are ASD with respect to this metric, after an appropriate choice of orientation. We conclude that \eqref{sigma_d_zero} and \eqref{Sigma_sym} alone suffice as defining equations for SDG.

From \eqref{Kähler}, we can clearly see that all 4d SD metrics are Kähler manifolds. What's more, they are Hyper-Kähler. In addition to $J$, one can also define two other complex structures $I$ and $K$, via
\begin{equation}
    I^\mu_{\;\; \nu} = \frac{1}{4} g^{\mu \rho} (\Sigma^{11} + \Sigma^{22} )_{\rho \nu} \, , \hspace{1 cm} K^\mu_{\;\; \nu} = \frac{i}{4} g^{\mu \rho} (\Sigma^{11} - \Sigma^{22})_{\rho \nu} \, ,
\end{equation}
which satisfy $I^2 = J^2 = K^2 = I J K = -1$.

Let us now give the proof of Plebański's second heavenly equation, repeating Plebański's original derivation \cite{plebanski1975some}. We use new coordinates $(u,\bu, w, \bw)$ satisfying
\begin{equation}
    u = - \pdv{\Omega}{\by^1}, \hspace{1 cm} w = \pdv{\Omega}{\by^2}, \hspace{1 cm} \bu = \by^1, \hspace{1 cm} \bw = \by^2 \, ,
\end{equation}
which can always be found because we have the Jacobian $|\partial(u, w)/\partial(y^1, y^2) | = 1$ from \eqref{plebanskis_first}. The tetrad now becomes
\begin{equation}
\begin{aligned}
    \theta^{1 \dot 1} = 2 d \bu \, ,& & & & & & & & &\theta^{2 \dot 1} = 2( d w - B d \bu - A d \bw )\, ,\\
    \theta^{1 \dot 2} =2 d \bw \, ,& & & & & & & & &\theta^{2 \dot 2} = 2( d u + C d \bu + B d \bw)\, ,
\end{aligned}
\end{equation}
with $A = \partial_{\by^2}^2 \Omega$, $B = \partial_{\by^1}\partial_{\by^2} \Omega$, $C = \partial_{\by^1}^2 \Omega$. With this tetrad,
\begin{equation}
\begin{aligned}
    \Sigma^{22} = 2 \theta^{2 \dot 2} \wedge \theta^{2 \dot 1} =&\, 8 (d u \wedge d w - A d u \wedge d \bw - C d \bu \wedge d w \\
    &- B(d u \wedge d \bu + d w \wedge d \bw) - (AC - B^2) d \bu \wedge d \bw )
\end{aligned}
\end{equation}
and $\dd \Sigma^{22} = 0$ implies
\begin{align}\label{pleb2proof1}
    \partial_w A = \partial_u B\, , &\hspace{1 cm} \partial_u C = \partial_w B \, ,  \\
    \partial_u( AC - B^2) - \partial_\bu A+\partial_\bw B = 0 \, , & \hspace{1 cm} \partial_w(AC - B^2) - \partial_{\bu} B + \partial_{\bw} C = 0\, . \label{pleb2proof2}
\end{align}
\eqref{pleb2proof1} implies that there exists a function $\tphi$ such that
\begin{equation}
    A = \partial_u^2 \tphi, \hspace{1 cm} B = \partial_u \partial_w \tphi, \hspace{1 cm} C = \partial_w^2 \tphi.
\end{equation}
Plugging the above equations into \eqref{pleb2proof2} gives
\begin{equation}
    \partial_u ( \Box \tphi - \{ \partial_u \tphi, \partial_w \tphi \} ) = \partial_w ( \Box \tphi - \{ \partial_u \tphi, \partial_w \tphi \} ) = 0.
\end{equation}
This implies $\Box \tphi - \{ \partial_u \tphi, \partial_w \tphi \} =\partial_\bu F(\bu, \bw)$ for some function $F(\bu, \bw)$. If we then define $\phi = \tphi - u F$, we have
\begin{equation}
    \Box \phi - \{ \partial_u \phi, \partial_w \phi \} = 0
\end{equation}
which is Plebański's second heavenly equation and $\phi$ is Plebański's second heavenly scalar. The tetrad is now
\begin{equation}
\begin{aligned}
    \theta^{1 \dot 1} &= 2 d \bu  \\
    \theta^{1 \dot 2} &= 2 d \bw  \\
    \theta^{2 \dot 1} &= 2( d w - (\partial_u \partial_w \phi ) d \bu - (\partial_u^2 \phi) d \bw ) \\
    \theta^{2 \dot 2} &= 2 (d u + ( \partial_w^2 \phi) d \bu + (\partial_u \partial_w \phi) d \bw )
\end{aligned}
\end{equation}
and the metric is
\begin{equation}
\begin{aligned}
    ds^2 &= \theta^{1 \dot 1} \theta^{2 \dot 2} - \theta^{1 \dot 2} \theta^{2 \dot 1} \\
    &= 4 (du \, d \bu - d w \, d \bw + (\partial_w^2 \phi) d \bu^2 + (\partial_u^2 \phi) d \bw^2 + 2(\partial_u \partial_w \phi) d \bu \, d \bw ).
\end{aligned}
\end{equation}
This completes the proof that all SD metrics can always locally be written in the above form for some $\phi$ satisfying the second heavenly equation.

\section{SDYM and the Parke-Taylor formula}\label{appD}

In this appendix we'll develop the theory of perturbiner expansions in SDYM in a way that mirrors our work in SDG. We'll show how it can be used to prove the Parke-Taylor formula for the Yang-Mills tree-level MHV amplitude \cite{Parke:1986gb}.

Perturbiners in SDYM were previously used to prove the Parke-Taylor formula by Selivanov and Rosly \cite{Rosly:1996vr,selivanov1997selfdual,Selivanov:1996gw}. See also \cite{Berends:1987me,Bardeen:1995gk,Cangemi:1996rx,Korepin:1996mm}. The only difference between our work below and that of Selivanov and Rosly is our use of the Chalmers-Siegel scalar, which we'll introduce below. In this appendix alone we redefine the polarization vectors to be
\begin{equation}
    \vep^\mu_{+,i} = - i \, \partial_z p^\mu_i,   \hspace{1 cm} \vep^\mu_{-,i} = \frac{i}{\omega^2} \partial_\bz p^\mu_i .
\end{equation}

\subsection{Perturbiners in SDYM}

Say we have a gauge field $A_\mu = A_\mu^a \bT^a$, where $\bT^a \in \mathfrak{g}$. The field strength is defined to be
\begin{equation}
    F_{\mu \nu} = \partial_\mu A_\nu - \partial_\nu A_\mu + [A_\mu, A_\nu].
\end{equation}
The equation for SDYM is
\begin{equation}
    F_{\mu \nu} = \frac{i}{2} \vep_{\mu \nu \rho \sigma} F^{\rho \sigma}.
\end{equation}
In lightcone coordinates \eqref{lc_coords}, $\vep_{u\bu w\bw} = 4i$ and the above equation reduces to
\begin{equation}\label{three_sdym_eq}
    F_{u w} = 0, \hspace{0.5 cm} F_{u \bu} = F_{w\bw}, \hspace{0.5 cm} F_{\bu \bw} = 0.
\end{equation}
If we set $A_u = 0$, then $F_{uw} = 0$ implies $A_w = 0$. $F_{u \bu} = F_{w \bw}$ then implies $\partial_u A_\bu = \partial_w A_\bw$ which can be solved by
\begin{equation}
    A_\bu = \partial_{w} \Phi, \hspace{1 cm} A_{\bw} =  \partial_u \Phi,
\end{equation}
where $\Phi$ is a $\mathfrak{g}$-valued scalar we shall call the Chalmers-Siegel scalar \cite{Parkes:1992rz, Bardeen:1995gk, Chalmers:1996rq}. We have therefore shown that the connection can be written as
\begin{equation}\label{A_Phi}
    A_\mu = \begin{pmatrix} A_u \\ A_\bu \\ A_w \\A_\bw \end{pmatrix} = \begin{pmatrix}
        0 \\ \partial_w \Phi \\ 0 \\ \partial_u \Phi
    \end{pmatrix}.
\end{equation}
Note this manifestly satisfies Lorenz gauge $\partial_\mu A^\mu = 0$.

The only remaining equation to use from \eqref{three_sdym_eq} is
\begin{equation}\label{Phi_eom}
\begin{aligned}
    F_{\bu \bw} = \Box\, \Phi - [\partial_u \Phi, \partial_w \Phi ] = 0 
\end{aligned}
\end{equation}
which is the equation of motion. Finding solutions perturbatively to the above equation is the goal of this section. We define the seed functions
\begin{equation}
    \Phi_i \equiv \epsilon_i \, \bT^{a_i} e^{i p_i \cdot X}
\end{equation}
which are positive helicity gluons with wavefunction $A_\mu = \epsilon_i( \vep_{+,i})_\mu\bT^{a_i} e^{i p_i \cdot X}$.

\begin{definition}
We denote
\begin{equation}
    \bigPhi ( \Phi_1, \ldots, \Phi_N ) \equiv \;\; \begin{matrix}\text{full perturbiner expansion of Chalmers-Siegel} \\ \text{scalar $\Phi$ with seed functions }\Phi_1, \ldots, \Phi_N. \end{matrix} 
\end{equation}
\end{definition}

\begin{definition}
    For some subset of indices $\{ i_1, \ldots, i_k\} \subset \{1, \ldots, N\}$ for $1 \leq k \leq N$, we define 
    \begin{equation*}
        \bigPhi^{(k)}( \Phi_{i_1}, \ldots, \Phi_{i_k})
    \end{equation*}
    to be the sum of terms in
    \begin{equation*}
        \bigPhi ( \Phi_1, \ldots, \Phi_N )
    \end{equation*}
    which contain the seed functions $\Phi_{i_1}, \ldots, \Phi_{i_k}$. This implies
    \begin{equation}
    \bigPhi\,(\Phi_1, \ldots, \Phi_N) = \sum_{k = 1}^N \sum_{ \substack{\{i_1, \ldots, i_k\}   \subset \{1, \ldots, N\} }  } \bigPhi^{(k)}( \Phi_{i_1}, \ldots, \Phi_{i_k}).
\end{equation}
\end{definition}

\begin{thm}\label{claimD1} The perturbiner expansion is given by
\begin{equation}
    \tcboxmath{ \bigPhi^{(N)}( \Phi_{1}, \ldots, \Phi_{N}) = \sum_{\sigma \, \in \, \mathrm{Sym}(N)} \frac{\Phi_{\sigma(1)} \Phi_{\sigma(2)} \ldots \Phi_{\sigma(N)}}{z_{\sigma(1) \sigma(2)} z_{\sigma(2) \sigma(3)} \ldots z_{\sigma(N-1) \sigma(N)}}. }
\end{equation}
\end{thm}

\begin{proof}

This proof will work by plugging the proposed perturbiner expansion into the equation of motion \eqref{Phi_eom} and showing that it is satisfied.

Let us begin by creating a visual representation of the terms in the perturbiner expansion. For a string of indices $s_1, \ldots, s_k$, we represent a single term in the expansion with a chain of nodes as follows.
\begin{center}
    \tikzset{every picture/.style={line width=0.75pt}} 

\begin{tikzpicture}[x=0.75pt,y=0.75pt,yscale=-1,xscale=1]

\draw  [fill={rgb, 255:red, 255; green, 255; blue, 255 }  ,fill opacity=1 ] (45.47,66) .. controls (45.47,59.1) and (51.06,53.5) .. (57.97,53.5) .. controls (64.87,53.5) and (70.47,59.1) .. (70.47,66) .. controls (70.47,72.9) and (64.87,78.5) .. (57.97,78.5) .. controls (51.06,78.5) and (45.47,72.9) .. (45.47,66) -- cycle ;
\draw  [fill={rgb, 255:red, 255; green, 255; blue, 255 }  ,fill opacity=1 ] (70.47,66) .. controls (70.47,59.1) and (76.06,53.5) .. (82.97,53.5) .. controls (89.87,53.5) and (95.47,59.1) .. (95.47,66) .. controls (95.47,72.9) and (89.87,78.5) .. (82.97,78.5) .. controls (76.06,78.5) and (70.47,72.9) .. (70.47,66) -- cycle ;
\draw  [fill={rgb, 255:red, 255; green, 255; blue, 255 }  ,fill opacity=1 ] (95.47,66) .. controls (95.47,59.1) and (101.06,53.5) .. (107.97,53.5) .. controls (114.87,53.5) and (120.47,59.1) .. (120.47,66) .. controls (120.47,72.9) and (114.87,78.5) .. (107.97,78.5) .. controls (101.06,78.5) and (95.47,72.9) .. (95.47,66) -- cycle ;
\draw  [fill={rgb, 255:red, 255; green, 255; blue, 255 }  ,fill opacity=1 ] (145.47,66) .. controls (145.47,59.1) and (151.06,53.5) .. (157.97,53.5) .. controls (164.87,53.5) and (170.47,59.1) .. (170.47,66) .. controls (170.47,72.9) and (164.87,78.5) .. (157.97,78.5) .. controls (151.06,78.5) and (145.47,72.9) .. (145.47,66) -- cycle ;
\draw  [fill={rgb, 255:red, 255; green, 255; blue, 255 }  ,fill opacity=1 ] (120.47,66) .. controls (120.47,59.1) and (126.06,53.5) .. (132.97,53.5) .. controls (139.87,53.5) and (145.47,59.1) .. (145.47,66) .. controls (145.47,72.9) and (139.87,78.5) .. (132.97,78.5) .. controls (126.06,78.5) and (120.47,72.9) .. (120.47,66) -- cycle ;

\draw (57.97,67) node  [font=\footnotesize]  {$s_{1}$};
\draw (82.97,67) node  [font=\footnotesize]  {$s_{2}$};
\draw (107.97,67) node  [font=\footnotesize]  {$s_{3}$};
\draw (157.97,67) node  [font=\footnotesize]  {$s_{k}$};
\draw (172.47,66) node [anchor=west] [inner sep=0.75pt]    {$=\dfrac{1}{z_{s_{1} s_{2}}}\dfrac{1}{z_{s_{2} s_{3}}} \dotsc \dfrac{1}{z_{s_{k-1} s_{k}}} \Phi _{s_{1}} \Phi _{s_{2}} \Phi _{s_{3}} \dotsc \Phi _{s_{k-1}} \Phi _{s_{k}}$};
\draw (132.97,66) node  [font=\normalsize]  {$\mycdots$};

\end{tikzpicture}
\end{center}

Let us now consider what happens when we act with the wave operator $\Box$ on the above term. Note that $\Box \Phi_i = 0$ and $\Box ( \Phi_i \Phi_j) = z_{ij} D^{ij} \Phi_i \Phi_j$, see \eqref{boxphiiphij}. From the product rule of differentiation, we have for instance
\begin{equation}
    \Box( \Phi_1 \Phi_2 \Phi_3) = ( z_{12} D^{12} + z_{23} D^{23} + z_{13} D^{13} ) ( \Phi_1 \Phi_2 \Phi_3).
\end{equation}
We can therefore think of the action of $\Box$ on a string of seed functions as a sum over all pairs of seed functions $\Phi_i$ and $\Phi_j$ in the string where we act by $z_{ij} D^{ij}$. We can represent this pairing with a double-line as follows.
\begin{center}
    \tikzset{every picture/.style={line width=0.75pt}} 

\begin{tikzpicture}[x=0.75pt,y=0.75pt,yscale=-1,xscale=1]

\draw    (141.47,58.08) .. controls (141.43,57.35) and (141.41,56.63) .. (141.41,55.93) .. controls (141.41,43.88) and (146.88,35.73) .. (154.22,31.47) .. controls (158.2,29.16) and (162.74,28) .. (167.3,28) .. controls (167.53,28) and (167.75,28) .. (167.98,28.01) .. controls (176.78,28.23) and (185.59,32.76) .. (190.52,41.73) .. controls (192.91,46.08) and (194.39,51.49) .. (194.47,57.98)(144.46,57.92) .. controls (144.43,57.25) and (144.41,56.58) .. (144.41,55.93) .. controls (144.41,45.22) and (149.16,37.87) .. (155.73,34.06) .. controls (159.24,32.02) and (163.26,31) .. (167.3,31) .. controls (167.5,31) and (167.7,31) .. (167.9,31.01) .. controls (175.71,31.2) and (183.52,35.22) .. (187.89,43.17) .. controls (190.07,47.15) and (191.4,52.1) .. (191.47,58.02) ;
\draw  [fill={rgb, 255:red, 255; green, 255; blue, 255 }  ,fill opacity=1 ] (80.47,58) .. controls (80.47,51.1) and (86.06,45.5) .. (92.97,45.5) .. controls (99.87,45.5) and (105.47,51.1) .. (105.47,58) .. controls (105.47,64.9) and (99.87,70.5) .. (92.97,70.5) .. controls (86.06,70.5) and (80.47,64.9) .. (80.47,58) -- cycle ;
\draw  [fill={rgb, 255:red, 255; green, 255; blue, 255 }  ,fill opacity=1 ] (130.47,58) .. controls (130.47,51.1) and (136.06,45.5) .. (142.97,45.5) .. controls (149.87,45.5) and (155.47,51.1) .. (155.47,58) .. controls (155.47,64.9) and (149.87,70.5) .. (142.97,70.5) .. controls (136.06,70.5) and (130.47,64.9) .. (130.47,58) -- cycle ;
\draw  [fill={rgb, 255:red, 255; green, 255; blue, 255 }  ,fill opacity=1 ] (230.47,58) .. controls (230.47,51.1) and (236.06,45.5) .. (242.97,45.5) .. controls (249.87,45.5) and (255.47,51.1) .. (255.47,58) .. controls (255.47,64.9) and (249.87,70.5) .. (242.97,70.5) .. controls (236.06,70.5) and (230.47,64.9) .. (230.47,58) -- cycle ;
\draw  [fill={rgb, 255:red, 255; green, 255; blue, 255 }  ,fill opacity=1 ] (180.47,58) .. controls (180.47,51.1) and (186.06,45.5) .. (192.97,45.5) .. controls (199.87,45.5) and (205.47,51.1) .. (205.47,58) .. controls (205.47,64.9) and (199.87,70.5) .. (192.97,70.5) .. controls (186.06,70.5) and (180.47,64.9) .. (180.47,58) -- cycle ;
\draw  [fill={rgb, 255:red, 255; green, 255; blue, 255 }  ,fill opacity=1 ] (205.47,58) .. controls (205.47,51.1) and (211.06,45.5) .. (217.97,45.5) .. controls (224.87,45.5) and (230.47,51.1) .. (230.47,58) .. controls (230.47,64.9) and (224.87,70.5) .. (217.97,70.5) .. controls (211.06,70.5) and (205.47,64.9) .. (205.47,58) -- cycle ;
\draw  [fill={rgb, 255:red, 255; green, 255; blue, 255 }  ,fill opacity=1 ] (155.47,58) .. controls (155.47,51.1) and (161.06,45.5) .. (167.97,45.5) .. controls (174.87,45.5) and (180.47,51.1) .. (180.47,58) .. controls (180.47,64.9) and (174.87,70.5) .. (167.97,70.5) .. controls (161.06,70.5) and (155.47,64.9) .. (155.47,58) -- cycle ;
\draw  [fill={rgb, 255:red, 255; green, 255; blue, 255 }  ,fill opacity=1 ] (105.47,58) .. controls (105.47,51.1) and (111.06,45.5) .. (117.97,45.5) .. controls (124.87,45.5) and (130.47,51.1) .. (130.47,58) .. controls (130.47,64.9) and (124.87,70.5) .. (117.97,70.5) .. controls (111.06,70.5) and (105.47,64.9) .. (105.47,58) -- cycle ;

\draw (92.97,59) node  [font=\footnotesize]  {$s_{1}$};
\draw (242.97,59) node  [font=\footnotesize]  {$s_{k}$};
\draw (142.97,59) node  [font=\footnotesize]  {$s_{i}$};
\draw (192.97,59) node  [font=\footnotesize]  {$s_{j}$};
\draw (257.47,58) node [anchor=west] [inner sep=0.75pt]    {$=\dfrac{1}{z_{s_{1} s_{2}}} \dotsc \dfrac{1}{z_{s_{k-1} s_{k}}}\left( z_{s_{i} s_{j}} D^{s_{i} s_{j}}\right) \Phi _{s_{1}} \dotsc \Phi _{s_{i}} \dotsc \Phi _{s_{j}} \dotsc \Phi _{s_{k}}$};
\draw (217.97,58) node    {$\mycdots $};
\draw (167.97,58) node    {$\mycdots $};
\draw (117.97,58) node    {$\mycdots $};

\end{tikzpicture}
\end{center}

Let us now define a new graphical element which is just a line (or `cut') drawn in between two seed functions in a chain. This line indicates that the chain of nodes before and after the cut should be evaluated independently and then multiplied together at the end (preserving order). 
\begin{center}
    \tikzset{every picture/.style={line width=0.75pt}} 

\begin{tikzpicture}[x=0.75pt,y=0.75pt,yscale=-1,xscale=1]

\draw  [fill={rgb, 255:red, 255; green, 255; blue, 255 }  ,fill opacity=1 ] (22.97,49) .. controls (22.97,42.1) and (28.56,36.5) .. (35.47,36.5) .. controls (42.37,36.5) and (47.97,42.1) .. (47.97,49) .. controls (47.97,55.9) and (42.37,61.5) .. (35.47,61.5) .. controls (28.56,61.5) and (22.97,55.9) .. (22.97,49) -- cycle ;
\draw  [fill={rgb, 255:red, 255; green, 255; blue, 255 }  ,fill opacity=1 ] (47.97,49) .. controls (47.97,42.1) and (53.56,36.5) .. (60.47,36.5) .. controls (67.37,36.5) and (72.97,42.1) .. (72.97,49) .. controls (72.97,55.9) and (67.37,61.5) .. (60.47,61.5) .. controls (53.56,61.5) and (47.97,55.9) .. (47.97,49) -- cycle ;
\draw  [fill={rgb, 255:red, 255; green, 255; blue, 255 }  ,fill opacity=1 ] (72.97,49) .. controls (72.97,42.1) and (78.56,36.5) .. (85.47,36.5) .. controls (92.37,36.5) and (97.97,42.1) .. (97.97,49) .. controls (97.97,55.9) and (92.37,61.5) .. (85.47,61.5) .. controls (78.56,61.5) and (72.97,55.9) .. (72.97,49) -- cycle ;
\draw  [fill={rgb, 255:red, 255; green, 255; blue, 255 }  ,fill opacity=1 ] (97.97,49) .. controls (97.97,42.1) and (103.56,36.5) .. (110.47,36.5) .. controls (117.37,36.5) and (122.97,42.1) .. (122.97,49) .. controls (122.97,55.9) and (117.37,61.5) .. (110.47,61.5) .. controls (103.56,61.5) and (97.97,55.9) .. (97.97,49) -- cycle ;
\draw  [fill={rgb, 255:red, 255; green, 255; blue, 255 }  ,fill opacity=1 ] (122.97,49) .. controls (122.97,42.1) and (128.56,36.5) .. (135.47,36.5) .. controls (142.37,36.5) and (147.97,42.1) .. (147.97,49) .. controls (147.97,55.9) and (142.37,61.5) .. (135.47,61.5) .. controls (128.56,61.5) and (122.97,55.9) .. (122.97,49) -- cycle ;
\draw  [fill={rgb, 255:red, 255; green, 255; blue, 255 }  ,fill opacity=1 ] (147.97,49) .. controls (147.97,42.1) and (153.56,36.5) .. (160.47,36.5) .. controls (167.37,36.5) and (172.97,42.1) .. (172.97,49) .. controls (172.97,55.9) and (167.37,61.5) .. (160.47,61.5) .. controls (153.56,61.5) and (147.97,55.9) .. (147.97,49) -- cycle ;
\draw [line width=0.75]    (97.97,26) -- (97.97,72) ;
\draw  [fill={rgb, 255:red, 255; green, 255; blue, 255 }  ,fill opacity=1 ] (248.97,179) .. controls (248.97,172.1) and (254.56,166.5) .. (261.47,166.5) .. controls (268.37,166.5) and (273.97,172.1) .. (273.97,179) .. controls (273.97,185.9) and (268.37,191.5) .. (261.47,191.5) .. controls (254.56,191.5) and (248.97,185.9) .. (248.97,179) -- cycle ;
\draw  [fill={rgb, 255:red, 255; green, 255; blue, 255 }  ,fill opacity=1 ] (273.97,179) .. controls (273.97,172.1) and (279.56,166.5) .. (286.47,166.5) .. controls (293.37,166.5) and (298.97,172.1) .. (298.97,179) .. controls (298.97,185.9) and (293.37,191.5) .. (286.47,191.5) .. controls (279.56,191.5) and (273.97,185.9) .. (273.97,179) -- cycle ;
\draw  [fill={rgb, 255:red, 255; green, 255; blue, 255 }  ,fill opacity=1 ] (298.97,179) .. controls (298.97,172.1) and (304.56,166.5) .. (311.47,166.5) .. controls (318.37,166.5) and (323.97,172.1) .. (323.97,179) .. controls (323.97,185.9) and (318.37,191.5) .. (311.47,191.5) .. controls (304.56,191.5) and (298.97,185.9) .. (298.97,179) -- cycle ;
\draw  [fill={rgb, 255:red, 255; green, 255; blue, 255 }  ,fill opacity=1 ] (323.97,179) .. controls (323.97,172.1) and (329.56,166.5) .. (336.47,166.5) .. controls (343.37,166.5) and (348.97,172.1) .. (348.97,179) .. controls (348.97,185.9) and (343.37,191.5) .. (336.47,191.5) .. controls (329.56,191.5) and (323.97,185.9) .. (323.97,179) -- cycle ;
\draw  [fill={rgb, 255:red, 255; green, 255; blue, 255 }  ,fill opacity=1 ] (348.97,179) .. controls (348.97,172.1) and (354.56,166.5) .. (361.47,166.5) .. controls (368.37,166.5) and (373.97,172.1) .. (373.97,179) .. controls (373.97,185.9) and (368.37,191.5) .. (361.47,191.5) .. controls (354.56,191.5) and (348.97,185.9) .. (348.97,179) -- cycle ;
\draw  [fill={rgb, 255:red, 255; green, 255; blue, 255 }  ,fill opacity=1 ] (373.97,179) .. controls (373.97,172.1) and (379.56,166.5) .. (386.47,166.5) .. controls (393.37,166.5) and (398.97,172.1) .. (398.97,179) .. controls (398.97,185.9) and (393.37,191.5) .. (386.47,191.5) .. controls (379.56,191.5) and (373.97,185.9) .. (373.97,179) -- cycle ;

\draw (174.97,49) node [anchor=west] [inner sep=0.75pt]    {$=\left(\dfrac{\Phi _{s_{1}} \dotsc \Phi _{s_{i}}}{z_{s_{1} s_{2}} \dotsc z_{s_{i-1} s_{i}}}\right)\left(\dfrac{\Phi _{s_{i+1}} \dotsc \Phi _{s_{k}}}{z_{s_{i+1} s_{i+2}} \dotsc z_{s_{n-1} ,s_{k}}} \ \right)$};
\draw (35.47,50) node  [font=\footnotesize]  {$s_{1}$};
\draw (85.47,50) node  [font=\footnotesize]  {$s_{i}$};
\draw (110.47,50) node  [font=\footnotesize]  {$s_{i+1}$};
\draw (160.47,50) node  [font=\footnotesize]  {$s_{k}$};
\draw (174.97,115.5) node [anchor=west] [inner sep=0.75pt]    {$=\dfrac{z_{s_{i} s_{i+1}}}{z_{s_{1} s_{2}} \dotsc z_{s_{k-1} s_{k}}} \Phi _{s_{1}} \dotsc \Phi _{s_{k}}$};
\draw (135.47,49) node    {$\mycdots $};
\draw (60.47,49) node    {$\mycdots $};
\draw (174.97,182) node [anchor=west] [inner sep=0.75pt]    {$=z_{s_{i} s_{i+1}} \times $};
\draw (261.47,180) node  [font=\footnotesize]  {$s_{1}$};
\draw (311.47,180) node  [font=\footnotesize]  {$s_{i}$};
\draw (336.47,180) node  [font=\footnotesize]  {$s_{i+1}$};
\draw (386.47,180) node  [font=\footnotesize]  {$s_{k}$};
\draw (361.47,179) node    {$\mycdots $};
\draw (286.47,179) node    {$\mycdots $};

\end{tikzpicture}
\end{center}
As we can see above, a cut drawn in between nodes $s_i$ and $s_{i+1}$ simply corresponds to multiplying the entire chain of nodes by $z_{s_i s_{i+1}}$.

We will define one further graphical element, which is a wiggly line. A wiggly line extending from node $s_i$ to node $s_j$ corresponds to the action of $D^{s_i s_j}$ on the string of seed functions. An arrow is drawn to denote orientation because $D^{s_i s_j} =- D^{s_j s_i }$.
\begin{center}
    \tikzset{every picture/.style={line width=0.75pt}} 

\begin{tikzpicture}[x=0.75pt,y=0.75pt,yscale=-1,xscale=1]

\draw    (107.47,65) .. controls (105.43,63.13) and (105.13,61.31) .. (106.58,59.53) .. controls (108.11,57.96) and (108.02,56.31) .. (106.3,54.56) .. controls (104.67,53.15) and (104.76,51.66) .. (106.57,50.08) .. controls (108.54,48.51) and (108.9,46.82) .. (107.67,45.01) .. controls (106.77,42.68) and (107.46,41.11) .. (109.75,40.28) .. controls (111.98,39.93) and (112.94,38.66) .. (112.63,36.45) .. controls (112.68,34.12) and (113.98,33.05) .. (116.53,33.25) .. controls (118.6,34) and (120.1,33.27) .. (121.03,31.06) .. controls (122.08,29) and (123.73,28.59) .. (125.96,29.84) .. controls (127.63,31.34) and (129.19,31.24) .. (130.66,29.53) .. controls (132.73,27.98) and (134.48,28.15) .. (135.89,30.03) .. controls (137.06,31.98) and (138.62,32.38) .. (140.57,31.21) .. controls (142.71,30.23) and (144.34,30.92) .. (145.45,33.27) .. controls (145.96,35.48) and (147.31,36.32) .. (149.5,35.8) .. controls (151.83,35.53) and (153.12,36.68) .. (153.37,39.25) .. controls (153.06,41.46) and (154,42.7) .. (156.21,42.97) .. controls (158.61,43.93) and (159.34,45.49) .. (158.41,47.65) .. controls (157.2,49.55) and (157.54,51.14) .. (159.43,52.41) .. controls (161.17,54.13) and (161.15,55.87) .. (159.37,57.62) .. controls (157.5,58.75) and (157.14,60.43) .. (158.29,62.66) -- (157.47,65) ;
\draw [shift={(137.31,30.31)}, rotate = 187.14] [fill={rgb, 255:red, 0; green, 0; blue, 0 }  ][line width=0.08]  [draw opacity=0] (10.72,-5.15) -- (0,0) -- (10.72,5.15) -- (7.12,0) -- cycle    ;
\draw  [fill={rgb, 255:red, 255; green, 255; blue, 255 }  ,fill opacity=1 ] (44.97,65) .. controls (44.97,58.1) and (50.56,52.5) .. (57.47,52.5) .. controls (64.37,52.5) and (69.97,58.1) .. (69.97,65) .. controls (69.97,71.9) and (64.37,77.5) .. (57.47,77.5) .. controls (50.56,77.5) and (44.97,71.9) .. (44.97,65) -- cycle ;
\draw  [fill={rgb, 255:red, 255; green, 255; blue, 255 }  ,fill opacity=1 ] (69.97,65) .. controls (69.97,58.1) and (75.56,52.5) .. (82.47,52.5) .. controls (89.37,52.5) and (94.97,58.1) .. (94.97,65) .. controls (94.97,71.9) and (89.37,77.5) .. (82.47,77.5) .. controls (75.56,77.5) and (69.97,71.9) .. (69.97,65) -- cycle ;
\draw  [fill={rgb, 255:red, 255; green, 255; blue, 255 }  ,fill opacity=1 ] (94.97,65) .. controls (94.97,58.1) and (100.56,52.5) .. (107.47,52.5) .. controls (114.37,52.5) and (119.97,58.1) .. (119.97,65) .. controls (119.97,71.9) and (114.37,77.5) .. (107.47,77.5) .. controls (100.56,77.5) and (94.97,71.9) .. (94.97,65) -- cycle ;
\draw  [fill={rgb, 255:red, 255; green, 255; blue, 255 }  ,fill opacity=1 ] (119.97,65) .. controls (119.97,58.1) and (125.56,52.5) .. (132.47,52.5) .. controls (139.37,52.5) and (144.97,58.1) .. (144.97,65) .. controls (144.97,71.9) and (139.37,77.5) .. (132.47,77.5) .. controls (125.56,77.5) and (119.97,71.9) .. (119.97,65) -- cycle ;
\draw  [fill={rgb, 255:red, 255; green, 255; blue, 255 }  ,fill opacity=1 ] (144.97,65) .. controls (144.97,58.1) and (150.56,52.5) .. (157.47,52.5) .. controls (164.37,52.5) and (169.97,58.1) .. (169.97,65) .. controls (169.97,71.9) and (164.37,77.5) .. (157.47,77.5) .. controls (150.56,77.5) and (144.97,71.9) .. (144.97,65) -- cycle ;
\draw  [fill={rgb, 255:red, 255; green, 255; blue, 255 }  ,fill opacity=1 ] (169.97,65) .. controls (169.97,58.1) and (175.56,52.5) .. (182.47,52.5) .. controls (189.37,52.5) and (194.97,58.1) .. (194.97,65) .. controls (194.97,71.9) and (189.37,77.5) .. (182.47,77.5) .. controls (175.56,77.5) and (169.97,71.9) .. (169.97,65) -- cycle ;
\draw  [fill={rgb, 255:red, 255; green, 255; blue, 255 }  ,fill opacity=1 ] (194.97,65) .. controls (194.97,58.1) and (200.56,52.5) .. (207.47,52.5) .. controls (214.37,52.5) and (219.97,58.1) .. (219.97,65) .. controls (219.97,71.9) and (214.37,77.5) .. (207.47,77.5) .. controls (200.56,77.5) and (194.97,71.9) .. (194.97,65) -- cycle ;

\draw (57.47,66) node  [font=\footnotesize]  {$s_{1}$};
\draw (107.47,66) node  [font=\footnotesize]  {$s_{i}$};
\draw (157.47,66) node  [font=\footnotesize]  {$s_{j}$};
\draw (207.47,66) node  [font=\footnotesize]  {$s_{k}$};
\draw (221.97,65) node [anchor=west] [inner sep=0.75pt]    {$=\dfrac{1}{z_{s_{1} s_{2}}} \dotsc \dfrac{1}{z_{s_{k-1} s_{k}}}\left( D^{s_{i} s_{j}}\right) \Phi _{s_{1}} \dotsc \Phi _{s_{i}} \dotsc \Phi _{s_{j}} \dotsc \Phi _{s_{k}}$};
\draw (82.47,65) node    {$\mycdots $};
\draw (132.47,65) node    {$\mycdots $};
\draw (182.47,65) node    {$\mycdots $};

\end{tikzpicture}
\end{center}

Now that we have defined all of these graphical elements, we will begin checking that
\begin{equation*}
    \Box \Phi \overset{?}{=} [ \partial_u \Phi, \partial_w \Phi]
\end{equation*}
for our proposed $\Phi$. To evaluate the commutator on  the RHS, via the product rule of differentiation we are to sum over the action of $\partial_u$ and $\partial_w$ on all possible nodes on either side of a ``cut'' as follows. 
\begin{center}
    \input{figures/gauge_dots_5}
\end{center}
In other words, $[\partial_u \Phi, \partial_w \Phi]$ corresponds to a sum over all strings of nodes where a cut can be placed anywhere within the string and a wiggly line can be placed anywhere that starts on the left of the cut and ends on the right of the cut. Note that this diagram corresponds to the term
\begin{center}
    \tikzset{every picture/.style={line width=0.75pt}} 

\begin{tikzpicture}[x=0.75pt,y=0.75pt,yscale=-1,xscale=1]

\draw    (127.47,85) .. controls (125.48,83.24) and (125.22,81.29) .. (126.68,79.15) .. controls (128.29,77.88) and (128.28,76.35) .. (126.65,74.56) .. controls (125.16,72.59) and (125.38,70.87) .. (127.32,69.39) .. controls (129.33,68.22) and (129.79,66.61) .. (128.71,64.58) .. controls (127.83,62.37) and (128.52,60.88) .. (130.77,60.13) .. controls (133,59.66) and (133.9,58.3) .. (133.45,56.06) .. controls (133.23,53.69) and (134.31,52.45) .. (136.69,52.34) .. controls (139.02,52.44) and (140.27,51.32) .. (140.46,48.99) .. controls (140.87,46.6) and (142.28,45.61) .. (144.69,46.01) .. controls (147.02,46.55) and (148.3,45.81) .. (148.53,43.79) .. controls (149.44,41.48) and (151.09,40.7) .. (153.48,41.47) .. controls (155.79,42.35) and (157.25,41.8) .. (157.85,39.81) .. controls (159.18,37.63) and (160.71,37.16) .. (162.42,38.41) .. controls (164.69,39.54) and (166.58,39.09) .. (168.11,37.05) .. controls (169.1,35.16) and (170.73,34.88) .. (172.99,36.19) .. controls (174.52,37.66) and (176.18,37.46) .. (177.97,35.59) .. controls (179.19,33.81) and (180.88,33.69) .. (183.03,35.23) .. controls (184.42,36.85) and (185.78,36.82) .. (187.1,35.13) .. controls (189.18,33.48) and (190.88,33.51) .. (192.2,35.22) .. controls (194.11,37) and (195.81,37.11) .. (197.28,35.56) .. controls (199.51,34.13) and (201.19,34.32) .. (202.32,36.15) .. controls (204.03,38.11) and (205.69,38.39) .. (207.29,36.99) .. controls (209.62,35.78) and (211.24,36.15) .. (212.16,38.08) .. controls (213.6,40.2) and (215.18,40.65) .. (216.89,39.42) .. controls (219.28,38.46) and (220.81,38.99) .. (221.46,41) .. controls (222.57,43.25) and (224.31,43.99) .. (226.68,43.23) .. controls (228.55,42.29) and (229.92,43) .. (230.79,45.36) .. controls (231.47,47.7) and (233,48.66) .. (235.38,48.24) .. controls (237.33,47.6) and (238.5,48.49) .. (238.88,50.91) .. controls (239.04,53.28) and (240.3,54.45) .. (242.65,54.44) .. controls (245.06,54.63) and (246.15,55.93) .. (245.91,58.32) .. controls (245.45,60.6) and (246.35,62.01) .. (248.6,62.56) .. controls (250.87,63.37) and (251.56,64.9) .. (250.68,67.15) .. controls (249.6,69.22) and (250.07,70.87) .. (252.08,72.1) .. controls (254.03,73.62) and (254.23,75.08) .. (252.7,76.49) .. controls (251.11,78.34) and (251.13,80.2) .. (252.75,82.07) -- (252.47,85) ;
\draw [shift={(194.23,35.33)}, rotate = 181.89] [fill={rgb, 255:red, 0; green, 0; blue, 0 }  ][line width=0.08]  [draw opacity=0] (10.72,-5.15) -- (0,0) -- (10.72,5.15) -- (7.12,0) -- cycle    ;
\draw  [fill={rgb, 255:red, 255; green, 255; blue, 255 }  ,fill opacity=1 ] (64.97,85) .. controls (64.97,78.1) and (70.56,72.5) .. (77.47,72.5) .. controls (84.37,72.5) and (89.97,78.1) .. (89.97,85) .. controls (89.97,91.9) and (84.37,97.5) .. (77.47,97.5) .. controls (70.56,97.5) and (64.97,91.9) .. (64.97,85) -- cycle ;
\draw  [fill={rgb, 255:red, 255; green, 255; blue, 255 }  ,fill opacity=1 ] (89.97,85) .. controls (89.97,78.1) and (95.56,72.5) .. (102.47,72.5) .. controls (109.37,72.5) and (114.97,78.1) .. (114.97,85) .. controls (114.97,91.9) and (109.37,97.5) .. (102.47,97.5) .. controls (95.56,97.5) and (89.97,91.9) .. (89.97,85) -- cycle ;
\draw  [fill={rgb, 255:red, 255; green, 255; blue, 255 }  ,fill opacity=1 ] (114.97,85) .. controls (114.97,78.1) and (120.56,72.5) .. (127.47,72.5) .. controls (134.37,72.5) and (139.97,78.1) .. (139.97,85) .. controls (139.97,91.9) and (134.37,97.5) .. (127.47,97.5) .. controls (120.56,97.5) and (114.97,91.9) .. (114.97,85) -- cycle ;
\draw  [fill={rgb, 255:red, 255; green, 255; blue, 255 }  ,fill opacity=1 ] (139.97,85) .. controls (139.97,78.1) and (145.56,72.5) .. (152.47,72.5) .. controls (159.37,72.5) and (164.97,78.1) .. (164.97,85) .. controls (164.97,91.9) and (159.37,97.5) .. (152.47,97.5) .. controls (145.56,97.5) and (139.97,91.9) .. (139.97,85) -- cycle ;
\draw  [fill={rgb, 255:red, 255; green, 255; blue, 255 }  ,fill opacity=1 ] (164.97,85) .. controls (164.97,78.1) and (170.56,72.5) .. (177.47,72.5) .. controls (184.37,72.5) and (189.97,78.1) .. (189.97,85) .. controls (189.97,91.9) and (184.37,97.5) .. (177.47,97.5) .. controls (170.56,97.5) and (164.97,91.9) .. (164.97,85) -- cycle ;
\draw  [fill={rgb, 255:red, 255; green, 255; blue, 255 }  ,fill opacity=1 ] (189.97,85) .. controls (189.97,78.1) and (195.56,72.5) .. (202.47,72.5) .. controls (209.37,72.5) and (214.97,78.1) .. (214.97,85) .. controls (214.97,91.9) and (209.37,97.5) .. (202.47,97.5) .. controls (195.56,97.5) and (189.97,91.9) .. (189.97,85) -- cycle ;
\draw  [fill={rgb, 255:red, 255; green, 255; blue, 255 }  ,fill opacity=1 ] (214.97,85) .. controls (214.97,78.1) and (220.56,72.5) .. (227.47,72.5) .. controls (234.37,72.5) and (239.97,78.1) .. (239.97,85) .. controls (239.97,91.9) and (234.37,97.5) .. (227.47,97.5) .. controls (220.56,97.5) and (214.97,91.9) .. (214.97,85) -- cycle ;
\draw  [fill={rgb, 255:red, 255; green, 255; blue, 255 }  ,fill opacity=1 ] (239.97,85) .. controls (239.97,78.1) and (245.56,72.5) .. (252.47,72.5) .. controls (259.37,72.5) and (264.97,78.1) .. (264.97,85) .. controls (264.97,91.9) and (259.37,97.5) .. (252.47,97.5) .. controls (245.56,97.5) and (239.97,91.9) .. (239.97,85) -- cycle ;
\draw  [fill={rgb, 255:red, 255; green, 255; blue, 255 }  ,fill opacity=1 ] (264.97,85) .. controls (264.97,78.1) and (270.56,72.5) .. (277.47,72.5) .. controls (284.37,72.5) and (289.97,78.1) .. (289.97,85) .. controls (289.97,91.9) and (284.37,97.5) .. (277.47,97.5) .. controls (270.56,97.5) and (264.97,91.9) .. (264.97,85) -- cycle ;
\draw  [fill={rgb, 255:red, 255; green, 255; blue, 255 }  ,fill opacity=1 ] (289.97,85) .. controls (289.97,78.1) and (295.56,72.5) .. (302.47,72.5) .. controls (309.37,72.5) and (314.97,78.1) .. (314.97,85) .. controls (314.97,91.9) and (309.37,97.5) .. (302.47,97.5) .. controls (295.56,97.5) and (289.97,91.9) .. (289.97,85) -- cycle ;
\draw [line width=0.75]    (189.97,62) -- (189.97,108) ;

\draw (77.47,86) node  [font=\footnotesize]  {$s_{1}$};
\draw (127.47,86) node  [font=\footnotesize]  {$s_{i}$};
\draw (177.47,86) node  [font=\footnotesize]  {$s_{\ell }$};
\draw (202.47,86) node  [font=\footnotesize]  {$s_{\ell +1}$};
\draw (252.47,86) node  [font=\footnotesize]  {$s_{j}$};
\draw (302.47,86) node  [font=\footnotesize]  {$s_{k}$};
\draw (316.97,85) node [anchor=west] [inner sep=0.75pt]    {$=\dfrac{1}{z_{s_{1} s_{2}}} \dotsc \dfrac{1}{z_{s_{n-1} s_{n}}}( z_{s_{\ell } s_{\ell +1}})\left( D^{s_{i} s_{j}}\right) \Phi _{s_{1}} \dotsc \Phi _{s_{k}}.$};
\draw (102.47,85) node    {$\mycdots $};
\draw (152.47,85) node    {$\mycdots $};
\draw (227.47,85) node    {$\mycdots $};
\draw (277.47,85) node    {$\mycdots $};

\end{tikzpicture}
\end{center}

To evaluate $\Box \Phi$ on the LHS, we simply sum over all placements of the double-lined edge that starts and ends on two distinct nodes.

The equality $\Box \Phi = [\partial_u \Phi, \partial_w \Phi]$ is then confirmed by the computation in the figure below, done with an example where the double-line/wiggly-lines start and end on $s_i$ and $s_{i+4}$, respectively. We sum over all placements of the cut on the RHS in between the two end-points of the wiggly line. Because the cut between nodes $s_{i+\ell}$ and $s_{i+\ell+1}$ multiplies the whole expression by $z_{s_{i+\ell} s_{i+\ell+1}}$, when we sum over the chain of cut placements in the example below we end up with the overall factor $z_{s_i s_{i+4}}$, matching the LHS. This computation completes the proof that the perturbiner expansion satisfies the equation of motion.

\begin{center}
    \input{figures/gauge_dots_7}
\end{center}

\end{proof}

\begin{exs}
    For $N=1$:
    \begin{equation}
        \bigPhi(\Phi_1) = \bigPhi^{(1)}(\Phi_1) = \Phi_1.
    \end{equation}
    For $N=2$:
        \begin{equation}
        \begin{aligned}
        \bigPhi(\Phi_1, \Phi_2) &= \bigPhi^{(1)}(\Phi_1)  + \bigPhi^{(1)}(\Phi_2) + \bigPhi^{(2)}(\Phi_1, \Phi_2)
        \end{aligned}
        \end{equation}
    with
    \begin{equation}
        \bigPhi^{(2)}(\Phi_i, \Phi_j) = \frac{1}{z_{ij}}\Phi_i \Phi_j + \frac{1}{z_{ji}}\Phi_j \Phi_i.
    \end{equation}
    For $N = 3$:
    \begin{equation}
    \begin{aligned}
        \bigPhi(\Phi_1, \Phi_2, \Phi_3) = & \, \bigPhi^{(1)}(\Phi_1) + \bigPhi^{(1)}(\Phi_2) +\bigPhi^{(1)}(\Phi_3) \\
        & + \bigPhi^{(2)} (\Phi_1, \Phi_2) + \bigPhi^{(2)} (\Phi_2, \Phi_3)  + \bigPhi^{(2)} (\Phi_1, \Phi_3)  \\
        & + \bigPhi^{(3)}(\Phi_1, \Phi_2, \Phi_3)
    \end{aligned}
    \end{equation}
    with
    \begin{equation}\label{phi3}
    \begin{aligned}
        \bigPhi^{(3)}(\Phi_i, \Phi_j, \Phi_k) = &\, \frac{1}{z_{ij}} \frac{1}{z_{jk}} \Phi_i \Phi_j \Phi_k + \frac{1}{z_{jk}} \frac{1}{z_{ki}} \Phi_j \Phi_k \Phi_i +\frac{1}{z_{ki}} \frac{1}{z_{ij}} \Phi_k \Phi_i \Phi_j \\
        & + \frac{1}{z_{ji}} \frac{1}{z_{ik}} \Phi_j \Phi_i \Phi_k + \frac{1}{z_{kj}} \frac{1}{z_{ji}} \Phi_k \Phi_j \Phi_i +\frac{1}{z_{ik}} \frac{1}{z_{kj}} \Phi_i \Phi_k \Phi_j.
    \end{aligned}
    \end{equation}

\end{exs}

Now that we have a formula for the perturbiner expansion in SDYM for any $N$, we will now present a recursive formula which allows for the perturbiner expansion of $N$ seed functions to be computed if one already knows the expansion for $N-1$. This formula is given below.

\begin{thm}
\begin{equation}
    \tcboxmath{ \bigPhi^{(N)}(\Phi_1, \ldots, \Phi_N) = \sum_{i=1}^{N-1} \eval{\bigPhi^{(N-1)}(\Phi_1, \ldots, \Phi_{N-1} )}_{\displaystyle \Phi_i \mapsto \frac{1}{z_{iN}} [\Phi_i, \Phi_N]} }
\end{equation}
\end{thm}

\begin{exs}
    For $N=2$ we have
    \begin{equation}
        \bigPhi^{(2)}(\Phi_1, \Phi_2) = \frac{1}{z_{12}} [\Phi_1, \Phi_2].
    \end{equation}
    For $N = 3$
    \begin{equation}
        \bigPhi^{(3)}(\Phi_1, \Phi_2, \Phi_3) = \frac{1}{z_{12}} [ \frac{1}{z_{13}}[\Phi_1, \Phi_3], \Phi_2] + \frac{1}{z_{12}}[\Phi_1, \frac{1}{z_{23}} [\Phi_2, \Phi_3]].
    \end{equation}
    If one expands out the above term, one finds agreement with \eqref{phi3}.
\end{exs}

\begin{proof}

In order to prove this recursive formula, let us use theorem \ref{claimD1} to write the $N-1$ case and show that the recursive formula correctly reproduces the $N$ case. 

Let us consider how a certain string of seed functions in the $N$ case can be reached by a previous string in the $N-1$ case via the recursive formula. Without loss of generality, let us consider how we can achieve the mutation of the string
\begin{equation*}
    1, \ldots, i, j, \ldots, N-1 \;\;\; \to \;\;\; 1, \ldots, i , N, j, \ldots, N-1
\end{equation*}
where $N$ is inserted in between some adjacent nodes $i$ and $j$.

If we start with the $1, \ldots, i, j, \ldots, N-1$ string
\begin{equation*}
    \frac{1}{z_{12}} \ldots \frac{1}{z_{ij}} \ldots \frac{1}{z_{N-2,N-1}} \Phi_1 \ldots \Phi_i \Phi_j \ldots \Phi_{N-1}
\end{equation*}
then only way to reach the desired string from the recursive rule is to add the two terms from the substitutions
\begin{equation*}
    \Phi_i \mapsto \frac{1}{z_{iN}}( \underbrace{\Phi_i \Phi_N}_{\substack{\text{only include} \\ \text{this term}}} \!\!\!\! - \;\Phi_N \Phi_i \;\; ), \hspace{0.75 cm} + \hspace{0.75 cm} \Phi_j \mapsto \frac{1}{z_{jN}}( \;\;  \Phi_j \Phi_N - \!\!\! \underbrace{\Phi_N \Phi_j}_{\substack{\text{only include} \\ \text{this term}}} )
\end{equation*}
under which the string transforms as
\begin{equation}
\begin{aligned}
    \frac{1}{z_{12}}& \ldots \frac{1}{z_{ij}} \ldots \frac{1}{z_{N-2,N-1}} \Phi_1 \ldots \Phi_i \Phi_j \ldots \Phi_{N-1} \\
    \mapsto&\, 
    \frac{1}{z_{12}} \ldots \frac{1}{z_{ij}} \ldots \frac{1}{z_{N-2,N-1}} \Phi_1 \ldots \Phi_i \left( \frac{1}{z_{iN}} \Phi_N - \frac{1}{z_{jN}} \Phi_N \right) \Phi_j \ldots \Phi_{N-1} \\
    &= \frac{1}{z_{12}} \ldots \frac{1}{z_{iN}} \frac{1}{z_{Nj}} \ldots \frac{1}{z_{N-2,N-1}} \Phi_1 \ldots \Phi_i \Phi_N \Phi_j \ldots \Phi_{N-1}
\end{aligned}
\end{equation}
giving the desired result. Here we have used the elementary identity
\begin{equation}
    \frac{1}{z_{ij}} \left( \frac{1}{z_{iN}} - \frac{1}{z_{jN}} \right) = \frac{1}{z_{iN}} \frac{1}{z_{Nj}}.
\end{equation}
    
\end{proof}

\subsection{Recursion operator in SDYM}

Linearized perturbations $\Phi + \delta \Phi$ to the Chalmers-Siegel scalar satisfy
\begin{equation}\label{sdym_linear_eom}
    \Box \delta \Phi - [\partial_u \delta \Phi, \partial_w \Phi] - [\partial_u \Phi, \partial_w \delta \Phi] = 0.
\end{equation}
There exists a recursion operator $\mR$ which acts on the space of linear perturbations $\delta \Phi$ and is defined as the solution to the pair of differential equations
\begin{equation}\label{sdym_R_def}
    \begin{aligned}
    \partial_u ( \mR \delta \Phi) &= \partial_\bw \delta \Phi + [ \partial_u \Phi, \delta \Phi ], \\
    \partial_w ( \mR \delta \Phi) &= \partial_\bu \delta \Phi + [ \partial_w \Phi, \delta \Phi ].
    \end{aligned}
\end{equation}
The compatibility of the two equations follows from the fact that $\delta \Phi$ is a linear perturbation of $\Phi$:
\begin{equation}
    \partial_u ( \partial_w ( \mR \delta \Phi) ) - \partial_w ( \partial_u ( \mR \delta \Phi) ) = \Box \delta \Phi - [ \partial_u \delta \Phi, \partial_w \Phi ] - [ \partial_u \Phi, \partial_w \delta \Phi ] = 0.
\end{equation}
Furthermore, $\mR \delta \Phi$ is guaranteed to solve the linearized e.o.m.\ by 
\begin{equation}
    \Box \, \mR \delta \Phi  - [ \partial_u \mR \delta \Phi, \partial_w \Phi] -[ \partial_u \Phi, \partial_w \mR \delta \Phi ]= [\Box \Phi - [\partial_u \Phi, \partial_w \Phi], \delta \Phi] = 0.
\end{equation}

A natural example of a linear perturbation $\delta \Phi$ is the change incurred by adding one extra seed function into the perturbiner expansion. To that end, we define
\begin{definition}
\begin{equation}
    \bigPhi ( \Phi_1, \ldots, \Phi_N | \Phi_I ) \equiv \bigPhi ( \Phi_1, \ldots, \Phi_N , \Phi_I ) - \bigPhi ( \Phi_1, \ldots, \Phi_N )
\end{equation}
\end{definition}
\noindent which just is the change in the perturbiner expansion coming from adding in the seed function $\Phi_I$, for some index $I \notin \{ 1, \ldots, N \}$.

In particular,
\begin{equation}\label{sdym_lin_pert_ex}
    \Phi = \bigPhi ( \Phi_1, \ldots, \Phi_N), \hspace{0.5 cm} \delta \Phi =  \bigPhi ( \Phi_1, \ldots, \Phi_N | \Phi_I)
\end{equation}
will solve the linearized e.o.m.\ \eqref{sdym_linear_eom}.

It turns out that there is a nice formula for the action of $\mR$ on the linear perturbation above, which we give below.

\begin{thm}\label{claim_R_sdym}
    \begin{equation}
        \tcboxmath{ \mR \, \bigPhi ( \Phi_1, \ldots, \Phi_N | \Phi_I) = - z_I \times \bigPhi ( \Phi_1, \ldots, \Phi_N | \Phi_I). }
    \end{equation}
\end{thm}

\begin{proof} First, we know what $\delta \Phi$ is because it is just a sum over all color-ordered strings which necessarily contain the node $I$. Therefore, we can prove the above formula by checking that \eqref{sdym_R_def} is satisfied if we assume $\mR \delta \Phi = - z_I \delta \Phi$. We will check the first equation of \eqref{sdym_R_def} and the second will follow by similar logic.

The term $\partial_u ( - z_I \delta \Phi)$ will be a sum of terms which each have $\partial_u$ acting on a seed function (which we'll denote $i$) and also contain the seed function $I$.

First, assume $i$ is to the left of $I$ in the color ordering.

On the RHS of \eqref{sdym_R_def} we have $\partial_\bw \delta \Phi + [\partial_u \Phi, \delta \Phi]$. For the $\partial_\bw \delta \Phi$ piece, note that
\begin{equation}\label{z_diff_sdym}
    \partial_{\bw} \Phi_i = - z_i \partial_u \Phi_i\, , \hspace{1 cm} \partial_\bu \Phi_i = - z_i \partial_w \Phi_i \, ,
\end{equation}
(see \eqref{z_diff_relation}). In the commutator $[\partial_u \Phi, \delta \Phi]$, the term $\partial_u \Phi \delta \Phi$ will have $\partial_u$ acting on a node (let's say $i$) to the left of $I$. Furthermore, diagrammatically there will be a `cut' or `line' in between the $i$ and $I$ nodes.

Below we show of a group terms in a particular example where all nodes are presented in the same order and all the derivatives are acting on $i$. Note that we have to sum over all of the positions of the `cut' in between $i$ and $I$.
\begin{center}
    \input{figures/gauge_dots_8}
\end{center}
It turns out that, using \eqref{z_diff_sdym} and summing over the chain of cuts, the diagram on the left does indeed equal the sum of diagrams on the right.

Next, assume that $i$ is to the right of $I$. We now look at the other term $- \delta \Phi \partial_u \Phi$ in the commutator $[\partial_u \Phi, \delta \Phi]$. An analogous computation is shown below.
\begin{center}
    \input{figures/gauge_dots_9}
\end{center}

Finally, if $i = I$, then the equality of diagrams follows in a straightforward manner. Therefore, the LHS of \eqref{sdym_R_def} does equal the RHS if $\mR \delta \Phi = - z_I \delta \Phi$.

\end{proof}

\subsection{Proof of the Parke Taylor formula}

We now present a proof of the Parke Taylor formula for gluon MHV scattering which is analogous to our proof of the graviton MHV formula. We start with the Yang Mills action with an added topological term. The topological term does not affect scattering amplitudes.
\begin{equation}
    S = - \frac{1}{2} \int d^4 X \Tr ( F^{\mu \nu} F_{\mu \nu} - \frac{i}{2} \vep_{\mu \nu \rho \sigma} F^{\mu \nu} F^{\rho \sigma} )
\end{equation}
We now write the basis of flat-space vierbeins $e_{A \dot A}$
\begin{equation}
    e_{1 \dot 1} = \partial_\bu, \hspace{0.5 cm} e_{1 \dot 2} = \partial_\bw, \hspace{0.5 cm} e_{2 \dot 1} = \partial_w, \hspace{0.5 cm} e_{2 \dot 2} = \partial_u.
\end{equation}
Changing to spinor indices by $A_{A \dot A} =e^\mu_{A \dot A} A_\mu$ and $F_{A \dot A B \dot B} = e_{A \dot A}^\mu e_{B \dot B}^\nu F_{\mu \nu}$, and decomposing the field strength into ASD and SD parts $F_{AB}$ and $\widetilde{F}_{\dot A \dot B}$ by \eqref{F_decomp}, the usual kinetic term and topological term then read
\begin{equation}
\begin{aligned}
    F^{\mu \nu}F_{\mu \nu} &=  2 F^{AB} F_{AB} + 2 \widetilde{F}^{\dot A \dot B} \widetilde{F}_{\dot A \dot B} \\
    \vep_{\mu \nu \rho \sigma} F^{\mu \nu} F^{\rho \sigma} &= 4 i \, ( F^{AB} F_{AB} - \widetilde{F}^{\dot A \dot B} \widetilde{F}_{\dot A \dot B})
\end{aligned}
\end{equation}
implying our action can be written purely in terms of the ASD part $F_{AB}$ as
\begin{equation}
    S = - 2 \int d^4 X \Tr( F^{AB} F_{AB} ).
\end{equation}

In a SD background, $F_{AB} = 0$. The above action therefore provides a convenient way to compute the MHV amplitude. Imagine a solution which contains $N$ positive helicity gluons and $2$ negative helicity gluons. The ``linear'' part of the solution should be
\begin{equation}
    A_\mu = \sum_{i = 1}^{N+2} \epsilon_i (\vep_i)_\mu \bT^{a_i} e^{i p_i \cdot X} + \ldots
\end{equation}
and the higher order terms in the infinitesimal parameters $\epsilon_i$ (the $\ldots$ above) are determined by requiring that $A_\mu$ solves the e.o.m.. From
\begin{equation}
    \mathcal{A}_{\rm MHV} =  \frac{\partial^{N+2}}{\partial \epsilon_1 \ldots \partial \epsilon_{N+2}}  \eval{i S}_{\epsilon_i = 0 \; \forall \epsilon_i}
\end{equation}
we can see that to find $\mathcal{A}_{\rm MHV}$ all we must do is compute $f^I_{AB}$ and $f^J_{AB}$, where $I = N+1$ and $J = N+2$ are the labels corresponding to the two negative helicity gluons and $f_{AB} = \delta F_{AB}$ denote linear ASD perturbations to the SD background of the $N$ positive helicity gluons. We then have
\begin{equation}
    \mathcal{A}_{\rm MHV} = - 4 i  \frac{\partial^{N+2}}{\partial \epsilon_1 \ldots \partial \epsilon_{N+2}} \eval{ \int d^4 X (f^I)^{AB} (f^J)_{AB} }_{\epsilon_i = 0 \; \forall \epsilon_i}.
\end{equation}

We now solve for the ASD perturbation $f^I_{AB}$. This can be done as follows. Defining the covariant derivative
\begin{equation}
    \mD_\mu = \partial_\mu + [A_\mu, \cdot \, ]
\end{equation}
the Yang Mills e.o.m.\ and Bianchi identity read
\begin{equation}
    \mD_\mu F^{\mu \nu} = 0 \, ,\hspace{1 cm} \mD_\mu (\star F)^{\mu \nu} = 0.
\end{equation}
These can be added and subtracted to give
\begin{equation}
    \mD^A_{\;\; \dot A} F_{AB} = 0 \, , \hspace{1 cm}\mD_A^{\;\; \dot A} \widetilde{F}_{\dot A \dot B} = 0.
\end{equation}
The first of the above equations gives us an equation of motion for $f_{AB}$, which we write explicitly as 
\begin{equation}
    \mD^A_{\;\; \dot A} f_{AB} = (e^A_{\;\; \dot B})^\mu \partial_\mu f_{AB } + [A^{A}_{\;\; \dot A}, f_{AB}] = 0.
\end{equation}
Note that here the connection $A_{A \dot A}$ denotes the connection of the SD background, given by \eqref{A_Phi}. The above equation can be thought of as two equations, one for $\dot A = 1$ and $\dot A = 2$, which we expand out below as
\begin{equation}\label{f_R}
\begin{aligned}
    \dot A = \dot 1: & \;\;\;\; 0 = - \partial_w f_{1 B} + \partial_\bu f_{2 B} + [\partial_w \Phi, f_{2 B}] \, ,\\
    \dot A = \dot 2: & \;\;\;\; 0 = - \partial_u f_{1 B} + \partial_\bw f_{2 B} + [\partial_u \Phi, f_{2 B}] \, .
\end{aligned}
\end{equation}
Comparing the above with \eqref{sdym_R_def}, we note that this implies that $f_{1 B}$ is given by the recursion operator $\mR$ acting on $f_{2B}$:
\begin{equation}
    f_{1B} = \mR \, f_{2B}.
\end{equation}
In fact, the two equations of \eqref{f_R} can be combined to show that
\begin{equation}
    \Box f_{AB} - [ \partial_u f_{AB}, \partial_w \Phi] - [\partial_u \Phi, \partial_w f_{AB}] = 0.
\end{equation}
Therefore, an ASD perturbation to a SD background solves the same e.o.m.\ as a SD perturbation! Thankfully we have already discussed how to solve for such perturbations using \eqref{sdym_lin_pert_ex}.

Using theorem \ref{claim_R_sdym} we know we can take $f_{1B}^I = - z_I f_{2 B}^I$ and from the symmetry $f_{AB} = f_{(AB)}$ we know we can solve for $f^I_{AB}$ with the formula 
\begin{equation}
    f^I_{AB} = \frac{1}{2} \kappa^I_A \kappa^I_B \bigPhi(\Phi_1, \ldots, \Phi_N | \Phi_I)
\end{equation}
where $\kappa^I_A = ( - z_I, 1)$. The proportionality factor of $1/2$ was chosen so that the negative helicity gluon has wavefunction $A_\mu = \epsilon_I (\vep^{-})_\mu \bT^{a_I} e^{i p_I \cdot X}$, as can be checked straightforwardly in the trivial $N=0$ case where the SD background is 0.

Plugging the above expression for $f_I$ and $f_J$ into the action, we have
\begin{equation}\label{S_sdym_IJ}
\begin{aligned}
    S = -  \langle I J \rangle^2 \int d^4 X \; \Tr \left( \bigPhi(\Phi_1, \ldots, \Phi_N | \Phi_I) \bigPhi(\Phi_1, \ldots, \Phi_N | \Phi_J)  \right).
\end{aligned}
\end{equation}
We can draw the terms arising from this trace in the following way.
\begin{center}
    \tikzset{every picture/.style={line width=0.75pt}} 

\begin{tikzpicture}[x=0.75pt,y=0.75pt,yscale=-1,xscale=1]

\draw  [fill={rgb, 255:red, 255; green, 255; blue, 255 }  ,fill opacity=1 ] (103.97,142) .. controls (103.97,135.1) and (109.56,129.5) .. (116.47,129.5) .. controls (123.37,129.5) and (128.97,135.1) .. (128.97,142) .. controls (128.97,148.9) and (123.37,154.5) .. (116.47,154.5) .. controls (109.56,154.5) and (103.97,148.9) .. (103.97,142) -- cycle ;
\draw  [fill={rgb, 255:red, 255; green, 255; blue, 255 }  ,fill opacity=1 ] (128.97,142) .. controls (128.97,135.1) and (134.56,129.5) .. (141.47,129.5) .. controls (148.37,129.5) and (153.97,135.1) .. (153.97,142) .. controls (153.97,148.9) and (148.37,154.5) .. (141.47,154.5) .. controls (134.56,154.5) and (128.97,148.9) .. (128.97,142) -- cycle ;
\draw  [fill={rgb, 255:red, 255; green, 255; blue, 255 }  ,fill opacity=1 ] (153.97,142) .. controls (153.97,135.1) and (159.56,129.5) .. (166.47,129.5) .. controls (173.37,129.5) and (178.97,135.1) .. (178.97,142) .. controls (178.97,148.9) and (173.37,154.5) .. (166.47,154.5) .. controls (159.56,154.5) and (153.97,148.9) .. (153.97,142) -- cycle ;
\draw  [fill={rgb, 255:red, 255; green, 255; blue, 255 }  ,fill opacity=1 ] (178.97,142) .. controls (178.97,135.1) and (184.56,129.5) .. (191.47,129.5) .. controls (198.37,129.5) and (203.97,135.1) .. (203.97,142) .. controls (203.97,148.9) and (198.37,154.5) .. (191.47,154.5) .. controls (184.56,154.5) and (178.97,148.9) .. (178.97,142) -- cycle ;
\draw  [fill={rgb, 255:red, 255; green, 255; blue, 255 }  ,fill opacity=1 ] (203.97,142) .. controls (203.97,135.1) and (209.56,129.5) .. (216.47,129.5) .. controls (223.37,129.5) and (228.97,135.1) .. (228.97,142) .. controls (228.97,148.9) and (223.37,154.5) .. (216.47,154.5) .. controls (209.56,154.5) and (203.97,148.9) .. (203.97,142) -- cycle ;
\draw  [fill={rgb, 255:red, 255; green, 255; blue, 255 }  ,fill opacity=1 ] (228.97,142) .. controls (228.97,135.1) and (234.56,129.5) .. (241.47,129.5) .. controls (248.37,129.5) and (253.97,135.1) .. (253.97,142) .. controls (253.97,148.9) and (248.37,154.5) .. (241.47,154.5) .. controls (234.56,154.5) and (228.97,148.9) .. (228.97,142) -- cycle ;
\draw  [fill={rgb, 255:red, 255; green, 255; blue, 255 }  ,fill opacity=1 ] (253.97,142) .. controls (253.97,135.1) and (259.56,129.5) .. (266.47,129.5) .. controls (273.37,129.5) and (278.97,135.1) .. (278.97,142) .. controls (278.97,148.9) and (273.37,154.5) .. (266.47,154.5) .. controls (259.56,154.5) and (253.97,148.9) .. (253.97,142) -- cycle ;
\draw  [fill={rgb, 255:red, 255; green, 255; blue, 255 }  ,fill opacity=1 ] (278.97,142) .. controls (278.97,135.1) and (284.56,129.5) .. (291.47,129.5) .. controls (298.37,129.5) and (303.97,135.1) .. (303.97,142) .. controls (303.97,148.9) and (298.37,154.5) .. (291.47,154.5) .. controls (284.56,154.5) and (278.97,148.9) .. (278.97,142) -- cycle ;
\draw [line width=0.75]    (178.97,119) -- (178.97,165) ;
\draw  [fill={rgb, 255:red, 255; green, 255; blue, 255 }  ,fill opacity=1 ] (378.5,107) .. controls (378.5,100.1) and (384.1,94.5) .. (391,94.5) .. controls (397.9,94.5) and (403.5,100.1) .. (403.5,107) .. controls (403.5,113.9) and (397.9,119.5) .. (391,119.5) .. controls (384.1,119.5) and (378.5,113.9) .. (378.5,107) -- cycle ;
\draw  [fill={rgb, 255:red, 255; green, 255; blue, 255 }  ,fill opacity=1 ] (354.81,116.81) .. controls (354.81,109.91) and (360.41,104.31) .. (367.31,104.31) .. controls (374.22,104.31) and (379.81,109.91) .. (379.81,116.81) .. controls (379.81,123.72) and (374.22,129.31) .. (367.31,129.31) .. controls (360.41,129.31) and (354.81,123.72) .. (354.81,116.81) -- cycle ;
\draw  [fill={rgb, 255:red, 255; green, 255; blue, 255 }  ,fill opacity=1 ] (402.19,116.81) .. controls (402.19,109.91) and (407.78,104.31) .. (414.69,104.31) .. controls (421.59,104.31) and (427.19,109.91) .. (427.19,116.81) .. controls (427.19,123.72) and (421.59,129.31) .. (414.69,129.31) .. controls (407.78,129.31) and (402.19,123.72) .. (402.19,116.81) -- cycle ;
\draw  [fill={rgb, 255:red, 255; green, 255; blue, 255 }  ,fill opacity=1 ] (412,140.5) .. controls (412,133.6) and (417.6,128) .. (424.5,128) .. controls (431.4,128) and (437,133.6) .. (437,140.5) .. controls (437,147.4) and (431.4,153) .. (424.5,153) .. controls (417.6,153) and (412,147.4) .. (412,140.5) -- cycle ;
\draw  [fill={rgb, 255:red, 255; green, 255; blue, 255 }  ,fill opacity=1 ] (402.19,164.19) .. controls (402.19,157.28) and (407.78,151.69) .. (414.69,151.69) .. controls (421.59,151.69) and (427.19,157.28) .. (427.19,164.19) .. controls (427.19,171.09) and (421.59,176.69) .. (414.69,176.69) .. controls (407.78,176.69) and (402.19,171.09) .. (402.19,164.19) -- cycle ;
\draw  [fill={rgb, 255:red, 255; green, 255; blue, 255 }  ,fill opacity=1 ] (378.5,174) .. controls (378.5,167.1) and (384.1,161.5) .. (391,161.5) .. controls (397.9,161.5) and (403.5,167.1) .. (403.5,174) .. controls (403.5,180.9) and (397.9,186.5) .. (391,186.5) .. controls (384.1,186.5) and (378.5,180.9) .. (378.5,174) -- cycle ;
\draw  [fill={rgb, 255:red, 255; green, 255; blue, 255 }  ,fill opacity=1 ] (354.81,164.19) .. controls (354.81,157.28) and (360.41,151.69) .. (367.31,151.69) .. controls (374.22,151.69) and (379.81,157.28) .. (379.81,164.19) .. controls (379.81,171.09) and (374.22,176.69) .. (367.31,176.69) .. controls (360.41,176.69) and (354.81,171.09) .. (354.81,164.19) -- cycle ;
\draw  [fill={rgb, 255:red, 255; green, 255; blue, 255 }  ,fill opacity=1 ] (345,140.5) .. controls (345,133.6) and (350.6,128) .. (357.5,128) .. controls (364.4,128) and (370,133.6) .. (370,140.5) .. controls (370,147.4) and (364.4,153) .. (357.5,153) .. controls (350.6,153) and (345,147.4) .. (345,140.5) -- cycle ;
\draw [line width=0.75]    (369.2,191.6) -- (385.8,152.1) ;
\draw [line width=0.75]    (386,128.4) -- (370,90.4) ;

\draw (75,142.39) node [anchor=west] [inner sep=0.75pt]    {$\mathrm{Tr}\Bigl( \ \ \ \ \ \ \ \ \ \ \ \ \ \ \ \ \ \ \ \ \ \ \ \ \ \ \ \ \ \ \ \ \ \ \ \ \ \ \  \Bigr)$};
\draw (191.47,142) node    {$3$};
\draw (216.47,142) node    {$4$};
\draw (116.47,142) node    {$1$};
\draw (166.47,142) node    {$2$};
\draw (141.47,142) node    {$I$};
\draw (241.47,142) node    {$J$};
\draw (266.47,142) node    {$5$};
\draw (291.47,142) node    {$6$};
\draw (309.97,142) node [anchor=west] [inner sep=0.75pt]    {$\ \ =$};
\draw (357.5,140.5) node    {$I$};
\draw (367.31,164.19) node    {$1$};
\draw (367.31,116.81) node    {$2$};
\draw (391,107) node    {$3$};
\draw (414.69,116.81) node    {$4$};
\draw (424.5,140.5) node    {$J$};
\draw (414.69,164.19) node    {$5$};
\draw (391,174) node    {$6$};

\end{tikzpicture}
\end{center}
Here, the `cut' on the LHS is drawn between the two terms in \eqref{S_sdym_IJ} which correspond to $f^I$ and $f^J$. On the RHS, a second cut is drawn to where the beginning and end of the chain are linked together by the trace.

We must sum over the positions of the two cuts, which are each placed between the $I$ and $J$ nodes on opposite sides. If we, say, hold one cut fixed and sum over the position of the other cut, the result is that we can just multiply the whole diagram without said cut by $z_{IJ}$, using the graphical computation shown below.
\begin{center}
    \input{figures/gauge_dots_11}
\end{center}
There are, however, two cuts whose positions must be summed over. Summing over the position of the first cut, as we've shown above, gets us a factor of $z_{IJ}$, and summing over the second cut would similarly get us a factor of $-z_{IJ}$. Therefore, summing over the positions of both cuts gets us a factor of $-(z_{IJ})^2$. Plugging this back into \eqref{S_sdym_IJ}, this will reproduce the overall power of $\langle I J \rangle^4$ out front. The integration over $\int d^4 X \, e^{i (p_1 + \ldots + p_{N+2}) \cdot X}$ then becomes the energy momentum conserving delta function.

In total, the final amplitude is
\begin{equation}
    \mathcal{A}_{\rm MHV} = i \, (2 \pi)^4 \delta^{(4)} \left( \sum_{i = 1}^{N+2}p_i \right)  \langle I J \rangle^4 \sum_{\sigma \in \mathrm{Sym}(N+2) / \mathbb{Z}_{N+2}} \frac{\Tr (\bT^{a_{\sigma(1)}} \ldots \bT^{a_{\sigma(N+2)}}) }{\langle \sigma(1) \sigma(2) \rangle \ldots\langle \sigma(N+2) \sigma(1) \rangle }
\end{equation}
which is the Parke Taylor formula.

\bibliography{sd_bib.bib}
\bibliographystyle{jhep}

\end{document}